\documentclass[letterpaper,12pt]{article}
\usepackage{xr}
\usepackage{rotating}
\usepackage{lscape}
\usepackage{rotating}
\usepackage{graphicx}
\usepackage{pgf}
\usepackage{subfig}
\usepackage[utf8]{inputenc}
\usepackage{geometry,setspace}
\usepackage{siunitx,booktabs,caption,widetable,threeparttable}
\usepackage{tikz,fullpage}
\usetikzlibrary{arrows,	petri, topaths}

\usepackage[authoryear]{natbib}
\usepackage{bibentry}
\usepackage[colorlinks = true, allcolors = blue]{hyperref}
\usepackage{listings}
\usepackage[reftex]{theoremref}
\usepackage{amsmath, amsthm, amssymb, bm, mathtools, amsfonts, subdepth}
\mathtoolsset{showonlyrefs}

\theoremstyle{remark}
\newtheorem{rem}{Remark}

\theoremstyle{definition}
\newtheorem{example}{Example}

\newtheorem{innercustomthm}{Example}
\newenvironment{customthm}[1]
{\renewcommand\theinnercustomthm{#1}\innercustomthm}
{\endinnercustomthm}

\theoremstyle{plain}
\newtheorem{thm}{Theorem}
\newtheorem{assumption}{Assumption}
\newtheorem{lem}{Lemma}

\newtheorem{innercustomass}{Assumption}
\newenvironment{customass}[1]
{\renewcommand\theinnercustomass{#1}\innercustomass}
{\endinnercustomass}

\newtheorem{theorem}{Theorem}[section]
\newtheorem{lemma}{Lemma}[section]

\newtheorem{res}{Result}[section]

\usepackage{fixme}

\usepackage{graphicx,float}
\usepackage{color}
\usepackage{pdfpages}
\newcommand{\diag}{\operatorname{diag}}
\newcommand{\E}{\operatorname{E}}

\newcommand{\ba}{\begin{eqnarray}}
\newcommand{\be}{\begin{equation}}
\newcommand{\ban}{\begin{eqnarray*}}
\newcommand{\ea}{\end{eqnarray}}
\newcommand{\ean}{\end{eqnarray*}}
\newcommand{\ee}{\end{equation}}

\def\1{\mathop{1_{}}\nolimits}

\def\diag{\mathop{{\rm diag}}\nolimits}

\def\E{\mathop{{\rm E}_{}}\nolimits}

\def\th{\hat\theta}

\renewcommand{\mid}{\,|\,}

\renewcommand{\Xi}{X_{i}}
\usepackage{algorithm}
\usepackage{algpseudocode}
\usepackage[title]{appendix}

\DeclarePairedDelimiter{\abs}{\lvert}{\rvert}
\DeclarePairedDelimiter{\norm}{\lVert}{\rVert}

\newcommand\inverse{^{-1}}
\newcommand{\R}{\mathbb{R}}
\newcommand{\N}{\mathbb{N}}

\newcommand{\superscript}[1]{\ensuremath{^{\textrm{#1}}}}

\renewcommand{\th}[0]{\superscript{th}}

\renewcommand{\E}{\mathbb{E}}
\renewcommand{\Pr}{\mathbb{P}}

\title{Leave-out estimation of variance components}

\author{Patrick Kline, Raffaele Saggio, Mikkel S\o lvsten\thanks{We thank Isaiah Andrews, Bruce Hansen, Whitney Newey, Anna Mikusheva, Jack Porter, Andres Santos, Azeem Shaikh and seminar participants at UC Berkeley, CEMFI, Chicago, Harvard, UCLA, MIT, Northwestern, NYU, Princeton, Queens, UC San Diego, Wisconsin, the NBER Labor Studies meetings, and the CEME Interactions workshop for helpful comments. The data used in this study was generously provided by the Fondazione Rodolfo De Benedetti and originally developed by the Economics Department of the Universit\`a Ca Foscari Venezia under the supervision of Giuseppe Tattara. We thank the Berkeley Institute for Research on Labor and Employment for funding support and Schmidt Futures, which provided financial assistance for this project through the Labor Science Initiative at the Berkeley Opportunity Lab.}}

\begin{document}
	
	\maketitle

	\begin{abstract}
		We propose leave-out estimators of quadratic forms designed for the study of linear models with unrestricted heteroscedasticity. Applications include analysis of variance and tests of linear restrictions in models with many regressors.
		An approximation algorithm is provided that enables accurate computation of the estimator in very large datasets. 
		We study the large sample properties of our estimator allowing the number of regressors to grow in proportion to the number of observations. Consistency is established in a variety of settings where plug-in methods and estimators predicated on homoscedasticity exhibit first-order biases. 
		For quadratic forms of increasing rank, the limiting distribution can be represented by a linear combination of normal and non-central $\chi^2$ random variables, with normality ensuing under strong identification. Standard error estimators are proposed that enable tests of linear restrictions and the construction of uniformly valid confidence intervals for quadratic forms of interest.
		We find in Italian social security records that leave-out estimates of a variance decomposition in a two-way fixed effects model of wage determination yield substantially different conclusions regarding the relative contribution of workers, firms, and worker-firm sorting to wage inequality than conventional methods. Monte Carlo exercises corroborate the accuracy of our asymptotic approximations, with clear evidence of non-normality emerging when worker mobility between blocks of firms is limited. 	
	\end{abstract}
	
	
	\newpage{}
	\onehalfspacing
	
	As economic datasets have grown large, so has the number of parameters employed in econometric models. Typically, researchers are interested in certain low dimensional summaries of these parameters that communicate the relative influence of the various economic phenomena under study. An important benchmark comes from \citet{fisher1925statistical}'s foundational work on analysis of variance (ANOVA) which he proposed as a means of achieving a ``separation of the variance ascribable to one group of causes, from the variance ascribable to other groups.''\footnote{See \citet{cochran1980fisher} for a discussion of the intellectual development of this early work.}
	
	A large experimental literature \citep{sacerdote2001peer,graham2008identifying,chetty2011does,angrist2014perils} employs variants of Fisher's ANOVA approach to	infer the degree of variability attributable to peer or classroom	effects. Related methods are often used to study heterogeneity across firms, workers, and schools in their responsiveness to exogenous regressors with continuous variation \citep{raudenbush1986hierarchical,raudenbush2002hierarchical,arellano2011identifying,graham2012identification}.
	In labor economics, log-additive models of worker and firm fixed effects are increasingly used to study worker-firm	sorting and the dispersion of firm specific pay premia \citep{abowd1999high,card2013workplace,card2016firms,song2015firming,sorkin2018ranking} and analogous methods have been applied to settings in health	economics \citep{finkelstein2016sources,silver2016essays} and the economics of education \citep{arcidiacono2012estimating}.

	This paper considers estimation of and inference on \emph{variance components}, which we define broadly as quadratic forms in the parameters of a linear model. Notably, this definition yields an important connection to the recent literature on testing linear restrictions in models with many regressors \citep{anatolyev2012inference,chao2014testing,cattaneo2017inference}. Traditional variance component estimators are predicated on the assumption that the errors in a linear model are identically distributed draws from a normal distribution. Standard references on this subject \citep[e.g.,][]{searle2009variance} suggest diagnostics for heteroscedasticity and non-normality, but offer little guidance regarding estimation and inference when these problems are encountered. A closely related literature on panel data econometrics proposes variance component estimators designed for fixed effects models that either restrict the dimensionality of the underlying group means \citep{bonhomme2019distributional} or the nature of the heteroscedasticity governing the errors \citep{andrews2008high,jochmans2016fixed}.
	
	%estimating
	
	Our first contribution is to propose a new variance component estimator designed for unrestricted linear models with 
	heteroscedasticity of unknown form.	The estimator is finite sample unbiased and can be written as a naive ``plug-in'' variance component estimator plus a bias correction term that involves ``cross-fit'' \citep{newey2018cross} estimators of observation-specific error variances. We also develop a representation of the estimator in terms of a covariance between outcomes and a ``leave-one-out'' generalized prediction \cite[e.g., as in][]{powell1989semiparametric}, which allows us to apply recent results on the behavior of second order U-statistics. Building on work by \cite{achlioptas2001database}, we propose a random projection method that enables computation of our estimator in very large
	datasets with little loss of accuracy.

	We study the asymptotic behavior of the proposed leave-out estimator in an environment where the number of regressors may be proportional to the sample size: a framework that has alternately been termed ``many covariates'' \citep{cattaneo2017inference} or ``moderate dimensional'' \citep{lei2016asymptotics} asymptotics. Verifiable design requirements are provided under which the estimator is consistent and we show in an Appendix that these conditions are weaker than those required by jackknife bias correction procedures \citep{quenouille1949approximate,hahn2004jackknife,dhaene2015split}. 
	A series of examples is discussed where the leave-out estimator is consistent, while estimators relying on jackknife or homoscedasticity-based bias corrections are not.
	
	We present three sets of theoretical results that enable inference based upon our estimator in a variety of settings. The first result concerns inference on quadratic forms of fixed rank, a problem which typically arises when testing a few linear restrictions in a model with many covariates \citep{cattaneo2017inference}. Familiar examples of such applications include testing that particular regressors are significant in a fixed effects model and conducting inference on the coefficients from a projection of fixed effects onto a low dimensional vector of covariates. Extending classic proposals by \cite{horn1975estimating} and \cite{mackinnon1985some}, we show that our leave-out approach can be used to construct an Eicker-White style variance estimator that is unbiased in the presence of unrestricted heteroscedasticity and that enables consistent inference on linear contrasts under weaker design restrictions than those considered by \citet{cattaneo2017inference}. 
	
	Next, we derive a result establishing asymptotic normality of quadratic forms of growing rank. Such quadratic forms typically arise when conducting analysis of variance but also feature in tests of model specification involving a large number of linear restrictions \citep{anatolyev2012inference,chao2014testing}. The large sample distribution of the estimator is derived using a variant of the arguments in \citet{chatterjee2008new} and \citet{soelvsten2017robust} and a standard error estimator is proposed that utilizes sample splitting formulations of the sort considered by \cite{newey2018cross}. This standard error estimator is shown to enable consistent inference on quadratic forms of growing rank in the presence of unrestricted heteroscedasticity when the regressor design allows for sample splitting and to provide conservative inference otherwise.
	
	%When certain conditions on the regressor design are met, this standard error estimator is shown to enable consistent inference on high rank quadratic forms in the presence of unrestricted heteroscedasticity When the requirements are not met, the standard error estimator leads to conservative inference.\fxnote{Pat review } 
	%I think we should foster the thought that we propose one standard error and that its behavior changes based on whether or not you are leave-one or two-out connected. Thus I changed the formulation. The old is here:
	%A conservative standard error estimator is provided for cases where the relevant design requirements are not met.

%without requiring smoothness conditions on
%under verifiable conditions on the regressor design 
	
	Finally, we present conditions under which the large sample distribution of our estimator is non-pivotal and can be represented by a linear combination of normal and non-central $\chi^{2}$ random variables, with the non-centralities of the $\chi^{2}$ terms serving as weakly identified nuisance parameters. This distribution arises in a two-way fixed effects model when there are ``bottlenecks'' in the mobility network. 
	Such bottlenecks are shown to emerge, for example, when worker mobility is governed by a stochastic block model with limited mobility between blocks. To construct asymptotically valid confidence intervals in the presence of nuisance parameters, we propose inverting a minimum distance test statistic.	Critical values are obtained via an application of the procedure of \cite{andrews2016geometric}. The resulting confidence interval is shown to be valid uniformly in the values of the nuisance parameters and to have a closed form representation in many settings, which greatly simplifies its computation.
	%The resulting confidence interval is shown to be uniformly valid and to have a closed form representation in many settings, which greatly simplifies its computation.

	We illustrate our results with an application of the two-way worker-firm
	fixed effects model of \citet{abowd1999high} to Italian social security records. 
	The proposed leave-out estimator
	finds a substantially smaller contribution of firms to wage inequality
	and much more assortativity in the matching of workers to firms than
	either the uncorrected plug-in estimator originally considered by \citet{abowd1999high}
	or the homoscedasticity-based correction procedure of \citet{andrews2008high}.
%	Applying a generalization of our estimator that leaves out all observations
%	associated with a worker, we find evidence of substantial 	
	When studying panels of length greater than two, we allow for
	serial correlation in the errors by employing a generalization of our estimator
	that leaves out all the observations in a worker-firm match. Failing to account for this dependence
	is shown to yield over-estimates of the variance of firm effects. 
	
	%Remarkably, we find no evidence of between-match dependence in the errors.	
%	that match. Remarkably, between by  corresponding to a . 
	%Doing so
	%yields results remarkably similar to those obtained when analyzing only two time periods.
	
	Projecting firm effect estimates onto measures of worker age and firm size, we find
	that older workers tend to be employed at firms offering higher firm wage effects;
	however, this phenomenon is largely explained by the tendency of older
	workers to sort to bigger firms. Leave-out standard errors for the coefficients
	of these linear projections are found to be several times larger than 
	%those produced by 
	a naive standard error predicated on the assumption that the estimated fixed effects are independent of each other.
%	Standard error estimates for the coefficients
%	of a linear projection of firm effects onto worker and firm characteristics are found to be several times larger than those produced by a naive two-step approach. 	
%	Older workers tend to be employed at firms offering higher firm wage effects. However,
%	a linear projection of firm effects onto measures of worker age and firm size suggests 
%	this phenomenon is largely explained by the tendency of older workers to sort to bigger firms.
%	Leave-out standard errors for these projection relationships are found to exceed
%	naive standard errors that treat firm effect estimates
%	as independent from one another by orders of magnitude.
%	Naive standard errors that treat firm effect estimates
%	as independent from one another are shown to
%	massively under-state the sampling variability of the coefficients associated with these projections.	
%	approaches to estimating
%	the standard errors of these projection coefficients 
%	Leave out standard error estimates for 
%	Developing a connection between our proposed test of linear restrictions and leave-out estimates of the coefficient of determination,
	Stratifying our analysis by birth cohort, we formally reject the null hypothesis
	that older and younger workers face identical vectors of firm effects. However, the two sets of firm effects are 
	estimated to have a correlation coefficient of nearly 0.9, while the plug-in estimate of correlation is only 0.54.
%	allowing age-dependent firm effects yields trivial
%	improvements in model fit. 

	To assess the accuracy of our asymptotic
	approximations, we conduct a series of Monte Carlo exercises utilizing the realized mobility patterns of workers between firms. Clear evidence of non-normality arises in the sampling distribution
	of the estimated variance of firm effects
	in settings where the worker-firm mobility network is weakly connected. The proposed confidence regions are shown to provide reliable size control in both strongly and weakly identified settings.

	\section{Unbiased Estimation of Variance Components}
	\label{sec:unbiased}
	
	Consider the linear model
	\begin{align}\label{eq:model}
		y_{i} = x_{i}'\beta + \varepsilon_{i} && (i = 1,\ldots,n)
	\end{align}
	where the regressors $x_{i}\in\mathbb{R}^{k}$ are non-random and the design matrix $S_{xx}=\sum_{i=1}^{n}x_{i}x_{i}'$ has full rank. The unobserved errors $\{ \varepsilon_{i} \}_{i=1}^{n}$ are mutually independent and obey $\E[\varepsilon_{i}]=0$, but may possess observation specific variances $\E[\varepsilon_{i}^2]=\sigma_{i}^{2}$. 
	
	Our object of interest is a quadratic form $\theta=\beta'A\beta$ for some known non-random symmetric matrix $A\in\mathbb{R}^{k\times k}$ of rank $r$. Following \cite{searle2009variance}, when $A$ is positive semi-definite $\theta$ is a \emph{variance} component, while when $A$ is non-definite $\theta$ may be referred to as a \emph{covariance} component. Note that linear restrictions on the parameter vector $\beta$ can be formulated in terms of variance components: for a non-random vector $v$, the null hypothesis $v'\beta=0$ is equivalent to the restriction $\theta=0$ when $A=vv'$. Examples from the economics literature where variance components are of direct interest are discussed in Section \ref{sec:examples}.

	\subsection{Estimator}
	A naive plug-in estimator of $\theta$ is given by the quadratic form $\hat \theta_{\text{PI}}= \hat{\beta}'A\hat{\beta}$,	where $\hat{\beta}=S_{xx}^{-1}\sum_{i=1}^{n}x_{i}y_{i}$ denotes the Ordinary Least Squares (OLS) estimator of $\beta$. Estimation error in $\hat \beta$ leads the plug-in estimator to exhibit a bias involving a linear combination of the unknown variances $\{ \sigma_{i}^{2}\} _{i=1}^{n}$. Specifically, standard results on quadratic forms imply that $\E [ \hat \theta ]=\theta + \text{trace}( A \mathbb{V}[\hat \beta])$, where
	\begin{align}\label{eq:bias}
		\text{trace}\left( A \mathbb{V}[\hat \beta]\right) = \sum_{i=1}^{n} B_{ii} \sigma_{i}^{2} & \quad \mbox{and} \quad B_{ii}=x_{i}' S_{xx}^{-1} A S_{xx}^{-1} x_{i}.
	\end{align}
	As discussed in Section \ref{sec:examples}, this bias can be particularly severe when the dimension of the regressors $k$ is large relative to the sample size.

	A bias correction can be motivated by observing that an unbiased estimator of the $i$-th error variance is 
	\begin{align}\label{eq:sigmahat}
		\hat{\sigma}_{i}^{2}=y_{i}\left(y_{i}-x_{i}'\hat{\beta}_{-i}\right),
	\end{align}	
	where $\hat \beta_{-i}=\left(S_{xx} - x_i x_i' \right)\inverse\sum_{\ell\neq i}x_{\ell}y_{\ell}$ denotes the leave-$i$-out OLS estimator of $\beta$. This insight suggests the following bias-corrected estimator of $\theta$:
	\begin{align}\label{eq:estimator}
		\hat{\theta}=\hat{\beta}'A\hat{\beta}-\sum_{i=1}^{n} B_{ii}\hat{\sigma}_{i}^{2}.
	\end{align}
	While \cite{newey2018cross} observe that ``cross-fit'' covariances relying on sample splitting can be used to remove bias of the sort considered here, we are not aware of existing estimators involving the leave-one-out estimators $\{\hat{\sigma}_{i}^{2}\} _{i=1}^{n}$.

	One can also motivate $\hat{\theta}$ via a change of variables argument.	Letting $\tilde{x}_{i}=A S_{xx}^{-1} x_{i}$ denote a vector of ``generalized'' regressors, we can write
	\begin{align}
		\theta = \beta' A \beta = \beta' S_{xx} S_{xx}\inverse A  \beta = \sum_{i=1}^{n}\beta'  x_{i} \tilde {x}_{i}'\beta  =  \sum_{i=1}^{n}\mathbb{E}\left[y_{i}\tilde{x}_{i}'\beta\right].
	\end{align}
	This observation suggests using the unbiased \emph{leave-out} estimator
	\begin{align}\label{eq:cov}
		\hat \theta = \sum_{i=1}^{n}y_{i}\tilde{x}_{i}'\hat{\beta}_{-i}.
	\end{align}	
	
	Note that direct computation of $\hat \beta_{-i}$ can be avoided by exploiting the representation
		\begin{align}\label{eq:Pii}
			y_{i}-x_{i}'\hat{\beta}_{-i}=\dfrac{y_i - x_i'\hat \beta}{1-P_{ii}},
		\end{align}
where $P_{ii}= x_i' S_{xx}\inverse x_i$ gives the leverage of observation $i$. Applying the Sherman-Morrison-Woodbury formula \citep{woodbury1949stability,sherman1950adjustment}, this representation also reveals that \eqref{eq:estimator} and \eqref{eq:cov} are numerically equivalent:
	\begin{align}
		y_i \tilde x_i'\hat \beta_{-i} &= \underbrace{y_i \tilde x_i' S_{xx}\inverse \sum_{\ell\neq i}x_{\ell}y_{\ell}}_{=y_i \tilde x_i'\hat \beta - B_{ii} y_i^2} +  \underbrace{\frac{y_i \tilde x_i' S_{xx} \inverse x_i x_i' S_{xx}\inverse}{1 - P_{ii}} \sum_{\ell\neq i}x_{\ell}y_{\ell}}_{=B_{ii} y_i x_i'\hat \beta_{-i}} = y_i \tilde x_i'\hat \beta - B_{ii} \hat \sigma_i^2.
	\end{align}
	A similar combination of a change of variables argument and a leave-one-out estimator was used by \cite{powell1989semiparametric} in the context of weighted average derivatives. The JIVE estimators proposed by \cite{phillips1977bias} and \cite{angrist1999jackknife} also use a leave-one-out estimator, though without the change of variables.\footnote{ The object of interest in JIVE estimation is a \emph{ratio} of quadratic forms $\beta_1'S_{xx}\beta_2/\beta_2'S_{xx}\beta_2$ in the two-equation model $y_{ij} = x_{i}'\beta_j + \varepsilon_{ij}$ for $j=1,2$. When no covariates are present, using leave-out estimators of both the numerator and denominator of this ratio yields the JIVE1 estimator of \cite{angrist1999jackknife}.}

	\begin{rem}
		The $\{\hat\sigma_i^2 \}_{i=1}^n$ can also be used to construct an unbiased variance estimator 
		\begin{align}
			\hat{\mathbb{V}}[\hat{\beta}] = S_{xx}^{-1}\left(\sum_{i=1}^{n} x_{i} x_{i}' \hat{\sigma}_{i}^{2} \right)S_{xx}^{-1}.
		\end{align}
		Section \ref{sec:InfLow} shows that $\hat{\mathbb{V}}[\hat{\beta}]$ can be used to perform asymptotically valid inference on linear contrasts in settings where existing Eicker-White estimators fail. Specifically, $\hat{\mathbb{V}}[\hat{\beta}]$ leads to valid inference under conditions where the MINQUE estimator of \cite{rao1970estimation} and the MINQUE-type estimator of \cite{cattaneo2017inference} do not exist \cite[see, e.g.,][]{horn1975estimating,verdier2016estimation}.
	\end{rem}

	\begin{rem}
		The quantity $\hat{\mathbb{V}}[\hat{\beta}]$ is closely related to the HC2 variance estimator of \cite{mackinnon1985some}. While the HC2 estimator employs observation specific variance estimators $\hat \sigma^2_{i,\text{HC2}}=\frac{(y_i - x_i'\hat \beta)^2}{1-P_{ii}}$, $\hat{\mathbb{V}}[\hat{\beta}]$ relies instead on $\hat \sigma^2_i=\frac{y_i(y_i - x_i'\hat \beta)}{1-P_{ii}}$
	\end{rem}

	\begin{rem}\th\label{rem:cluster}
		In some cases it may be important to allow dependence in the errors in addition to heteroscedasticity. A common case arises when the data are organized into mutually exclusive and independent ``clusters'' within which the errors may be dependent \citep{moulton1986random}. The same change of variables argument implies that an estimator of the form $\sum_{i=1}^{n}y_{i}\tilde{x}_{i}'\hat{\beta}_{-c\left(i\right)}$
		will be unbiased in such settings, where $\hat{\beta}_{-c\left(i\right)}$ is the OLS estimator obtained after leaving out all observations in the cluster to which observation $i$ belongs. 
	\end{rem}

\subsection{Large Scale Computation} \label{sec:comp}
From \eqref{eq:estimator} and \eqref{eq:Pii}, computation of $\hat \theta$ relies on the values $\{B_{ii},P_{ii}\}_{i=1}^n$. Section \ref{sec:examples} provides some canonical examples where these quantities can be computed in closed form. When closed forms are unavailable, a number of options exist for accelerating computation. For example, in the empirical application of Section \ref{sec:application}, we make use of a preconditioned conjugate gradient algorithm suggested by \cite{koutis2011combinatorial} to compute exact leave-out variance decompositions in a two-way fixed effects model involving roughly one million observations and hundreds of thousands of parameters (see Appendix \ref{sec:Lambda} for details). However, in very large scale applications involving tens or hundreds of millions of parameters, exact computation of $\{B_{ii},P_{ii}\}_{i=1}^n$ is likely to become infeasible. Fortunately, it is possible to quickly approximate $\hat \theta$ in such settings using a variant of the random projection method introduced by \cite{achlioptas2001database}. We refer to this method as the Johnson-Lindenstrauss approximation (JLA) for its connection to the work of \cite{johnson1984extensions}.

%that is designed for accelerating the solution of symmetric and diagonally dominant linear systems. As described in Appendix \ref{sec:Lambda}, this algorithm allows us 
% even when employing acceleration methods

	%Johnson-Lindenstrauss Lemma
	
	JLA can be described by the following algorithm: fix a $p \in \N$ and generate the matrices $R_B, R_P \in \R^{p \times n}$, where $(R_B,R_P)$ are composed of mutually independent Rademacher random variables that are independent of the data, i.e., their entries take the values $1$ and $-1$ with probability $1/2$. Next decompose $A$ into $A=\frac{1}{2}( A_1'A_2 + A_2'A_1)$ for $A_1,A_2 \in \R^{n \times k}$ where $A_1=A_2$ if $A$ is positive semi-definite.\footnote{Interpretable choices of $A_1$ and $A_2$ are typically suggested by the structure of the problem; see, for instance, the discussion in \th\ref{ex:AKM} of Section \ref{sec:examples}.} Let
	\begin{align}
		\hat P_{ii} = \frac{1}{p}\norm*{ R_P X S_{xx}\inverse x_i}^2
		\quad \text{and} \quad
		\hat B_{ii} =\frac{1}{p} \left(R_B A_1 S_{xx}\inverse x_i\right)'\left(R_B A_2 S_{xx}\inverse x_i\right)
	\end{align}
	where $X =(x_1,\dots,x_n)'$. % and $A_1,A_2 \in \R^{n \times k}$ obey $ \frac{1}{2}( A_1'A_2 + A_2'A_1)=A$. %In the examples of the next section, $A_1$ and $A_2$ can be chosen as demeaned regressors. 
	The Johnson-Lindenstrauss approximation to $\hat \theta$ is 
	\begin{align}
		\hat \theta_{JLA} = \hat \beta'A\hat \beta - \sum_{i=1}^n \hat B_{ii} \hat \sigma_{i,JLA}^2,
	\end{align}
	where $\hat \sigma_{i,JLA}^2 = \frac{y_{i}\left(y_{i}-x_{i}'\hat{\beta}\right)}{1-\hat P_{ii}}\left(1- \frac{1}{p}\frac{3\hat P_{ii}^3 + \hat P_{ii}^2}{1-\hat P_{ii}}\right)$. The term $\frac{1}{p}\frac{3\hat P_{ii}^3 + \hat P_{ii}^2}{1-\hat P_{ii}}$ removes a non-linearity bias introduced by approximating $P_{ii}$. 
	
	Section \ref{sec:conspaper} establishes asymptotic equivalence between $\hat \theta_{JLA}$ and $\hat \theta$. Appendix \ref{sec:Lambda} discusses implementation details and numerically illustrates the trade-off between computation time and the bias introduced by JLA for different choices of $p$ under a range of sample sizes. Notably, we show
	that JLA allows us to accurately compute a variance decomposition in a two-way fixed effects model with roughly 15 million parameters -- a scale comparable to the study of \cite{card2013workplace} -- in under an hour. A MATLAB package \citep{kssCODE} implementing both the exact and JLA versions of our estimator in the two-way fixed effects model is available online.

	\subsection{Relation to Existing Approaches}		
	As discussed in Section \ref{sec:examples}, several literatures make use of bias corrections nominally predicated on homoscedasticity. A common ``homoscedasticity-only'' estimator takes the form
	\begin{align}\label{eq:HO}
		\hat \theta_{\text{HO}} = \hat{\beta}'A\hat{\beta}-\sum_{i=1}^{n} B_{ii} \hat{\sigma}_{\text{HO}}^{2}
	\end{align}
	where $\hat{\sigma}_{\text{HO}}^{2} = \frac{1}{n-k} \sum_{i=1}^n (y_i - x_i' \hat \beta)^2$ is the degrees-of-freedom corrected variance estimator. A sufficient condition for unbiasedness of $\hat \theta_{\text{HO}}$ is that there be no empirical covariance between $\sigma^{2}_{i}$ and $(B_{ii},P_{ii})$. This restriction is in turn implied by the special cases of \emph{homoscedasticity} where $\sigma_i^2$ does not vary with $i$ or \emph{balanced design} where $(B_{ii},P_{ii})$ does not vary with $i$. In general, however, this estimator will tend to be biased \citep[see, e.g.,][chapter 10, or Appendix \ref{app:rel}]{scheffe1959analysis}.

	A second estimator, closely related to $\hat{\theta}$, relies upon a jackknife bias-correction \citep{quenouille1949approximate} of the plug-in estimator. This estimator can be written
	\begin{align}
		\hat \theta_{\text{JK}} = n \hat \theta_\text{PI} - \frac{n-1}{n} \sum_{i=1}^n \hat \theta_{\text{PI},-i}
		\quad \text{where} \quad
		\hat \theta_{\text{PI},-i} = \hat \beta_{-i}' A \hat \beta_{-i}.
	\end{align}
	In Appendix \ref{app:rel} we illustrate that jackknife bias-correction tends to over-correct and produce a first order bias in the opposite direction of the bias in the plug-in estimator. This is analogous to the upward bias in the jackknife estimator of $\mathbb{V}[\hat \beta]$ which was derived by \cite{efron1981} and shown by \cite{el2018can} to be of first order importance for inference with many Gaussian regressors.

	There are several proposed adaptations of the jackknife to long panels that can decrease bias under stationarity restrictions on the regressors. Letting $t(i)\in\{1,...,T\}$ denote the time period in which an observation is observed, we can write the panel jackknife of \cite{hahn2004jackknife} as
	\begin{align}
		\hat \theta_\text{PJK} = T \hat \theta_{\text{PI}} - \frac{T-1}{T} \sum_{t=1}^T \hat \theta_{\text{PI},-t}
		\quad \text{where} \quad 
		\hat \theta_{\text{PI},-t} = \hat \beta_{-t}' A\hat \beta_{-t}
	\end{align}
	and $\hat \beta_{-t} =(\sum_{i : t(i) \neq t} x_i x_i')\inverse \sum_{i : t(i) \neq t} x_i y_i$
	is the OLS estimator that excludes all observations from period $t$. \cite{dhaene2015split} propose a closely related split panel jackknife
	\begin{align}
		\hat \theta_\text{SPJK} = 2\hat \theta_{\text{PI}} - \frac{\hat \theta_{\text{PI},1} + \hat \theta_{\text{PI},2}}{2}
		\quad \text{where} \quad 
		\hat \theta_{\text{PI},j} = \hat \beta_{j}' A\hat \beta_{j}
	\end{align}
	and $\hat \beta_{1}$ (and $\hat \beta_{2}$) are OLS estimators based on the first half (and the last half) of an even number of time periods. In Appendix \ref{app:rel}, we illustrate how short panels can lead these adaptations of the jackknife to produce first order biases in the opposite direction of the bias in the plug-in estimator.
	
	\subsection{Finite Sample Properties}\label{fsample}
	
	We now study the finite sample properties of the leave-out estimator $\hat \theta$ and its infeasible analogue $\theta^* = \hat \beta'A\hat \beta - \sum_{i=1}^n B_{ii}\sigma_i^2$, which uses knowledge of the individual error variances. First, we note that $\hat \theta$ is unbiased whenever each of the leave-one-out estimators $\hat \beta_{-i}$ exists, which can equivalently be expressed as the requirement that $\max_i P_{ii} <1$. This condition turns out to also be necessary for the existence of unbiased estimators, which highlights the need for additional restrictions on the model or sample whenever some leverages equal one.

	\begin{lem}
		\hangindent\leftmargini
		\th\label{lem:unbiased}
		\text{1.} If $\max_i P_{ii} <1$, then $\E[\hat \theta] = \theta$.
		\begin{enumerate}
			\setcounter{enumi}{1}
			\item Unbiased estimators of $\theta = \beta'A\beta$ exist for all $A$ if and only if $\max_{i} P_{ii} < 1$.
		\end{enumerate}
	\end{lem}

	Next, we show that when the errors are normal, the infeasible estimator $\theta^*$ is a weighted sum of a series of non-central $\chi^2$ random variables. This second result provides a useful point of departure for our asymptotic approximations and highlights the important role played by the matrix
	\begin{align}
		\tilde A = S_{xx}^{-1/2} A S_{xx}^{-1/2},
	\end{align}
	which encodes features of both the target parameter (which is defined by $A$) and the design matrix $S_{xx}$.

	Let $\lambda_1,\dots,\lambda_r$ denote the non-zero eigenvalues of $\tilde A$, where $\lambda_1^2 \ge\dots \ge \lambda_r^2$ and each eigenvalue appears as many times as its algebraic multiplicity. We use $Q$ to refer to the corresponding matrix of orthonormal eigenvectors so that $\tilde A = Q D Q'$ where $D = \text{diag}(\lambda_1,\dots,\lambda_r)$. With these definitions we have
	\begin{align}
		\hat \beta'A\hat \beta = \sum_{\ell = 1}^r \lambda_\ell \hat b_\ell^2, 
		%\quad \text{where} \quad
		%\hat b = (\hat b_1,\dots,\hat b_r)' = Q'S_{xx}^{1/2} \hat \beta.
	\end{align}
	where $\hat b = (\hat b_1,\dots,\hat b_r)' = Q'S_{xx}^{1/2} \hat \beta$ contains $r$ linear combinations of the elements in $\hat \beta$.
	%The $r$-dimensional 
	The random vector $\hat b$ and the eigenvalues $\lambda_1,\dots,\lambda_r$ are central to both the finite sample distribution provided below in \th\ref{lem:fin} and the asymptotic properties of $\hat \theta$ as studied in Sections \ref{sec:InfLow}--\ref{sec:weak}. Each eigenvalue of $\tilde A$ can be thought of as measuring how strongly $\theta$ depends on a particular linear combination of the elements in $\beta$ relative to the difficulty of estimating that combination (as summarized by $S_{xx}^{-1}$). As discussed in Section \ref{sec:weak}, when a few of these eigenvalues are large relative to the others, a form of weak identification can arise.
	
	\begin{lem}\th\label{lem:fin}
		If $\varepsilon_{i} \sim \mathcal{N}(0, \sigma_i^2)$, then
		\begin{enumerate}
			\item $\hat b \sim \mathcal{N}\left( b , \mathbb{V}[\hat b ] \right)$ where $b = Q' S_{xx}^{1/2} \beta$,
			\item $\theta^* =  \sum_{\ell = 1}^{r} \lambda_\ell\left(\hat b_\ell^2-\mathbb{V}[\hat b_\ell ]\right)$
		\end{enumerate}
	\end{lem}

	The distribution of $\theta^*$ is a sum of $r$ potentially dependent non-central $\chi^2$ random variables with non-centralities $b = (b_1,\dots,b_r)'$. In the special case of homoscedasticity $(\sigma_i^2 = \sigma^2)$ and no signal $(b=0)$ we have that $\hat b \sim \mathcal{N}\left( 0, \sigma^2 I_r \right)$, which implies that the distribution of $\theta^*$ is a weighted sum of $r$ \emph{independent} central $\chi^2$ random variables. The weights are the eigenvalues of $\tilde A$, therefore consistency of $\theta^*$ follows whenever the sum of the squared eigenvalues converges to zero. The next subsection establishes that the leave-out estimator remains consistent when a signal is present ($b\neq0$) and the errors exhibit unrestricted heteroscedasticity.
	
	\subsection{Consistency}\label{sec:conspaper}
	
	We now drop the normality assumption and provide conditions under which $\hat \theta$ remains consistent. To accommodate high dimensionality of the regressors we allow all parts of the model to change with $n$: 
	\begin{align}
	y_{i,n} = x_{i,n}'\beta_n + \varepsilon_{i,n} && (i = 1,\ldots,n)
	\end{align}
	where $x_{i,n}\in\mathbb{R}^{k_n}$, $S_{xx,n}=\sum_{i=1}^{n}x_{i,n}x_{i,n}'$, $\E[\varepsilon_{i,n}]=0$, $\E[\varepsilon_{i,n}^2]=\sigma_{i,n}^{2}$ and $\theta_n=\beta_n'A_n\beta_n$ for some sequence of known non-random symmetric matrices $A_n\in\mathbb{R}^{k_n\times k_n}$ of rank $r_n$. By treating $x_{i,n}$ and $A_n$ as sequences of constants, all uncertainty derives from the disturbances $\left\{ \varepsilon_{i,n} : 1\le i \le n, n \ge 1 \right\}$. This \emph{conditional} perspective is common in the statistics literatures on ANOVA \citep{scheffe1959analysis,searle2009variance} and allows us to be agnostic about the potential dependency among the $\{x_{i,n}\}_{i=1}^n$ and $A_n$.\footnote{An \emph{unconditional} analysis might additionally impose distributional assumptions on $A_n$ and consider $\bar \theta = \beta'\E_{A_n}[A_n] \beta$ as the object of interest. The uncertainty in $\hat \theta - \bar \theta$ can always be decomposed into components attributable to $\hat \theta - \theta$ and $\theta - \bar \theta$. Because the behavior of $\theta - \bar \theta$ depends entirely on model choices, we leave such an analysis to future work.} Following standard practice we drop the $n$ subscript in what follows. All limits are taken as $n$ goes to infinity unless otherwise noted.
	
	Our analysis makes heavy use of the following assumptions.	
	\begin{assumption}\th\label{ass:reg} 
		(i) $\max_i \left( \E[\varepsilon_i^4] + \sigma^{-2}_i \right) = O(1)$, (ii) there exist a $c <1$ such that $\max_i P_{ii} \le c$ for all $n$, and (iii) $\max_i (x_i'\beta)^2 = O(1)$.
	\end{assumption}		
	Part $(i)$ of this condition limits the thickness of the tails in the error distribution, as is typically required for OLS estimation \citep[see, e.g.,][page 10]{cattaneo2017inference}. The bounds on $(x_i'\beta)^2$ and $P_{ii}$ imply that $\hat \sigma_i^2$ has bounded variance. Part $(iii)$ is a technical condition that can be relaxed to allow $\max_i (x_i'\beta)^2$ to %be unbounded as the sample size grows as discussed further in Section \ref{sec:verify}.
	increase slowly with sample size as discussed further in Section \ref{sec:verify}.
	From $(ii)$ it follows that $\frac{k}{n} \le c < 1$ for all $n$.

	The following \th\nameref{lem:cons} establishes consistency of $\hat \theta$. 
	\begin{lem}\th\label{lem:cons}
		If \th\ref{ass:reg} and one of the following conditions hold, then $\hat{\theta} - \theta \overset{p}{\rightarrow} 0$.
		\begin{enumerate}
			\item[(i)] $A$ is positive semi-definite, $\theta = \beta'A\beta= O(1)$, and $\text{trace}(\tilde A^2) = \sum_{\ell =1}^r \lambda_{\ell}^2= o(1)$.
			
			\item[(ii)] $A=\frac{1}{2}( A_1'A_2 + A_2'A_1)$ where $\theta_1 = \beta'A_1'A_1\beta$ and $\theta_2 = \beta'A_2'A_2\beta$ satisfy (i).
		\end{enumerate}
%		\begin{lem}\th\label{lem:cons}
%		Suppose \th\ref{ass:reg} holds.
%		\begin{enumerate}
%			\item If $A$ is positive semi-definite, (i) $\theta= O(1)$, (ii) $\text{trace}(\tilde A^2) = \sum_{\ell =1}^r \lambda_{\ell}^2= o(1)$, then $\hat{\theta} - \theta \overset{p}{\rightarrow} 0$.
%			
%			\item If $A=\frac{1}{2}( A_1'A_2 + A_2'A_1)$ where $\varTheta_\ell = \beta'A_\ell'A_\ell\beta$ satisfies (i) and (ii) for $\ell=1,2$, then $\hat{\theta} - \theta \overset{p}{\rightarrow} 0$.
%		\end{enumerate}
		
	\end{lem}
	The first condition of \th\ref{lem:cons} establishes consistency of variance components given boundedness of $\theta$ and a joint condition on the design matrix $S_{xx}$ and the matrix $A$. The second condition shows that consistency of covariance components follows from consistency of variance components that dominate them via the Cauchy-Schwarz inequality, i.e., $\theta^2 = (\beta'A_1'A_2\beta)^2 \le \theta_1 \theta_2$. In several of the examples discussed in the next section, $\text{trace}(\tilde A^2)$ is of order $r/n^2$, which is necessarily small in large samples. A more extensive discussion of primitive conditions that yield $\text{trace}(\tilde A^2) = o(1)$ is provided in Section \ref{sec:verify}.

	We conclude this section by establishing asymptotic equivalence between the leave-out estimator $\hat \theta$ and its approximation $\hat \theta_{JLA}$ under the condition that $p^4$ is large relative to sample size.

	\begin{lem}\th\label{lem:JLA}
	If \th\ref{ass:reg} is satisfied, $n/p^4 = o(1)$, $\mathbb{V}[\hat  \theta]\inverse = O(n)$, and one of the following conditions hold, then $\mathbb{V}[\hat  \theta]^{-1/2}{(\hat \theta_{JLA} - \hat \theta - \mathrm{B}_p)}{} = o_p(1)$ where $\abs{\mathrm{B}_p} \le \frac{1}{p} \sum_{i=1}^n P_{ii}^2 \abs{B_{ii}} \sigma_i^2$.

	\begin{enumerate}
		\item[(i)] $A$ is positive semi-definite and $\E[\hat \beta'A\hat \beta] - \theta = \sum_{i=1}^n B_{ii} \sigma_i^2 = O(1)$.
		
		\item[(ii)] $A=\frac{1}{2}( A_1'A_2 + A_2'A_1)$ where $\theta_1 = \beta'A_1'A_1\beta$ and $\theta_2 = \beta'A_2'A_2\beta$ satisfy (i) and $\frac{\mathbb{V}[\hat  \theta_1]\mathbb{V}[\hat  \theta_2]}{n\mathbb{V}[\hat  \theta]^2} = O(1)$.
	\end{enumerate}
	
	\end{lem}
	
%	\begin{lem}\th\label{lem:JLA}
%		Suppose that \th\ref{ass:reg} holds, $n/p^4 = o(1)$, and $\mathbb{V}[\hat  \theta]\inverse = O(n)$. Let $\mathrm{B}_p = \frac{1}{p} \sum_{i=1}^n B_{ii} \sigma_i^2  \frac{2\sum_{\ell \neq i}^n P_{i\ell}^4 - P_{ii}^2(1-P_{ii})^2}{(1-P_{ii})^2}$.
%		\begin{enumerate}
%			\item If $A$ is positive semi-definite and (i) $\E[\hat \theta_{\text{PI}}] - \theta = O(1)$, then $\frac{\hat \theta_{JLA} - \hat \theta - \mathrm{B}_p}{\mathbb{V}[\hat  \theta]^{1/2}} = o_p(1)$ and $\abs{\mathrm{B}_p} \le \frac{1}{p} \sum_{i=1}^n P_{ii}^2 B_{ii} \sigma_i^2$.
%			
%			\item If $A$ is non-definite with $A=\frac{1}{2}( A_1'A_2 + A_2'A_1)$ where $\varTheta_\ell = \beta'A_\ell'A_\ell\beta$ satisfies (i) for $\ell=1,2$ and $\frac{\mathbb{V}[\hat  \varTheta_1]\mathbb{V}[\hat  \varTheta_2]}{n\mathbb{V}[\hat  \theta]^2} = O(1)$, then $\frac{\hat \theta_{JLA} - \hat \theta - \mathrm{B}_p}{\mathbb{V}[\hat  \theta]^{1/2}} = o_p(1)$ and $\abs{\mathrm{B}_p} \le \frac{1}{p} \sum_{i=1}^n P_{ii}^2 \abs{B_{ii}} \sigma_i^2$.
%		\end{enumerate}
%
%	\end{lem}

	\th\ref{lem:JLA} requires that $\hat \theta$ is not super-consistent and that the bias in the plug-in estimator is asymptotically bounded, assumptions which can be shown to be satisfied in the examples introduced in the next section. For variance components, the \th\nameref{lem:JLA} characterizes an approximation bias $\mathrm{B}_p$ in $\hat \theta_{JLA}$ of order $1/p$ and provides an interpretable bound on $\mathrm{B}_p$: the approximation bias is at most $1/p$ times the bias in the plug in estimator $\hat \beta 'A \hat \beta$. For covariance components, asymptotic equivalence follows when the variance components defined by $A_1'A_1$ and $A_2'A_2$ do not converge at substantially slower rates than $\hat \theta$. Under this condition, the approximation bias is at most $1/p$ times the average of the biases in the plug in estimators $\hat \beta 'A_1'A_1 \hat \beta$ and $\hat \beta 'A_2'A_2 \hat \beta$.
	
	These bounds on the approximation bias suggests that a $p$ of a few hundred should suffice for point estimation. However, unless $n/p^2 = o(1)$, the resulting approximation bias needs to be accounted for when conducting inference. Specifically, one can lengthen the tails of the confidence sets proposed in Sections \ref{sec:dist} and \ref{sec:variance} by $\frac{1}{p} \sum_{i=1}^n \hat P_{ii}^2 \abs{\hat B_{ii}} \hat \sigma_{i,JLA}^2$ when relying on JLA. 
	
%	The conditions above require that $\hat \theta$ is not super-consistent and that the bias in the plug-in estimator is asymptotically bounded, assumptions which can be shown to be satisfied in the examples introduced in the next section. The lemma characterizes an approximation bias $\mathrm{B}_p$ in $\hat \theta_{JLA}$ of order $1/p$ and provides an interpretable bound on $\mathrm{B}_p$: the approximation bias is at most $1/p$ times the bias in the plug in estimator $\hat \beta 'A \hat \beta$ when $A$ is positive semi-definite. This bound suggests that a $p$ of a few hundred should suffice for point estimation. However, the resulting approximation bias needs to be accounted for when conducting inference. Specifically, one can enlarge the confidence sets proposed in Sections \ref{sec:dist} and \ref{sec:variance} using $\frac{1}{p} \sum_{i=1}^n \hat P_{ii}^2 \abs{\hat B_{ii}} \hat \sigma_{i,JLA}^2$ when relying on JLA. 
	%While the result allows for a relatively small $p$, it seems advisable to use as large a $p$ as is computationally feasible. 

	\section{Examples}\label{sec:examples}
	
	We now consider four commonly encountered empirical examples where our proposed estimation strategy provides an advantage over existing methods.

	\begin{example}[Coefficient of determination]\th\label{ex:R2}
		\textit{}
		
		Sewall \citet{wright1921correlation} proposed measuring the explanatory power of a linear model using the coefficient of determination. When $x_i$ includes an intercept, the object of interest and its corresponding plug-in estimator can be written
		\begin{align}
		R^2 &= \frac{\beta' A \beta }{\beta' A \beta + \frac{1}{n}\sum_{i=1}^{n} \sigma_i^2} = \frac{\sigma^2_{X\beta}}{\sigma^2_y} 
		\quad \text{and} \quad
		\hat R_\text{PI}^2 = \frac{\hat \beta' A \hat \beta}{\frac{1}{n}\sum_{i=1}^{n} (y_i - \bar y)^{2}} = \frac{\hat \sigma^2_{X\beta,\text{PI}}}{\hat \sigma^2_y} 
		\shortintertext{where}
		A &= \frac{1}{n}\sum_{i=1}^{n} (x_i-\bar x)(x_i-\bar x)', \quad \bar x = \frac{1}{n} \sum_{i=1}^n x_i, \quad \bar y = \frac{1}{n} \sum_{i=1}^n y_i.
		\end{align}
		\cite{theil1961economic} noted that the plug-in estimator of $\sigma^2_{X\beta}$ is biased and proposed an adjusted $R^2$ measure that utilizes the homoscedasticity-only estimator in \eqref{eq:HO}. The above choice of $A$ yields $B_{ii} = \frac{1}{n} (P_{ii} - \frac{1}{n}),$ which implies $\sum_{i=1}^n B_{ii} = \frac{k-1}{n}$. Hence, Theil's proposal can be written
		\begin{align}
		\hat R^2_\text{adj} = \frac{\hat \sigma^2_{X\beta,\text{HO}}}{\hat \sigma^2_y} = \frac{\hat \beta' A \hat \beta - \frac{k-1}{n}\hat{\sigma}_{\text{HO}}^{2}}{\hat \sigma^2_y}.
		\end{align}
		A rearrangement gives the familiar representation $\frac{1-\hat R^2_\text{adj}}{1-\hat R_{\text{PI}}^2} = \frac{n-1}{n-k}$ which highlights that the adjusted estimator of $R^2$ relates to the unadjusted one through a degrees-of-freedom correction. 
		
		The leave-out estimator of $\sigma^2_{X\beta}$ allows for unrestricted heteroscedasticity and can be found by noting that $\tilde x_i = A S_{xx}\inverse x_i = \frac{1}{n}(x_i - \bar x)$, which yields
		\begin{align}
		\hat R^2 = \frac{\hat \sigma_{X\beta}^2}{\hat \sigma^2_y}
		\quad \text{where} \quad
		\hat \sigma^2_{X\beta}=\frac{1}{n}\sum_{i=1}^{n}y_i (x_{i} - \bar x)'\hat\beta_{-i}.
		\end{align}
		In general, this estimator does not have an interpretation in terms of degrees-of-freedom corrections. Instead, the explanatory power of the linear model is assessed using the empirical covariance between leave-one-out predictions $(x_{i} - \bar x)'\hat \beta_{-i}$ and the left out observation $y_i$.
	\end{example}

	\begin{example}[Analysis of covariance]\th\label{ex:ANOVA}
		\textit{}
		
		Since the work of \citet{fisher1925statistical}, it has been common to summarize the effects of experimentally assigned treatments on outcomes with estimates of variance components. Consider a dataset comprised of observations on $N$ groups with $T_{g}$ observations in the ${g}$-th group. The ``analysis of covariance'' model posits that outcomes can be written
		\begin{align}\label{eq:panel}
		y_{{g} t} =\alpha_{{g}} + x_{{g} t}'\delta + \varepsilon_{{g} t} && ({g}=1,\dots,N, \ t = 1,\dots,T_{g} \ge 2),
		\end{align}
		where $\alpha_g$ is a group effect and $x_{gt}$ is a vector of strictly exogenous covariates.

		A prominent example comes from \cite{chetty2011does} who study the adult earnings $y_{gt}$ of $n = \sum_{{g}=1}^N T_{g}$ students assigned experimentally to one of $N$ different classrooms. Each student also has a vector of predetermined background characteristics $x_{gt}$. The variability in student outcomes attributable to classrooms can be written: 
		\begin{align}
		\sigma_{\alpha}^{2}=\frac{1}{n}\sum_{{g}=1}^{N} T_{g} \left(\alpha_{g}-\bar{\alpha}\right)^{2}
		\end{align}
		where $\bar{\alpha}=\frac{1}{n}\sum_{{g}=1}^{N} T_{g} \alpha_{{g}}$ gives the (enrollment-weighted) mean classroom effect. 
		
		This model and object of interest can written in the notation of the preceding section ($y_i = x_i'\beta + \varepsilon_{i}$ and $\sigma_{\alpha}^{2} = \beta'A\beta$) by letting $i = i(g,t)$ where $i(\cdot,\cdot)$ is bijective with inverse denoted $(g(\cdot),t(\cdot))$, $y_i= y_{{g} t}$, $\varepsilon_{i} = \varepsilon_{{g} t}$,
		%if, for each $(g,t)$, we
		\begin{align}
		x_i &= (d_i',x_{{g} t}')', 
		\quad
		\beta = (\alpha',\delta')', \quad \alpha = (\alpha_1,\dots,\alpha_N)',
		\quad
		d_i = (\mathbf{1}_{\{{g}=1\}},\dots,\mathbf{1}_{\{{g}=N\}})',
		\shortintertext{and}
		A &= \mbox{\scriptsize $\begin{bmatrix} 		A_{dd} & 0 \\ 0 & 0  	\end{bmatrix} $}
		\quad \text{where} \quad 
		A_{dd}=\frac{1}{n}\sum_{i=1}^{n} (d_{i} - \bar d) (d_{i} - \bar d)', \quad \bar d=\frac{1}{n} \sum_{i=1}^{n} d_{i}.
		\end{align}
		\cite{chetty2011does} estimate $\sigma_{\alpha}^{2}$ using a random effects ANOVA estimator \cite[see e.g.,][]{searle2009variance} which is of the homoscedasticity-only type given in \eqref{eq:HO}. As discussed in Section \ref{sec:unbiased} and Appendix \ref{app:rel}, this estimator is in general first order biased when the errors are heteroscedastic and group sizes are unbalanced.

		\noindent \textbf{\emph{Special Case: No Common Regressors}} 
		When there are no common regressors ($x_{{g} t}=0$ for all ${g},t$), the leave-out estimator of $\sigma_{\alpha}^{2}$ has a particularly simple representation:
		\begin{align}\label{eq:AP}
		\hat{\sigma}_{\alpha}^{2} 
		&= \dfrac{1}{n}\sum_{g=1}^{N} \left( T_g \left(\hat \alpha_{g}-\hat{\bar{\alpha}}\right)^{2} - \left(1-\frac{T_g}{n}\right) \hat \sigma_g^2 \right)
		\quad \text{for} \quad \hat \sigma_g^2 = \frac{1}{T_g-1} \sum_{t=1}^{T_g} (y_{gt} - \hat \alpha_g)^2,
		\end{align}
		where $\hat \alpha_{g}= \frac{1}{T_g}\sum_{t=1}^{T_g} y_{gt}$, and $\hat{\bar{\alpha}}= \frac{1}{n}\sum_{{g}=1}^{N} T_{g} \hat \alpha_{{g}}$. This representation shows that if the model consists only of group specific intercepts, then the leave-out estimator relies on group level degrees-of-freedom corrections.  The statistic in \eqref{eq:AP} was analyzed by \cite{akritas2004heteroscedastic} in the context of testing the null hypothesis that $\sigma_\alpha^2=0$ while allowing for heteroscedasticity at the group level.
		
		\noindent \textbf{\emph{Covariance Representation}} 
		Another instructive representation of the leave-out estimator is in terms of the empirical covariance
		\begin{align}
			\hat{\sigma}_{\alpha}^{2} &= \sum_{i=1}^{n} y_i \tilde{d}_i'\hat \alpha_{-i}
			\quad \text{where} \quad \hat \beta_{-i} = (\hat \alpha_{-i}',\hat \delta_{-i}').
		\end{align}
		The generalized regressor $\tilde{d}_i$ can be described as follows: if there are no common regressors then $\tilde{d}_i = \frac{1}{n} (d_i - \bar d)$, which is analogous to \th\ref{ex:R2}. If the model includes common regressors then $\tilde d_i  = \frac{1}{n} \left( (d_i - \bar d) - \hat \varGamma'(x_{g(i)t(i)} - \bar x_{g(i)}) \right)$ where $\bar x_g = \frac{1}{T_{g}} \sum_{t=1}^{T_g} x_{gt}$ and $\hat \varGamma$ is the coefficient vector from an instrumental variables (IV) regression of $d_i - \bar d$ on $x_{g(i)t(i)} - \bar x_{g(i)}$ using $x_{g(i)t(i)}$ as an instrument. The IV residual $\tilde{d}_i$ is uncorrelated with $x_{{g}(i) t(i)}$ and the covariance between ${d_i}$ and $\tilde d_i$ is $A_{dd}$, which ensures that the empirical covariance between $y_i = d_i'\alpha + x_{{g}(i) t(i)}'\delta + \varepsilon_i$ and the generalized prediction $\tilde{d_i}'\hat \alpha_{-i}$ is an unbiased estimator of ${\sigma}_{\alpha}^{2}$.
				
%, because $d_i$ is uncorrelated with $x_{g(i)t(i)} - \bar x_{g(i)}$,
		
	\end{example}

	\begin{example}[Random coefficients]\th\label{ex:RC}
		\textit{}
		
		Group memberships are often modeled as influencing slopes in addition to intercepts \citep{kuh1959validity,hildreth1968some,raudenbush2002hierarchical,arellano2011identifying,graham2012identification,graham2016quantile}. Consider the following ``random coefficient'' model:
		\begin{align}\label{random_coeff}
		y_{{g} t} = \alpha_{g} + z_{{g} t}\gamma_{g} + \varepsilon_{{g} t} && ({g}=1,\dots,N, \ t=1,\dots,T_{g} \ge 3).
		\end{align}
		
		An influential example comes from \cite{raudenbush1986hierarchical}, who model student mathematics
		scores as a ``hierarchical'' linear function of socioeconomic status (SES) with school-specific intercepts $(\alpha_{g} \in \R)$ and slopes $(\gamma_{g} \in \R)$. Letting  $\bar{\gamma}=\frac{1}{n}\sum_{{g}=1}^{N} T_{g} \gamma_{{g}}$ for $n = \sum_{{g}=1}^N T_{g}$, the student-weighted variance of slopes can be written:
		\begin{align}
		\sigma_{\gamma}^{2}=\frac{1}{n}\sum_{{g}=1}^N T_{g}\left(\gamma_{{g}}-\bar{\gamma}\right)^{2}.
		\end{align}
		In the notation of the preceding section we can write $y_i = x_i'\beta + \varepsilon_{i}$ and $\sigma_\gamma^2 = \beta'A\beta$ where
		\begin{align}
		x_i = (d_i', d_i' z_{{g} t})', \qquad \beta = (\alpha',\gamma')', \qquad \gamma = (\gamma_1,\dots,\gamma_N)', 
		\qquad A = \mbox{\scriptsize $\begin{bmatrix} 		A_{dd} & 0 \\ 0 & 0  	\end{bmatrix} $}
		\end{align}
		for $y_i$, $\varepsilon_i$, $d_i$, $A_{dd}$, and $\alpha$ as in the preceding example.
		
		\cite{raudenbush1986hierarchical} use a maximum likelihood estimator of $\sigma_{\gamma}^{2}$ predicated upon normality and homoscedastic errors. \cite{swamy1970efficient} considers an estimator of $\sigma^{2}_{\gamma}$ that relies on group-level degrees-of-freedom corrections and is unbiased when the error variance is allowed to vary at the group level, but not with the level of $z_{{g} t}$. By contrast, the leave-out estimator is unbiased under arbitrary patterns of heteroscedasticity.
		
		\noindent \textbf{\emph{Covariance Representation} }
		The leave-out estimator can be represented in terms of the empirical covariance
		\begin{align}
			\hat \sigma_{\gamma}^{2}=\sum_{i=1}^{n} y_i \tilde z_i \tilde d_i'\hat \gamma_{-i}
			\quad \text{where} \quad \tilde d_i = \frac{1}{n}(d_i - \bar d), \quad \tilde z_i = \frac{z_{g(i)t(i)}-\bar z_{g(i)}}{\sum_{t=1}^{T_{g(i)}}( z_{g(i)t}-\bar z_{g(i)} )^2},
		\end{align}
		and $\bar z_{g} = \frac{1}{T_g}\sum_{t=1}^{T_{g}} z_{gt}$. Demeaning $z_{g(i)t(i)}$ at the group level makes $\tilde d_i \tilde z_i$ uncorrelated with $d_i$ and scaling by the group variability in $z_{g(i)t}$ ensures that the covariance between $\tilde d_i \tilde z_i$ and $d_i z_{g(i)t(i)}$ is $A_{dd}$. This implies that the empirical covariance between $y_i = d_i'\alpha +  z_{g(i)t(i)} d_i'\gamma + \varepsilon_i$ and the generalized prediction $\tilde z_i \tilde d_i'\hat \gamma_{-i}$ is an unbiased estimator of $\sigma_{\gamma}^{2}$.
	\end{example}

	\begin{example}[Two-way fixed effects]\th\label{ex:AKM}
		\textit{}
		
		Economists often study settings where units possess two or more group memberships, some of which can change over time. A prominent example comes from \cite{abowd1999high} (henceforth AKM) who propose a panel model of log wage determination that is additive in worker and firm fixed effects. This so-called ``two-way'' fixed effects model takes the form:
		\begin{align}\label{eq:AKM}
		y_{{g} t} =  \alpha_{{g}} + \psi_{j({g},t)} + x_{{g} t}'\delta +  \varepsilon_{{g} t}  && ({g}=1,\dots,N, \ t=1,\dots,T_{g} \ge 2)
		\end{align}
		where the function $j(\cdot,\cdot):\{ 1,\dots,N\}\times\{1,\dots, \max_g T_{g}\} \rightarrow \{ 0,\dots,J\} $ allocates each of $n = \sum_{{g} =1}^N T_{g}$ person-year observations to one of $J+1$ firms. Here $\alpha_{{g}}$ is a ``person effect'', $\psi_{j({g},t)}$ is a ``firm effect'', $x_{{g} t}$ is a time-varying covariate, and $\varepsilon_{{g} t}$ is a time-varying error. In this context, the mean zero assumption on the errors $\varepsilon_{{g} t}$ can be thought of as requiring both the common covariates $x_{gt}$ and the firm assignments $j(\cdot,\cdot)$ to obey a strict exogeneity condition. 
		
		Interest in such models often centers on understanding how much of the variability in log wages is attributable to firms \citep[see, e.g.,][]{card2013workplace,song2015firming}. AKM summarize the firm contribution to wage inequality via the following two parameters:  
		\begin{align}
		\sigma_{\psi}^{2}=\frac{1}{n}\sum_{{g}=1}^{N}\sum_{t=1}^{T_{g}}\left(\psi_{j\left({g},t\right)}-\bar{\psi}\right)^{2}
		\quad \text{and} \quad 
		\sigma_{\alpha,\psi}=\frac{1}{n}\sum_{{g}=1}^{N}\sum_{t=1}^{T_{g}}\left(\psi_{j\left({g},t\right)}-\bar{\psi}\right)\alpha_{{g}}
		\end{align}
		where $\bar \psi = \frac{1}{n}\sum_{{g}=1}^{N} \sum_{t=1}^{T_{g}} \psi_{j({g},t)}$.	The variance component $\sigma_{\psi}^{2}$ measures the contribution of firm wage variability to inequality, while the covariance component $\sigma_{\alpha,\psi}$ measures the additional contribution of systematic sorting of high wage workers to high wage firms.

		To represent this model and the corresponding objects of interest in the notation of the preceding section ($y_i = x_i'\beta + \varepsilon_{i}$, $\sigma_{\psi}^{2}= \beta'A_\psi \beta$, and $\sigma_{\alpha,\psi}= \beta'A_{\alpha,\psi}\beta$), let 
		\begin{align}
		x_i = (d_i',f_i',x_{{g} t}')', \ \beta = (\alpha',\psi',\delta')', \ \alpha = (\alpha_1,\dots,\alpha_N)' + \mathbf{1}_N'\psi_0, \ \psi = (\psi_1\,\dots,\psi_J)' - \mathbf{1}_J'\psi_0,		
		\end{align}
		for $y_i$, $\varepsilon_i$, and $d_i$ as in the preceding examples, $f_i = (\mathbf{1}_{\{j({g},t)=1\}},\dots,\mathbf{1}_{\{j({g},t)=J\}})',$
		\begin{align}
		A_\psi
		&= \mbox{\scriptsize $\begin{bmatrix} 0 & 0 & 0 \\ 0 & A_{ff} & 0 \\ 0 & 0 & 0 \end{bmatrix}$ }
		\quad \text{where} \quad 
		A_{ff}=\frac{1}{n}\sum_{i=1}^{n} (f_{i}-\bar f)(f_{i} - \bar f)', \quad
		\bar f=\frac{1}{n}\sum_{i=1}^{n} f_{i} ,
		\shortintertext{and}
		A_{\alpha,\psi}
		&=\frac{1}{2}  \mbox{\scriptsize $\begin{bmatrix} 0 & A_{df} & 0 \\ A_{df}' & 0 & 0 \\ 0 & 0 & 0\end{bmatrix}$}
		\quad \text{where} \quad
		A_{df}=\frac{1}{n}\sum_{i=1}^{n} (d_{i} - \bar d)(f_{i} - \bar f)'.
		\end{align}
		Computation of the Johnson-Lindenstrauss approximation can be facilitated using the representations $A_\psi = A_f'A_f$ and $A_{\alpha,\psi} = \frac{1}{2}(A_d'A_f + A_f'A_d)$ where
		\begin{align}
			A_f' = \mbox{\scriptsize $ \frac{1}{\sqrt{n}}\begin{bmatrix} 0 & 0 & 0 \\ f_1 - \bar f & \dots & f_n-\bar f \\ 0 & 0 & 0 \end{bmatrix}$} \quad \text{and} \quad A_d' = \mbox{\scriptsize$\frac{1}{\sqrt{n}}\begin{bmatrix}  d_1 - \bar d & \dots & d_n-\bar d \\ 0 & 0 & 0 \\ 0 & 0 & 0 \end{bmatrix}$ }.
		\end{align}
		
		Addition and subtraction of $\psi_0$ in $\beta$ amounts to the normalization, $\psi_0 = 0$, which has no effect on the variance components of interest. As \cite{abowd1999high,abowd2002computing} note, least squares estimation of \eqref{eq:AKM} requires one normalization of the $\psi$ vector within each set of firms connected by worker mobility. 
		For simplicity, we assume all firms are connected so that only a single normalization is required.\footnote{\cite{bonhomme2019distributional} study a closely related model where workers and firms each belong to one of a finite number of types and each pairing of worker and firm type is allowed a different mean wage. These mean wage parameters are shown to be identified
		when each worker type moves between each firm type with positive probability, enabling estimation even when many firms are not connected.}
		
		%\footnote{\cite{bonhomme2019distributional} study a related model where workers and firms each belong to one of a finite number of types. They show that the parameters of this model are identified
		%when each worker \emph{type} moves between each firm \emph{type} with positive probability.}

		%Note, however, that this condition is not strictly weaker than
		%the AKM model's connectedness requirement. A limiting example arises when, as in the AKM model, each worker and firm is allowed its
		%own type. In this example, \cite{bonhomme2019distributional}'s identification condition requires that each worker be at risk of moving between each firm, which
%		is not necessary for identification in the AKM model.
		
%		becomes more restrictive than the AKM connectivity requirement as each worker must have
		
%		thereby circumventing the need to normalize firm wage effects within each connected set of firms. The AKM model, by contrast, treats each worker and firm as its own type but imposes (log-)additive separability between worker and firm heterogeneity.

		\noindent \textbf{\emph{Covariance Representation}} 
		\cite{abowd1999high} estimated $\sigma_{\psi}^{2}$ and $\sigma_{\alpha,\psi}$ using the naive plug-in estimators $\hat{\beta}'A_{\psi}\hat{\beta}$ and $ \hat{\beta}'A_{\alpha,\psi}\hat{\beta}$ which are, in general, biased. \cite{andrews2008high} proposed the ``homoscedasticity-only'' estimators of \eqref{eq:HO}. These estimators are unbiased when the errors $\varepsilon_i$ are independent and have common variance. \cite{bonhomme2019distributional} propose a two-step estimation approach that is consistent in the presence of heteroscedasticity when the support of firm wage effects is restricted to a finite number of values and each firm grows large with the total sample size $n$. Our leave-out estimators, which avoid both the homoscedasticity requirement on the errors and any cardinality restrictions on the support of the firm wage effects, can be written compactly as covariances taking the form
		\begin{align}\label{eq:varest}
		\hat{\sigma}_{\psi}^{2}  =\sum_{i=1}^{n} y_{i} x_{i}' S_{xx}^{-1} A_{\psi} \hat{\beta}_{-i}, 
		\qquad 
		\hat{\sigma}_{\alpha,\psi}  =\sum_{i=1}^{n} y_{i} x_{i}' S_{xx}^{-1} A_{\alpha,\psi} \hat{\beta}_{-i}.
		\end{align}
	Notably, these estimators are unbiased whenever the leave out estimator $\hat \beta_{-i}$ can be computed, regardless of the distribution of firm sizes.
		
		\noindent \textbf{\emph{Special Case: Two time periods}}
		A simpler representation of $\hat{\sigma}_{\psi}^{2}$ is available in the case where only two time periods are available and no common regressors are present ($T_{g} =2$ and $x_{{g} t} =0$ for all ${g},t$). Consider this model in first differences
		\begin{align}
		\label{fd_model}
		\Delta y_{{g}} = \Delta f_{{g}}' \psi + \Delta \varepsilon_{{g}}
		&&
		({g}=1,\dots,N)
		\end{align}
		where $\Delta y_{{g}}=y_{{g} 2}-y_{{g} 1}$, $\Delta \varepsilon_{{g}} = \varepsilon_{{g} 2} - \varepsilon_{{g} 1}$, and 
		$\Delta f_{{g}} = f_{i(g,2)} -  f_{i(g,1)}$. The leave-out estimator of $\sigma_{\psi}^2$ applied to this differenced representation of the model is:
		%estimator $\hat{\sigma}_{\psi}^{2}$ equals the
		\begin{align}
		\hat \sigma_{\psi}^2 =\sum_{{g} =1}^N \Delta y_{{g} }  \Delta \tilde f_g'\hat{\psi}_{-{g}} 
		\quad \text{where} \quad
		\Delta \tilde f_g = A_{ff} S_{\Delta f\Delta f}^{-1} \Delta f_g.
		\end{align}		
		Note that the quantities $S_{\Delta f\Delta f}$ and $\hat{\psi}_{-{g}}$ correspond respectively to $S_{xx}$ and $\hat{\beta}_{-i}$ in the first differenced model. 
		
		\begin{rem}\th\label{rem:nocluster}The leave-out representation above reveals that $\hat{\sigma}_{\psi}^{2}$ is not only unbiased under arbitrary heteroscedasticity and design unbalance, but also under arbitrary correlation between $\varepsilon_{{g} 1}$ and $\varepsilon_{{g} 2}$. The same can be shown to hold for $\hat{\sigma}_{\alpha,\psi}$. Furthermore, this representation highlights that $\hat{\sigma}_{\psi}^{2}$ only depends upon observations with $\Delta f_{{g} } \neq 0$ (i.e., firm ``movers'').
		\end{rem}
		
	\end{example}

	\section{Inference on Quadratic Forms of Fixed Rank}\label{sec:InfLow}
	
	While the previous section emphasized variance components where the rank $r$ of $A$ was increasing with sample size, we first study the case where $r$ is fixed. Problems of this nature often arise when testing a few linear restrictions or when conducting inference on linear combinations of the regression coefficients, say $v'\beta$. In the case of two-way fixed effects models of wage determination, the quantity $v'\beta$ might correspond to the difference in mean values of firm effects between male and female workers \citep{card2015bargaining} or to the coefficient from a projection of firm effects onto firm size \citep{bloom2018disappearing}.\fxnote{does bloom actually run such a regression?} A third use case, discussed at length by \cite{cattaneo2017inference}, is where $v'\beta$ corresponds to a linear combination of a few common coefficients in a linear model with high dimensional fixed effects that are regarded as nuisance parameters.

	To characterize the limit distribution of $\hat \theta$ when $r$ is small, we rely on a representation of $\theta$ as a weighted sum of squared linear combinations of the data: $\hat \theta = \sum_{\ell = 1}^{r} \lambda_\ell\left(\hat b_\ell^2-\hat{\mathbb{V}}[\hat b_\ell ]\right)$ where 
	\begin{align}
		\hat b &= \sum_{i=1}^n w_i y_i
		\quad \text{and} \quad 
		\hat{\mathbb{V}}[\hat b] = \sum_{i=1}^n  w_{i} w_{i}' \hat \sigma_{i}^2
	\end{align}
	for $w_i = (w_{i1},\dots,w_{ir})'= Q' S_{xx}^{-1/2} x_i$.  The following theorem characterizes the asymptotic distribution of $\hat \theta$ while providing conditions under which $\hat b$ is asymptotically normal and $\hat{\mathbb{V}}[\hat b]$ is consistent.
	
	%The distribution of $\theta$ is characterized by
	
	\begin{thm}\th\label{thm2}				
		If \th\ref{ass:reg} holds, $r$ is fixed, and $\max_i  w_i'w_i = o(1)$, then 
		\begin{enumerate}
			\item $\mathbb{V}[\hat b]^{-1/2}(\hat b - b) \xrightarrow{d} \mathcal{N}\left( 0 , I_r \right)$ where $b = Q' S_{xx}^{1/2} \beta$,
			\item ${\mathbb{V}}[\hat b]\inverse \hat{\mathbb{V}}[\hat b] \xrightarrow{p} I_r$,
			\item $\hat \theta = \sum_{\ell = 1}^{r} \lambda_\ell\left(\hat b_\ell^2-\mathbb{V}[\hat b_\ell]\right) + o_p(\mathbb{V}[\hat \theta]^{1/2})$,
		\end{enumerate}
	\end{thm}
	
	The high-level requirement of this theorem that $\max_i w_i'w_i = o(1)$ is a Lindeberg condition ensuring that no observation is too influential. One can think of $\max_i w_i'w_i$ as measuring the inverse effective sample size available for estimating $b$: when the weights are equal across $i$, the equality $\sum_{i=1}^n w_i w_i' = I_r$ implies that $w_{i\ell}^2 = \frac{1}{n}$. Since $\frac{1}{n} \sum_{i=1}^n w_i'w_i = \frac{r}{n}$, the requirement that  $\max_i w_i'w_i = o(1)$ is implied by a variety of primitive conditions that limit how far a maximum is from the average \cite[see, e.g.,][Appendix A.1]{anatolyev2012inference}. Note that \th\ref{thm2} does not apply to settings where $r$ is proportional to $n$ because $\max_i w_i'w_i \ge \frac{r}{n}$.

	In the special case where $A = vv'$ for some non-random vector $v$, \th\ref{thm2} establishes that the variance estimator $\hat{\mathbb{V}}[\hat \beta] = S_{xx}\inverse \left(\sum_{i=1}^n x_i x_i' \hat \sigma_i^2\right)S_{xx}\inverse$ enables consistent inference on the linear combination $v'\beta$ using the approximation
	\begin{align}
	\frac{v'(\hat \beta - \beta)}{\sqrt{v'\hat{\mathbb{V}}[\hat \beta]v}} \xrightarrow{d} \mathcal{N}(0,1). \label{lincom}
	\end{align}
	To derive this result we assumed that $\max_i P_{ii} \le c$ for some $c<1$, whereas standard Eicker-White variance estimators generally require that $\max_i P_{ii} \rightarrow 0$ and \cite{cattaneo2017inference} establish an asymptotically valid approach to inference in settings where $\max_i P_{ii} \le 1/2$. Thus $\hat{\mathbb{V}}[\hat \beta]$ leads to valid inference under weaker conditions than existing versions of Eicker-White variance estimators.

	\begin{rem}
		\th\ref{thm2} extends classical results on hypothesis testing of a few linear restrictions, say, $H_0 : R\beta=0$, to allow for many regressors and heteroscedasticity. A convenient choice of $A$ for testing purposes is $\frac{1}{r} R' (R S_{xx}\inverse R')\inverse R$ where $r$, the rank of $R \in \R^{r \times k}$, is fixed. Under $H_0$, the asymptotic distribution of $\hat \theta$ is an equally weighted sum of $r$ central $\chi^2$ random variables. This distribution is known up to $\mathbb{V}[\hat b ]$ and a critical value can be found through simulation. For a recent contribution to this literature, see \cite{anatolyev2012inference} who allows for many regressors but considers the special case of homoscedastic errors.
	\end{rem}

	\section{Inference on Quadratic Forms of Growing Rank}\label{sec:dist}

	We now turn to the more challenging problem of conducting inference on $\theta$ when $r$ increases with $n$, as in the examples discussed in Section \ref{sec:examples}. These results also enable tests of many linear restrictions. For example, in a model of gender-specific firm effects of the sort considered by \cite{card2015bargaining}, testing the hypothesis that men and women face identical sets of firm fixed effects entails as many equality restrictions as there are firms.

	\subsection{Limit Distribution}

	In order to describe the result we introduce $\check x_i = \sum_{\ell=1}^n M_{i\ell} \frac{B_{\ell\ell}}{1-P_{\ell \ell}} x_\ell$ where $M_{i\ell} = \mathbf{1}_{\{i=\ell\}} - x_i S_{xx}\inverse x_\ell$. Note that $\check x_i $ gives the residual from a regression of $\frac{B_{ii}}{1-P_{ii}} x_i$ on $x_i$. Therefore, $\check x_i =0$ when the regressor design is balanced. The contribution of $\check x_i$ to the behavior of $\hat \theta$ is through the estimation of $\sum_{i=1}^n B_{ii} \sigma_i^2$, which can be ignored in the case where the rank of $A$ is bounded.  When the rank of $A$ is large, as implied by condition $(ii)$ of \th\ref{thm3} below, this estimation error can resurface in the asymptotic distribution. One can think of the eigenvalue ratio in $(ii)$ as the inverse effective rank of $\tilde A$: when all the eigenvalues are equal $\frac{ \lambda_1^2}{\sum_{\ell=1}^r \lambda_\ell^2} = \frac{1}{r}$.
	\begin{thm}\th\label{thm3}
		Recall that $\tilde x_i = A S_{xx}\inverse x_i$ where $\hat \theta = \sum_{i=1}^n y_i \tilde x_i'\hat \beta_{-i}$. If \th\ref{ass:reg} holds and the following conditions are satisfied
		\begin{align}
		(i) \ \mathbb{V}[\hat \theta]\inverse \max_i \left( (\tilde x_i'\beta)^2 + (\check x_i'\beta)^2 \right) = o(1), \quad (ii) \ \frac{ \lambda_1^2}{\sum_{\ell=1}^r \lambda_\ell^2} = o(1),
		\end{align}
		then $\mathbb{V}[\hat \theta]^{-1/2} (\hat \theta - \theta ) \xrightarrow{d} \mathcal{N}(0,1)$.
	\end{thm}
	The proof of \th\ref{thm3} relies on a variation of Stein's method developed in \cite{soelvsten2017robust} and a representation of $\hat \theta$ as a second order U-statistic, i.e.,
	\begin{align}\label{eq:Ustat}
	\hat \theta = \sum_{i=1}^n \sum_{\ell \neq i} C_{i\ell} y_i y_\ell
	\end{align}
	where $C_{i\ell} = B_{i\ell} - 2\inverse M_{i\ell}\left(M_{ii}\inverse B_{ii} + M_{\ell\ell}\inverse B_{\ell\ell}\right)$ and $B_{i\ell} = x_i' S_{xx}\inverse A S_{xx}\inverse x_\ell$. The proof shows that the ``kernel'' $C_{i \ell}$ varies with $n$ in such a way that $\hat \theta$ is asymptotically normal whether or not $\hat \theta$ is a degenerate U-statistic (i.e., whether or not $\beta$ is zero).

	One representation of the variance appearing in \th\ref{thm3} is
	\begin{align}
	\mathbb{V}[\hat \theta] &= \sum_{i=1}^n \left( 2\tilde x_i'\beta - \check x_i'\beta \right)^2 \sigma_i^2 + 2 \sum_{i=1}^n \sum_{\ell \neq i} C_{i\ell}^2 \sigma_i^2 \sigma_\ell^2.
	\end{align}
	Note that this variance is bounded from below by $\min_i \sigma_i^2 \sum_{i=1}^n ( 2\tilde x_i'\beta)^2 + (\check x_i'\beta)^2$ since $\sum_{i=1}^n \tilde x_i'\beta \check x_i'\beta = 0$. Therefore $(i)$ will be satisfied whenever $\max_i \left( (\tilde x_i'\beta)^2 + (\check x_i'\beta)^2 \right)$ is not too large compared to $\sum_{i=1}^n (\tilde x_i'\beta)^2 + (\check x_i'\beta)^2$. As in \th\ref{thm2}, $(i)$ is implied by a variety of primitive conditions that limit how far a maximum is from the average, but since $(i)$ involves a one dimensional function of $x_i$ it can also be satisfied when $r$ is large. A particularly simple case where $(i)$ is satisfied is when $\beta=0$; further cases are discussed in Section \ref{sec:verify}. 
	
	\begin{rem}\th\label{rem:testbig}
		\th\ref{thm3} can be used to test a large system of linear restrictions of the form $H_0: R\beta=0$ where $r \rightarrow \infty$ is the rank of $R \in \R^{r \times k}$. Under this null hypothesis, choosing $A=\frac{1}{r} R' (R S_{xx}\inverse R')\inverse R$ implies ${\mathbb{V}}[\hat \theta]^{-1/2} \hat \theta \xrightarrow{d} \mathcal{N}(0,1)$ since all the non-zero eigenvalues of $\tilde A$ are equal to $\frac{1}{r}$. The existing literature allows for either heteroscedastic errors and moderately few regressors \citep[][$k^3/n \rightarrow 0$]{donald2003empirical} or homoscedastic errors and many regressors \citep[][$k/n \le c <1$]{anatolyev2012inference}. When coupled with the estimator of ${\mathbb{V}}[\hat \theta]$ presented in the next subsection, this result enables tests with heteroscedastic errors and many regressors.
	\end{rem}
	
	\begin{rem}
		\th\ref{thm3} extends some common results in the literature on many and many weak instruments \cite[see, e.g.,][]{chao2012asymptotic} where the estimators are asymptotically equivalent to quadratic forms. The structure of that setting is such that $\tilde A = I_r/r$ and $r \rightarrow \infty$, in which case condition $(ii)$ of \th\ref{thm3} is automatically satisfied. 
		%and therefore does not feature prominently in that literature. 
	\end{rem}

	\subsection{Variance Estimation}
	\label{sec:variance_estimation}

	In order to conduct inference based on the normal approximation in \th\ref{thm3} we now propose an estimator of $\mathbb{V}[\hat \theta]$.	The U-statistic representation of $\hat \theta$ in \eqref{eq:Ustat} implies that the variance of $\hat \theta$ is
	\begin{align}
	\mathbb{V}[\hat \theta] = 4\sum_{i=1}^n \left(\sum_{\ell \neq i} C_{i\ell} x_\ell'\beta \right)^2 \sigma_i^2 + 2\sum_{i=1}^n \sum_{\ell \neq i} C_{i\ell}^2 \sigma_i^2 \sigma_\ell^2.
	\end{align}
	Naively replacing $\{x_i'\beta,\sigma_i^2\}_{i=1}^n$ with $\{y_i,\hat \sigma_i^2\}_{i=1}^n$ in the above formula to form a plug-in estimator of $\mathbb{V}[\hat \theta]$ will, in general, lead to invalid inferences as $\hat \sigma_i^2 \hat \sigma_\ell^2$ is a biased estimator of $\sigma_i^2 \sigma_\ell^2$. For this reason, we consider estimators of the error variances that rely on leaving out more than one observation. Since this approach places additional restrictions on the design, Appendix \ref{sec:cons} describes a simple adjustment which leads to conservative inference in settings where these restrictions do not hold.
	
	%	While $\mathbb{V}[\hat \theta]$ can, in principle, be estimated without bias using leave-\textit{three}-out estimators, this approach will be computationally intractable in many settings. Moreover, in the case of Example 4, it is likely that discarding particular triples of observations will lead the mobility network to become disconnected, making it impossible to compute estimates of $\beta$. Instead we propose a sample splitting approach to inference that is computationally tractable and delivers consistent inference under restrictions on the covariates that are straightforward to verify.
	
	%Suppose that we can construct two linear unbiased estimators $\widehat{x_i'\beta}_{-i,1}$ and $\widehat{x_i'\beta}_{-i,2}$ of $x_i'\beta$ that are independent of one another and of $y_i$. 
	\noindent \textbf{\emph{Sample Splitting}} Our specific proposal is an estimator that exploits two independent unbiased estimators of $x_i'\beta$ that are also independent of $y_i$. We denote these estimators $\widehat{x_i'\beta}_{-i,s} = \sum_{\ell \neq i}^n P_{i\ell,s} y_\ell$ for $s = 1,2$, where $P_{i\ell,s}$ does not (functionally) depend on the $\{y_i\}_{i=1}^n$. To ensure independence between $\widehat{x_i'\beta}_{-i,1}$ and $\widehat{x_i'\beta}_{-i,2}$, we require that $P_{i\ell,1} P_{i\ell,2} =0$ for all $\ell$. Employing these split sample estimators, we create a new set of unbiased estimators for $\sigma_i^2$:
	\begin{align}
	\tilde \sigma_{i}^2 = \left(y_i - \widehat{x_i'\beta}_{-i,1}\right)\left(y_i - \widehat{x_i'\beta}_{-i,2}\right) 
	\quad \text{and} \quad 
	\hat \sigma_{i,-\ell}^2 = \begin{cases}
	y_i(y_i - \widehat{x_i'\beta}_{-i,1}), & \text{if } P_{i\ell,1}= 0, \\
	y_i(y_i - \widehat{x_i'\beta}_{-i,2}), & \text{if } P_{i\ell,1}\neq 0,
	\end{cases}
	\end{align}
	where $\hat \sigma_{i,-\ell}^2 $ is independent of $y_\ell$ and $\tilde \sigma_{i}^2$ is a cross-fit estimator of the form considered in \cite{newey2018cross}. These cross-fit estimators can be used to construct an estimator of $\sigma_i^2 \sigma_\ell^2$ that, under certain design conditions, will be unbiased.
	 Letting $P_{im,-\ell} = P_{im,1}1_{\{P_{i\ell,1}= 0\}} + P_{im,2}1_{\{P_{i\ell,1} \neq 0\}}$ denote the weight observation $m$ receives in $\hat \sigma_{i,-\ell}^2$ and $\tilde C_{i\ell} = C_{i\ell}^2 + 2\sum_{m =1 }^n C_{mi} C_{m\ell} (P_{mi,1} P_{m\ell,2} + P_{mi,2} P_{m\ell,1})$, we define
	\begin{align}
		\widehat{\sigma_i^2 \sigma_\ell^2} &= 
		\begin{cases}
		\hat \sigma_{i,-\ell}^2 \cdot \hat \sigma_{\ell,-i}^2, & \text{if }  P_{im,-\ell}P_{\ell m,-i}=0 \text{ for all } m, \\
		\tilde \sigma_{i}^2 \cdot \hat \sigma_{\ell,-i}^2,  & \text{else if } P_{i\ell,1} + P_{i\ell,2}=0, \\
		\hat \sigma_{i,-\ell}^2 \cdot	\tilde \sigma_{\ell}^2,  & \text{else if } P_{\ell i,1} + P_{\ell i,2}=0, \\
		\hat \sigma_{i,-\ell}^2 \cdot (y_\ell - \bar y)^2 \cdot 1_{\{\tilde C_{i\ell} < 0\}}, & \text{otherwise.}
		\end{cases}
	\end{align}
	The first three cases in the above definition correspond respectively to pairs where (i) $\hat \sigma_{i,-\ell}^2$ and $\hat \sigma_{\ell,-i}^2$ are independent, (ii) $\widehat{x_i'\beta}_{-i,1}$ and $\widehat{x_i'\beta}_{-i,2}$ are independent of $y_\ell$, and (iii) $\widehat{x_\ell'\beta}_{-\ell,1}$ and $\widehat{x_\ell'\beta}_{-\ell,2}$ are independent of $y_i$. When any of these three cases apply, we obtain an unbiased estimator of $\sigma_i^2 \sigma_\ell^2$. For the remaining set of pairs $\mathcal{B} = \{(i,\ell) : P_{im,-\ell}P_{\ell m,-i}\neq 0 \text{ for some } m, \  P_{i\ell,1} + P_{i\ell,2} \neq 0, \ P_{\ell i,1} + P_{\ell i,2} \neq 0 \}$ that comprise the fourth case we rely on an unconditional variance estimator which leads to a biased estimator of $\sigma_i^2 \sigma_\ell^2$ and conservative inference.

	\noindent \textbf{\emph{Design Requirements}}
	Constructing the above split sample estimators places additional requirements on the design matrix $S_{xx}$.  We briefly discuss these requirements in the context of \th\ref{ex:ANOVA,ex:RC,ex:AKM}. In the ANOVA setup of \th\ref{ex:ANOVA}, leave-one-out estimation requires a minimum group size of two, whereas existence of $\{\widehat{x_i'\beta}_{-i,s}\}_{s=1,2}$ requires groups sizes of at least three. Conservative inference can be avoided (i.e., the set $\mathcal{B}$ will be empty) when the minimum group size is at least four. In the random coefficients model of \th\ref{ex:RC}, minimum group sizes of three and five are sufficient to ensure feasibility of leave-one-out estimation and existence of $\{\widehat{x_i'\beta}_{-i,s}\}_{s=1,2}$, respectively. Conservativeness can be avoided with a minimum group size of seven. 
	
	In the first differenced two-way fixed effects model of \th\ref{ex:AKM}, the predictions $\{\widehat{x_i'\beta}_{-i,s}\}_{s=1,2}$ are associated with particular paths in the worker-firm mobility network and independence requires that these paths be edge-disjoint. Menger's theorem \citep{menger1927allgemeinen} implies that $\{\widehat{x_i'\beta}_{-i,s}\}_{s=1,2}$ exists if the design matrix has full rank when any two observations are dropped. Menger's theorem also implies that conservativeness can be avoided if the design matrix has full rank when any three observations are dropped. In our application, we use Dijkstra's algorithm to find the paths that generate $\{\widehat{x_i'\beta}_{-i,s}\}_{s=1,2}$ (see Appendix \ref{sec:samplesplitalg} for further details).

	\noindent \textbf{\emph{Consistency}}
	The following lemma shows that $\widehat{\sigma_i^2 \sigma_\ell^2} $ can be utilized to construct an estimator of $\mathbb{V}[\hat \theta]$ that delivers consistent inference when sufficiently few pairs fall into $\mathcal{B}$ and provides conservative inference otherwise. 
	
	\begin{lem}\th\label{lem:varCS}
		For $s =1,2$, suppose that $\widehat{x_i'\beta}_{-i,s}$ satisfies (unbiasedness) $\sum_{\ell \neq i}^n P_{i\ell,s} x_\ell'\beta=x_i'\beta$, (sample splitting) $P_{i\ell,1} P_{i\ell,2} =0$ for all $\ell$, and (projection property) $\lambda_{\max}(P_sP_s') = O(1)$ where $P_s = (P_{i\ell,s})_{i,\ell}$ is the hat-matrix corresponding to $\widehat{x_i'\beta}_{-i,s}$. Let
		\begin{align}
		\hat{\mathbb{V}}[\hat \theta] &= 4\sum_{i=1}^n \left( \sum_{\ell \neq i} C_{i\ell} y_\ell  \right)^2 \tilde \sigma_{i}^2 - 2 \sum_{i=1}^n \sum_{\ell \neq i} \tilde C_{i\ell} \widehat{\sigma_i^2 \sigma_\ell^2}.
		\end{align}
		\begin{enumerate}
			\item If the conditions of \th\ref{thm3} hold and %${\mathbb{V}}[\hat \theta] \inverse \sum_{(i,\ell) \in \mathcal{B}} \abs{\tilde C_{i\ell}} = o(1)$
			$\abs{\mathcal{B}} = O(1)$, then 
			$\frac{\hat \theta - \theta}{\hat {\mathbb{V}}[\hat \theta]^{1/2}} \xrightarrow{d} \mathcal{N}(0,1).$
			\item If the conditions of \th\ref{thm3} hold, then $\liminf_{n \rightarrow \infty}\Pr\left( \theta \in \left[ \hat \theta \pm z_{\alpha} \hat{\mathbb{V}}[\hat \theta]^{1/2}  \right] \right) \ge 1-\alpha$ where $z_{\alpha}^2$ denotes the $(1-\alpha)$'th quantile of a central $\chi^2_{1}$ random variable.
		\end{enumerate}
		
	\end{lem}
	
	In the formula for $\hat{\mathbb{V}}[\hat \theta]$, the first term can be seen as a plug-in estimator and standard results for quartic forms imply that the expectation of this term is ${\mathbb{V}}[\hat \theta] + 2 \sum_{i=1}^n \sum_{\ell \neq i} \tilde C_{i\ell} {\sigma_i^2 \sigma_\ell^2}$. Hence, the second term is a bias correction which completely removes the bias when $\mathcal{B} = \emptyset$ and leaves a positive bias otherwise. In Appendix \ref{sec:cons} we establish validity of an adjustment to $\hat{\mathbb{V}}[\hat \theta]$ that utilizes an upward biased unconditional variance estimator for observations where it is not possible to construct $\{\widehat{x_i'\beta}_{-i,s}\}_{s=1,2}$. 
	
	\begin{rem}\th\label{rem:std}
		The purpose of the condition $\abs{\mathcal{B}} = O(1)$ in the above lemma is to ensure that the bias of $\hat{\mathbb{V}}[\hat \theta]$ grows small with the sample size. Because the bias of $\hat{\mathbb{V}}[\hat \theta]$ is non-negative, inference based on $\hat{\mathbb{V}}[\hat \theta]$ remains valid even when this condition fails, as stated in the second part of \th\ref{lem:varCS}. In practice, it may be useful for researchers to calculate the fraction of pairs that belong to $\mathcal{B}$ to gauge the extent to which inference might be conservative. Similarly, it may be useful to compute the share of observations where it is not possible to construct $\{\widehat{x_i'\beta}_{-i,s}\}_{s=1,2}$ to investigate whether upward bias in the standard error could lead to power concerns. %In our application to worker-firm wage effects we compute these quantities and find $\abs{\mathcal{B}}$ to be vanishing small compared to ${n}\choose{2}$ in all samples considered and that our algorithm always produces the two predictions $\{\widehat{x_i'\beta}_{-i,s}\}_{s=1,2}$ on a leave-two-out connected set of firms.
	\end{rem}

%	\begin{rem}\label{rem:std}
%		The purpose of the condition ${\mathbb{V}}[\hat \theta] \inverse \sum_{(i,\ell) \in \mathcal{B}} \abs{\tilde C_{i\ell}} = o(1)$ in the above lemma is to ensure that the bias of $\hat{\mathbb{V}}[\hat \theta]$ grows small with the sample size. Because the bias of $\hat{\mathbb{V}}[\hat \theta]$ is non-negative, inference based on $\hat{\mathbb{V}}[\hat \theta]$ remains valid even when this condition fails as stated in the second part of the Lemma. In practice, it may be useful for researchers to calculate the fraction of pairs that belong to $\mathcal{B} $ or an empirical analogue of $ {\mathbb{V}}[\hat \theta] \inverse \sum_{(i,\ell) \in \mathcal{B}} \abs{\tilde C_{i\ell}} $ to gauge the extent to which inference might be conservative. In our application to worker-firm wage effects we compute these quantities and find both to be vanishing small in a leave-two-out connected set of firms.
%	\end{rem}

	\section{Weakly Identified Quadratic Forms of Growing Rank}\label{sec:weak}
	
	In some settings where $r$ grows with the sample size, condition (ii) of \th\ref{thm3} may not apply. For example in two-way fixed effects models, it is possible that ``bottlenecks'' arise in the mobility network that lead the largest eigenvalues to dominate the others.
	
	 This section provides a theorem which covers the case where some of the squared eigenvalues $\lambda_1^2,\dots,\lambda_r^2$ are large relative to their sum $\sum_{\ell=1}^r \lambda_\ell^2$. To motivate this assumption, note that each eigenvalue of $\tilde A$ measures how strongly $\theta$ depends on a particular linear combination of the elements of $\beta$ relative to the difficulty of estimating that combination (as summarized by $S_{xx}^{-1}$). From \th\ref{lem:cons}, $\text{trace}(\tilde A^2) = \sum_{\ell=1}^r \lambda_\ell^2$ governs the total variability in $\hat \theta$. Therefore, \th\ref{thm4} covers the case where $\theta$ depends strongly on a few linear combinations of $\beta$ that are imprecisely estimated relative to the overall sampling uncertainty in $\hat \theta$. The following assumption formalizes this setting.
		\begin{assumption}\th\label{ass:eig}
			There exist a $c >0$  and a known and fixed $q \in \{1,\dots,r-1\}$ such that 
			\begin{align}
				\frac{ \lambda_{q+1}^2}{\sum_{\ell=1}^r \lambda_\ell^2} = o(1) \quad \text{and} \quad \frac{\lambda_{q}^2}{\sum_{\ell=1}^r \lambda_\ell^2} \ge c \quad \text{for all } n.
			\end{align}
		\end{assumption}
		
		\th\ref{ass:eig} defines $q$ as the number of squared eigenvalues that are large relative to their sum. Equivalently, $q$ indexes the number of nuisance parameters in $b$ that are \emph{weakly identified} relative to their influence on $\theta$ and the uncertainty in $\hat \theta$. The assumption that $q$ is known is motivated by our discussion of Examples \ref{ex:R2}--\ref{ex:AKM} in Section \ref{sec:verify} and the theoretical literature on weak identification, which typically makes an ex-ante distinction between strongly and weakly identified parameters \citep[e.g.,][]{andrews2012estimation}. In Section \ref{sec:getq} we offer some guidance on choosing $q$ in settings where it is unknown.

		\subsection{Limit Distribution}

		Given knowledge of $q$, we can split $\hat \theta$ into a known function of $\hat {\mathsf{b}}_q= (\hat b_1,\dots,\hat b_{q})'$ and $\hat \theta_q$ where $\hat b_1,\dots,\hat b_{q}$ are OLS estimators of the weakly identified nuisance parameters:
		\begin{align}
		\hat {\mathsf{b}}_q & = \sum_{i=1}^n \mathsf{w}_{iq} y_i, & \mathsf{w}_{iq} &= ( w_{i1},\dots, w_{iq})', \\
		\hat \theta_q &= \hat \theta - \sum_{\ell=1}^{q} \lambda_\ell (\hat b_\ell^2 - \hat{\mathbb{V}}[\hat b_\ell] ), & \hat{\mathbb{V}}[\hat b] &= \sum_{i=1}^n w_{i} w_{i}' \hat \sigma_{i}^2.
		\end{align}

		The main difficulty in proving the following \th\nameref{thm4} is to show that the joint distribution of $(\hat{\mathsf{b}}_q',\hat \theta_q)'$ is normal, which we do using the same variation of Stein's method that was employed for \th\ref{thm3}. The high-level conditions involve $\tilde x_{iq}$ and $\check x_{iq}$ which are the parts of $\tilde x_{i}$ and $\check x_{i}$ that pertain to $\hat \theta_q$ and are defined in the proof of \th\ref{thm4}.

	\begin{thm}\th\label{thm4}
		If $\max_i  \mathsf{w}_{iq}'\mathsf{w}_{iq} = o(1)$, $\mathbb{V}[\hat \theta_q]\inverse \max_i \left( (\tilde x_{iq}'\beta)^2 + (\check x_{iq}'\beta)^2 \right) = o(1)$, and \th\ref{ass:reg,ass:eig} hold, then
		\begin{enumerate}
			\item $\mathbb{V}[(\hat{\mathsf{b}}_q',\hat \theta_q)']^{-1/2}
			\left(
			(\hat{\mathsf{b}}_q',\hat \theta_q)'
			- \E[(\hat{\mathsf{b}}_q',\hat \theta_q)']
			\right)
			\xrightarrow{d} \mathcal{N}\left(0, I_{q+1} \right)$
			\item \label{eq:thm1} $\hat \theta = \sum_{\ell = 1}^{q} \lambda_\ell\left(\hat b_\ell^2-\mathbb{V}[\hat b_{\ell}]\right) + \hat \theta_q + o_p(\mathbb{V}[\hat \theta]^{1/2})$
		\end{enumerate}
	\end{thm}
	
	\th\ref{thm4} provides an approximation to $\hat \theta$ in terms of a quadratic function of $q$ asymptotically normal random variables and a linear function of one asymptotically normal random variable. Here, the non-centralities $\E[\hat {\mathsf{b}}_q] = (b_1,\dots,b_q)'$ serve as nuisance parameters that influence both $\theta$ and the shape of the limiting distribution of $\hat \theta - \theta$. The next section proposes an approach to dealing with these nuisance parameters that provides asymptotically valid inference on $\theta$ for any value of $q$.

	\subsection{Variance Estimation}
	\label{sec:varq}
	
	In \th\ref{thm4} the relevant variance is $\varSigma_q := \mathbb{V}[(\hat{\mathsf{b}}_q',\hat \theta_q)']$,
	\begin{align}
	\varSigma_q &= \sum_{i=1}^n \begin{bmatrix}
	\mathsf{w}_{iq} \mathsf{w}_{iq}' \sigma_i^2 & 2\mathsf{w}_{iq}\left(\sum_{\ell \neq i} C_{i\ell q} x_\ell'\beta \right) \sigma_i^2 \\
	2 \mathsf{w}_{iq}'\left(\sum_{\ell \neq i} C_{i\ell q} x_\ell'\beta \right) \sigma_i^2  & 4 \left(\sum_{\ell \neq i} C_{i\ell q} x_\ell'\beta \right)^2 \sigma_i^2 + 2 \sum_{\ell \neq i} C_{i\ell q}^2 \sigma_i^2 \sigma_\ell^2
	\end{bmatrix},
	\end{align}
	where $C_{i\ell q}$ is defined in the proof of \th\ref{thm4}. Our estimator of this variance reuses the split sample estimators introduced for \th\ref{thm3}:
	\begin{align}
	\hat\varSigma_q &= \sum_{i=1}^n \begin{bmatrix}
	\mathsf{w}_{iq} \mathsf{w}_{iq}' \hat \sigma_i^2 & 2 \mathsf{w}_{iq}\left(\sum_{\ell \neq i} C_{i\ell q} y_\ell \right) \tilde \sigma_i^2 \\
	2 \mathsf{w}_{iq}'\left(\sum_{\ell \neq i} C_{i\ell q} y_\ell \right) \tilde \sigma_i^2  & 4 \left(\sum_{\ell \neq i} C_{i\ell q} y_\ell \right)^2 \tilde \sigma_i^2 - 2 \sum_{\ell \neq i} \tilde C_{i\ell q}^2 \widetilde{\sigma_i^2 \sigma_\ell^2}
	\end{bmatrix}
	\end{align}
	where $\tilde C_{i\ell q}$ and $\widetilde{\sigma_i^2 \sigma_\ell^2}$ are defined in the proof of the next lemma which shows consistency of this variance estimator.
	
	\begin{lem}\th\label{lem:4}
		For $s =1,2$, suppose that $\widehat{x_i'\beta}_{-i,s}$ satisfies $\sum_{\ell \neq i}^n P_{i\ell,s} x_\ell'\beta=x_i'\beta$, $P_{i\ell,1} P_{i\ell,2} =0$ for all $\ell$, and $\lambda_{\max}(P_sP_s') = O(1)$. If the conditions of \th\ref{thm4} hold and $\abs{\mathcal{B}} = O(1)$, then $\varSigma_q\inverse
		\hat\varSigma_q \xrightarrow{p} I_{q+1}.$
	\end{lem}

	\begin{rem}
		As in the case of variance estimation for \th\ref{thm3}, it may be that the design does not allow for construction of the predictions $\widehat{x_i'\beta}_{-i,1}$ and $\widehat{x_i'\beta}_{-i,2}$ used in $\hat\varSigma_q$. For such cases, Appendix \ref{sec:cons} proposes an adjustment to $\hat\varSigma_q$ which has a positive definite bias and therefore leads to valid (but conservative) inference when coupled with the inference method discussed in the next section.
	\end{rem}

%	 In small samples the ``rule-of-thumb'' treats eigenvalue ratios above $\frac{1}{10}$ as large, and then it slowly shrinks that threshold to allow for additional parameters in $b$ to be treated as weakly identified when considering a larger sample.
	
%		\begin{rem}
%		An important advantage of providing asymptotic approximations that treat the covariate design as fixed is that one may inspect some of the conditions of \th\ref{thm4} without compromising the validity of subsequent inferences. Conditions that can be inspected are that (a) the rank of $S_{xx}$ is full, (b) $\max_{i} P_{ii} \le c < 1$, (c) $\frac{ \lambda_{q+1}^2}{\sum_{\ell=1}^r \lambda_\ell^2} = o(1)$, and (d) $\max_i \mathsf{w}_{iq}' \mathsf{w}_{iq} = o(1)$. Conditions (a) and (b) ensure that the point estimator is defined and in our empirical application we discuss how one can exclude parameters from the linear model in order to ensure that both conditions are satisfied. Conditions (c) and (d) can be thought of as restricting inverse effective sample sizes, with smaller values being preferable. Our Monte Carlo study suggests good performance of the asymptotic approximations for values in (c) and (d) as high as $\frac{1}{10}$ which is in line with statistics folklore \cite*[see, e.g.,][page 20]{lei2016asymptotics}. %See section \ref{sec:application} for further guidance regarding the choice of $q$.
%	\end{rem}

	\section{Inference with Nuisance Parameters}\label{sec:variance}
	
	 In this section, we develop a two-sided confidence interval for $\theta$ that delivers asymptotic size control conditional on a choice of $q$. Our proposal involves inverting a minimum distance statistic in $\hat{\mathsf{b}}_q$ and $\hat \theta_q$, which \th\ref{thm4} implies are jointly normally distributed. To avoid the conservatism associated with standard projection methods \cite[e.g.,][]{dufour2001finite}, we seek to adjust the critical value downwards to deliver size control on $\theta$ rather than $\E[(\hat{\mathsf{b}}_q', \hat \theta_q)']$. However, unlike in standard projection problems (e.g., the problem of subvector inference), $\theta$ is a nonlinear function of $\E[\hat{\mathsf{b}}_q]$. To accommodate this complication, we use a critical value proposed by \cite{andrews2016geometric} that depends on the curvature of the problem.
	
	\subsection{Inference With Known $q$} 	
	% When $q=0$, this procedure simplifies to a standard two-sided confidence interval based on $\hat \theta$ and asymptotic normality. If $q=1$, the resulting confidence interval has a closed form solution given in Appendix \ref{app:conf}, and for $q>1$, inference relies on solving two quadratic optimization problems that involve $q+1$ unknowns. 

	 %Here we focus on the cases of $q=0$ and $q=1$ and relegate the full description of the case where $q>1$ to Appendix \ref{app:conf}. 

	The confidence interval we consider is based on inversion of a minimum-distance statistic for $(\hat{\mathsf{b}}_q',\hat \theta_q)'$ using the critical value proposed in \cite{andrews2016geometric}. For a specified level of confidence, $1-\alpha$, we consider the interval
	\begin{align}
		\hat C_{\alpha,q}^\theta &= \left[ \min_{(\dot b_1,\dots,\dot b_q,\dot \theta_q)'\in \hat{\mathsf{E}}_{\alpha,q}} \sum_{\ell=1}^{q} \lambda_\ell \dot b_\ell^2 + \dot \theta_q , \max_{(\dot b_1,\dots,\dot b_q,\dot \theta_q)'\in \hat{\mathsf{E}}_{\alpha,q}} \sum_{\ell=1}^{q} \lambda_\ell \dot b_\ell^2 + \dot \theta_q \right]
		\shortintertext{where}
		\hat{\mathsf{E}}_{\alpha,q} &= \left\{ (\mathsf{b}_q',\theta_q)' \in \R^{q+1} : \begin{pmatrix}
		\hat{\mathsf{b}}_q - \mathsf{b}_q \\ \hat \theta_q - \theta_q
		\end{pmatrix}'\hat \varSigma_q \inverse \begin{pmatrix}
		\hat{\mathsf{b}}_q - \mathsf{b}_q \\ \hat \theta_q - \theta_q
		\end{pmatrix} \le z_{\alpha,\hat \kappa_q}^2\right\}.
	\end{align}
	%and $\hat \varSigma_q =\hat{\mathbb{V}}[(\hat{\mathsf{b}}_q',\hat \theta_q)']$ is an estimator of $\varSigma_q = {\mathbb{V}}[(\hat{\mathsf{b}}_q',\hat \theta_q)']$ given in the next subsection. 
	
	The critical value function, $z_{\alpha,\kappa}$, depends on the maximal curvature, $\kappa$, of a certain manifold (exact definitions of $z_{\alpha,\kappa}$ and $\kappa$ are given in Appendix \ref{app:conf}). Heuristically, $\kappa$ can be thought of as summarizing the influence of the nuisance parameter $\E[\hat {\mathsf{b}}_q]$ on the shape of $\hat \theta$'s limiting distribution. Accordingly, $z_{\alpha}^2 := z_{\alpha,0}^2$ is equal to the $(1-\alpha)$'th quantile of a central $\chi^2_{1}$ random variable. As $\kappa \rightarrow \infty$, $z_{\alpha,\kappa}^2$ approaches the $(1-\alpha)$'th quantile of a central $\chi^2_{q+1}$ random variable. This upper limit on $z_{\alpha,\kappa}$ is used in the projection method in its classical form as popularized in econometrics by \cite{dufour2001finite}, while the lower limit $z_{\alpha}$ would yield size control if $\theta$ were linear in $\E[\hat{\mathsf{b}}_q]$.

	When $q=0$, the maximal curvature is zero and $\hat C_{0}^\theta$ simplifies to $[\hat \theta \pm z_{\alpha} \hat{\mathbb{V}}[\hat \theta]^{1/2} ]$. %where the standard error $\hat{\mathbb{V}}[\hat \theta]^{1/2}$ is given in the next subsection.
	When $q=1$, the maximal curvature is $\hat \kappa_1 = \frac{2\abs{\lambda_1}\hat{\mathbb{V}}[\hat b_{1}]}{\hat{\mathbb{V}}[\hat \theta_1]^{1/2}(1-\hat \rho^2)^{1/2}}$ where $\hat \rho$ is the estimated correlation between $\hat b_{1}$ and $\hat \theta_1$. This curvature measure is intimately related to eigenvalue ratios previously introduced, as $\hat \kappa_1^2$ is approximately equal to $\frac{2\lambda_1^2}{\sum_{\ell=2}^{r} \lambda_\ell^2}$ when the error terms are homoscedastic and $\beta=0$. A closed form expression for the $q=1$ confidence interval is provided in Appendix \ref{app:conf}. When $q>1$, inference relies on solving two quadratic optimization problems that involve $q+1$ unknowns, which can be achieved reliably using standard quadratic programming routines.
	
	The following lemma shows that a consistent variance estimator as proposed in \th\ref{lem:4} suffices for asymptotic validity under the conditions of \th\ref{thm4} and Appendix \ref{sec:cons} establishes validity when only a conservative variance estimator is available.
	\begin{lem}\th\label{lem:inf}
		If $\varSigma_q \inverse \hat \varSigma_q \xrightarrow{p} I_{q+1}$ and the conditions of \th\ref{thm4} hold, then 
		\begin{align}
		\liminf_{n \rightarrow \infty} \Pr\left( \theta \in \hat C_{\alpha,q}^\theta \right) \ge 1-\alpha.
		\end{align}
	\end{lem}
	
	The confidence interval studied in \th\ref{lem:inf} constructs a $q+1$ dimensional ellipsoid $\hat{\mathsf{E}}_{\alpha,q}$ and maps it through the quadratic function $(\dot b_1,\dots,\dot b_q,\dot \theta_q) \mapsto \sum_{\ell=1}^{q} \lambda_\ell \dot b_\ell^2 + \dot \theta_q$. This approach ensures uniform coverage over any possible values of the nuisance parameters $b_1,\dots,b_q$ which are imprecisely estimated relative to overall sampling uncertainty in $\hat \theta$.
	
	\begin{rem}
		An alternative to \th\ref{lem:inf} is to conduct inference using a first-order Taylor expansion of $\sum_{\ell=1}^{q} \lambda_\ell \hat b_\ell^2 + \hat \theta_q$. This so-called ``Delta method'' approach is asymptotically equivalent to using the confidence interval $[\hat \theta \pm z_{\alpha} \hat{\mathbb{V}}[\hat \theta]^{1/2} ]$ studied in Section \ref{sec:dist}. However, the Delta method is not uniformly valid in the presence of nuisance parameters as approximate linearity can fail when $\min_{\ell \le q} b_\ell^2 = O(1)$. Section \ref{sec:verify} introduces a stochastic block model with $q=1$ and characterizes $b_1^2$ as the squared difference in average firm effects across two blocks multiplied by the number of between block movers. Thus the Delta method will potentially undercover unless there are strong systematic differences between the two blocks.
	\end{rem}
%	
%	\begin{rem}
%		When the nuisance parameters $\E[\hat{\mathsf{b}}_q]= (b_1,\dots,b_q)'$ are large, i.e., $\min_{\ell \in \{1,\dots,q\}} b_\ell^2 \rightarrow \infty$, it follows from \th\ref{thm4} and the Delta method that $\hat C_{\alpha,0}^\theta = \left[\hat \theta \pm z_{\alpha} \hat{\mathbb{V}}[\hat \theta]^{1/2} \right]$ delivers size control even when $q$ is non-zero. The interval $\hat C_{\alpha,q}^\theta$ will also provide size control, but will tend to be longer (and conservative) as $z_{\alpha,\hat \kappa_q} > z_{\alpha}$.  Note that $b_1,\dots,b_q$ are linear combinations of $\beta$ rescaled so that their estimators $\hat b_1,\dots,\hat b_q$ have a non-vanishing variance. Thus $\min_{\ell \in \{1,\dots,q\}} b_\ell^2 \rightarrow \infty$ will be satisfied if the corresponding unscaled linear combinations are estimated consistently and bounded away from zero. In Section \ref{sec:application} we illustrate that $\hat C_{\alpha,0}^\theta$ can undercover in a setting where $q > 0$ and $\min_{\ell \in \{1,\dots,q\}} b_\ell^2$ is bounded, which serves to illustrate the fragility of the Delta method.
%	\end{rem}\fxnote{this Remark is long and seems to cover multiple points. Can it be shortened? Can the first half be pulled out of the remark into text? Is the last sentence still relevant?}

	\subsection{Choosing $q$}\label{sec:getq}

	The preceding discussion of inference considered a setting where the number of weakly identified parameters was known in advance. In some applications, it may not be clear ex ante what value $q$ takes. In such situations researchers may wish to report confidence intervals for two consecutive values of $q$ (or their union). 
	This heuristic serves to minimize the influence of the specific value of $q$ picked, and both our simulations and empirical application suggest that $\hat C_{\alpha,q}^\theta$ barely varies with $q$ when $\frac{ \lambda_{ q+1}^2}{\sum_{\ell=1}^r \lambda_\ell^2} < \frac{1}{10}$. Consequently, little power is sacrificed by taking the union. 
	
	This observation also suggests a heuristic threshold for choosing $q$; namely, to let $q$ be such that $\frac{ \lambda_{q}^2}{\sum_{\ell=1}^r \lambda_\ell^2} \ge \frac{1}{10}$ and $\frac{ \lambda_{ q+1}^2}{\sum_{\ell=1}^r \lambda_\ell^2} < \frac{1}{10}$, with $q=0$ when $\frac{ \lambda_{1}^2}{\sum_{\ell=1}^r \lambda_\ell^2} < \frac{1}{10}$. %Note that $\frac{ \lambda_{q+1}^2}{\sum_{\ell=1}^r \lambda_\ell^2}$ can be thought of as summarizing inverse effective sample size for the weighted average $\hat \theta_{q}$ of $\hat b_{q +1},\dots,\hat b_{r}$, so we are suggesting 
	A similar threshold rule can be motivated under a slight strengthening of \thref{ass:eig} which allows one to learn $q$ from the data.
	\begin{customass}{2$^\prime$}\th\label{ass:eig2}
		There exist a $c >0$, an $\epsilon >0$, and a fixed $q \in \{1,\dots,r-1\}$ such that 
		\begin{align}
		\frac{ \lambda_{q+1}^2}{\sum_{\ell=1}^r \lambda_\ell^2} = O(r^{-\varepsilon}) \quad \text{and} \quad \frac{\lambda_{q}^2}{\sum_{\ell=1}^r \lambda_\ell^2} \ge c \quad \text{for all } n.
		\end{align}
	\end{customass}
	A threshold based choice of $q$ is the unique $\hat q$ for which
	\begin{align}
	\frac{ \lambda_{\hat q+1}^2}{\sum_{\ell=1}^r \lambda_\ell^2} < c_r 
	\quad \text{and} \quad 
	\frac{\lambda_{\hat q}^2}{\sum_{\ell=1}^r \lambda_\ell^2} \ge c_r
	\quad \text{for some } c_r \rightarrow 0,
	\end{align}
	with $\hat q=0$ when $\frac{ \lambda_{1}^2}{\sum_{\ell=1}^r \lambda_\ell^2} < c_r$.
	Under Assumption 2$^\prime$, $\hat q = q$ in sufficiently large samples provided that $c_r$ is chosen so that $c_r r^{\varepsilon} \rightarrow \infty$. This condition is satisfied when $c_r$ shrinks slowly to zero, e.g., when $c_r \propto 1/\log(r)$.

	\section{Verifying Conditions}\label{sec:verify}
	
	We now revisit the examples of Section 2 and verify the conditions required to apply our theoretical results. Appendix \ref{app:sufficiency} provides further details on these calculations.
	
	\begin{customthm}{1}(Coefficient of determination, continued)
		Recall that $\theta = \sigma_{X\beta}^2 = \beta'A\beta$ where $A = \frac{1}{n}\sum_{i=1}^{n} (x_i-\bar x)(x_i-\bar x)'$ and $\tilde A = \frac{1}{n}(I_k - n S_{xx}^{-1/2} \bar x \bar x' S_{xx}^{-1/2}  )$. Suppose \th\ref{ass:reg} holds.
		
		\noindent \textbf{\emph{Consistency}} Consistency follows from \th\ref{lem:cons} since $\lambda_\ell = \frac{1}{n}$ for $\ell = 1,\dots,r$ where $r = \text{dim}(x_i)-1$. Thus $\text{trace}(\tilde A^2)=r/n^2 \le 1/n= o(1)$.
		
		\noindent \textbf{\emph{Limit Distribution}} If $\text{dim}(x_i)$ is fixed, then $w_i'w_i = P_{ii} - \frac{1}{n}$ and \th\ref{thm2} applies under the standard ``textbook'' condition that $\max_{i} P_{ii} = o(1)$. If $\text{dim}(x_i) \rightarrow \infty$, then \th\ref{thm3} applies if $\mathbb{V}[\hat \theta]\inverse \max_i (\check x_i'\beta )^2 = o(1)$ which follows if, e.g., $\max_{i} \frac{1}{\sqrt{r}}\sum_{\ell=1}^n \abs{M_{i\ell}} = o(1)$ where $M_{i\ell} = \mathbf{1}_{\{i = \ell\}} - x_i' S_{xx}\inverse x_\ell$ (this condition holds in the next two examples). Equality among all eigenvalues excludes the weak identification setting of \th\ref{thm4}. 
		
		\noindent \textbf{\emph{Unbounded Mean Function}} Inspection of the proofs reveal that \th\ref{ass:reg}(iii), $\max_i (x_i'\beta)^2 = O(1)$, can be dropped if the above conditions are strengthened to $\max_{i,\ell} P_{ii} (x_\ell'\beta)^2 = o(1)$ when $\text{dim}(x_i)$ is fixed or $\max_{i,j} \frac{\abs{x_j'\beta}\left( 1+\sum_{\ell=1}^n \abs{M_{i\ell}}\right)}{\sqrt{r}} = o(1)$ when $\text{dim}(x_i) \rightarrow \infty$.
	\end{customthm}
	
	\begin{customthm}{2}(Analysis of covariance, continued) 
		Recall that $\theta = \sigma_{\alpha}^{2}=\frac{1}{n}\sum_{{g}=1}^{N} T_{g} \left(\alpha_{g}-\bar{\alpha}\right)^{2}$ where $y_{gt}= \alpha_g + x_{gt}'\delta + \varepsilon_{gt}$, $g$ index the $N$ groups, and $T_g$ is group size.
		
		\noindent \textbf{\emph{No Common Regressors}} This is a special case of the previous example with $r=N-1$, $P_{ii} = T_{g(i)}\inverse$ and $\check x_i =0$.  Assumption 1(ii),(iii) requires $T_g \ge 2$ and $\max_g \alpha_g^2 = O(1)$. \th\ref{thm2} applies if the number of groups is fixed and $\min_g T_g \rightarrow \infty$, while \th\ref{thm3} applies if the number of groups is large. \th\ref{thm4} cannot apply as all eigenvalues are equal to $\frac{1}{n}$.

		\noindent \textbf{\emph{Common Regressors}} To accommodate common regressors of fixed dimension, assume $\norm{\delta}^2 + \max_{g,t} \norm{x_{gt}}^2 = O(1)$ and that $\frac{1}{n}\sum_{g=1}^N \sum_{t=1}^{T_g} (x_{gt} - \bar x_g)(x_{gt} - \bar x_g)'$ converges to a positive definite limit. This is a standard assumption in basic panel data models \cite[see, e.g.,][Chapter 10]{wooldridge2010econometric}. Allowing such common regressors does not alter the previous conclusions: \th\ref{thm2} applies if $N$ is fixed and $\min_g T_g \rightarrow \infty$ since $w_i'w_i \le P_{ii} = T_{g(i)}\inverse + O(n\inverse)$, \th\ref{thm3} applies if $N \rightarrow \infty$ since $\sum_{\ell=1}^n \abs{M_{i\ell}} = O(1)$, and \th\ref{thm4} cannot apply since $n \lambda_\ell \in [c_1,c_2]$ for $\ell =1,\dots,r$ and some $c_2 \ge c_1 > 0$ not depending on $n$. 
		
		\noindent \textbf{\emph{Unbounded Mean Function}} All conclusions continue to hold if $\max_{g,t} \alpha_g^2 + \norm{x_{gt}}^2 = O(1)$ is replaced with $\frac{\max_{g,t} \alpha_g^2 + \norm{x_{gt}}^2}{ \max \{N,\min_g T_g\}} = o(1)$ and $\sigma_\alpha^2 + \frac{1}{n}\sum_{g=1}^N \sum_{t=1}^{T_g} \norm{x_{gt}}^2 = O(1)$.
	\end{customthm}
		
	\begin{customthm}{3}(Random coefficients, continued)
		For simplicity, consider the \textit{uncentered} second moment $\theta = \frac{1}{n} \sum_{g=1}^N T_g \gamma_g^2$ where $y_{gt}= \alpha_g + z_{gt}'\gamma_g + \varepsilon_{gt}$. Suppose \th\ref{ass:reg} holds and assume that $\max_{g,t} \alpha_g + \gamma_g^2 + z_{gt}^2 =O(1)$ and $\min_g S_{zz,g} \ge c > 0$ where $S_{zz,g}= \sum_{t=1}^{T_g}(z_{gt} - \bar z_{g})^2$. Note that $\min_g S_{zz,g} > 0$ is equivalent to full rank of $S_{xx}$ and $S_{zz,g}$ indexes how precisely $\gamma_g$ can be estimated.
		
		\noindent \textbf{\emph{Consistency}} The $N$ eigenvalues of $\tilde A$ are $\lambda_g = \frac{T_g}{n}S_{zz,g}\inverse$ for $g=1,\dots,N$ where the group indexes are ordered so that $\lambda_1 \ge \dots \ge \lambda_N$. Consistency follows from \th\ref{lem:cons} if $\lambda_1\inverse = n\frac{S_{zz,1}}{T_1}  \rightarrow \infty$. This is automatically satisfied with many groups of bounded size.
		 
		\noindent \textbf{\emph{Limit Distribution}} If $N$ is fixed and $\min_g S_{zz,g} \rightarrow \infty$, then \th\ref{thm2} applies. If $\frac{\sqrt{N}}{T_1} S_{zz,1} \rightarrow \infty$, then \th\ref{thm3} applies. If $\frac{\sqrt{N}}{T_2} S_{zz,2} \rightarrow \infty$, $ \frac{\sqrt{N}}{T_{1}} S_{zz,{1}} = O(1)$, and $S_{zz,1} \rightarrow \infty$, then \th\ref{thm4} applies with $q=1$. In this case, $\gamma_1$ is weakly identified relative to its influence on $\theta$ and the overall variability of $\hat \theta$. This is expressed through the condition $ \frac{\sqrt{N}}{T_{1}} S_{zz,{1}} = O(1)$ where $S_{zz,{1}}$ is the identification strength of $\gamma_1$, $T_1$ provides the influence of $\gamma_1$ on $\theta$ and $1/\sqrt{N}$ indexes the variability of $\hat \theta$.
	\end{customthm}
	
	\begin{customthm}{4}(Two-way fixed effects, continued)
		In this final example, we restrict attention to the first-differenced setting $\Delta y_{{g}} = \Delta f_{{g}}' \psi + \Delta \varepsilon_{{g}}$ with $T_g=2$ and a large number of firms, $J \rightarrow \infty$. Our target parameter is the variance of firm effects $\theta = \sigma_{\psi}^{2} = \frac{1}{n}\sum_{{g}=1}^{N}\sum_{t=1}^{T_{g}} \left(\psi_{j\left({g},t\right)}-\bar{\psi}\right)^{2}$ and we consider \th\ref{ass:reg} satisfied; in particular, $\max_j \abs{\psi_j} = O(1)$. 
		
		\noindent \textbf{\emph{Leverages}} The leverage $P_{gg}$ of observation $g$ is less than one if the origin and destination firms of worker $g$ are connected by a path not involving $g$. Letting $n_g$ denote the number of edges in the shortest such path, one can show that $P_{gg} \le \frac{n_g}{1+n_g}$. Therefore, if $\max_g n_g < 100$ then \th\ref{ass:reg}(ii) is satisfied with $\max P_{gg} \le .99$. In our application we find $\max_g n_g =12$, leading to a somewhat smaller bound on the maximal leverage. The same consideration implies a bound on the model in levels since $P_{i(g,t)i(g,t)} = \frac{1}{2}(1+P_{gg})$.
		
		\noindent \textbf{\emph{Eigenvalues}} The eigenvalues of $\tilde A$ satisfy the equality
		\begin{align}
		\lambda_\ell =\frac{1}{n \dot \lambda_{J+1-\ell}} \qquad \text{for} \quad \ell = 1,\dots,J
		\end{align}
		where $\dot \lambda_1 \ge \dots \ge \dot \lambda_J$ are the non-zero eigenvalues of the matrix $E^{1/2} \mathcal{L} E^{1/2}$. $\mathcal{L}$ is the normalized Laplacian of the employer mobility network and connectedness of the network is equivalent to full rank of $S_{xx}$ (see Appendix \ref{app:sufficiency} for definitions). $E$ is a diagonal matrix of employer specific ``churn rates'', i.e., the number of moves in and out of a firm divided by the total number of employees in the firm. $E$ and $\mathcal{L}$ interact in determining the eigenvalues of $\tilde A$. In \th\ref{ex:RC}, the quantities $\{T_\ell\inverse S_{zz,\ell}\}_{\ell=1}^N$ played a role directly analogous to the churn rates in $E$, so in this example we focus on the role of $\mathcal{L}$ by assuming that the diagonal entries of $E$ are all equal to one.
		
		\noindent \textbf{\emph{Strongly Connected Network}} The employer mobility network is \emph{strongly connected} if $\sqrt{J} \mathcal{C} \rightarrow \infty$ where $\mathcal{C} \in (0,1]$ is Cheeger's constant for the mobility network \cite[see, e.g.,][]{mohar1989isoperimetric,jochmans2016fixed}. Intuitively, $\mathcal{C}$ measures the most severe ``bottleneck'' in the network, where a bottleneck is a set of movers that upon removal from the data splits the mobility network into two disjoint blocks. The severity of the bottleneck is governed by the number of movers removed divided by the smallest number of movers in either of the two disjoint blocks. The inequalities $\dot \lambda_J \ge 1 - \sqrt{1-\mathcal{C}^2}$ \cite[][Theorem 2.3]{chung1997spectral} and ${\lambda_1^2}/{\sum_{\ell=1}^J \lambda_\ell^2} \le 4(\sqrt{J} \dot \lambda_J)^{-2}$ imply that a strongly connected network yields $q=0$, which rules out application of \th\ref{thm4}. Furthermore, a strongly connected network is sufficient (but not necessary) for consistency of $\hat \theta$ as $\sum_{\ell=1}^J \lambda_\ell^2 \le \frac{J}{n} ( \sqrt{n} \dot \lambda_J)^{-2}$.
	
		\noindent \textbf{\emph{Weakly Connected Network}} When $\sqrt{J} \mathcal{C}$ is bounded, the network is \emph{weakly connected} and can contain a sufficiently severe bottleneck that a linear combination of the elements of $\psi$ is estimated imprecisely relative to its influence on $\theta$ and the total uncertainty in $\hat \theta$. The weakly identified linear combination in this case is a difference in average firm effects across the two blocks on either side of the bottleneck, which contributes a $\chi^2$ term to the asymptotic distribution. Below we use a stochastic block model to further illustrate this phenomenon. Our empirical application demonstrates that weakly connected networks can appear in practically relevant settings. 
			
		\noindent \textbf{\textit{Stochastic Block Model}} Consider a stochastic block model of network formation where firms belong to one of two blocks and a set of workers switch firms, possibly by moving between blocks. Workers' mobility decisions are independent: with probability $p_b$ a worker moves between blocks and with probability $1-p_b$ she moves within block. For simplicity, we further assume that the two blocks contain equally many firms and consider a semi-sparse network where $\frac{J\log(J)}{n} + \frac{\log(J)}{np_b} \rightarrow 0$.\footnote{The semi-sparse stochastic block model is routinely employed in the statistical literature on spectral clustering, see, e.g, \cite{sarkar2015role}.} In this model the asymptotic behavior of $\hat \theta$ is governed by $p_b$: the most severe bottleneck is between the two blocks and has a Cheeger's constant proportional to $p_b$. In Appendix  \ref{app:sufficiency}, we use this model to verify the high-level conditions leading to \th\ref{thm3,thm4} and show that \th\ref{thm3} applies when $\sqrt{J} p_b \rightarrow \infty$, while \th\ref{thm4} applies with $q=1$ otherwise. The argument extends to any finite number of blocks, in which case $q$ is the number of blocks minus one. Finally, we show that $\hat \theta$ is consistent even when the network is weakly connected. To establish consistency we only impose $\frac{\log(J)}{np_b} \rightarrow 0$, which requires that the number of movers across the two blocks is large.
	\end{customthm}

	\section{Application}\label{sec:application}

Consider again the problem of estimating variance components in two-way fixed effect models of wage determination. \cite{card2016firms} note that plug-in wage decompositions of the sort introduced by AKM typically attribute $15\%$--$25\%$ of overall wage variance to variability in firm fixed effects. Given the bias and potential sampling
variability associated with plug-in estimates, however, it has been difficult to infer whether firm effects play a differentially important role in certain markets or among particular demographic groups. 

In this section, we use Italian social security records to compute leave-out estimates of 
the AKM wage decomposition and contrast them with estimates based upon the plug-in estimator of \citet{abowd1999high} and the homoscedasticity-corrected estimator of \cite{andrews2008high}. We then investigate whether
the variance components that comprise the AKM decomposition differ across age groups. While it is well known that wage inequality increases with age \citep{mincer1974schooling, lemieux2006increasing}, less is known about the extent to which firm pay premia mediate this phenomenon. Standard wage posting models \citep[e.g., ][]{burdett1998wage} suggest older workers have had more time to climb (and fall off) the job ladder and to receive outside offers \citep{bagger2014tenure}, which may result in more dispersed firm wage premia. But older workers have also had more time to develop professional reputations revealing their relative productivity, which should generate a large increase in the variance of person effects \citep{gibbons1992does,gibbons2005comparative}. The tools developed in this paper allow us to formally study these hypotheses.

\subsection{Sample Construction}

The data used in our analysis come from the Veneto Worker History (VWH) file, which provides the annual earnings 
and days worked associated
with each covered employment spell taking place in the Veneto region of Northeast Italy
over the years 1984-2001. The VWH data have been used in a number of recent studies \citep{card2014rent,bartolucci2018identifying,serafinelli2019good,devicienti2019collective} and are well suited to the analysis of age differences because they provide precise information on dates of birth. These data are also notable for being publicly available, making the costs of replicating our analysis unusually low.\footnote{See \url{http://www.frdb.org/page/data/scheda/inps-data-veneto-workers-histories-vwh/doc_pk/11145} for information on obtaining the VHW.}

Our baseline sample consists of workers with employment spells taking place in the years 1999 and 2001, which provides us with a three year horizon over which to measure job mobility. In Section \ref{sec:bigT_AKM} we analyze a longer unbalanced sample spanning the years 1996--2001 and find that it yields similar results. For each worker-year pair, we retain the unique employment spell yielding the highest earnings in that year. Wages in each year are defined as earnings in the selected spell divided by the spell length in days.  Workers are divided into two groups of roughly equal size according to their year of birth: ``younger'' workers born in the years 1965-1983 (aged 18-34 in 1999) and ``older'' workers born in the years 1937-1964 (aged 35-64 in 1999). Further details on our processing of the VWH records is provided in Appendix \ref{app:data}.

%birth cohorts: 1937-1964 is old, 1965-1983 is young

%born before 1965 is old, before 1965+ is young

Table \ref{table1} reports the number of person-year observations available among
workers employed by firms in the region's largest connected set, along with the largest connected
set for each age group. Workers are classified as ``movers'' if they switch firms between 1999 and 2001. Comparing the number of movers to half the number of person-year observations reveals that roughly $21\%$ of all workers are movers. The movers share rises to $26\%$ among younger workers while only $16\%$ of older workers are movers, reflecting the tendency of mobility rates to decline with age. The average number of movers per connected
firm ranges from nearly 3 in the pooled sample to roughly 2 in the thinner age-specific samples, suggesting that many firms are 
associated with only a single mover.

Our leave-out estimation strategy requires that each firm effect remain estimable
after removing any single observation. The second panel of Table \ref{table1} enforces this requirement by 
restricting to firms that remain connected when any 
mover is dropped (see Appendix \ref{sec:pruning} for computational details).
Pruning the sample in this way drops roughly half of the firms but less than a third of the movers and eliminates roughly 30\% of all workers regardless of their mobility status. These additional restrictions raise mean wages by roughly 5\% and lower the variance of wages by 5--10\% depending on the sample.  

To assess the potential influence of these sample restrictions on our estimands of interest, we construct a third sample that further requires the firm effects to remain estimable after removing any two observations.\footnote{We thank an anonymous referee for this suggestion.} This ``leave-two-out connected set'' is also of theoretical interest because it provides a setting where the requirements for consistency of the variance estimator of Lemma \ref{lem:varCS} appear to be satisfied. On average, the leave-two-out connected sets have roughly half as many firms and 20\% fewer movers than the corresponding leave-one-out sets, and the average number of movers per firm
ranges from approximately 5.6 in the sample of older workers to 4.3 in the sample of younger workers. Restricting the sample in this way further raises mean wages by 3--4\% but yields negligible changes in variance, except among the sample of older workers, which experiences a nearly 7\% \emph{increase} in variance. We investigate below the extent to which these changes in unconditional variances reflect changes in the variance of underlying firm wage effects.

\subsection{AKM Model and Design Diagnostics}
\label{sec:estimates_AKM}
Consider the following simplified version of the AKM model:
\begin{align}
\label{AKM_appli}
y_{{g} t} =  \alpha_{{g}} + \psi_{j({g},t)} +  \varepsilon_{{g} t}.  && ({g}=1,\dots,N, \ t=1,2)
\end{align}
We fit models of this sort to the VHW data after having pre-adjusted log wages for year effects in a first step. This adjustment is obtained by estimating an augmented version of the above model by OLS that includes a dummy control for the year 2001. Hence, $y_{{g} t}$ gives the log wage in year $t$ minus a year 2001 dummy times its estimated coefficient. This two-step approach simplifies computation without compromising consistency because the year effect is estimated at a $\sqrt{N}$ rate.

The bottom of Table \ref{table1} reports for each sample the maximum leverage $(\max_i P_{ii})$ of any person-year observation (Appendix \ref{sec:Lambda} discusses the computation of these leverages). While our pruning procedure ensures $\max_i P_{ii}<1$, it is noteworthy that $\max_i P_{ii}$ is still quite close to one, indicating that certain person-year observations remain influential on the parameter estimates. This finding highlights the inadequacy of asymptotic approximations that require the dimensionality of regressors to grow slower than the sample size, which would lead the maximum leverage to tend to zero.

The asymptotic results of Section \ref{sec:weak} emphasize the importance of not only the maximal leverage, but the number and severity of any bottlenecks in the mobility network. Figure \ref{fig:net} illustrates the leave-two-out connected set for older workers. Each firm is depicted as a dot, with the size of the dot proportional to the total number of workers employed at the firm over the years 1999 and 2001. Dots are connected when a worker moves between the corresponding pair of firms. The figure highlights the two most severe bottlenecks in this network, which divide the firms into three distinct blocks. Each block's firms have been shaded a distinct color. The blue block consists of only five firms, four of which are quite small, which limits its influence on the asymptotic behavior of our estimator. However, the green block has 51 firms with a non-negligible employment share of $9.5\%$. \th\ref{thm4} and the discussion in Section \ref{sec:verify} therefore suggest that the bottleneck between the green and the larger red block will generate weak identification and asymptotic non-normality, predictions we explore in detail below.

\subsection{Variance Decompositions}
Table \ref{table2} reports the results of applying to our samples three estimators of the AKM variance decomposition: the naive plug-in (PI) estimator $\hat \theta_{\text{PI}}$ originally proposed by AKM, the homoscedasticity-only (HO) estimator $\hat \theta_{\text{HO}}$ of \cite{andrews2008high}, and the leave-out (KSS) estimator $\hat \theta$. The PI estimator finds that the variance of firm effects in the pooled leave-one-out connected set accounts for roughly 20\% of the total variance of wages, while among younger workers firm effect variability is found to account for $31\%$ of overall wage variance. Among older workers, variability in firm effects is estimated to account for only $16\%$ of the variance of wages in the leave-one-out connected set.

Are these age differences driven by biases attributable to estimation error? Applying the HO estimator of \cite{andrews2008high} reduces the estimated variances of firm effects by roughly $18\%$ in the age-pooled sample,  $27\%$ in the sample of younger workers, and $16\%$ in the sample of older workers. However, the KSS estimator yields further, comparably sized, reductions in the estimated firm effect variance relative to the HO estimator, indicating the presence of substantial heteroscedasticity in these samples. For instance, in the pooled leave-one-out sample, the KSS estimator finds a variance of firm effects that accounts for only 13\% of the overall variance of wages, while the HO estimator finds that firm effects account for 16\% of wage variance.

Moreover, while the plug-in estimates suggested that the firm effect variance was greater among older than younger workers, the KSS estimator finds the opposite pattern. The KSS estimator also finds that the pooled variance of firm effects exceeds the corresponding variance in either age-specific sample, a sign that mean firm effects differ by age. We explore this between age group component of firm variability in greater depth below.

A potential concern with analyzing the leave-one-out connected set is that worker and firm behavior in this sample may be non-representative of the broader (just-)connected set. To assess this possibility, we also report estimates for the leave-two-out connected set. Remarkably, the KSS estimator finds negligible differences in the variance of firm effects between the leave-one-out and leave-two-out samples for both the pooled sample and the sample of younger workers. Among older workers the estimated firm effect variance falls by about 11\% in the leave-two-out sample, though we show below that this difference may be attributable to sampling variation. The broad similarity between leave-one-out and leave-two-out KSS estimates is likely attributable to the fact that trimmed firms tend to be small and therefore contribute little to the person-year weighted variance of firm effects that has been the focus of the literature. 

PI estimates of person effect variances are much larger than the corresponding estimates of firm effect variance, accounting for $66\%$--$88\%$ of the total variance of wages depending on the sample. The PI estimator also finds that person effects are much more dispersed among older than younger workers, which is in accord with standard models of human capital accumulation and employer learning. The estimated ratio of older to younger person effect variances in the leave-one-out sample is roughly 2.6. Applying the HO estimator reduces the magnitude of the person effect variance among all age groups, but boosts the ratio of older to younger person effect variances to 3.2. The KSS estimator yields further downward corrections to estimated person effect variances, leading the contribution of person effect variability to range from only $50\%$ to $80\%$ of total wage variance. Proportionally, however, the variability of older workers remains stable at 3.2 times that of younger workers. 

PI estimates of the covariance between worker and firm effects are negative in both age-restricted samples, though not in the pooled sample. When converted to correlations, these figures suggest there is mild negative assortative matching of workers to firms. Applying the HO estimator leads the covariances to change sign in both age-specific samples, while generating a mild increase in the estimated covariance of the pooled sample. In all three samples, however, the HO estimates indicate very small correlations between worker and firm effects. By contrast, the KSS estimator finds a rather strong positive correlation of 0.21 among younger workers, 0.27 among older workers, and 0.28 in the pooled leave-one-out sample, indicating the presence of non-trivial positive assortative matching between workers and firms. While the patterns in the leave-two-out sample are broadly similar,  the KSS correlation estimate among older workers is substantially smaller in the leave-two-out than the leave-one-out sample (0.18 vs 0.27).

Finally, we examine the overall fit of the two-way fixed effects model using the coefficient of determination. 
The PI estimator of $R^2$ suggests the two-way fixed effects model explains more than $95\%$ of wage variation in the pooled sample, $91\%$ in the sample of younger workers, and 97\% in the sample of older workers. The HO estimator of $R^2$ is equivalent to the adjusted $R^2$ measure of \cite{theil1961economic}. The adjusted $R^2$ indicates that the two-way fixed effects model explains roughly $90\%$ of the variance of wages in the pooled sample, which is quite close to the figures reported in \cite{card2013workplace} for the German labor market. Applying the KSS estimator yields very minor changes in estimated explanatory power relative to the HO estimates. Interestingly, a sample size weighted average of the age group specific KSS $R^2$ estimates lies slightly below the pooled KSS estimate of $R^2$, which suggests allowing firm effects to differ by age group fails to appreciably improve the model's fit. We examine this hypothesis more carefully in Section \ref{sec:sorting}.

\subsection{Multiple Time Periods and Serial Correlation}
\label{sec:bigT_AKM}
Thus far, our analysis has relied upon panels with only two time periods. Table \ref{table3} reports KSS estimates of the variance of firm effects in an unbalanced panel spanning the years 1996--2001. To analyze this longer panel, we expand our set of time varying covariates to include unrestricted year effects and a third order polynomial in age normalized to have slope zero at age 40 as discussed in \cite{card2016firms}.\footnote{Pre-adjusting for age has negligible effects on the variance decompositions reported in Table \ref{table2} but is quantitatively more important in this longer panel. Age adjustments are particularly pronounced among younger workers who generally exhibit greater wage growth and tend to move rapidly to higher paying firms.} Allowing up to six wage observations per worker yields a substantially larger estimation sample with roughly three times more person-year observations in the age-pooled leave-one-out connected set than was found in Table \ref{table1}. For older workers, who have especially low mobility rates, allowing more time periods raises the number of person-year observations in the leave-one-out connected set by a factor of roughly 5.7 and more than triples the number of firms. 

While these additional observations will tend to reduce the bias in the plug-in estimator, using longer panels may present two distinct sets of complications. First, the equivalence discussed in \th\ref{rem:nocluster} no longer holds, which implies that leaving a single person-year observation out is unlikely to remove the bias in estimates of the variance of firm effects when the errors are serially correlated. Second, pooling many years of data may change the target parameter if firm or person effects ``drift'' with time. The bottom rows of Table \ref{table3} probe for the importance of serial correlation by leaving out ``clusters'' of observations -- as described in \th\ref{rem:cluster} -- defined successively as all observations within the same worker-firm ``match'' and all observations belonging to the same worker; see Appendix \ref{app:cluster} for computational details. Because worker $g$'s person effect is not estimable when leaving that worker's entire wage history out, we estimate a within-transformed specification that eliminates the person effects in a first step.  

Leaving out the match yields an important reduction in the variance of firm effects relative to leaving out a single person-year observation, indicating the presence of substantial serial correlation within match. By contrast, leaving out the worker turns out to have negligible effects on the estimated variance of firm effects, suggesting that serial correlation across-matches is negligible. As expected, pooling several years of data reduces the bias of the PI estimator: the magnitude of the difference between the PI estimates of the variance of firm effects and the leave-worker-out estimates tends to be smaller than the corresponding difference between the PI and KSS estimates of the variance of firm effects reported in Table \ref{table2}.

Remarkably, the firm effect variance estimates that result from leaving out either the match or worker are nearly identical to the KSS estimates reported in Table \ref{table2} for both the age-pooled samples and the samples of younger workers, suggesting the firm effects are relatively stable over this longer horizon. Among older workers, the leave-cluster-out estimates of the variance of firm effects are higher than those reported in Table \ref{table2}, which is unsurprising given that the number of firms under consideration more than tripled in this longer panel. Reassuringly, however, Table \ref{table3} reveals that the KSS estimates of the variance of firm effects among older workers in the leave-one-out and leave-two-out connected sets are very close to one another. The general stability of the KSS estimates of firm effect variances to alternate panel lengths may be attributable to the relatively placid macroeconomic conditions present in Veneto over this period, see the discussion in \cite{devicienti2019collective}.

Our leave-cluster-out exercises suggest researchers seeking to analyze longer panels may be able to avoid biases stemming from serial correlation by simply collapsing the data to match means in a first step and then analyzing these means using the leave-one-observation-out estimator. This two-step approach should substantially reduce computational time while generating only mild efficiency losses due to equal weighting of matches. In what follows, we revert to our baseline sample with exactly two observations per worker.

\subsection{Sorting and Wage Structure}\label{sec:sorting}

The KSS estimates reported in Table \ref{table2} indicate that older workers exhibit somewhat less variable firm effects and a stronger correlation between person and firm effects than younger workers. These findings might reflect lifecycle differences in the sorting of workers to firms or differences in the structure of firm wage effects across the two age groups. 

Table \ref{table4} explores the sorting channel by projecting the pooled firm effects from the leave-one-out sample onto a constant, an indicator for being an older worker, the log of firm size, and the interaction of the indicator with log firm size. Because these projection coefficients are linear combinations of the estimated firm effects, we use the KSS standard errors proposed in equation \eqref{lincom} and analyzed in \th\ref{thm2}. For comparison, we also report a naive standard error that treats the firm effect estimates as independent observations and computes the usual Eicker-White ``robust'' standard errors. In all cases, the KSS standard error is at least twice the corresponding naive standard error and in one case roughly 24 times larger. In light of the consistency results of \th\ref{thm2}, this finding suggests the standard practice of regressing firm effect estimates on observables in a second step without adjusting the standard errors for correlation across firm effects can yield highly misleading inferences. 

The first column of Table \ref{table4} shows that older workers tend to work at firms with higher average firm effects. Evidently older workers do occupy the upper rungs of the job ladder. The second column shows that this sorting relationship is largely mediated by firm size. An older worker at a firm with a single employee is estimated to have a mean firm wage effect 0.16 log points lower than a younger worker at a firm of the same size, an economically insignificant difference that is also revealed to be statistically insignificant when using the KSS standard error. As firm size grows, older workers begin to enjoy somewhat larger firm wage premia. Evaluated at the median firm size of 12 workers, the predicted gap between older and younger workers rises to 0.54 log points, a gap that we can distinguish from zero at the 5\% level using the KSS standard error but is still quite modest. We conclude that the tendency of older workers to be employed at larger firms is a quantitatively important driver of the firm wage premia they enjoy.

Figure \ref{fig:reg} investigates to what extent the firm wage effects differ between age groups. Using the age-restricted leave-one-out connected sets, we obtain a pair of age group specific firm effect estimates $\{\hat \psi_j^Y, \hat \psi_j^O\}_{j\in\mathcal{J}}$ for the set $\mathcal{J}$ of 8,578 firms present in both samples (see Appendix \ref{app:testOY} for details). Figure \ref{fig:reg} plots the person-year weighted averages of $\hat \psi_j^Y$ and $\hat \psi_j^O$ within each centile bin of $\hat \psi_j^O$. A person-year weighted projection of $\hat \psi_j^Y$ onto $\hat \psi_j^O$ yields a slope of only 0.501. To correct this plug-in slope estimate for attenuation bias, we multiply the unadjusted slope by the ratio of the PI estimate of the person-year weighted variance of $\psi_{j}^O$ to the corresponding KSS estimate of this quantity. Remarkably, this exercise yields a projection slope of 0.987, suggesting that, were it not for the estimation error in $\hat \psi_j^O$, the conditional averages depicted in Figure \ref{fig:reg} would be centered around the dashed 45 degree line. Converting this slope into a correlation using the KSS estimate of the person-year weighted variance of $\psi_{j}^Y$ yields a person-year weighted correlation between the two sets of firm effects of 0.89, which indicates the underlying $(\psi_{j}^Y,\psi_{j}^O)$ pairs are tightly clustered around this 45 degree line. 

%The close agreement between the slope and the correlation reveals that the two sets of firm effects also exhibit roughly equal variability.

Theorem \ref{thm3} allows us to formally test the joint null hypothesis that the two sets of firm effects are actually identical, i.e., that both the slope and $R^2$ from a projection of $\psi_j^Y$ onto $\psi_j^O$ are one. We can state this hypothesis as $ H_0: \psi_j^O=\psi_j^Y$ for all $j\in \mathcal{J}.$ Using the test suggested in \th\ref{rem:testbig} we obtain a realized test statistic of $3.95$ which, when compared to the right tail of a standard normal distribution, yields a p-value on $H_0$ of less than $0.1\%$.
Hence, we can decisively reject the null hypothesis that older and younger workers face exactly the same vectors of firm effects. However, our earlier correlation results suggest that $H_0$ nonetheless provides a fairly accurate approximation to the structure of firm effects, at least among those firms that employ movers of both age groups.

\subsection{Inference}

We now study more carefully the problem of inference on the variance of firm effects. For convenience, the top row of Table \ref{table5} reprints our earlier KSS estimates of the variance of firm effects in each sample. Below each estimate of firm effect variance is a corresponding standard error estimate, computed according to the approach described in \th\ref{lem:varCS}. As noted in \th\ref{rem:std}, these standard errors will be somewhat conservative when there is a large share of observations for which no split sample predictions can be created. In the leave-one-out samples this share varies between 15\% and 22\%, indicating that the standard errors are likely upward biased. In the leave-two-out samples, however, this source of bias is not present as the split sample predictions always exist. The standard errors will also tend to be conservative when there is a large share of observation pairs in the set $\mathcal{B}$, for which there is upward bias in the estimator of the error variance product.  However, for both the leave-one-out and leave-two-out samples, this share varies between only 0.03\% and 0.46\%, suggesting only a small degree of upward bias stems from this source.\fxnote{please check}

The next panel of Table \ref{table5} reports the $95\%$ confidence intervals that arise from setting $q=0$, $q=1$, or $q=2$. While the first interval employs a normal approximation, the latter two allow for weak identification by employing non-standard limiting distributions involving linear combinations of normal and $\chi^2$ random variables. We also report estimates of the curvature parameters $(\kappa_1,\kappa_2)$ used to construct the weak identification robust intervals. In the pooled samples both curvature parameters are estimated to be quite small, indicating that a normal approximation is likely to be accurate. Accordingly, setting $q>0$ has little discernible effect on the resulting confidence intervals in these samples. However, among older workers, particularly in the leave-two-out sample, we find stronger curvature coefficients suggesting weak identification may be empirically relevant. Setting $q>0$ in this sample widens the confidence interval somewhat and also changes its shape: mildly shortening the lower tail of the interval but lengthening the upper tail.

Treating the samples of younger and older workers as independent, the fact that the confidence intervals for the two age group samples overlap implies we cannot reject the null hypothesis that the firm effect variances are identical at the $(1-0.95^2)\times100 = 9.75\%$ level. The significance of the 0.23 log point difference between the leave-one-out and leave-two-out estimates of firm effect variance in the sample of older workers turns out to more difficult to assess. By the Cauchy-Schwartz inequality, the covariance between the leave-one-out and leave-two-out estimators is at most $(0.0026)^2(0.0014)^2=3.64\times10^{-6}$. Hence the standard error on the difference between the two estimators is at least 0.0012, which implies a maximal t-statistic of 1.92. Therefore, even when using a normal approximation, we find rather weak evidence against the null that the leave-one-out and leave-two-out estimands are equal. However, because the leave-one-out standard error estimator is likely upward biased, this finding is somewhat less conclusive than would typically be the case.

Theorem \ref{thm4} suggests two important diagnostics for the asymptotic behavior of our estimator are the Lindeberg statistics $\{ \max_i \mathsf{w}_{is}^2\}_{s=1,2} $ and the top eigenvalue shares $\{\lambda_{s}^2/\sum_{\ell=1}^r \lambda_\ell^2 \}_{s = 1,2,3}$. 
The bottom panel of Table \ref{table5} reports these statistics for each sample.  The top eigenvalue shares are fairly small in the pooled sample and among younger workers. A small top eigenvalue share indicates the estimator does not depend strongly on any particular linear combination of firm effects and hence that a normal distribution should provide a suitable approximation to the estimator's asymptotic behavior (i.e. that $q=0$). Accordingly, we find that the confidence intervals are virtually identical for all values of $q$ in both the pooled samples and the two samples of younger workers. 

Among older workers the top eigenvalue share is 31\% in the leave-one-out sample and 58\% in the leave-two-out sample. The next largest eigenvalue share is, in both cases, less than 5\%, which suggests this is a setting where $q=1$. In line with this view, confidence intervals based upon the $q=1$ and $q=2$ approximations are nearly identical in both samples of older workers. The accuracy of these weak-identification robust confidence intervals hinges on the Lindeberg condition of Theorem \ref{thm4} being satisfied. One can think of the Lindeberg statistic $ \max_i  \mathsf{w}_{is}^2 $ as giving an inverse measure of effective sample size available for estimating the linear combination of firm effects associated with the $s$'th largest eigenvalue. The fact that these statistics are all less than or equal to 0.05 implies an effective sample size of at least 20. We study in the Monte Carlo exercises below whether this effective sample size is sufficient to provide accurate coverage. Reassuringly, the sum of squared eigenvalues is quite small in all six samples considered, indicating that the leave out estimator is consistent also in our weakly identified settings.

\subsection{Monte Carlo Experiments}
\label{sec:MC}

We turn now to studying the finite sample behavior of the leave-out estimator of firm effect variance and its associated confidence intervals under a particular data generating process (DGP). Data were generated from the following first differenced model based upon equation (\ref{fd_model}):	
\begin{align}
\Delta y_{g} = \Delta f_{g}' \hat \psi^{scale} + \Delta \varepsilon_{{g}},  && ({g}=1,\dots,N).
\end{align}
Here $\hat \psi^{scale}$ gives the $J\times1$ vector of OLS firm effect estimates found in the pooled leave-one-out sample, rescaled to match the KSS estimate of firm effect variance for that sample. 
%of 0.0240
The errors ${\Delta \varepsilon_{{g}}}$ were drawn independently from a normal distribution with variances given by the following model of heteroscedasticity:
\begin{align}
\mathbb{V}[\Delta \varepsilon_{{g}}] = \exp(a_0 + a_1 B_{gg} + a_2 P_{gg} + a_3 \ln L_{g2} + a_4 \ln L_{g1}),
\end{align}
where $L_{gt}$ gives the size of the firm employing worker $g$ in period $t$. To choose the coefficients of this model, we estimated a nonlinear least squares fit to the ${\hat \sigma_g^2}$ in the pooled leave-one-out sample, which yielded the following estimates:
\begin{align}
\hat a_0=-3.3441, \quad \hat a_1=1.3951, \quad \hat a_2=-0.0037, \quad \hat a_3=-0.0012, \quad \hat a_4=-0.0086.
\end{align}
For each sample, we drew from the above DGP 1,000 times while holding firm assignments fixed at their sample values. 

Table \ref{table6} reports the results of this Monte Carlo experiment. In accord with theory, the KSS estimator of firm effect variances is unbiased while the PI and HO estimators are biased upwards. As expected, the KSS standard error estimator exhibits a modest upward bias in the leave-one-out samples ranging from 15\% in the sample of older workers to 44\% among younger workers. In the leave-two-out sample, however, the standard error estimator exhibits biases of only 6\% or less. Unsurprisingly then, the $q=0$ confidence interval over-covers in both the pooled leave-one-out sample and the leave-one-out sample of younger workers. In the corresponding leave-two-out samples, however, coverage is very near its nominal level, both for the normal based $(q=0)$ and the weak identification robust $(q=1)$ intervals. 

In the samples of older workers, the normal distribution provides a poor approximation to the shape of the estimator's sampling distribution, which is to be expected given the large top eigenvalues found in these designs. This non-normality generates substantial under-coverage by the $q=0$ confidence interval in the leave-two-out sample. Applying the weak identification robust interval in the leave-two-out sample of older workers yields coverage very close to nominal levels despite the fact that the effective sample size available for the top eigenvector is only about $20$.

In sum, the Monte Carlo experiments demonstrate that confidence intervals predicated on the assumption that $q=1$ can provide accurate size control in leave-two-out samples when the realized mobility network exhibits a severe bottleneck. We also achieved size control in leave-one-out samples, albeit at the cost of moderate over-coverage. Hence, in applications where statistical power is a first-order consideration, it may be attractive to restrict attention to leave-two-out samples, which tend to yield estimates of variance components very close to those found in leave-one-out samples but with substantially less biased standard errors.

\section{Conclusion}

We propose a new estimator of quadratic forms with applications to several areas of economics. The estimator is finite sample unbiased in the presence of unrestricted heteroscedasticity and can be accurately approximated in very large datasets via random projection methods. Consistency is established under verifiable design requirements in an environment where the number of regressors may grow in proportion to the sample size. The estimator enables tests of linear restrictions of varying dimension under weaker conditions than have been explored in previous work. A new distributional theory highlights the potential for the proposed estimator to exhibit deviations from normality when some linear combinations of coefficients are imprecisely estimated relative to others. 

In an application to Italian worker-firm data, we showed that ignoring heteroscedasticity can substantially bias conclusions about the relative contribution of workers, firms, and worker-firm sorting to wage inequality. Accounting for serial correlation within a worker-firm match was found to be empirically important, while across match correlation appears to be negligible. Consequently, those studying longer panels may wish to collapse their data down to match level means and then apply the leave-observation-out estimator. Alternately, researchers can simply extract and analyze separately balanced panels of length two, which also facilitates analysis of the temporal stability of the firm and person effect variances.

% Breaking a long panel into shorter sub-panels also enables a ``rolling-AKM'' analysis that accommodates the potential for firm and person effects to drift \citep{card2013workplace}.

Leave-out standard error estimates for the coefficients of a linear projection of firm effects onto worker and firm observables were found to be several times larger than standard errors that naively treat the estimated firm effects as independent. These results strongly suggest that researchers seeking to identify the observable correlates of high-dimensional fixed effects should consider employing the proposed standard errors, including when studying settings falling outside the traditional worker-firm setup \citep[e.g., ][]{finkelstein2016sources,chetty2018impacts}. Stratifying our analysis by birth cohort, we formally rejected the null hypothesis that older and younger workers face identical vectors of firm effects but found that the two sets of firm effects were highly correlated. Corresponding techniques can be used to study multivariate models.

%Older workers tend to be employed at firms offering higher firm wage effects
%but these differences are largely explained by the tendency of older workers to sort to bigger firms. 

A Monte Carlo analysis demonstrated that bottlenecks in the worker-firm mobility network can generate quantitatively important deviations from normality. The proposed inference procedure captured these deviations accurately with a weak identification robust confidence interval. In cases where the mobility network was strongly connected, accurate inferences were obtained with a normal approximation. Our results suggest that in typical worker-firm applications, the normal approximation is likely to suffice. However, when studying small areas, or sub-populations with limited mobility, accounting for weak identification can be quantitatively important.

%%%%%%%%%%%%%%%%%%%%%%%%%%%%%%%%%%%%%%%%%%%%%%%%%%
%%%%%%%%%%%%%%%%%%%%%%%%%%%%%%%%
\bibliographystyle{chicago}
\newpage
%	\nocite{*}
\bibliography{paper}
%%%%%%%%%%%%%%%%%%%%%%%%%
\newpage

\begin{figure}[H]
	\begin{center}
	\caption{Realized Mobility Network: Older workers}
	\label{fig:net}
	\includegraphics[width=\textwidth]{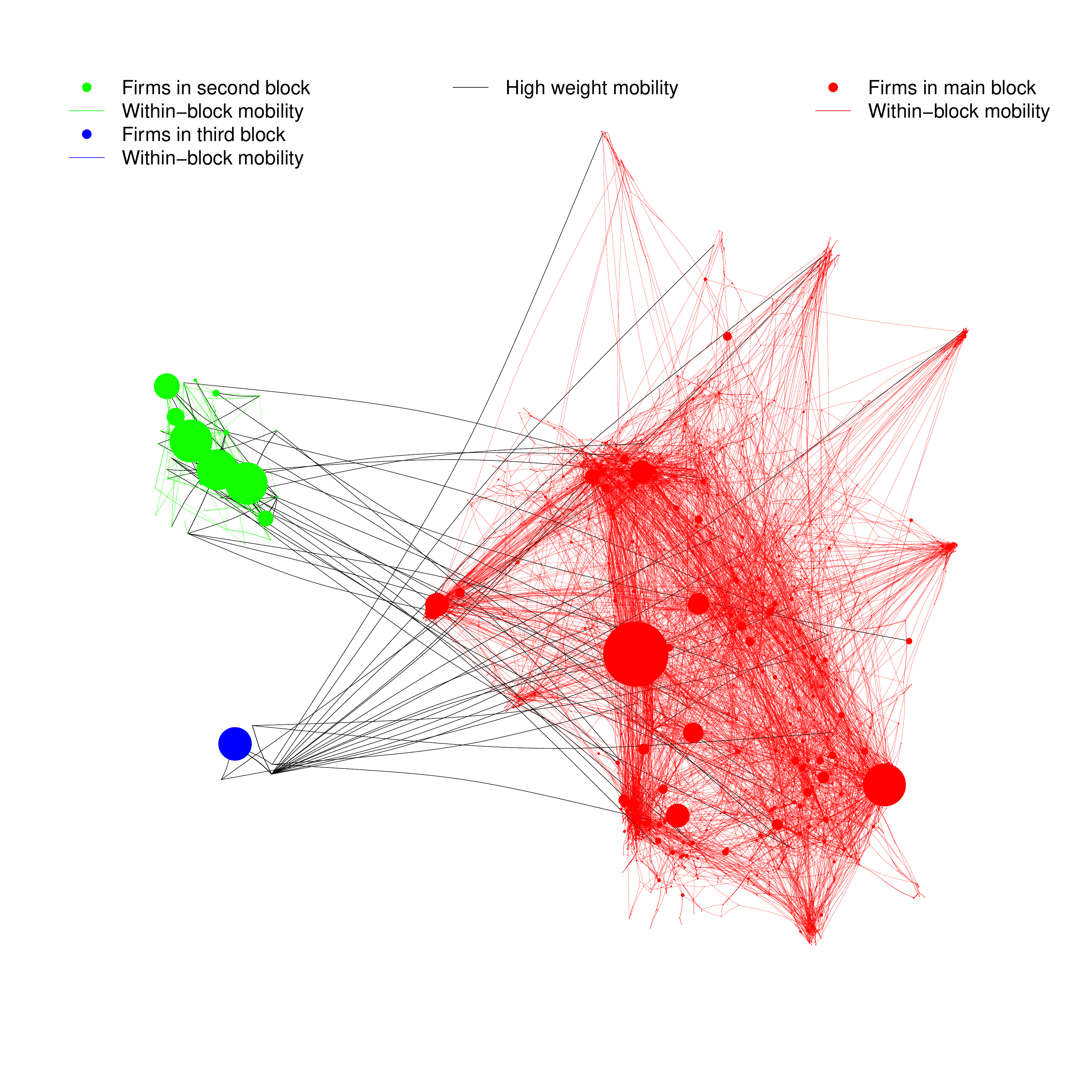}
	\end{center}
	%\begin{singlespace}
		{\footnotesize \underline{Note:} This figure provides a visualization of the design matrix $S_{xx}$ for the leave-two-out sample of older workers (see Table \ref{table1} for reference). The graph is plotted in the statistical software R using the \textit{igraph} package and the large-scale graph layout (DrL) using the option to concentrate firms from the same blocks. High weight mobility refers to observations that have $w_{i1}^2$ or $w_{i2}^2$ above $1/500$ and these observations form the bottlenecks between the three blocks. 
		}
	%\end{singlespace}
\end{figure}

\addcontentsline{toc}{section}{Tables and Figures}

\newpage

\begin{figure}[H]
\begin{center}
	\caption{Do Firm Effects Differ Across Age Groups?}
	\label{fig:reg}
	{\includegraphics[width = \textwidth]{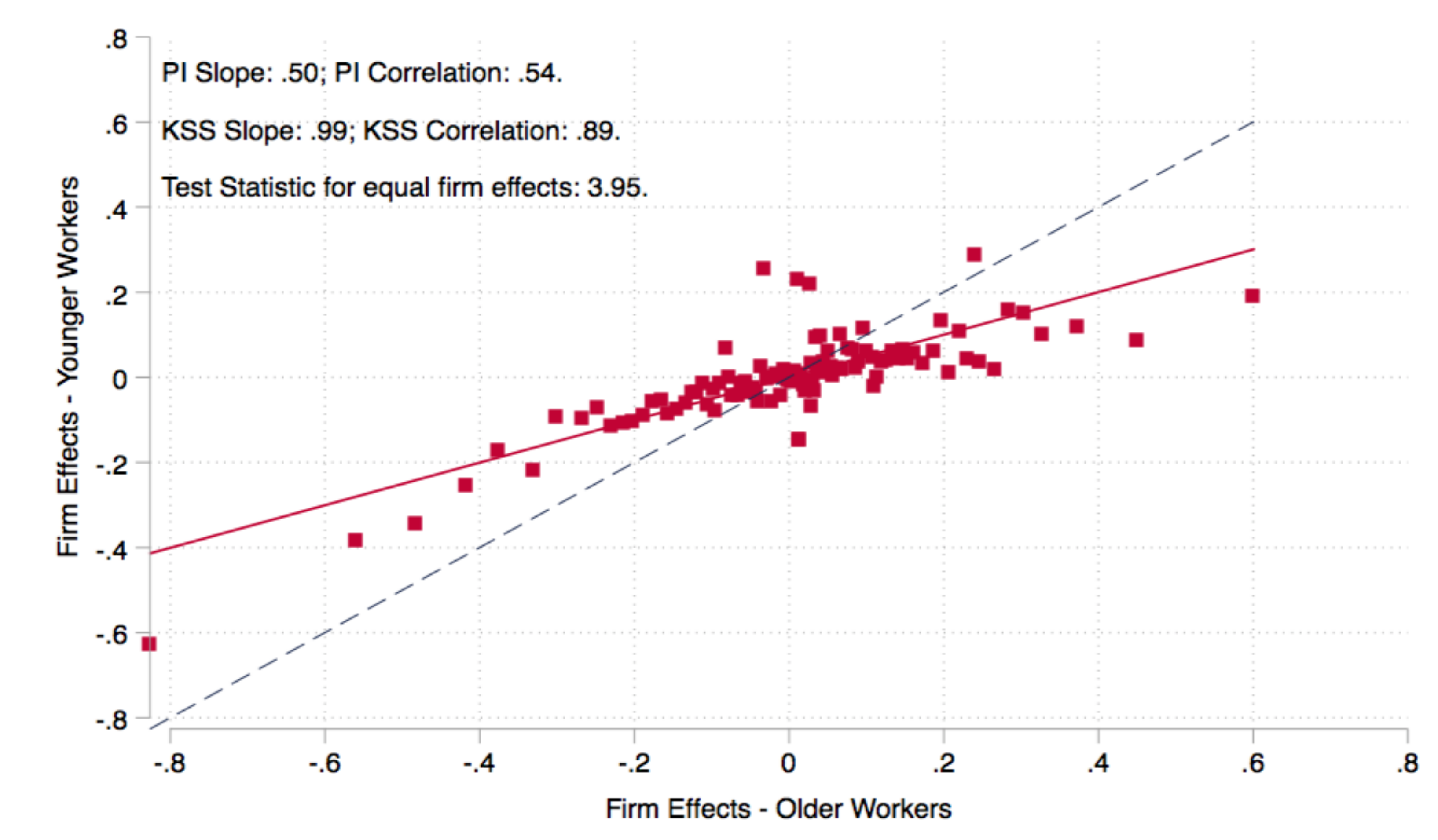}} \\ 
	\end{center}
	%\begin{singlespace}
		{\footnotesize \underline{Note:} This figure plots the mean of the estimated firm effects for younger workers ($\hat \psi_j^Y$) by centiles of the estimated firm effects for older workers ($\hat \psi_j^O$) in the sample of 8,578 firms for which both sets of effects are leave-one-out identified. Both sets of firm effects are demeaned within this estimation sample. ``PI slope'' gives the coefficient from a person-year weighted projection of $\hat \psi_j^Y$ onto $\hat \psi_j^O$. ``KSS slope'' adjusts for attenuation bias by multiplying the PI slope by the ratio of the plug-in estimate of the person-year weighted variance of $ \psi_j^O$ to the KSS adjusted estimate of the same quantity. ``PI correlation" gives the person-year weighted sample correlation between $\hat \psi_j^O$ and $\hat \psi_j^Y$ while ``KSS correlation" adjusts this correlation for sampling error in both $\hat \psi_j^O$ and $\hat \psi_j^Y$ using leave out estimates of the relevant variances. ``Test statistic'' refers to the realization of ${\hat \theta_{H_0}}/{\sqrt{\hat{\mathbb{V}}[\hat \theta_{H_0}]}}$ where $\hat \theta_{H_0}$ is the quadratic form associated with the null hypothesis that the firm effects are equal across age groups, see \th\ref{rem:testbig} and Appendix \ref{app:testOY} for details. From \th\ref{thm3}, ${\hat \theta_{H_0}}/{\sqrt{{\mathbb{V}}[\hat \theta_{H_0}]}}$ converges to a $\mathcal{N}(0,1)$ under the null hypothesis that $\psi_j^O=\psi_j^Y$ for all 8,578 firms. %The last line reports the difference in the coefficient of determination between the restricted model which imposes $\hat \psi_j^Y = \hat \psi_j^{O}$ and the unrestricted model which imposes a different firm effects across young and older workers, see for further details.
		}
	%\end{singlespace}
\end{figure}



**Table 1:** Summary Statistics

|  | Pooled<br>[1] | Younger Workers<br>[2] | Older Workers<br>[3] |
|---|---|---|---|
| ***Largest Connected Set*** |  |  |  |
| Number of Observations | 1,859,459 | 1,011,111 | 643,020 |
| Number of Movers | 197,572 | 133,627 | 53,035 |
| Number of Firms | 73,933 | 62,848 | 26,606 |
|  |  |  |  |
| Mean Log Daily Wage | 4.7507 | 4.6741 | 4.8925 |
| Variance Log Daily Wage | 0.1985 | 0.1321 | 0.2722 |
|  |  |  |  |
| ***Leave One Out Sample*** |  |  |  |
| Number of Observations | 1,319,972 | 661,528 | 425,208 |
| Number of Movers | 164,203 | 102,746 | 35,467 |
| Number of Firms | 42,489 | 33,151 | 10,733 |
|  |  |  |  |
| Mean Log Daily Wage | 4.8066 | 4.7275 | 4.9455 |
| Variance Log Daily Wage | 0.1843 | 0.1200 | 0.2591 |
|  |  |  |  |
| Maximum Leverage ( $P_{ii}$ ) | 0.9365 | 0.9437 | 0.9513 |
|  |  |  |  |
| ***Leave Two Out Sample*** |  |  |  |
| Number of Observations | 1,126,964 | 514,284 | 301,948 |
| Number of Movers | 128,705 | 70,703 | 21,066 |
| Number of Firms | 24,424 | 16,263 | 3,756 |
|  |  |  |  |
| Mean Log Daily Wage | 4.8307 | 4.7528 | 4.9834 |
| Variance Log Daily Wage | 0.1898 | 0.1232 | 0.2760 |
|  |  |  |  |
| Maximum Leverage ( $P_{ii}$ ) | 0.8333 | 0.8481 | 0.8604 |

**Note:** Data in column 1 corresponds to VHW observations in the years 1999 and 2001. Column 2 restricts to workers born in the years 1965-1983. Column 3 considers workers born in the years 1937-1964.  The largest connected set gives the largest sample in which all firms are connected by worker mobility. The leave out sample is the largest connected set such that every firm remains connected even after removing a worker from the sample, see Appendix B1 for details. The leave two out sample is the largest connected set such that every firm remains connected even after removing two workers from the sample, see Appendix B2 for details. Statistics on log daily wages are person-year weighted.

**Table 2:** Variance Decomposition

| | Pooled | | Younger Workers | | Older Workers | |
|---|---|---|---|---|---|---|
| | Leave one out sample | Leave two out sample | Leave one out sample | Leave two out sample | Leave one out sample | Leave two out sample |
| Variance of Log Wages | 0.1843 | 0.1898 | 0.1200 | 0.1232 | 0.2591 | 0.2760 |
| | | | | | | |
| ***Variance of Firm Effects*** | | | | | | |
| Plug in (PI) | 0.0358 | 0.0316 | 0.0368 | 0.0314 | 0.0415 | 0.0304 |
| Homoscedasticity Only (HO) | 0.0295 | 0.0271 | 0.0270 | 0.0251 | 0.0350 | 0.0243 |
| Leave Out (KSS) | 0.0240 | 0.0238 | 0.0218 | 0.0221 | 0.0204 | 0.0180 |
| | | | | | | |
| ***Variance of Person Effects*** | | | | | | |
| Plug in (PI) | 0.1321 | 0.1341 | 0.0843 | 0.0827 | 0.2180 | 0.2406 |
| Homoscedasticity Only (HO) | 0.1173 | 0.1214 | 0.0647 | 0.0663 | 0.2046 | 0.2298 |
| Leave Out (KSS) | 0.1119 | 0.1179 | 0.0596 | 0.0634 | 0.1910 | 0.2221 |
| | | | | | | |
| ***Covariance Firm, Person Effects*** | | | | | | |
| Plug in (PI) | 0.0039 | 0.0077 | -0.0058 | -0.0008 | -0.0032 | -0.0006 |
| Homoscedasticity Only (HO) | 0.0097 | 0.0117 | 0.0030 | 0.0049 | 0.0040 | 0.0041 |
| Leave Out (KSS) | 0.0147 | 0.0149 | 0.0075 | 0.0075 | 0.0171 | 0.0115 |
| | | | | | | |
| ***Correlation of Person, Firm Effects*** | | | | | | |
| Plug in (PI) | 0.0565 | 0.1189 | -0.1040 | -0.0155 | -0.0334 | -0.0068 |
| Homoscedasticity Only (HO) | 0.1649 | 0.2033 | 0.0726 | 0.1205 | 0.0475 | 0.0542 |
| Leave Out (KSS) | 0.2830 | 0.2803 | 0.2092 | 0.2016 | 0.2744 | 0.1826 |
| | | | | | | |
| ***Coefficient of Determination*** | | | | | | |
| Plug in (PI) | 0.9546 | 0.9559 | 0.9183 | 0.9183 | 0.9774 | 0.9781 |
| Homoscedasticity Only (HO) | 0.9029 | 0.9078 | 0.8184 | 0.8256 | 0.9524 | 0.9551 |
| Leave Out (KSS) | 0.8976 | 0.9047 | 0.8091 | 0.8209 | 0.9489 | 0.9536 |

**Note:** See Table 1 for sample definitions. All variance components are person-year weighted. Wages have been pre-adjusted for a year fixed effect.

**Table 3:** Variance of Firm Effects under Different Leave-Out Strategies

| | Pooled | | Younger Workers | | Older Workers | |
|---|---|---|---|---|---|---|
| | Leave one out sample | Leave two out sample | Leave one out sample | Leave two out sample | Leave one out sample | Leave two out sample |
| ***Descriptive Statistics*** | | | | | | |
| Number of Observations | 5,163,446 | 4,758,713 | 2,632,596 | 2,290,223 | 2,016,202 | 1,652,729 |
| Number of Movers | 440,323 | 391,313 | 276,338 | 229,013 | 123,777 | 92,568 |
| Number of Firms | 86,646 | 60,896 | 71,673 | 46,084 | 32,809 | 16,396 |
| | | | | | | |
| Mean Log Daily Wage | 4.7511 | 4.7662 | 4.6644 | 4.6799 | 4.8841 | 4.9158 |
| Variance Log Daily Wage | 0.1839 | 0.1854 | 0.1156 | 0.1162 | 0.2477 | 0.2579 |
| | | | | | | |
| Maximum Leverage ( $P_{ii}$ ) | 0.9021 | 0.7713 | 0.9056 | 0.7775 | 0.9174 | 0.7639 |
| | | | | | | |
| ***Variance of firm Effects*** | | | | | | |
| Variance of Firm Effects (PI) | 0.0304 | 0.0278 | 0.0303 | 0.0273 | 0.0376 | 0.0336 |
| | | | | | | |
| Variance of Firm Effects (Leave Person-Year Out) | 0.0296 | 0.0273 | 0.0302 | 0.0272 | 0.0314 | 0.0322 |
| | | | | | | |
| Variance of Firm Effects (Leave Match Out) | 0.0243 | 0.0231 | 0.0221 | 0.0218 | 0.0265 | 0.0280 |
| | | | | | | |
| Variance of Firm Effects (Leave Worker Out) | 0.0241 | 0.0233 | 0.0227 | 0.0218 | 0.0270 | 0.0268 |

**Note:** These samples consist of VHW data spanning the years 1996-2001 (T=6). "Pooled" considers all workers, "Younger workers" were born in the years 1965-1983, "Older workers" were born in the years 1937-1964. Variance of Firm Effects (PI) denotes the uncorrected plug-in estimate of the variance of firm effects. Variance of Firm Effects (Leave Person-Year Out) computes the leave-out bias correction by leaving a single person-year observation out. Variance of Firm Effects (Leave Match Out) computes the leave-out bias correction by leaving entire worker-firm matches out. Variance of Firm Effects (Leave Worker Out) computes the leave-out bias correction by leaving out each worker's entire employment history. See Appendix B3.1 for computational details. Wages have been pre-adjusted for a cubic in age and year fixed effects.

**Table 4:** Projecting Firm Effects on Covariates

|  | (1) | (2) |
|---|---|---|
| Older Worker | 0.0272 | -0.0016 |
|  | (0.0009) | (0.0024) |
|  | [0.0003] | [0.0001] |
|  |  |  |
| Log Firm Size |  | 0.0276 |
|  |  | (0.0007) |
|  |  | [0.0001] |
|  |  |  |
| Older Worker x Log Firm Size |  | 0.0028 |
|  |  | (0.0005) |
|  |  | [0.0002] |
|  |  |  |
| Predicted Gap in Firm Effects Older vs. Younger Worker | 0.0272 | 0.0054 |
|  | (0.0009) | (0.0019) |
|  | [0.0003] | [0.0008] |
| Number of Observations | 1,319,972 | 1,319,972 |

**Note:** This table reports projections of the firm effects estimated in the pooled leave one out sample defined in Table 1 onto worker and firm characteristics. In Column 1, the firm effects are projected onto a constant and a dummy for the workers being born in the years 1937-1964. In Column 2, we regress firm effects on a constant, a dummy for workers born in the years 1937-1964, log firm size and an interaction between these two covariates. The brackets show the standard error associated with a given linear projection coefficient. In round brackets we report the standard error based on Theorem 1. In square brackets we report the naive standard error that one would obtain when regressing firm effects on a given covariate and using standard Eicker-White heteroskedastic robust standard errors. "Predicted Gap in Firm Effects" reports the predicted difference in firm effects between older and younger workers according to either Column 1 or Column 2 evaluated at the median firm size (12 Workers). See text for details.

**Table 5:** Inference on the Variance of Firm Effects

| | Pooled | | Younger Workers | | Older Workers | |
|---|---|---|---|---|---|---|
| | Leave one out sample | Leave two out sample | Leave one out sample | Leave two out sample | Leave one out sample | Leave two out sample |
| ***Variance of Firm Effects*** | | | | | | |
| Leave out estimate | 0.0240 | 0.0238 | 0.0218 | 0.0221 | 0.0204 | 0.0180 |
| | (0.0006) | (0.0006) | (0.0006) | (0.0006) | (0.0025) | (0.0013) |
| | | | | | | |
| Sum of Squared Eigenvalues | 2.1141E-06 | 1.6236E-06 | 3.9720E-06 | 2.9863E-06 | 0.0001 | 0.0001 |
| | | | | | | |
| Percentage of movers with no split sample estimator | 14.92% | 0.00% | 21.00% | 0.00% | 22.00% | 0.00% |
| Percentage of mover pairs in $\mathcal{B}$ | 0.04% | 0.05% | 0.03% | 0.05% | 0.21% | 0.46% |
| | | | | | | |
| ***Confidence Intervals*** | | | | | | |
| 95% Confidence Intervals - Strong id (q=0) | [0.0228; 0.0251] | [0.0227; 0.0249] | [0.0207; 0.0230] | [0.0210; 0.0232] | [0.0155; 0.0254] | [0.0154; 0.0205] |
| 95% Confidence Intervals - Weak id (q=1) | [0.0228; 0.0251] | [0.0227; 0.0249] | [0.0207; 0.0230] | [0.0209; 0.0234] | [0.0157; 0.0267] | [0.0158; 0.0221] |
| 95% Confidence Intervals - Weak id (q=2) | [0.0228; 0.0251] | [0.0227; 0.0251] | [0.0207; 0.0231] | [0.0209; 0.0234] | [0.0158; 0.0272] | [0.0161; 0.0231] |
| | | | | | | |
| Curvature - $\kappa_1$ | 0.0182 | 0.0416 | 0.0197 | 0.0689 | 0.3333 | 0.7845 |
| Curvature - $\kappa_2$ | 0.0271 | 0.1458 | 0.0444 | 0.0678 | 0.2779 | 1.1088 |
| | | | | | | |
| ***Diagnostics*** | | | | | | |
| Lindeberg Condition - q=1 | 0.1878 | 0.0865 | 0.2061 | 0.0371 | 0.0359 | 0.0503 |
| Lindeberg Condition - q=2 | 0.0866 | 0.1604 | 0.0237 | 0.1639 | 0.0491 | 0.0479 |
| | | | | | | |
| ***Eigenvalues Ratios*** | | | | | | |
| Eigenvalue Ratio - 1 | 0.0135 | 0.0233 | 0.0189 | 0.0211 | 0.3132 | 0.5883 |
| Eigenvalue Ratio - 2 | 0.0131 | 0.0202 | 0.0058 | 0.0135 | 0.0387 | 0.0499 |
| Eigenvalue Ratio - 3 | 0.0112 | 0.0131 | 0.0050 | 0.0087 | 0.0314 | 0.0187 |

**Note:** This table conducts inference on the variance of firm effects using the samples described in Table 1. The round brackets report standard error estimates constructed according to the procedure described in Section 4.2. Associated confidence intervals are computed under different assumptions on q, as described in Section 6.1. Curvature reports the maximal curvature, see Appendix C6.1 for further details. "Eigenvalue ratio - 1" gives the ratio of the squared largest eigenvalue of the matrix $\tilde{A}$ to the sum of all its squared eigenvalues. "Eigenvalue Ratio - 2" and "Eigenvalue Ratio - 3" report the same ratio using the second and third largest eigenvalues respectively. "Percentage of movers with no split sample estimator" reports the percentage of movers for which it is impossible to find two independent unbiased estimators of their conditional mean. "Percentage of mover pairs in $\mathcal{B}$ " gives the fraction of mover pairs where we use an unconditional variance estimate, see Section 4.2 for details.

**Table 6:** Montecarlo Results

| | Pooled | | Younger Workers | | Older Workers | |
|---|---|---|---|---|---|---|
| | Leave one out sample | Leave two out sample | Leave one out sample | Leave two out sample | Leave one out sample | Leave two out sample |
| True Variance of Firm effects | 0.0240 | 0.0208 | 0.0219 | 0.0186 | 0.0222 | 0.0199 |
| ***Mean, Standard deviation across Simulations*** | | | | | | |
| Variance of Firm Effects - Leave Out (KSS) | 0.0239 | 0.0208 | 0.0219 | 0.0186 | 0.0223 | 0.0198 |
| | (0.00037) | (0.00037) | (0.00040) | (0.00039) | (0.00221) | (0.00276) |
| Variance of Firm Effects - Homescedatic Only (HO) | 0.0336 | 0.0260 | 0.0328 | 0.0242 | 0.0345 | 0.0264 |
| | (0.00036) | (0.00036) | (0.00037) | (0.00038) | (0.00216) | (0.00275) |
| Variance of Firm Effects - Plug in (PI) | 0.0352 | 0.0270 | 0.0352 | 0.0256 | 0.0360 | 0.0273 |
| | (0.00036) | (0.00036) | (0.00036) | (0.00038) | (0.00216) | (0.00275) |
| Mean Estimated Standard error: | 0.00050 | 0.00039 | 0.00058 | 0.00042 | 0.00254 | 0.00273 |
| ***Coverage Rate*** | | | | | | |
| Leave out - Strong id (q=0) | 0.9890 | 0.9610 | 0.9970 | 0.9600 | 0.9720 | 0.8820 |
| Leave out - Weak id (q=1) | 0.9900 | 0.9640 | 0.9970 | 0.9630 | 0.9800 | 0.9590 |

**Note:** See Section 8.7 for description of data generating process (DGP). "True variance of firm effects" gives the variance of the firm effects in the DGP. KSS, PI, and HO report the average of the estimated variance of firm effects across simulations for the three different estimators. Round brackets report the standard deviation of estimator across simulations. ``Leave Out - Strong Id" reports the coverage rate of confidence interval for the variance of firms effects based upon KSS standard error and a normal approximation. ``Leave Out - Weak Id" reports the coverage rate of the confidence interval described in Section 6.1 under q=1. All results based upon 1,000 Monte Carlo draws.



%\tableofcontents
%\listoffigures
%\listoftables

\begin{appendices}
\section{Data}

This Appendix describes construction of the data used in the application of Section \ref{sec:application}. 

\subsection{Veneto Workers History}\label{app:data}
Our data come from the Veneto Workers History (VWH) file, which provides social security based earnings records on annual job spells for all workers employed in the Italian region of Veneto at any point between the years 1975 and 2001.  Each job-year spell in the VWH lists a start date, an end date, the number of days worked that year, and the total wage compensation received by the employee in that year. The earnings records are not top-coded. We also observe the gender of each worker and several geographic variables indicating the location of each employer. See \cite*{card2014rent} and \cite{serafinelli2019good} for additional discussion and analysis of the VWH.

We consider data from the years 1984--2001 as prior to that information on days worked tend to be of low quality. To construct the person-year panel used in our analysis, we follow the sample selection procedures described in \cite*{card2013workplace}. First, we drop employment spells in which the worker's age lies outside the range 18--64. The average worker in this sample has 1.21 jobs per year. To generate unique worker-firm assignments in each year, we restrict attention to spells associated with ``dominant jobs'' where the worker earned the most in each corresponding year.  From this person-year file, we then exclude workers that (i) report a daily wage less than 5 real euros or have zero days worked (1.5\% of remaining person-year observations) (ii) report a log daily wage change one year to the next that is greater than 1 in absolute value (6\%) (iii) are employed in the public sector (10\%) or (iv) have more than 10 jobs in any year or that have gender missing (0.1\%).

%\includepdf[scale=1,landscape=true]{../tables/paper_revision_KSS_tableA1.pdf} 
\section{Computation}\label{app:computation}
\renewcommand\thefigure{\thesection.\arabic{figure}}    
\setcounter{figure}{0}   

This Appendix describes the key computational aspects of the leave-out estimator $\hat \theta$, with an emphasis on the application to two-way fixed effects models with two time periods discussed in \th\ref{ex:AKM} and Section \ref{sec:application}.

\subsection{Leave-One-Out Connected Set}\label{sec:pruning}

Existence of $\hat \theta$ requires $P_{ii}<1$ (see \th\ref{lem:unbiased}) and the following describes an algorithm which prunes the data to ensure that $P_{ii}<1$. In the two-way fixed effects model of Section \ref{sec:estimates_AKM}, this condition requires that the bipartite network formed by worker-firm links remains connected when any one worker is removed. This boils down to finding workers that constitute cut vertices or \textit{articulation points} in the corresponding bipartite network. %Additionally, it is clear that $P_{ii}<1$ will fail if a firm is connected to exactly one \textit{mover}, i.e., a worker that either left or joined the firm across the, so the algorithm initially removes all firms with only one mover. 

The algorithm below takes as input a connected bipartite network $\mathcal{G}$ where workers and firms are vertices. Edges between two vertices correspond to the realization of a match between a worker and a firm \cite[see][for discussion]{jochmans2016fixed,bonhomme2017econometric}. In practice, one typically starts with a $\mathcal{G}$ corresponding to the \textit{largest connected component} of a given bipartite network \cite[see, e.g.,][]{card2013workplace}. The output of the algorithm is a subset of $\mathcal{G}$ where removal of any given worker does not break the connectivity of the associated graph.

The algorithm relies on existing functions that efficiently finds articulation points and largest connected components. In MATLAB such functions are available in the \textit{Boost Graph Library} and in R they are available in the \textit{igraph} package.

% \begin{algorithm}
%	\caption{Leave-One-Out Connected Set}
%	\begin{algorithmic}[1]
%		\Function{PruningNetwork}{$G$}\Comment{$\mathcal{G}\equiv$ Connected graph from bipartite network of firms and workers}
%		\State $a=1$
%		\While {$a>0$}
%		\State $\mathcal{G}^{bad}=\emptyset$.
%		\For {$g=1,\hdots N$}
%		\State Add $g$ to $\mathcal{G}^{bad}$ if removal of worker $g$ from $\mathcal{G}$ disconnects the resulting graph.	
%		\EndFor
%		\State $a= \vert \mathcal{G}^{bad} \vert.$ 
%		\State \text{Update $\mathcal{G}$ by finding the largest connected set after removing workers in $\mathcal{G}^{bad}$.} 
%		\EndWhile
%		\EndFunction
%		
%	\end{algorithmic}
%\end{algorithm}

 \begin{algorithm}
   \caption{Leave-One-Out Connected Set}
    \begin{algorithmic}[1]
      \Function{PruningNetwork}{$\mathcal{G}$}\Comment{$\mathcal{G}\equiv$ Connected bipartite network of firms and workers}
		\State Construct $\mathcal{G}_1$ from $\mathcal{G}$ by deleting all workers that are articulation points in $\mathcal{G}$
		\State Let $\mathcal{G}$ be the largest connected component of $\mathcal{G}_1$
		\State Return $\mathcal{G}$
       \EndFunction

\end{algorithmic}
\end{algorithm}

The algorithm typically completes in less than a minute for datasets of the size considered in our application. Furthermore, the vast majority of firms removed using this algorithm are only associated with one mover.

\subsection{Leave-Two-Out Connected Set}\label{sec:pruning2}

We also introduced a leave-two-out connected set, which is a subset of the original data such that removal of any \textit{two} workers does not break the connectedness of the bipartite network formed by worker-firm links. The following algorithm proceeds by applying the idea in Algorithm 1 to each of the networks constructed by dropping one worker. A crucial difference from Algorithm 1 is that \textit{two} workers who do not break connectedness in the input network may break connectedness when other workers have been removed. For this reason, the algorithm runs in an iterative fashion until it fails to remove any additional workers.

%To compute the leave-two-out connected set in the two-way fixed effects model of Section \ref{sec:estimates_AKM}, we proceed in a similar way as in Appendix \ref{sec:pruning}. A leave-two-out connected set is defined as a bipartite graph formed by worker-firm links that remains connected when any \textit{two} workers are removed. Therefore, this condition will not be satisfied whenever there exists in the graph a firm that is associated with only two movers.

%To find the largest leave-two-out connected set we use a simple extension of Algorithm 1.  

 \begin{algorithm}
	\caption{Leave-Two-Out Connected Set}
	\begin{algorithmic}[1]
		\Function{PruningNetwork2}{$\mathcal{G}$}\quad \Comment{$\mathcal{G}\equiv$ Leave-one-out connected bipartite network of firms and workers}
		\State $a=1$
		\While {$a>0$}
		\State $\mathcal{G}^{del}=\emptyset$
		\For {$g=1,\dots, N$}
		\State Construct $\mathcal{G}_1$ from $\mathcal{G}$ by deleting worker $g$
		\State Add all workers that are articulation points in $\mathcal{G}_1$ to $\mathcal{G}^{del}$
		\EndFor
		\State $a=  \abs{\mathcal{G}^{del}}$ 
		\If{$a>0$}
		\State Construct $\mathcal{G}_1$ from $\mathcal{G}$ by deleting all workers in $\mathcal{G}^{del}$
		\State Let $\mathcal{G}_2$ be the largest connected component of $\mathcal{G}_1$
		\State Let $\mathcal{G}$ be the output of applying Algorithm 1 to $\mathcal{G}_2$
		\EndIf
		\EndWhile
		\State Return $\mathcal{G}$
		\EndFunction
		
	\end{algorithmic}
\end{algorithm}

% \begin{algorithm}
%   \caption{Leave-Two-Out Connected Set}
%    \begin{algorithmic}[1]
%      \Function{PruningNetwork2}{$\mathcal{G}$}\quad \Comment{$\mathcal{G}\equiv$ Leave-one-out connected bipartite network of firms and workers}
%       \State $a=1$
%        \While {$a>0$}
%		\State $\mathcal{G}^{art}=\emptyset$.
%		\For {$g=1,\hdots N$}
%			\State Construct graph $\mathcal{G}_{g}=\mathcal{G}\setminus g,$
%			\State Add $g'$ to $\mathcal{G}^{bad}$ if removal of worker $g'$ from $\mathcal{G}_{g}$ disconnects the resulting graph.	
%		\EndFor
%		\State $a= \vert \mathcal{G}^{bad} \vert.$ 
%		\State \text{Update $\mathcal{G}$ by applying Algorithm 1 after removing workers in  $\mathcal{G}^{bad}$ from $\mathcal{G}$.  } 
%			        \EndWhile
%       \EndFunction
%
%\end{algorithmic}
%\end{algorithm}

\subsection{Computing $\hat \theta$}\label{sec:Lambda}

Our proposed leave-out estimator is a function of the $2n$ quadratic forms
\begin{align}
	P_{ii} = x_i' S_{xx}\inverse x_i \qquad B_{ii} = x_i'S_{xx}\inverse A S_{xx}\inverse x_i \qquad \text{for } i=1,...,n.
\end{align}
The estimates reported in Section \ref{sec:application} of the paper rely on exact computation of these quantities. In our application, $k$ is on the order of hundreds of thousands, making it infeasible to compute $S_{xx}\inverse$ directly. To circumvent this obstacle, we instead compute the $k$-dimensional vector $z_{i,exact}=S_{xx}\inverse x_i$ separately for each $i=1,..,n$. That is, we solve separately for each column of $Z_{exact}$ in the system 
\begin{align}
\underset{k \times k}{S_{xx}} \underset{k \times n}{Z_{exact}}   = \underset{k \times n}{X'}.
\end{align}
We then form $P_{ii} = x_{i}' z_{i,exact}$ and $B_{ii} = z_{i,exact}' A z_{i,exact}$. The solution $z_{i,exact}$ is computed via MATLAB's preconditioned conjugate gradient routine \textit{pcg}. In computing this solution, we utilize the preconditioner developed by \cite{koutis2011combinatorial}, which is optimized for diagonally dominant design matrices $S_{xx}$. These column-specific calculations are parallelized across different cores using MATLAB's \textit{parfor} command. 

\subsubsection{Leaving a Cluster Out}\label{app:cluster}

Table \ref{table3} applies the leave-cluster-out estimator introduced in \th\ref{rem:cluster} to estimate the variance of firm effects with more than two time periods and potential serial correlation. The estimator takes the form $\hat \theta_{cluster} = \sum_{i=1}^{n}y_{i}\tilde{x}_{i}'\hat{\beta}_{-c(i)}$ where $\hat{\beta}_{-c(i)}$ is the OLS estimator obtained after leaving out all observations in the cluster to which observation $i$ belongs. A representation of $\hat \theta_{cluster}$ that is useful for computation takes the observations in the $c$-th cluster and collect their outcomes in $y_c$ and their regressors in $X_c$. The leave-cluster-out estimator is then
\begin{align}\label{eq:cluster}
\hat \theta_{cluster} = \hat \beta' A \hat \beta - \sum_{c=1}^C y_c' B_{c} (I-P_{c})\inverse (y_c - X_c\hat{\beta}),
\end{align}
where $C$ denotes the total number of clusters, $P_c = X_c  S_{xx}\inverse X_c'$, and $B_c = X_c  S_{xx}\inverse A S_{xx}\inverse X_c'$. Since the entries of $P_c$ and $B_{c}$ are of the form $P_{i\ell} = x_i' S_{xx}\inverse x_\ell$ and $B_{i\ell} = x_i' S_{xx}\inverse A S_{xx}\inverse x_\ell$, computation can proceed in a similar fashion as described earlier for the leave-one-out estimator.

When defining the cluster as a worker-firm match, Table \ref{table3} applies $\hat \theta_{cluster}$ to the two-way fixed effects model in \eqref{eq:AKM}. When defining the cluster as a worker, the individual effects can not be estimated after leaving a cluster out. Table \ref{table3} therefore applies $\hat \theta_{cluster}$ after demeaning at the individual level. This transformation removes the individual effects so that the resulting model can be estimated after leaving a cluster out.

\subsubsection{Johnson-Lindenstrauss Approximation}
When $n$ is on the order of hundreds of millions and $k$ is on the order of tens of millions, the exact algorithm may no longer be tractable. The JLA simplifies computation of $P_{ii}$ considerably by only requiring the solution of $p$ systems of $k$ linear equations. That is, one need only solve for the columns of $Z_{JLA}$ in the system
\begin{align}
\underset{k \times k}{S_{xx}} \underset{k\times p}{Z_{JLA}}   = \underset{k \times p}{(R_PX)'},
\end{align} 
which reduces computation time dramatically when $p$ is small relative to $n$. 

To compute $B_{ii}$, it is necessary to solve linear systems involving both $A_1$ and $A_2$, leading to $2p$ systems of equations when $A_1 \neq A_2$. However, for variance decompositions like the ones considered in Section \ref{sec:estimates_AKM}, the same $2p$ systems can be reused for all three variance components, leading to a total of $3p$ systems of equations for the full variance decomposition. This is so because the three variance components use the matrices $A_{\psi} = A_f'A_f$, $A_{\alpha,\psi} = \frac{1}{2}(A_d'A_f + A_f'A_d)$, and $A_{\alpha} = A_d'A_d$ where
\begin{align}
A_f' = \mbox{\scriptsize $ \frac{1}{\sqrt{n}}\begin{bmatrix} 0 & 0 & 0 \\ f_1 - \bar f & \dots & f_n-\bar f \\ 0 & 0 & 0 \end{bmatrix}$} \quad \text{and} \quad 
A_d' = \mbox{\scriptsize$\frac{1}{\sqrt{n}}\begin{bmatrix}  d_1 - \bar d & \dots & d_n-\bar d \\ 0 & 0 & 0 \\ 0 & 0 & 0 \end{bmatrix}$ }.
\end{align}
Based on these insights, Algorithm 3 below takes as inputs $X$, $A_{f}$, $A_{d}$, and $p$, and returns $\hat P_{ii}$ and three different $\hat B_{ii}$'s which are ultimately used to construct the corresponding variance component $\hat \theta_{JLA}$ as defined in Section \ref{sec:comp}.

\begin{algorithm}
	\caption{Johnson-Lindenstrauss Approximation for Two-Way Fixed Effects Models}
	\begin{algorithmic}[1]
		\Function{JLA}{$X$,$A_{f}$,$A_d$,$p$}
		\State Generate $R_B, R_P \in \R^{p \times n}$, where $(R_B,R_P)$ are composed of mutually independent Rademacher entries
		\State Compute $(R_PX)'$, $(R_BA_{f})'$, $(R_BA_d)' \in \R^{k \times p}$
		\For {$\kappa=1,\dots,p$}
		\State Let $r_{\kappa,0}$, $r_{\kappa,1},$ $r_{\kappa,2} \in \R^{k}$ be the $\kappa$-th columns of $(R_PX)'$, $(R_BA_{f})'$, $(R_BA_{d})'$
		\State Let $z_{\kappa,\ell} \in \R^{k}$ be the solution to $S_{xx} z = r_{\kappa,\ell}$ for $\ell =0,1,2$
		\EndFor 
		\State Construct $Z_\ell=(z_{1,\ell}, \hdots, z_{p,\ell}) \in \R^{k \times p}$ for $\ell =0,1,2$
		\State Construct $\hat P_{ii} = \frac{1}{p} \norm*{Z_0'x_i}^2$, $\hat B_{ii,\psi} = \frac{1}{p} \norm*{Z_1'x_i}^2$, $\hat B_{ii,\alpha} = \frac{1}{p} \norm*{Z_2'x_i}^2$, $\hat B_{ii,\alpha\psi} = \frac{1}{p} (Z_1'x_i)'(Z_2'x_i)$ for $i=1,\dots,n$
		\State Return $\{\hat P_{ii},\hat B_{ii,\psi},\hat B_{ii,\alpha},\hat B_{ii,\alpha\psi}\}_{i=1}^n$
		\EndFunction
	\end{algorithmic}
\end{algorithm}

%Our MATLAB implementation at \url{https://github.com/rsaggio87/LeaveOutTwoWay} allows the user to avoid specifying $A_2$, in which case it only returns $\hat P_{ii}$ and a $\hat B_{ii}$ for the variance component defined by $A_1'A_1$. However, the current implementation (as of Aug 2019) is tailored towards the two-way fixed effects model without additional covariates as it relies on the fast solver introduced by \cite{koutis2011combinatorial} to quickly find the solutions to the systems in the sixth line of the algorithm.\fxnote{This paper is not a Github manual}

\subsubsection{Performance of the JLA} 
\label{sec:quality_of_approx}
Figure \ref{fig:increasing_T} evaluates the performance of the Johnson-Lindenstrauss approximation across 4 VWH samples that correspond to different (overlapping) time intervals (2000--2001; 1999--2001; 1998--2001; 1997--2001). The $x$-axis in Figure \ref{fig:increasing_T} reports the total number of person and firm effects associated with a particular sample. 

Figure \ref{fig:increasing_T} shows that the computation time for exact computation of $(B_{ii},P_{ii})$ increases rapidly as the number of parameters of the underlying AKM model grow; in the largest dataset considered -- which involves more than a million worker and firm effects --  exact computation takes about 8 hours. Computation of JLA complete in markedly shorter time: in the largest dataset considered computation time is less than 5 minutes when $p=500$ and slightly over 6 minutes when $p=2500$. Notably, the JLA delivers estimates of the variance of firm effects almost identical to those computed via the exact method, with the quality of the approximation increasing for larger $p$. For instance, in the largest dataset, the exact estimate of variance of firm effects is 0.028883. By comparison, the JLA estimate equals 0.028765 when $p=500$ and 0.0289022 when $p=2500$. 

In summary: for a sample with more than a million worker and firm effects, the JLA cuts computation time by a factor of 100 while introducing an approximation error of roughly $10^{-4}$.
%taking only 1\% of the time needed to perform an exact calculation of $(P_{ii}, B_{ii})$.   

\subsubsection{Scaling to Very Large Datasets}
\label{sec:quality_of_approx_CHK}
We now study how the JLA scales to much larger datasets of the dimension considered by \cite{card2013workplace} who fit models involving tens of millions of worker and firm effects to German social security records. To study the computational burden of a model of this scale, we rely on a synthetic dataset constructed from our original leave-one-out sample analyzed in Column 1 of Table \ref{table2}, i.e., the pooled Veneto sample comprised of wage observations from the years 1999 and 2001. We scale the data by creating replicas of this base sample. To connect the replicas, we draw at random 10\% of the movers and randomly exchange their period 1 firm assignments across replicas. By construction, this permutation maintains each (replicated) firm's size while ensuring leave-one-out connectedness of the resulting network.
%We then create replicas of this base sample and connect them by picking at random $10\%$ of movers whose origin and destination firms are in different replicas. This is done keeping both firm sizes and the number of periods in which we observe a given worker fixed.
%We then create replicas of this base sample and generate at random links across these replicas while keeping fixed both firm sizes and the number of periods in which we observe a given worker.\fxnote{what does this mean?} 

Wage observations are drawn from a variant of the DGP described in Section \ref{sec:MC} adapted to the levels formulation of the model. Specifically, each worker's wage is the sum of a rescaled person effect, a rescaled firm effect, and an error drawn independently in each period from a normal with variance $\frac{1}{2}\exp(\hat a_0 + \hat a_1 B_{gg} + \hat a_2 P_{gg} + \hat a_3 \ln L_{g2} + \hat a_4 \ln L_{g1})$.  As highlighted by Figure \ref{fig:increasing_T}, computing the exact estimator in these datasets would be extremely costly. Drawing from a stable DGP allows us to instead benchmark the JLA estimator against the true value of the variance of firm effects.

%We now study whether JLA can scale to extremely large datasets of the dimension considered by \cite{card2013workplace}, who analyzed data from the universe of German social security records and fit models involving tens of millions of worker and firm effects. Because we do not have access to the data considered by \cite{card2013workplace} we rely on a synthetic dataset constructed from the VHW data. Specifically, we take our original leave one out sample analyzed in Column 1 of Table \ref{table2}, i.e., the pooled Veneto sample over 1999--2001 ($T=2$). We then create replicas of this base sample and generate at random links across these replicas while keeping fixed both firm size and the number of periods in which we observe a given worker. Wage observations are drawn from a DGP of the sort analyzed in Table \ref{table6}. As highlighted by Figure \ref{fig:increasing_T}, computing the exact estimator in these datasets would be extremely costly. Drawing from a stable DGP allows us to instead benchmark the JLA estimator against the true value of the variance of firm effects.

Figure \ref{fig:CHK_data} displays the results. When setting $p=250$, the JLA delivers a variance of firm effects remarkably close to the true variance of firm effects defined by our DGP. As expected, the distance between our approximation and the true variance component decreases with the sample size for a fixed $p$. Remarkably, we are able to compute the AKM variance decomposition in a dataset with approximately 15 million person and year effects in only 35 minutes. Increasing the number of simulated draws in the JLA to $p=500$ delivers estimates of the variance of firm effects nearly indistinguishable from the true value. This is achieved in approximately one hour in the largest simulated dataset considered. The results of this exercise strongly suggest the leave-out estimator can be scaled to extremely large datasets involving the universe of administrative wage records in large countries such as Germany or the United States.

%the VWH dataset comes from a single region of Italy, we are not able to analyze a realized network of the size considered by  \cite{card2013workplace}. 
%To overcome both of these constraints, we rely on a simulation based approach. In particular, we take our original leave one out sample analyzed in Table 2, Column 1, i.e., the pooled Veneto sample over 1999-2001 (T=2). We then create (5,10,15,20) replicas of this starting sample generating at random links across replicas while keeping fixed both firm size and the number of periods in which we observe a given  worker, i.e. $T_{g}=2, \forall g={1, \hdots, N}$. We then simulate a DGP of the same type analyzed in Table 6. This allows us to know the true value of the variance of firm effects and therefore provides us with an useful benchmark to assess the quality of the corresponding JLA.  

%Moreover, as suggested by Figure \ref{fig:increasing_T}, computing the exact algorithm in a dataset of this scale is computationally infeasible.

\newpage
\thispagestyle{empty}
\begin{figure}[H]
\caption{Performance of the JLA Algorithm} 
\vspace{-15pt}
\label{fig:increasing_T}
\begin{center}
\subfloat[{Computation Time}]{\includegraphics[width = .8\columnwidth]{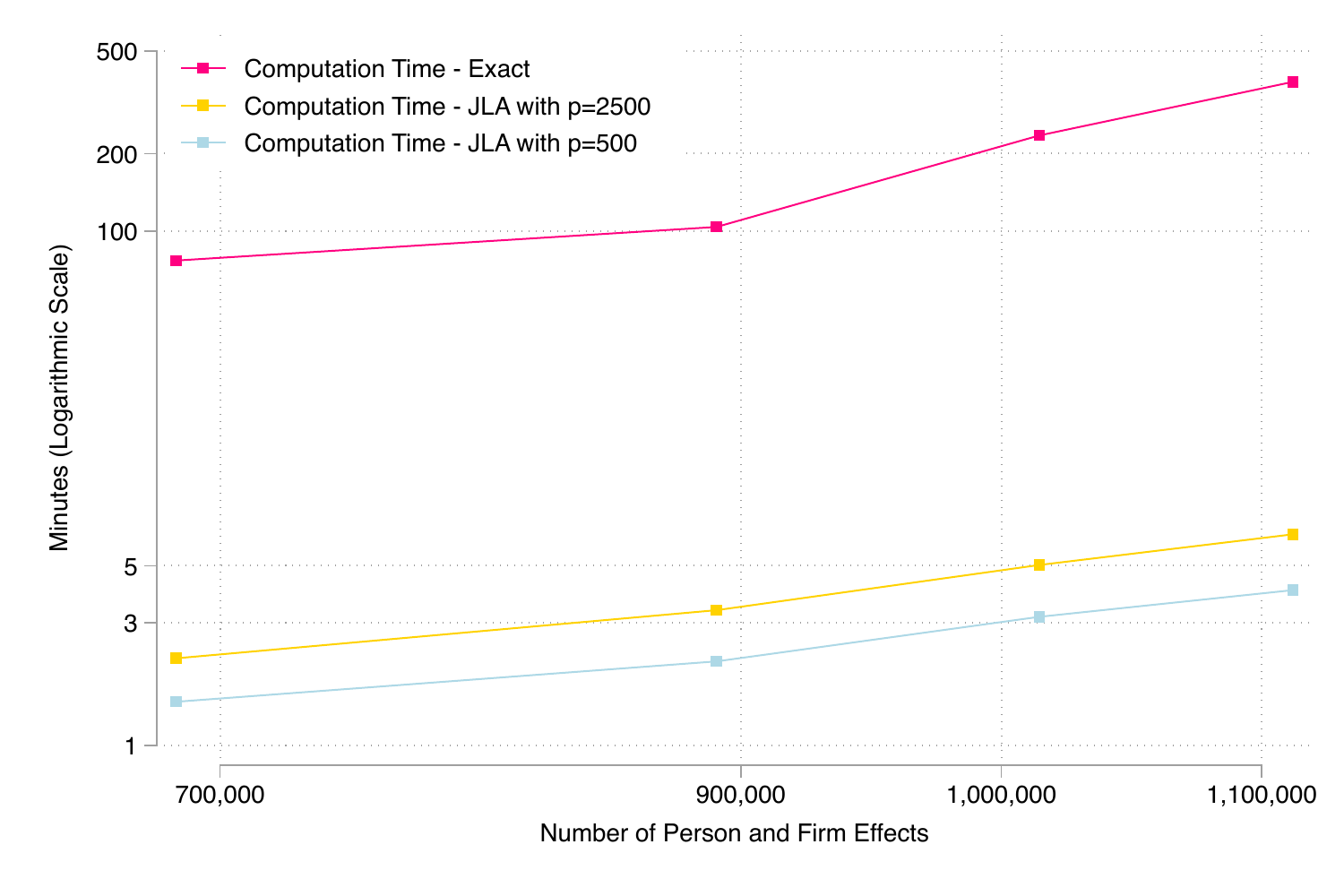}} \\ 
\vspace{-15pt}
\subfloat[{Quality of the Approximation}]{\includegraphics[width = .8\columnwidth]{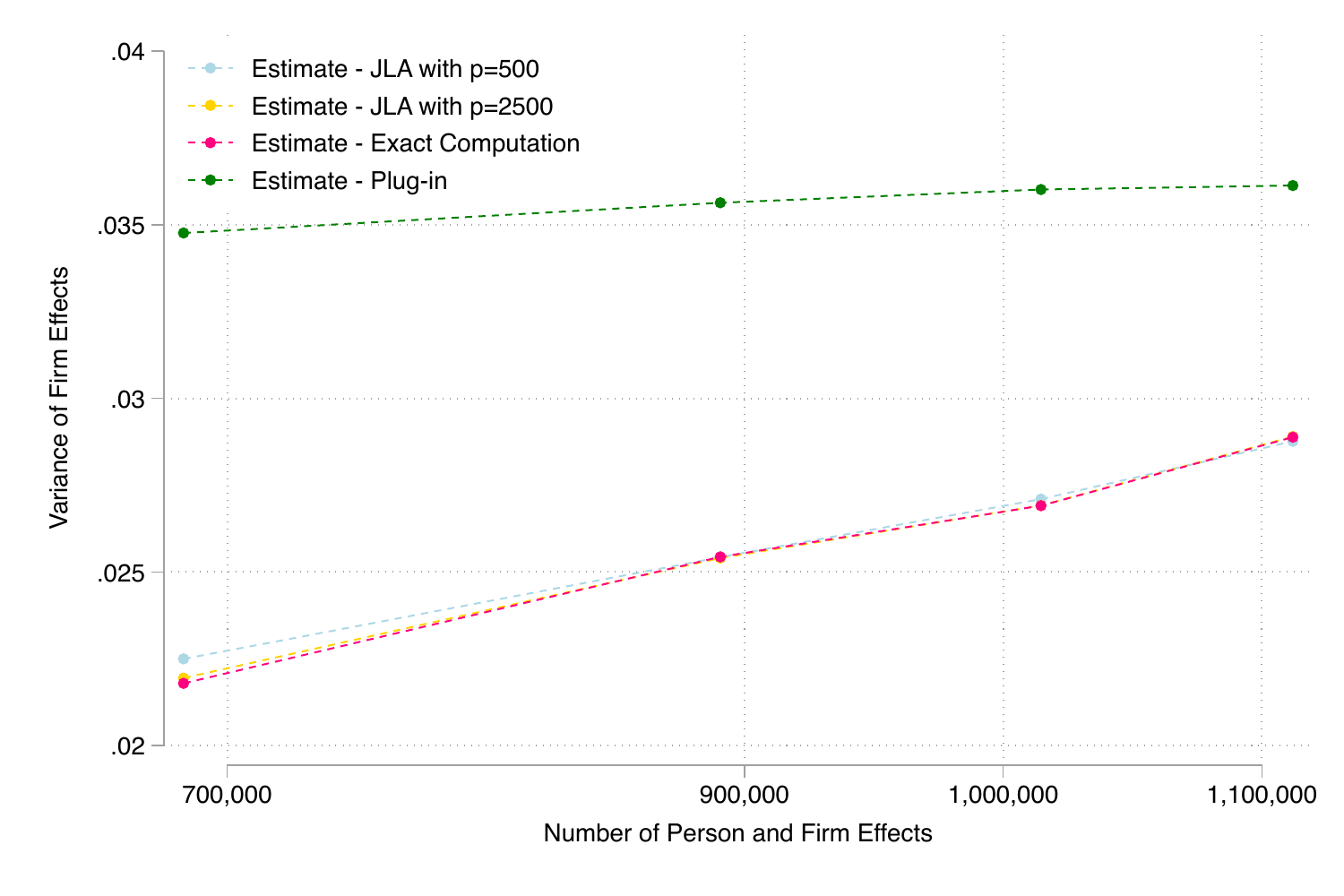}}
\end{center}
\vspace{-25pt}
\begin{singlespace}
	{\footnotesize \underline{Note:} Both panels consider 4 different samples of increasing length. The four samples contain data from the years 2000--2001, 1999--2001, 1998--2001, and 1997--2001, respectively. The $x$-axis reports the number of person and firm effects in each sample. 
		Panel (a) shows the time to compute the KSS estimate when relying on either exact computation of $\{B_{ii},P_{ii}\}_{i=1}^n$ or the Johnson-Lindenstrauss approximation (JLA) of these numbers using a $p$ of either 500 or 2500. Panel (b) shows the resulting estimates and the plug-in estimate. Computations performed on a 32 core machine with 256 GB of dedicated memory. Source: VWH dataset.
			%The first sample contains data from the years  2000-2001, the second sample spans the years 1999-2001, the third sample spans the years 1998-2001, and the last sample spans the years 1997-2001. 
			%We then show the computation time (Panel A) and the resulting estimates (Panel B) of three algorithms. The exact algorithm computes $(B_{ii},P_{ii})$ for each observation. We then consider two versions of the JLA algorithm: one that sets the tuning parameter $p$, defined in Algorithm 3, to 500 and one that instead sets $p=2500$. Computations performed on a 32 core machine with 256 GB of dedicated memory.
	}
\end{singlespace}
\end{figure}
%exactly by computing, say, $P_{ii}=x_{i}'S_{xx}^{-1}x_{i}$
\newpage
\thispagestyle{empty}
\begin{figure}[H]
\caption{Scaling to Very Large Datasets}
\vspace{-15pt}
\label{fig:CHK_data}
\begin{center}
\subfloat[{Computation Time}]{\includegraphics[width = .8\columnwidth]{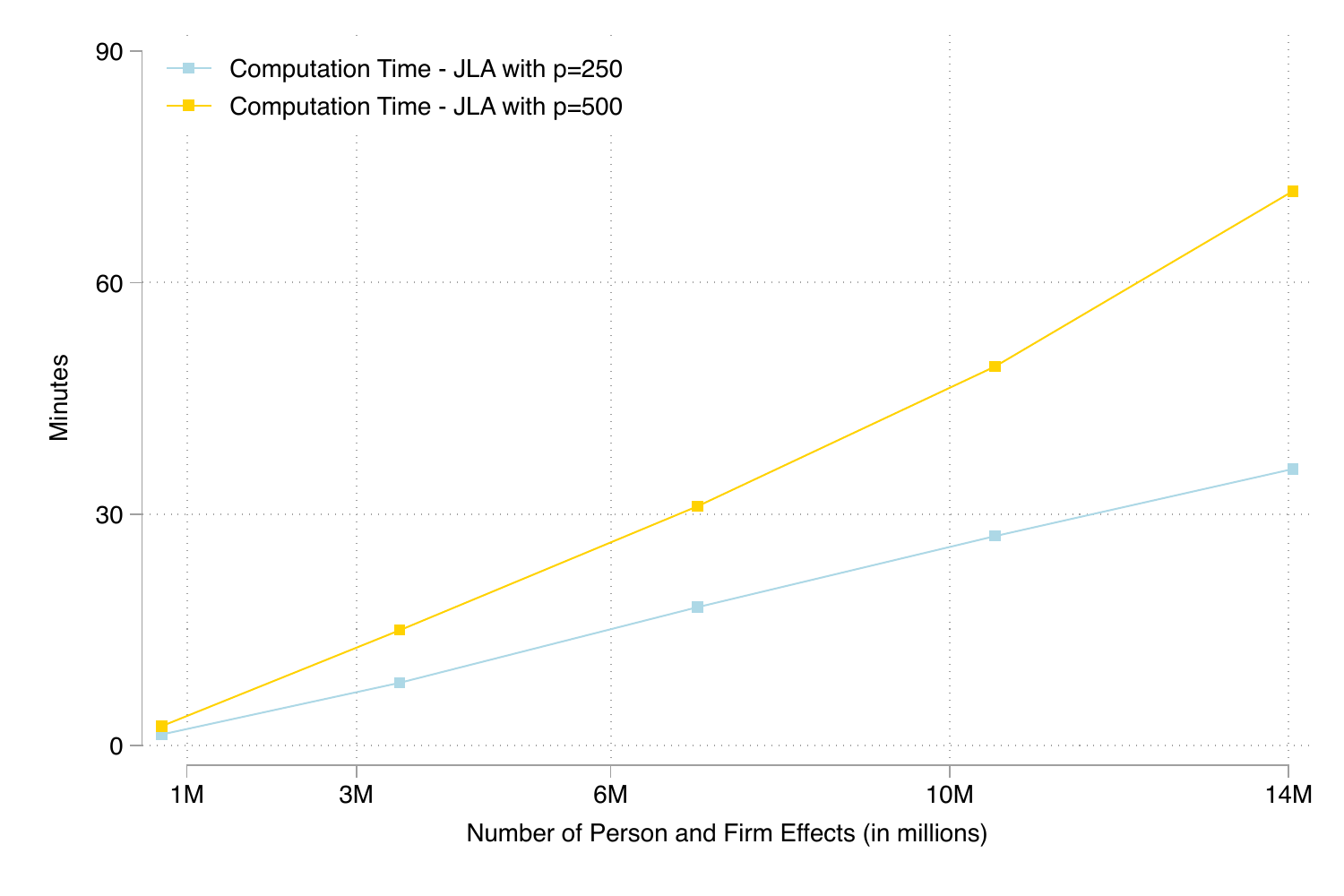}}\\
\vspace{-15pt}
\subfloat[{Quality of the Approximation}]{\includegraphics[width = .8\columnwidth]{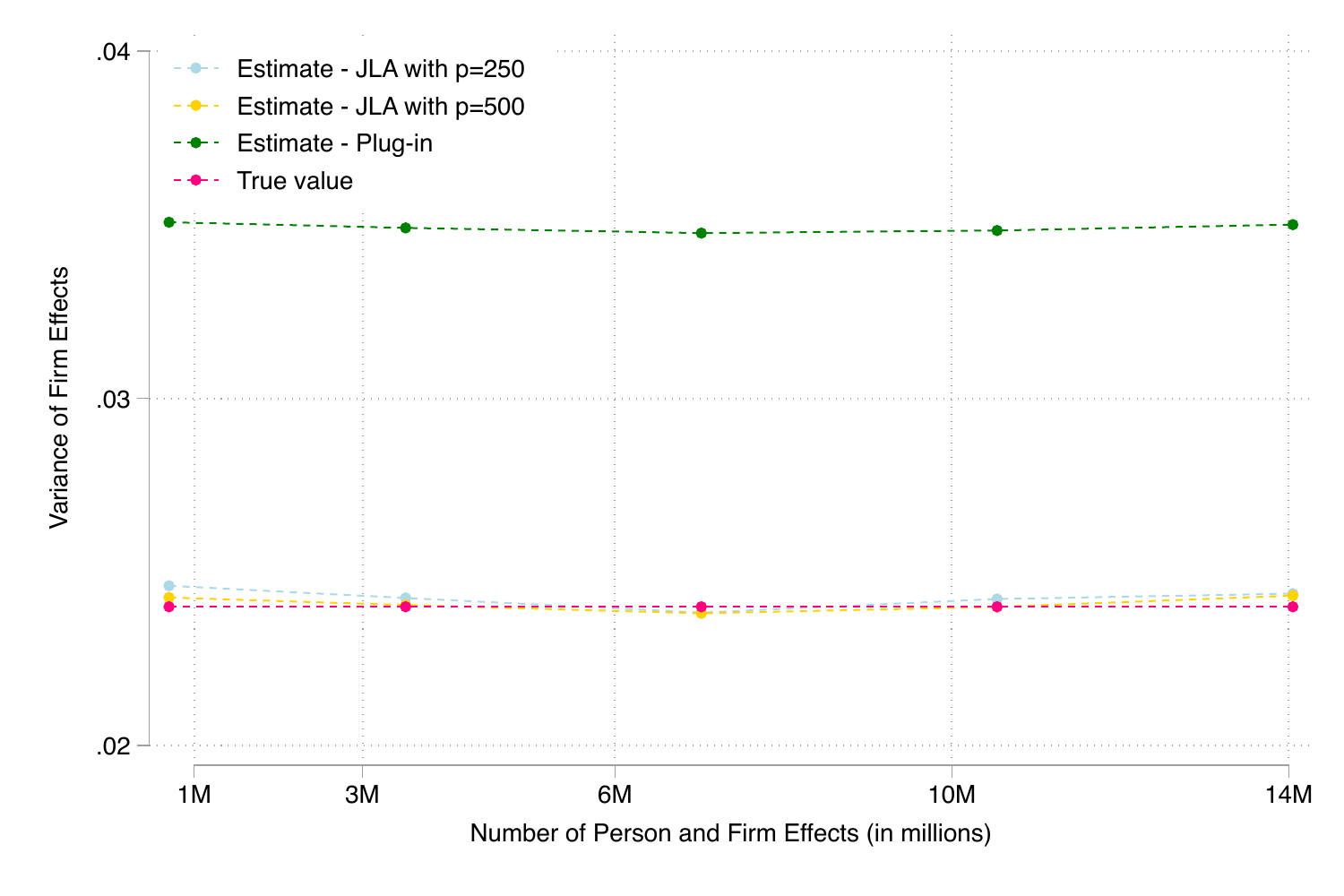}}
\end{center}
\vspace{-25pt}
\begin{singlespace}
	{\footnotesize \underline{Note:} Both panels consider synthetic datasets created from the pooled Veneto data in column 1 of Table \ref{table2} with $T=2$. It considers $\{1,5,10,15,20\}$ replicas of this sample while generating random links across replicas such that firm size and $T$ are kept fixed. Outcomes are generated from a DGP of the sort considered in Table \ref{table6}. The $x$-axis reports the number of person and firm effects in each sample.
		Panel (a) shows the time to compute the Johnson-Lindenstrauss approximation $\hat \theta_{JLA}$ using a $p$ of either 250 or 500. Panel (b) shows the resulting estimates, the plug-in estimate, and the true value of the variance of firm effects for the DGP. Computations performed on a 32 core machine with 256 GB of dedicated memory. Source: VWH dataset.
%		
%		This figure evaluates the performance of the JLA of the variance of firm effects computed across these increasing synthetic datasets, where the x-axis report the total number of firm and worker effects present in a given sample.  Panel (A) reports the computation time of the JLA when $p=250$, $p=500$. Panel (B) reports the corresponding JLA estimates of the variance of firm effects for $p=250$, $p=500$, the PI estimate along with the true value of the variance of firm effects in the DGP. Computations performed on a 32 core machine with 256 GB of dedicated memory. Source: VWH dataset.
	}
\end{singlespace}
\end{figure}

\subsection{Split Sample Estimators}
\label{sec:samplesplitalg} 
Sections \ref{sec:variance_estimation} and \ref{sec:varq} proposed standard error estimators predicated on being able to
construct independent split sample estimators $\widehat{x_i'\beta}_{-i,1}$ and $\widehat{x_i'\beta}_{-i,2}$. This section 
describes an algorithm for construction of these split sample estimators in the two-way fixed effects model of \th\ref{ex:AKM}. We restrict attention
to the case with $T_{g}=2$ and consider the model in first differences: $\Delta y_{g}=\Delta f_{g}'\psi + \Delta \varepsilon_{g}$ for $g=1,\dots,N$. When worker $g$ moves from firm $j$ to $j'$, we can estimate $\Delta f_{g}'\psi = \psi_{j'} - \psi_{j}$ without bias using OLS on any sub-sample where firms $j$ and $j'$ are connected, i.e., on any sample where there exist a path between firm $j$ and $j'$. To construct two disjoint sub-samples where firms $j$ and $j'$ are connected we therefore use an algorithm to find disjoint paths between these firms and distribute them into two sub-samples which will be denoted $\mathcal{S}_1$ and $\mathcal{S}_2$. %Using multiple paths or shorter paths can yield more precise estimators. 
Because it can be computationally prohibitive to characterize all possible paths, we use a version of Dijkstra's algorithm to find many short paths.\footnote{The algorithm presented below keeps running until it cannot find any additional paths. In our empirical implementation we stop the algorithm when it fails to find any new paths or as soon as one of the two sub-samples reach a size of at least 100 workers. We found that increasing this cap on the sub-sample size has virtually no effect on the estimated confidence intervals, but tends to increase computation time substantially.} 

Our algorithm is based on a network where firms are vertices and two firms are connected by an edge if one or more workers moved between them. This view of the network is the same as the one taken in Section \ref{sec:verify}, but different from the one used in Sections \ref{sec:pruning} and \ref{sec:pruning2} where both firms and workers were viewed as vertices. We use the adjacency matrix $\mathcal{A}$ to characterize the network in this section. To build the sub-samples $\mathcal{S}_1$ and $\mathcal{S}_2$, the algorithm successively drops workers from the network, so $\mathcal{A}_{-\mathcal{S}}$ will denote the adjacency matrix after dropping all workers in the set $\mathcal{S}$.

Given a network characterized by $\mathcal{A}$ and two connected firms $j$ and $j'$ in the network, we let $\dot P_{jj'}(\mathcal{A})$ denote the shortest path between them.\footnote{Many statistical software packages provide functions that can find shortest paths. In R they are available in the \textit{igraph} package while in MATLAB a package that builds on the work of \cite{yen1971finding} is available at \url{https://www.mathworks.com/matlabcentral/fileexchange/35397-k-shortest-paths-in-a-graph-represented-by-a-sparse-matrix-yen-s-algorithm?focused=3779015&tab=function}.} If $j$ and $j'$ are not connected $\dot P_{jj'}(\mathcal{A})$ is empty. Each edge in the path $\dot P_{jj'}(\mathcal{A})$ may have more than one worker associated with it. For each edge in $\dot P_{jj'}(\mathcal{A})$ the first step of the algorithm picks at random a single worker associated with that edge and places them in $\mathcal{S}_1$, while later steps place all workers associated with the shortest path in one of $\mathcal{S}_1$ and $\mathcal{S}_2$. This special first step ensures that the algorithm finds two independent unbiased estimators of $\Delta f_{g}'\psi $ whenever the network $\mathcal{A}$ is leave-two-out connected.

For a given worker $g$ with firm assignments $j=j(g,1),j'=j(g,2)$ and a leave-two-out connected network $\mathcal{A}$ the algorithm returns the $\{P_{g\ell,1}, P_{g\ell,2}\}_{\ell=1}^N$ introduced in Section \ref{sec:variance_estimation}. Specifically, $\widehat{\Delta f_{g}'\psi}_{-g,1} = \sum_{\ell=1}^N P_{g\ell,1} \Delta y_\ell$ and $\widehat{\Delta f_{g}'\psi}_{-g,2} = \sum_{\ell=1}^N P_{g\ell,2} \Delta y_\ell$ are independent unbiased estimators of $\Delta f_{g}'\psi$ that are also independent of $\Delta y_{g}$. If $\mathcal{A}$ is only leave-one-out connected then the algorithm may only find one path connecting $j$ and $j'$. When this happens the algorithm sets $P_{g\ell,2}=0$ for all $\ell$ as required in the formulation of the conservative standard errors proposed in Appendix \ref{sec:cons}.

 \begin{algorithm}[H]
	\caption{Split Sample Estimator for Inference}
	\begin{algorithmic}[1]
		\Function{splitsampleestimator}{$g,j,j',\mathcal{A}$}\Comment{$\mathcal{A}\equiv$ Leave-one-out connected network}
		\State Let $\mathcal{S}_{1}=\emptyset$ and  $\mathcal{S}_{2}=\emptyset$
		\State For each edge in $\dot P_{jj'}(\mathcal{A}_{-g})$, pick at random one worker from $\mathcal{A}_{-g}$ who is associated with that edge and add that worker to $\mathcal{S}_{1}$
		\State Add to $\mathcal{S}_{2}$ all workers from $\mathcal{A}_{-\{g,\mathcal{S}_{1}\}}$ who are associated with an edge in $\dot P_{jj'}(\mathcal{A}_{-\{g,\mathcal{S}_{1}\}})$ 
		\State Add to $\mathcal{S}_{1}$ all workers from $\mathcal{A}_{-\{g,\mathcal{S}_{1},\mathcal{S}_{2}\}}$ who are associated with an edge in $\dot P_{jj'}(\mathcal{A}_{-g})$
		\State Let $stop= 1\{\dot P_{jj'}(\mathcal{A}_{-\{g,\mathcal{S}_{1},\mathcal{S}_{2}\}})=\emptyset\}$ and $s=1$
		\While{$stop<1$}
		\State Add to $\mathcal{S}_{s}$ all workers from $\mathcal{A}_{-\{g,\mathcal{S}_{1},\mathcal{S}_{2}\}}$ who are associated with an edge in $\dot P_{jj'}(\mathcal{A}_{-\{g,\mathcal{S}_{1},\mathcal{S}_{2}\}})$ 
		\State Let $stop= 1\{\dot P_{jj'}(\mathcal{A}_{-\{g,\mathcal{S}_{1},\mathcal{S}_{2}\}})=\emptyset\}$ and update $s$ to  $1 + 1\{s=1\}$
		\EndWhile
		\State For $s = 1,2$ and $\ell =1,\dots,N$, let $P_{g\ell,s}=1\{\ell \in \mathcal{S}_{s}\} \Delta f_{\ell}' (\sum_{m \in \mathcal{S}_{s}} \Delta f_{m}\Delta f_{m}')^{\dagger} \Delta f_{g}$
		\State Return $\{P_{g\ell,1}, P_{g\ell,2}\}_{\ell=1}^N$
		\EndFunction
		
	\end{algorithmic}
\end{algorithm}

In line 5, all workers associated with the shortest path in line 3 are added to $\mathcal{S}_{1}$ if they were not added to $\mathcal{S}_{2}$ in line 4. This step ensures that all workers associated with $\dot P_{jj'}(\mathcal{A}_{-g})$ are used in the predictions.
In line 11, $P_{g\ell,s}$ is constructed as the weight observation $\ell$ receives in the prediction $\Delta f_{g}'\hat \psi_{s}$ where $\hat \psi_{s}$ is the OLS estimator of $\psi$ based on the sub-sample $\mathcal{S}_{s}$.

\subsection{Test of Equal Firm Effects}\label{app:testOY}
This section describes computation and interpretation of the test of the hypothesis that firm effects for ``younger" workers are equal to firm effects for the ``older" workers which applies \th\ref{rem:testbig} of the main text.

The hypothesis of interest corresponds to a restricted and unrestricted model which when written in matrix notation are
\begin{align}
\Delta y & = \Delta F\psi + \Delta \varepsilon \label{restricted} \\
\Delta y &= \Delta F_{O}\psi^{O}+\Delta F_{Y}\psi^{Y} + \Delta F_{3} \psi_{3} + \Delta \varepsilon = X\beta + \Delta \varepsilon \label{unrestricted}
\end{align}
where $\Delta y$ and $\Delta F$ collects the first differences $\Delta y_g$ and $\Delta f_g$ across $g$. $\Delta F_{O}$ represents $\Delta F$ for ``doubly connected'' firms present in each age group's leave-one-out connected set interacted with a dummy for whether the worker is ``old"; $\Delta F_{Y}$ represents $\Delta F$ for doubly connected firms interacted with a dummy for young; $\Delta F_{3}$ represents $\Delta F$ for firms that are associated with either younger movers or older movers but not both. Finally, we let $X=(\Delta F_{O}, \Delta F_{Y}, \Delta F_{3})$, $\beta=(\psi^{O\prime}, \psi^{Y\prime}, \psi_{3}')'$, and $\psi = (\psi^{O\prime},\psi_3')'$. 

The hypothesis in question is $\psi^{O} - \psi^{Y} =0$ or equivalently $R\beta = 0$ for $R=[I_{r},-I_{r},0]$ and $r = \abs{\mathcal{J}} = \text{dim}(\psi^O)$. Thus we can create the numerator of our test statistic by applying  \th\ref{rem:testbig} to \eqref{unrestricted} yielding
\begin{align}
\label{myquadratic}
\hat{\theta}=\hat{\beta}'A\hat{\beta}-\sum_{g=1}^{N}B_{gg}\hat{\sigma}_{g}^2
\end{align}
where $A=\frac{1}{r}R'(R S_{xx}^{-1}R')^{-1}R$; $B_{gg}$ and $\hat{\sigma}^2_{g}$ are defined as in Section \ref{sec:unbiased}.

%To make progress from a computational perspective we make 

Two insights help to simplify computation. First, since $\Delta F_{O}'\Delta F_{Y}=0$, $\Delta F_{O}'\Delta F_{3}=0$ and $\Delta F_{Y}'\Delta F_{3}=0$, we can estimate equation \eqref{unrestricted} via two separate regressions, one on the leave-one-out connected set for younger workers and the other on the leave-one-out connected set for older workers. We normalize the firm effects so that the same firm is dropped in both leave-one-out samples.

Second, we note that $\hat \beta'A \hat \beta = y'By$ where
\begin{align}
\label{condition_cck}
B = XS_{xx}^{-1}AS_{xx}^{-1}X'=\dfrac{P_{X}-P_{\Delta F}}{r},
\end{align}
$P_{X}=XS_{xx}^{-1}X'$, and $P_{\Delta F}=\Delta F(\Delta F'\Delta F)^{-1}\Delta F'$. Equation \eqref{condition_cck} therefore implies that $B_{ii}$ in \eqref{myquadratic} is simply a scaled difference between two statistical leverages: the first one obtained in the unrestricted model \eqref{unrestricted}, say $P_{X,gg}$, and the other on the restricted model of \eqref{restricted}, say $P_{\Delta F,gg}$. Section \ref{sec:Lambda} describes how to efficiently compute these statistical leverages. To conduct inference on the quadratic form in \eqref{myquadratic} we apply the routine described in Section \ref{sec:variance_estimation}.% and its computation in Section \ref{sec:samplesplitalg}. 

%The representation in \eqref{condition_cck} further implies that 
%\begin{align}
%	\frac{\hat \theta}{\frac{N}{r} \hat \sigma_{\Delta y}^2} = \frac{\hat \sigma^2_{X\beta}}{\hat \sigma_{\Delta y}^2}-\frac{\hat \sigma^2_{\Delta F\psi}}{\hat \sigma_{\Delta y}^2} 
%\end{align}
%where  $\hat \sigma^2_{X\beta} = \hat \sigma^2_{X\beta,\text{PI}} - \frac{1}{N}\sum_{g=1}^N (P_{X,gg} - \frac{1}{n}) \hat \sigma_g^2$, $\hat \sigma^2_{\Delta F\beta} = \hat \sigma^2_{\Delta F,\text{PI}} - \frac{1}{N}\sum_{g=1}^N (P_{\Delta F,gg} - \frac{1}{n}) \hat \sigma_g^2$, and the plug-in estimators are defined in Example 1. Thus we see that $\hat{\theta}$ is proportional to a difference between two leave-out estimators of the coefficient of determination; one for the unrestricted model \eqref{unrestricted} and one for the restricted one \eqref{restricted}. For this reason, we use $\hat{\theta}$ to test the statistical significance of $\psi^O-\psi^Y$ while we use $\frac{\hat \theta}{\frac{N}{r} \hat \sigma_{\Delta y}^2}$ to gauge the practical significance of that difference.

\section{Proofs}

	This Appendix contains all technical details and proofs that where left out of the paper. The material is primarily presented in the order it appears in the paper and under the same headings.

	\subsection{Unbiased Estimation of Variance Components}
	
	\subsubsection{Estimator}
		
	\begin{lemma}
		It follows from the Sherman-Morrison-Woodbury formula that the two representations of $\hat \theta$ given in \eqref{eq:estimator} and \eqref{eq:cov} are numerically identical, i.e., that $\hat{\beta}'A\hat{\beta}-\sum_{i=1}^{n} B_{ii}\hat{\sigma}_{i}^{2} = \sum_{i=1}^{n}y_{i}\tilde{x}_{i}'\hat{\beta}_{-i}$ whenever $S_{xx}$ has full rank and $\max_i P_{ii} <1$.
	\end{lemma}

	\begin{proof}
		The Sherman-Morrison-Woodbury formula states that if $S_{xx}$ has full rank and $P_{ii} <1$, then
		\begin{align}
			S_{xx}\inverse + \frac{S_{xx}\inverse x_i x_i' S_{xx}\inverse}{1-x_i'S_{xx}\inverse x_i} = \left(S_{xx} - x_ix_i'\right)\inverse.
		\end{align}
		Furthermore, we have that $\tilde x_i'S_{xx}\inverse x_i = x_i S_{xx}\inverse A S_{xx}\inverse x_i= B_{ii}$ so
		\begin{align}
			y_i \tilde x_i'\hat \beta_{-i} &= y_i \tilde x_i'\left(S_{xx} - x_ix_i'\right)\inverse \sum_{\ell\neq i}x_{\ell}y_{\ell} 
			=
			y_i \tilde x_i' S_{xx}\inverse \sum_{\ell\neq i}x_{\ell}y_{\ell} +  \frac{y_i \tilde x_i' S_{xx} \inverse x_i x_i' S_{xx}\inverse}{1 - x_i' S_{xx} \inverse x_i} \sum_{\ell\neq i}x_{\ell}y_{\ell} \\
			&= y_i \tilde x_i'\hat \beta - B_{ii} y_i^2 + y_i B_{ii} \underbrace{x_i' \frac{S_{xx}\inverse}{1 - x_i' S_{xx} \inverse x_i} \sum_{\ell\neq i}x_{\ell}y_{\ell}}_{=x_i'\hat \beta_{-i}}
			= y_i \tilde x_i'\hat \beta - B_{ii} y_i(y_i - x_i'\hat \beta_{-i})
		\end{align}
		where the last expression equals $y_i \tilde x_i'\hat \beta - B_{ii} \hat \sigma_i^2$. This finishes the proof since $\hat{\beta}'A\hat{\beta} = \sum_{i=1}^{n}y_{i}\tilde{x}_{i}'\hat{\beta}$. In the above the Sherman-Morrison-Woodbury formula was also used to establish that
		\begin{align}
			x_i'\hat \beta_{-i} &= x_i'\left(S_{xx} - x_ix_i'\right)\inverse \sum_{\ell\neq i}x_{\ell}y_{\ell} 
			= x_i' \frac{S_{xx}\inverse}{1 - x_i' S_{xx} \inverse x_i} \sum_{\ell\neq i}x_{\ell}y_{\ell},
		\end{align}
		and from this it follows that $y_{i}-x_{i}'\hat{\beta}_{-i}=\dfrac{y_i - x_i'\hat \beta}{1-P_{ii}}$ as claimed in the paper.
	\end{proof}

	\subsubsection{Large Scale Computation}
	All discussions of the computational aspects are collected in Appendix \ref{app:computation}.
	
	\subsubsection{Relation To Existing Approaches}\label{app:rel}
		Next we verify that the bias of $\hat \theta_{\text{HO}}$ is a function of the covariation between $\sigma_i^2$ and $(B_{ii},P_{ii})$.
		\begin{lemma}\th\label{lem:HObias}
		The bias of $\hat \theta_{\text{HO}}$ is $\sigma_{nB_{ii},\sigma_i^2} +S_B \frac{n}{n-k} \sigma_{P_{ii},\sigma_i^2}$ where
			\begin{align}
				\sigma_{nB_{ii},\sigma_i^2} = \sum_{i=1}^n B_{ii}(\sigma_i^2 - \bar \sigma^2), \ \ \bar \sigma^2 = \frac{1}{n} \sum_{i=1}^n \sigma_i^2, \ \ S_B = \sum_{i=1}^n B_{ii}, \ \
				\sigma_{P_{ii},\sigma_i^2} = \frac{1}{n}\sum_{i=1}^n P_{ii} (\sigma_i^2 - \bar \sigma^2).
			\end{align}
		\end{lemma}
	
		\begin{proof}
			Since $\hat{\sigma}^{2} = \frac{1}{n-k} \sum_{i=1}^n (y_i - x_i' \hat \beta)^2 = \frac{1}{n-k} \sum_{i=1}^n \sum_{\ell =1}^n M_{i\ell} \varepsilon_i \varepsilon_\ell$ we get that
			\begin{align}
				\E[\hat \theta_{\text{HO}}] - \theta &= \sum_{i=1}^n B_{ii} \sigma_i^2 - \left(\sum_{i=1}^n B_{ii}\right)\frac{1}{n-k} \sum_{i=1}^n M_{ii} \sigma_i^2  \\
				&=\sum_{i=1}^n B_{ii} (\sigma_i^2 - \bar \sigma^2) - S_B \frac{1}{n-k} \sum_{i=1}^n M_{ii} (\sigma_i^2-\bar \sigma^2) \\
				&=\sigma_{nB_{ii},\sigma_i^2} + S_B \frac{n}{n-k} \sigma_{P_{ii},\sigma_i^2}. \qedhere
			\end{align}
		\end{proof}
	
		\subsubsection*{Comparison to Jackknife Estimators}\label{app:JK}
		
		This subsection compares the leave-out estimator $\hat \theta$ to estimators predicated on jackknife bias corrections. We start by introducing some of the high-level assumptions that are typically used to motivate jackknife estimators. We then consider some variants of Examples 2 and 3 where these high-level conditions fail to hold and establish that the jackknife estimators have first order biases while the leave-out estimator retains consistency. 
		%We also provide a simple example where the homoscedasticity-only estimator has a first order bias. 
		
		\noindent \textbf{\textit{High-level Conditions}}
		Jackknife bias corrections are typically motivated by the high-level assumption that the bias of a plug-in estimator $\hat \theta_\text{PI}$ shrinks with the sample size in a known way and that the bias of $\frac{1}{n} \sum_{i=1}^n \hat \theta_{\text{PI},-i}$ depends on sample size in an identical way, i.e.,
		\begin{align}\label{eq:jk1}
		\E[\hat \theta_{\text{PI}}] = \theta + \frac{\text{D}_1}{n} + \frac{\text{D}_2}{n^2}, 
		\quad 
		\E\left[ \frac{1}{n} \sum_{i=1}^n \hat \theta_{\text{PI},-i}\right] = \theta + \frac{\text{D}_1}{n-1} + \frac{\text{D}_2}{(n-1)^2}
		\quad \text{for some } \text{D}_1,\text{D}_2.
		\end{align}
		Under \eqref{eq:jk1}, the jackknife estimator $\hat \theta_\text{JK} = n\hat \theta_{\text{PI}} - \frac{n-1}{n} \sum_{i=1}^n \hat \theta_{\text{PI},-i}$ has a bias of $-\frac{\text{D}_2}{n(n-1)}$. 
		
		For some long panel settings the bias in $\hat \theta_{\text{PI}}$ is shrinking in the number of time periods $T$ such that
		\begin{align}
		\E[\hat \theta_{\text{PI}}] = \theta + \frac{\dot{\text{D}}_1}{T} + \frac{\dot{\text{D}}_2}{T^2} \quad \text{for some } \dot{\text{D}}_1,\dot{\text{D}}_2.
		\end{align}
		In such settings, it may be that the biases of $\frac{1}{T} \sum_{t=1}^T \hat \theta_{\text{PI},-t}$ and $\frac{1}{2}(\hat \theta_{\text{PI},1} + \hat \theta_{\text{PI},2})$ depend on $T$ in an identical way, i.e.,
		\begin{align}
		\E\left[ \frac{1}{T} \sum_{t=1}^T \hat \theta_{\text{PI},-t} \right] =\theta + \frac{\dot{\text{D}}_1}{T-1} + \frac{\dot{\text{D}}_2}{(T-1)^2}
		\quad \text{and} \quad
		\E\left[ \frac{1}{2}(\hat \theta_{\text{PI},1} + \hat \theta_{\text{PI},2})\right] = \theta + \frac{2\dot{\text{D}}_1}{T} + \frac{4\dot{\text{D}}_2}{T^2}.
		\end{align}
		From here it follows that the panel jackknife estimator $\hat \theta_\text{PJK} = T\hat \theta_{\text{PI}} - \frac{T-1}{T} \sum_{t=1}^T \hat \theta_{\text{PI},-t}$ has a bias of $-\frac{\dot{\text{D}}_2}{T(T-1)}$ and that the split panel jackknife estimator $\hat \theta_\text{SPJK} = 2\hat \theta_{\text{PI}} - \frac{1}{2}(\hat \theta_{\text{PI},1} + \hat \theta_{\text{PI},2})$ has a bias of $-\frac{2\dot{\text{D}}_2}{T^2}$, both of which shrink faster to zero than $\frac{\dot{\text{D}}_1}{T}$ if $T \rightarrow \infty$. Typical sufficient conditions for bias-representations of this kind to hold (to second order) are that (i) $T \rightarrow \infty$, (ii) the design is stationary over time, and (iii) that $\hat \theta_{\text{PI}}$ is asymptotically linear \cite[see, e.g.,][]{hahn2004jackknife,dhaene2015split}. Below we illustrate that jackknife corrections can be inconsistent in Examples 2 and 3 when (i) and/or (ii) do not hold. Finally we note that $\hat \theta_{\text{PI}}$ (a quadratic function) need not be asymptotically linear as is evident from the non-normal asymptotic distribution of $\hat \theta$ derived in \th\ref{thm4} of this paper.
		
		\subsubsection*{Examples of Jackknife Failure}
		
		\begin{customthm}{2}[Special case]
			Consider the model
			\begin{align}
			y_{{g} t} =\alpha_{{g}} + \varepsilon_{{g} t} && ({g}=1,\dots,N, \ t = 1,\dots,T \ge 2),
			\end{align}
			where $\sigma^2_{gt} = \sigma^2$ and suppose the parameter of interest is $\theta = \frac{1}{N} \sum_{g=1}^N \alpha_g^2.$ For $T$ even, we have the following bias calculations:
			\begin{align}
			\E[\hat \theta_{\text{PI}}] &= \theta + \frac{\sigma^2}{T}, &
			\E\left[ \frac{1}{n} \sum_{i=1}^n \hat \theta_{\text{PI},-i}\right] &= \theta + \frac{\sigma^2}{T} + \frac{\sigma^2}{n(T-1)}, \\
			\E\left[ \frac{1}{T} \sum_{t=1}^T \hat \theta_{\text{PI},-t} \right] &=\theta + \frac{\sigma^2}{T-1}, &
			\E\left[ \frac{1}{2}(\hat \theta_{\text{PI},1} + \hat \theta_{\text{PI},2})\right] &= \theta + \frac{2\sigma^2}{T}.
			\end{align}
			The jackknife estimator $\hat \theta_{JK}$ has a first order bias of $-\frac{\sigma^2}{T(T-1)}$, which when $T=2$ is as large as that of $\hat \theta_{\text{PI}}$ but of opposite sign. By contrast, both of the panel jackknife estimators, $\hat \theta_{PJK}$ and the leave-out estimator are exactly unbiased and consistent as $n \rightarrow \infty$ when $T$ is fixed.
		\end{customthm}
		
		This example shows that the jackknife estimator can fail when applied to a setting where the number of regressors is large relative to sample size. Here the number of regressors is $N$ and the sample size is $NT$, yielding a ratio of $1/T$ and we see that $1/T \rightarrow 0$ is necessary for consistency of $\hat \theta_\text{JK}$. While the panel jackknife corrections appear to handle the presence of many regressors, this property disappears in the next example which adds the ``random coefficients'' of Example 3.
		
		\begin{customthm}{3}[Special case]
			Consider the model
			\begin{align}
			y_{{g} t} =\alpha_{{g}} + x_{gt}\delta_g + \varepsilon_{{g} t} && ({g}=1,\dots,N, \ t = 1,\dots,T \ge 3)
			\end{align}
			where $\sigma^2_{gt} = \sigma^2$ and $\theta = \frac{1}{N} \sum_{g=1}^N \delta_g^2$.
			
			An analytically convenient example arises when the regressor design is ``balanced'' across groups as follows:
			\begin{align}
			(x_{g 1},x_{g 2},\dots,x_{g T}) = (x_{1},x_{2},\dots,x_{T}),
			\end{align}
			where $x_1,x_2,x_3$ take distinct values and $\sum_{t=1}^T x_t=0$. The leave-out estimator is unbiased and consistent for any $T \ge 3$, whereas for even $T \ge 4$ we have the following bias calculations:
			\begin{align}
			\E[\hat \theta_{\text{PI}}] &= \theta + \frac{\sigma^2}{\sum_{t=1}^T x_t^2}, \\
			\E\left[ \frac{1}{T} \sum_{t=1}^T \hat \theta_{\text{PI},-t} \right] &=\theta + \frac{\sigma^2}{T} \sum_{t=1}^T \frac{1}{\sum_{s \neq t} ( x_s - \bar x_{-t})^2} , \\
			\E\left[ \frac{1}{2}(\hat \theta_{\text{PI},1} + \hat \theta_{\text{PI},2})\right] &= \theta + \frac{\sigma^2}{2\sum_{t=1}^{T/2}(x_t-\bar x_1)^2} + \frac{\sigma^2}{2\sum_{t=T/2+1}^{T}(x_t-\bar x_2)^2},
			\end{align}
			where $\bar x_{-t} = \frac{1}{T-1} \sum_{s \neq t} x_s$, $\bar x_1 = \frac{2}{T} \sum_{t=1}^{T/2} x_t$, and $\bar x_2 = \frac{2}{T} \sum_{t=T/2+1}^{T} x_t$. 
			
			The calculations above reveal that non-stationarity in either the level or variability of $x_t$ over time can lead to a negative bias in panel jackknife approaches, e.g.,
			\begin{align}
			\E\left[\hat \theta_\text{SPJK}\right] - \theta \le \frac{2\sigma^2}{\sum_{t=1}^T x_t^2} -\frac{\sigma^2}{2\sum_{t=1}^{T/2}x_t^2} - \frac{\sigma^2}{2\sum_{t=T/2+1}^{T}x_t^2}  \le 0
			\end{align}
			where the first inequality is strict if $\bar x_1 \neq \bar x_2$ and the second if $\sum_{t=1}^{T/2}x_t^2 \neq \sum_{t=T/2+1}^T x_t^2$. In fact, the following example
			\begin{align}
			(x_{1},x_{2},\dots,x_{T}) = (-1,2,0,\dots,0,-1)
			\end{align}
			renders the panel jackknife corrections inconsistent for small or large $T$:
			\begin{align}
			\E[\hat \theta_\text{PJK} ] &= \theta - \frac{7/5}{6}\sigma^2 + O\left(\frac{1}{T}\right)
			\quad \text{and} \quad 
			\E[\hat \theta_\text{SPJK} ] = \theta - \frac{8/5}{6}\sigma^2 + O\left(\frac{1}{T}\right).
			\end{align}
			Inconsistency results here from biases of first order that are negative and larger in magnitude than the original bias of $\hat \theta_{\text{PI}}$ (which is $\frac{\sigma^2}{6}$).

		\end{customthm}

		\noindent \textbf{\textit{Computations}}
		For this special case of example 2 we have that $A = \frac{I_N}{N}$ and $S_{xx} = T I_N$ so that $\tilde A = \frac{I_N}{NT}$ and $\text{trace}(\tilde A^2) = \frac{1}{NT^2} = o(1)$ which implies consistency of $\hat \theta$. Similarly we have that the bias of $\tilde \theta$ is
		\begin{align}
		&\frac{1}{n} \sum_{g=1}^N T_g \mathbb{V}[\hat \alpha_g]=\frac{1}{n} \sum_{g=1}^N \sigma^2 = \frac{\sigma^2}{T} \qquad \text{where } \hat \alpha_g = \frac{1}{T_g} \sum_{t=1}^{T_g} y_{gt}.
		\end{align}
		The same types of calculations lead to the other biases reported in the paper.
		
		For this special case of example 3 we have that $A = \begin{bmatrix} 0 & 0 \\ 0 & \frac{I_N}{N}\end{bmatrix}$ and $S_{xx} = \begin{bmatrix} T I_N & 0 \\ 0 & I_N \sum_{t=1}^T x_t^2 \end{bmatrix}$ which implies that $\text{trace}(\tilde A^2) = \frac{1}{N\left(\sum_{t=1}^T x_t^2\right)^2} = o(1)$ and therefore consistency of $\hat \theta$. Similarly we have that the bias of $\tilde \theta$ is
		\begin{align}
		&\frac{1}{n} \sum_{g=1}^N T_g \mathbb{V}[\hat \delta_g] = \frac{\sigma^2}{\sum_{t=1}^T x_t^2} \qquad \text{where } \hat \delta_g = \frac{\sum_{t=1}^{T_g} x_t y_{gt}}{\sum_{t=1}^T x_t^2}.
		\end{align}
		The same types of calculations lead to the other biases reported above. Now for the numerical example $(x_{1},x_{2},\dots,x_{T}) = (-1,2,0,\dots,0,-1)$ we have $\sum_{t=1}^T x_t^2 = 6$, $\sum_{t=T/2+1}^{T}(x_t-\bar x_2)^2  = 1 - \frac{2}{T},$ $\sum_{t=1}^{T/2}(x_t-\bar x_1)^2 = 2\sum_{t=1}^{T/2}x_t^2-T \bar x_1^2 = 5 - \frac{2}{T},$ and
		\begin{align}
		\sum_{s \neq t} ( x_s - \bar x_{-t})^2 &= \begin{cases}
		2 - \frac{4}{T-1} & \text{if } t =2, \\
		5 - \frac{1}{T-1} & \text{if } t \in \{1,T\}, \\
		6 & \text{otherwise},
		\end{cases}
		\end{align}
		Thus
		\begin{align}
		\E[\hat \theta_\text{PJK} ] - \theta &= \frac{T\sigma^2}{\sum_{t=1}^T x_t^2} - \sigma^2\frac{(T-1)}{T} \sum_{t=1}^T \frac{1}{\sum_{s \neq t} ( x_s - \bar x_{-t})^2} \\
		&=\sigma^2\frac{T}{6} - \sigma^2\frac{T-1}{T}\left(\frac{2}{5 - \frac{1}{T-1}} + \frac{1}{2 - \frac{4}{T-1}} + \frac{T-3}{6}\right)  \\
		&=\sigma^2\left(\frac{2}{3} - \frac{4}{6T} - \frac{T-1}{T}\frac{2}{5 - \frac{1}{T-1}} - \frac{T-1}{T}\frac{1}{2 - \frac{4}{T-1}} \right) 
		=  - \frac{7}{30}\sigma^2 + O\left(\frac{1}{T}\right) \\
		\quad \text{and} \quad 
		\E[\hat \theta_\text{SPJK} ] - \theta &= \frac{2\sigma^2}{\sum_{t=1}^T x_t^2} - \frac{\sigma^2}{2\sum_{t=1}^{T/2}(x_t-\bar x_1)^2} + \frac{\sigma^2}{2\sum_{t=T/2+1}^{T}(x_t-\bar x_2)^2} \\
		&= \sigma^2\left( \frac{1}{3} - \frac{1}{10 - \frac{4}{T}} - \frac{1}{2 - \frac{4}{T}} \right) =  - \frac{8}{30}\sigma^2 + O\left(\frac{1}{T}\right).
		\end{align}

	\subsubsection{Finite Sample Properties}
	
	Here we provide a restatement and proof of \th\ref{lem:unbiased,lem:fin} together with a characterization of the finite sample distribution of $\hat \theta$ which was excluded from the main text.
	
		\begin{lemma}\th\label{lem:finS}
			Recall that $\theta^* = \hat \beta'A\hat \beta - \sum_{i=1}^n B_{ii}\sigma_i^2$.
			\begin{enumerate}
				\item If $\max_i P_{ii} <1$, then $\E[\hat \theta] = \theta$.
				
				\item Unbiased estimators of $\theta = \beta'A\beta$ exist for all $A$ if and only if $\max_{i} P_{ii} < 1$.
				
				\item  If $\varepsilon_{i} \sim \mathcal{N}(0, \sigma_i^2)$, then $\theta^* = \sum_{\ell = 1}^{r} \lambda_\ell\left(\hat b_\ell^2-\mathbb{V}[\hat b_\ell]\right)$ and $\hat b \sim \mathcal{N}\left( b, \mathbb{V}[\hat b] \right)$.
				
				\item If $\max_i P_{ii} < 1$ and $\varepsilon_{i} \sim \mathcal{N}(0, \sigma_i^2)$, then $\hat \theta = \sum_{\ell=1}^{r_C} \lambda_\ell\left( C \right)\left(\hat y_\ell^2-V_{\ell\ell}\right)$ where $\hat y \sim \mathcal{N}\left( \mu , V \right)$, $\mu = Q'_{C} X\beta$, $V = Q_C' \varSigma Q_C$, $C = (C_{i\ell})_{i,\ell}$, $\varSigma = \text{diag}(\sigma_1^2,\dots,\sigma_n^2)$, and $C = Q_C D_C Q_C'$ is a spectral decomposition of $C$ such that $D_C = \text{diag}(\lambda_{1}(C),\dots,\lambda_{r_C}(C)$ and $r_C$ is the rank of $C$.
			\end{enumerate}
	\end{lemma}
	
	\begin{proof}
		First note that $\hat \beta' A \hat \beta =\sum_{i=1}^n \sum_{\ell=1}^n B_{i\ell}y_{i} y_\ell$ and $\hat \sigma_i^2 = y_i(y_i - x_i' \hat \beta_{-i}) = y_i M_{ii}\inverse \sum_{\ell=1}^n M_{i\ell} y_\ell$, so
		\begin{align}
			\hat \theta 
			&= \sum_{i=1}^n \sum_{\ell=1}^n B_{i\ell}y_{i} y_\ell - B_{ii} M_{ii}\inverse M_{i\ell} y_i y_\ell \\
			&= \sum_{i=1}^n \sum_{\ell=1}^n \left(B_{i\ell}- 2 \inverse M_{i\ell}\left( B_{ii} M_{ii}\inverse + B_{\ell\ell} M_{\ell\ell}\inverse \right) \right)y_i y_\ell 
			= \sum_{i=1}^n \sum_{\ell \neq i} C_{i\ell} y_{i} y_\ell.
		\end{align}
		The errors are mean zero and uncorrelated across observations, so
		\begin{align}
			\E[\hat \theta] = \sum_{i=1}^n \sum_{\ell \neq i} C_{i\ell} x_{i}'\beta x_\ell'\beta = \sum_{i=1}^n \sum_{\ell=1}^n B_{i\ell}x_{i}'\beta x_\ell'\beta - B_{ii} M_{ii}\inverse M_{i\ell} x_{i}'\beta x_\ell'\beta = \theta,
		\end{align}
		since $\sum_{i=1}^n \sum_{\ell=1}^n B_{i\ell}x_{i} x_\ell'=A$ and $\sum_{\ell=1}^n M_{i\ell} x_\ell =0$. This shows the first claim of the lemma.
		
		It suffices to show that no unbiased estimator of $\beta' S_{xx} \beta$ exist when $\max_i P_{ii}=1$. Any potential unbiased estimator must have the representation $y' D y + U$ where $\E[U]=0$ and $D = (D_{i\ell})_{i,\ell}$ satisfies (i) $D_{ii}=0$ for all $i$ and (ii) $X'DX=S_{xx}$ for $X = (x_1,\dots,x_n)'$. (ii) implies that $D$ must be $D = I + P\tilde D M + M\tilde DP + M\tilde DM$ for some $\tilde D$ where $P=(P_{i\ell})_{i,\ell}$ and $M=(M_{i\ell})_{i,\ell}$. If the exist a $i$ with $P_{ii}=1$, then $\sum_{\ell =1}^n P_{i\ell}^2 = P_{ii}$ yields $M_{i\ell}=0$ for all $\ell$ which implies that $D_{ii}$ must equal $1$ to satisfy (ii). However, this makes it impossible to satisfy (i). This shows the second claim.
		
		Recall the spectral decomposition $\tilde A=Q D Q'$ and definition of $\hat b = Q' S_{xx}^{1/2} \hat \beta$ which satisfies that $\hat b \sim \mathcal{N}(b,\mathbb{V}[\hat b])$ when $\varepsilon_{i} \sim \mathcal{N}(0, \sigma_i^2)$.
		We have that $\theta^* =  \sum_{\ell =1}^r \lambda_\ell \left( \hat b_\ell^2 - \mathbb{V}[\hat b_\ell]\right)$ since
		\begin{align}
			\hat \beta' A \hat \beta &= \hat \beta' S_{xx}^{1/2} \tilde A S_{xx}^{1/2} \hat \beta = \hat b' D \hat b = \sum_{\ell =1}^r \lambda_\ell \hat b_\ell^2,
			\shortintertext{and}
			\sum_{i=1}^n B_{ii} \sigma_i^2 &= \text{trace}(B \varSigma) = \text{trace}(A\mathbb{V}[\hat \beta]) = \text{trace}(D\mathbb{V}[\hat b]) = \sum_{\ell =1}^r \lambda_\ell\mathbb{V}[\hat b_\ell].
		\end{align}
		where $B = (B_{i\ell})_{i,\ell}$. This shows the third claim.
		
		The matrix $C$ is is well-defined as $\max_i P_{ii}<1$. Define $\hat y = Q_{C}' (y_1,\dots,y_n)'$ which satisfies that $\hat y \sim \mathcal{N}(\mu,V)$ when $\varepsilon_{i} \sim \mathcal{N}(0, \sigma_i^2)$. Furthermore,
		\begin{align}
			\hat \theta &=  y'  C  y = \hat y' D_C \hat y = \sum_{\ell=1}^{r_C} \lambda_{\ell}(C) \hat y_\ell^2,
		\end{align}
		and $C_{ii}=0$ for all $i$, so that $\sum_{\ell} \lambda_\ell\left( C \right)V_{\ell\ell}= \text{trace}(C\varSigma) =0$. This shows the last claim.
	\end{proof}

	\subsubsection{Consistency}
	
	The next result provides a restatement and proof of \th\ref{lem:cons}.
	
		\begin{lemma}\th\label{lem:consS}
			If \th\ref{ass:reg} and one of the following conditions hold, then $\hat{\theta} - \theta \overset{p}{\rightarrow} 0$.
			\begin{enumerate}
				\item[(i)] $A$ is positive semi-definite, $\theta = \beta'A\beta= O(1)$, and $\text{trace}(\tilde A^2) = \sum_{\ell =1}^r \lambda_{\ell}^2= o(1)$.
				
				\item[(ii)] $A=\frac{1}{2}( A_1'A_2 + A_2'A_1)$ where $\theta_1 = \beta'A_1'A_1\beta,$ $\theta_2 = \beta'A_2'A_2\beta$ satisfy (i).
			\end{enumerate}
			
		\end{lemma}
	
		\begin{proof}

			Suppose that $A$ is positive semi-definite. The difference between $\hat \theta$ and $\theta$ is
			\begin{align}
			\hat \theta -\theta &= 2 \sum_{i=1}^n \sum_{\ell =1}^n B_{i\ell} x_{\ell}'\beta \varepsilon_i + \sum_{i=1}^n \sum_{\ell \neq i} B_{i\ell} \varepsilon_{i} \varepsilon_\ell + \sum_{i=1}^n B_{ii}(\varepsilon_i^2 - \hat \sigma_i^2),
			\end{align}
			and each term has mean zero so we show that their variances are small in large samples. The variance of the first term is
			\begin{align}
			4 \sum_{i=1}^n \left(\sum_{\ell =1}^n B_{i\ell} x_{\ell}'\beta\right)^2 \sigma_i^2 
			&\le 4\max_i \sigma_i^2 \beta' X' B^2 X \beta = 4\max_i \sigma_i^2 \beta' A S_{xx}\inverse A \beta 
			\le 4\max_i \sigma_i^2 \theta \lambda_1= o(1)
			\end{align}
			where $B = (B_{i\ell})_{i,\ell}$, the last inequality follows from positive semi-definiteness of $A$, and the last equality follows from $\theta = O(1)$ and $\lambda_1 \le \text{trace}(\tilde A^2)^{1/2} = o(1)$. The variance of the second term is 
			\begin{align}
			2\sum_{i=1}^n \sum_{\ell \neq i} B_{i\ell}^2 \sigma^2_{i} \sigma^2_\ell \le 2\max_i \sigma_i^4 \sum_{i=1}^n \sum_{\ell =1}^n B_{i\ell}^2 = 2\max_i \sigma_i^4 \text{trace}(\tilde A^2) = o(1).
			\end{align}
			Finally, the variance of the third term is
			\begin{align}
			& \sum_{i=1}^n \left(\sum_{\ell = 1}^n M_{ll}\inverse B_{\ell \ell}  M_{i\ell} x_\ell'\beta \right)^2 \sigma_i^2 + 2\sum_{i=1}^n \sum_{\ell \neq i} M_{ii}^{-2}B_{ii}^2 M_{i\ell}^2 \sigma^2_{i} \sigma^2_\ell \\
			\le & \frac{1}{c^2} \max_i \sigma_i^2 \max_i (x_i'\beta)^2 \sum_{i=1}^n B_{ii}^2+ \frac{2}{c}\max_i \sigma_i^4 \sum_{i=1}^n B_{ii}^2 = o(1)
			\end{align}
			where $\min_i M_{ii} \ge c >0$ and $\sum_{i=1}^n B_{ii}^2 \le \text{trace}(\tilde A^2) = o(1)$. This shows the first claim of the lemma.
			
			When $A$ is non-definite, we write $A = \frac{1}{2}\left( A_1'A_2 + A_2'A_1\right)$ and note that
			\begin{align}
			\beta' A S_{xx}\inverse A \beta  \le \frac{1}{2}\left( \varTheta_1 \lambda_{\max}(\tilde A_2) + \varTheta_2 \lambda_{\max}(\tilde A_1)  \right) \quad \text{and} \quad  
			\text{trace}(\tilde A^2) \le \text{trace}(\tilde A_1^2)^{1/2}\text{trace}(\tilde A_2^2)^{1/2}
			\end{align}
			where $\tilde A_\ell = S_{xx}^{-1/2} A_k'A_k S_{xx}^{-1/2}$ for $\ell=1,2$ and $\lambda_{\max}(\tilde A_2)$ is the largest eigenvalue of $\tilde A_2$. Thus consistency of $\hat \theta$ follows from $\varTheta_1 = O(1)$, $\varTheta_2 = O(1)$, $\text{trace}(\tilde A_1^2)=o(1)$, and $\text{trace}(\tilde A_2^2) = o(1)$.
		\end{proof}
	
	The next result provides a restatement and proof of \th\ref{lem:JLA}.
	
	\begin{lemma}\th\label{lem:JLAS}
			If \th\ref{ass:reg}, $n/p^4 = o(1)$, $\mathbb{V}[\hat  \theta]\inverse = O(n)$, and one of the following conditions hold, then $\mathbb{V}[\hat  \theta]^{-1/2}{(\hat \theta_{JLA} - \hat \theta - \mathrm{B}_p)}{} = o_p(1)$ where $\abs{\mathrm{B}_p} \le \frac{1}{p} \sum_{i=1}^n P_{ii}^2 \abs{B_{ii}} \sigma_i^2$.
		
		\begin{enumerate}
			\item[(i)] $A$ is positive semi-definite and $\E[\hat \beta'A\hat \beta] - \theta = O(1)$.
			
			\item[(ii)] $A=\frac{1}{2}( A_1'A_2 + A_2'A_1)$ where $\theta_1 = \beta'A_1'A_1\beta,$ $\theta_2 = \beta'A_2'A_2\beta$ satisfy (i) and $\frac{\mathbb{V}[\hat  \theta_1]\mathbb{V}[\hat  \theta_2]}{n\mathbb{V}[\hat  \theta]^2} = O(1)$.
		\end{enumerate}
	\end{lemma}

	\begin{proof}
		Define $\mathrm{B}_p = \frac{1}{p} \sum_{i=1}^n B_{ii} \sigma_i^2  \frac{2\sum_{\ell \neq i}^n P_{i\ell}^4 - P_{ii}^2(1-P_{ii})^2}{(1-P_{ii})^2}$. Letting $(\hat \theta_{JLA} - \hat \theta)_2$ be a second order approximation of $\hat \theta_{JLA} - \hat \theta$, we first show that $\E\left[(\hat \theta_{JLA} - \hat \theta)_2\right] = \mathrm{B}_p $ and $\frac{\mathbb{V}[(\hat \theta_{JLA} - \hat \theta)_2]}{\mathbb{V}[\hat  \theta]} = O(\frac{1}{p})$. Then we finish the proof of the first claim by showing that the approximation error is ignorable. The bias bound follows immediately from the equality $\sum_{\ell \neq i}^n P_{i\ell}^2 = P_{ii}(1-P_{ii})$ which leads to $0 \le \sum_{\ell \neq i}^n P_{i\ell}^4 \le P_{ii}^2(1-P_{ii})^2$.

		We have $\hat \theta_{JLA} - \hat \theta = (\hat \theta_{JLA} - \hat \theta)_2 + AE_2$ where 
		\begin{align}
			(\hat \theta_{JLA} - \hat \theta)_2 &= \sum_{i=1}^n \hat \sigma_i^2\left(  B_{ii} -  \hat B_{ii}  - \hat B_{ii} \hat a_i - \hat B_{ii} \left(\hat a_i^2 - \frac{1}{p}\frac{3P_{ii}^3 + P_{ii}^2}{1-P_{ii}} \right) \right)
		\end{align}
		for $\hat a_i = \frac{\hat P_{ii} - P_{ii}}{1-P_{ii}}$ and approximation error
		\begin{align}
			AE_2 = \sum_{i=1}^n \hat \sigma_i^2 \hat B_{ii}\left(\frac{1}{p}\frac{3 \hat P_{ii}^2 + \hat P_{ii}^2 -(3P_{ii}^2 + P_{ii}^2)(1+\hat a_i)^2}{(1+\hat a_i)^2(1-P_{ii})} -\frac{\hat a_i^3}{1+\hat a_i} \right).
		\end{align}
		For the mean calculation involving $(\hat \theta_{JLA} - \hat \theta)_2$ we use independence between $\hat B_{ii}$, $\hat P_{ii}$, and $\hat \sigma_i^2$, unbiasedness of $\hat B_{ii}$, $\hat P_{ii}$, and $\hat \sigma_i^2$, and the variance formula
		\begin{align}
			\mathbb{V}[\hat a_i] &= \frac{2}{p} \frac{P_{ii}^2 - \sum_{\ell = 1}^n P_{i\ell}^4}{(1-P_{ii})^2} = \frac{1}{p}\frac{3P_{ii}^3 + P_{ii}^2}{1-P_{ii}} + \frac{P_{ii}^2(1-P_{ii})^2 - 2\sum_{\ell \neq i}^n P_{i\ell}^4}{p(1-P_{ii})^2}.
		\end{align}
		Taken together this implies that
		\begin{align}
			\E\left[(\hat \theta_{JLA} - \hat \theta)_2\right] = -\sum_{i=1}^n B_{ii} \sigma_i^2 \left( \mathbb{V}[\hat \alpha_i] - \frac{1}{p}\frac{3P_{ii}^3 + P_{ii}^2}{1-P_{ii}} \right) = \mathrm{B}_p.
		\end{align}
		For the variance calculation we proceed term by term. We have for $y=(y_i,\dots,y_n)'$ that
		{\footnotesize
		\begin{align}
			\mathbb{V}\left[ \sum_{i=1}^n \hat \sigma_i^2 (B_{ii} -  \hat B_{ii})\right] &= \E\left[ \mathbb{V}\left[\sum_{i=1}^n \hat \sigma_i^2 \hat B_{ii} \mid y\right]  \right] \le \frac{2}{p} \sum_{i=1}^n \sum_{\ell=1}^n B_{i\ell}^2  \E\left[ \hat \sigma_i^2 \hat \sigma_{\ell}^2  \right]
			= O\left( \frac{\text{trace}(\tilde A^2)}{p}  \right), \\
			\mathbb{V}\left[ \sum_{i=1}^n \hat \sigma_i^2 \hat B_{ii} \hat a_i \right] &= \E\left[ \mathbb{V}\left[\sum_{i=1}^n \hat \sigma_i^2 \hat B_{ii} \hat a_i \mid y,R_B\right]  \right] \le \frac{2}{p} \sum_{i=1}^n \sum_{\ell=1}^n P_{i\ell}^2 \tfrac{\E\left[\hat B_{ii} \hat B_{\ell \ell}\right]\E\left[ \hat \sigma_i^2 \hat \sigma_{\ell}^2  \right]}{(1-P_{ii})(1-P_{\ell \ell})}   \\
			&= O\left( \frac{\text{trace}(\tilde A^2)}{p} + \frac{\text{trace}(\tilde A_1^2)^{1/2}\text{trace}(\tilde A_2^2)^{1/2}}{p^2}  \right) 
			\shortintertext{where $\tilde A_\ell = S_{xx}^{-1/2} A_\ell'A_\ell S_{xx}^{-1/2}$ for $\ell=1,2$,}
			\mathbb{V}\left[ \sum_{i=1}^n \hat \sigma_i^2 \hat B_{ii} \left(\hat a_i^2 - \mathbb{V}[\hat a_i]\right) \right] &= \sum_{i=1}^n \sum_{\ell=1}^n \E\left[\hat B_{ii} \hat B_{\ell \ell}\right]\E\left[ \hat \sigma_i^2 \hat \sigma_{\ell}^2  \right] Cov\left(\hat a_i^2,\hat a_\ell^2 \right)\\
			&= O\left( \frac{\text{trace}(\tilde A^2)}{p^2} + \frac{\text{trace}(\tilde A_1^2)^{1/2}\text{trace}(\tilde A_2^2)^{1/2}}{p^3}  \right)  \\
			\mathbb{V}\bigg[ \sum_{i=1}^n \hat \sigma_i^2 \left( \hat B_{ii} - B_{ii} \right) &\frac{2\sum_{\ell \neq i}^n P_{i\ell}^4 - P_{ii}^2(1-P_{ii})^2}{p(1-P_{ii})^2}  \bigg] = O\left( \frac{\text{trace}(\tilde A^2)}{p^3}  \right) \\
			\mathbb{V}\bigg[ \sum_{i=1}^n B_{ii} \left(\hat \sigma_i^2 - \sigma_i^2 \right) &\frac{2\sum_{\ell \neq i}^n P_{i\ell}^4 - P_{ii}^2(1-P_{ii})^2}{p(1-P_{ii})^2}  \bigg] = O\left( \frac{\mathbb{V}[\hat \theta]}{p^2}  \right)
		\end{align}
		}%
		From this it follows that $\mathbb{V}[\hat  \theta]^{-1/2}\left( (\hat \theta_{JLA} - \hat \theta)_2 - \mathrm{B}_p  \right) = o_p(1)$ since $\text{trace}(\tilde A^2) = O(\mathbb{V}[\hat \theta])$ and $ \frac{\mathbb{V}[\hat \varTheta_1]\mathbb{V}[\hat \varTheta_2]}{p^4 \mathbb{V}[\hat \theta]^2}=o(1)$. 
		
		We now treat the approximation error while utilizing that $\E[\hat a_i^3] = O\left(\frac{1}{p^2}\right)$, $\E[\hat a_i^4] = O\left(\frac{1}{p^2}\right)$, and $\max_i \abs{\hat a_i} = o_p( \log(n)/\sqrt{p} )$ which follows from \cite[][Theorem 1.1 and its proof]{achlioptas2001database}. Proceeding term by term, we list the conclusions
		{\footnotesize
		\begin{align}
			\sum_{i=1}^n \hat \sigma_i^2 \hat B_{ii} \hat a_i^3 + \sum_{i=1}^n \hat \sigma_i^2 \hat B_{ii} \hat a_i^4 &= O_p\left( \frac{\E[\hat \varTheta_{1,\text{PI}} - \varTheta_{1}] + \E[\hat \varTheta_{2,\text{PI}} - \varTheta_{2}]}{p^2} \right) \\
			\sum_{i=1}^n \hat \sigma_i^2 \hat B_{ii} \frac{\hat a_i^5}{1+\hat a_i} &= O_p\left( \frac{\log(n)}{\sqrt{p}} \frac{\E[\hat \varTheta_{1,\text{PI}} - \varTheta_{1}] + \E[\hat \varTheta_{2,\text{PI}} - \varTheta_{2}]}{p^2} \right) \\
			\frac{1}{p}\sum_{i=1}^n \hat \sigma_i^2 \hat B_{ii} \frac{3 \hat P_{ii}^2 + \hat P_{ii}^2 -(3P_{ii}^2 + P_{ii}^2)(1+\hat a_i)^2}{(1+\hat a_i)^2(1-P_{ii})} &= O_p\left( \left(1+\frac{\log(n)}{\sqrt{p}} \right) \frac{\E[\hat \varTheta_{1,\text{PI}} - \varTheta_{1}] + \E[\hat \varTheta_{2,\text{PI}} - \varTheta_{2}]}{p^2} \right)
		\end{align}
		}%
		which finishes the proof.
	\end{proof}

	\subsection{Examples}
	All mathematical discussions of the examples are collected in Appendix \ref{app:sufficiency}.
		
	\subsection{Quadratic Forms of Fixed Rank}\label{app:dist1}
		The next result provides a restatement and proof of \th\ref{thm2}.
		
		\begin{theorem}\th\label{thm2S}				
			If \th\ref{ass:reg} holds, $r$ is fixed, and $\max_i  w_i'w_i = o(1)$, then 
			\begin{enumerate}
				\item $\mathbb{V}[\hat b]^{-1/2}(\hat b - b) \xrightarrow{d} \mathcal{N}\left( 0 , I_r \right)$ where $b = Q' S_{xx}^{1/2} \beta$,
				\item ${\mathbb{V}}[\hat b]\inverse \hat{\mathbb{V}}[\hat b] \xrightarrow{p} I_r$,
				\item $\hat \theta = \sum_{\ell = 1}^{r} \lambda_\ell\left(\hat b_\ell^2-\mathbb{V}[\hat b_\ell]\right) + o_p(\mathbb{V}[\hat \theta]^{1/2})$,
			\end{enumerate}
		\end{theorem}
		
		\begin{proof}
			The proof has two steps: First, we write $\hat \theta$ as $\sum_{\ell = 1}^{r} \lambda_\ell\left(\hat b_\ell^2-\mathbb{V}[\hat b_\ell]\right)$ plus an approximation error which is of smaller order than $\mathbb{V}[\hat \theta]$. This argument establishes the last two claims of the lemma. Second, we use Lyapounov's CLT to show that $\hat b \in \R^r$ is jointly asymptotically normal.
			
			\noindent \textbf{\textit{Decomposition and Approximation}}
				From the proof of \th\ref{lem:fin} it follows that 
				\begin{align}
					\hat \theta &= \sum_{\ell = 1}^{r} \lambda_\ell\left(\hat b_\ell^2-\mathbb{V}[\hat b_\ell]\right) + \sum_{i=1}^n B_{ii} (\sigma_i^2 -\hat \sigma_i^2)
				\end{align}
				where we now show that the mean zero random variable $\sum_{i=1}^n B_{ii} (\sigma_i^2 -\hat \sigma_i^2)$ is $o_p( \mathbb{V}[\hat \theta]^{1/2} )$.
				
				We have
				\begin{align}
					\sum_{i=1}^n B_{ii} (\hat \sigma_i^2 - \sigma_i^2)
					&= \sum_{i=1}^n B_{ii} \sum_{\ell =1}^n M_{ii}\inverse x_i'\beta M_{i\ell} \varepsilon_\ell
					+ \sum_{i=1}^n B_{ii} (\varepsilon_i^2 -\sigma^2_i)
					 + \sum_{i=1}^n B_{ii} \sum_{\ell \neq i} M_{ii}\inverse M_{i\ell} \varepsilon_i \varepsilon_\ell. \label{eq:term3}
				\end{align}
				The variances of these three terms are
				{\small
				\begin{align}
					\sum_{\ell=1}^n \sigma_\ell^2 \left( \sum_{i=1}^n M_{i\ell} B_{ii} M_{ii}\inverse x_i'\beta \right)^2 
					\le \max_i \sigma_i^2 \sum_{i=1}^n B_{ii}^2 M_{ii}^{-2} (x_i'\beta)^2 
					&\le \max_i \sigma_i^2 \max_i (x_i'\beta)^2 M_{ii}^{-2}  \times \sum_{i=1}^n B_{ii}^2, \\
					\sum_{i=1}^n B_{ii}^2 \mathbb{V}[\varepsilon_{i}^2]  &\le \max_i \E[\varepsilon_{i}^4]  \times \sum_{i=1}^n B_{ii}^2, \\
					\sum_{i=1}^n \sum_{\ell \neq i} \left( B_{ii}^2 M_{ii}^{-2} + B_{ii} M_{ii} \inverse B_{\ell\ell} M_{\ell\ell} \inverse \right) M_{i\ell}^2 \sigma_i^2 \sigma_\ell^2
					&\le 2 \max_i \sigma_i^4 M_{ii}^{-2} \times \sum_{i=1}^n B_{ii}^2.
				\end{align}
				}%
				Furthermore, we have that 
				\begin{align}
					\mathbb{V}[\hat \theta] \inverse \sum_{i=1}^n B_{ii}^2 \le \max_i w_i'w_i \mathbb{V}[\hat \theta] \inverse \sum_{l=1}^r \lambda_{l}^2 (\tilde A) \le \max_i w_i'w_i \max_i \sigma_i^{-4} = o(1),
				\end{align}
				so each of the three variances are of smaller order than $\mathbb{V}[\hat \theta]$. 
				
				For the second claim it suffices to show that $\delta(v) := \frac{\hat{\mathbb{V}}[v'\hat b] - \mathbb{V}[v'\hat b]}{\mathbb{V}[v'\hat b]} = o_p(1)$ for all nonrandom $v \in \R^{r}$ with $v'v=1$. Let $v \in \R^{r}$ be nonrandom with $v'v=1$. As above we have that $\delta(v) = \sum_{i=1}^n w_i(v) (\hat \sigma_i^2 - \sigma_i^2)$ is a mean zero variable which is $o_p(1)$ if $\sum_{i=1}^n w_i(v)^4 = o(1)$ where $w_i(v) = \frac{(v'w_i)^2}{\sum_{i=1}^n \sigma_i^2 (v'w_i)^2}.$ But this follows from
				\begin{align}
				\sum_{i=1}^n w_i(v)^4 \le \max_{i} \sigma_{i}^{-4} \max_{i} w_i'w_i = o(1)
				\end{align}
				where the inequality is implied by $\max_i  w_i'w_i = o(1)$, $v'v=1$, and $\sum_{i=1}^n w_i w_i' = I_r$.

				\noindent \textbf{\textit{Asymptotic Normality}}		
				Next we show that all linear combinations of $\hat b$ are asymptotically normal. Let $v \in \R^r$ be a non-random vector with $v'v=1$. 
				Lyapunov's CLT implies that $\mathbb{V}[v'\hat b]^{-1/2}v'(\hat b - b) \xrightarrow{d}N(0,1)$ if
				\begin{align}\label{eq:CLT}
					\mathbb{V}[v'\hat b]^{-2}\sum_{i=1}^n \E[\varepsilon_{i}^4] (v'Q'S_{xx}^{-1/2} x_i)^4 = \mathbb{V}[v'\tilde \beta]^{-2}\sum_{i=1}^n \E[\varepsilon_{i}^4] (v'w_i)^4 = o(1).
				\end{align}
				We have that $\max_i w_i'w_i = o(1)$ implies \eqref{eq:CLT} since $\max_{i} (v'w_i)^2 \le \max_{i} w_i'w_i$ and
				\begin{align}
					\sum_{i=1}^n (v'w_i)^2=1, \qquad \mathbb{V}[v'\tilde \beta]\inverse &\le \max_i \sigma_i^{-2} = O(1),  \qquad \max_i \E[\varepsilon_{i}^4] = O(1),
				\end{align}
				by definition of $w_i$ and \th\ref{ass:reg}.
			\end{proof}

		\subsection{Quadratic Forms of Growing Rank}\label{app:dist2}
		
		This appendix provides restatements and proofs of \th\ref{thm3,thm4}. The proofs relies on an auxiliary lemma which extends a central limit theorem given in \cite{soelvsten2017robust}.

		\subsubsection{A Central Limit Theorem}
	
		The proofs of \th\ref{thm3} and \th\ref{thm4} is based on the following lemma. Let $\{v_{n,i}\}_{i,n}$ be a triangular array of row-wise independent random variables with $\E[v_{n,i}]=0$ and $\mathbb{V}[v_{n,i}] = \sigma_{n,i}^2$, let $\{\dot w_{n,i}\}_{i,n}$ be a triangular array of non-random weights that satisfy $\sum_{i=1}^n \dot w_{n,i}^2 \sigma_{n,i}^2 =1$ for all $n$, and let  $(W_n)_n$ be a sequence of symmetric non-random matrices in $\R^{n \times n}$ with zeroes on the diagonal that satisfy $2\sum_{i=1}^n \sum_{\ell \neq i} W_{n,i\ell}^2 \sigma_{n,i}^2 \sigma_{n,\ell}^2 =1$. For simplicity, we drop the subscript $n$ on $v_{n,i}$, $\sigma^2_{n,i}$, $\dot w_{n,i}$ and $W_n$. Define 
		\begin{align}
		\mathcal{S}_n = \sum_{i=1}^n \dot w_{i} v_i \quad  \text{and} \quad \mathcal{U}_n = \sum_{i=1}^n \sum_{\ell \neq i} W_{i\ell} v_i v_\ell.
		\end{align}

		\begin{lemma}\th\label{lem:clt}
			If $\max_i \E[v_i^4] + \sigma_i^{-2} = O(1)$,
			\begin{align}
			(i) \ \max_i \dot w_i^2 = o(1), \qquad 
			(ii) \ \lambda_{\max}(W^2) = o(1),
			\end{align}
			then $(\mathcal{S}_n,\mathcal{U}_n)' \xrightarrow{d} \mathcal{N}(0,I_2)$.
		\end{lemma}

		This lemma extends the main result of Appendix A2 in \cite{soelvsten2017robust} to allow for $\{v_i\}_i$ to be an array of non-identically distributed variables and presents the conclusion in a way that is tailored to the application in this paper. The proof requires no substantially new ideas compared to \cite{soelvsten2017robust}, but we give it at the end of the next section for completeness.

		\subsubsection{Limit Distributions}
	
		\begin{theorem}\th\label{thm3S}
			If
			\begin{align}
			(i) \ \mathbb{V}[\hat \theta]\inverse \max_i \left( (\tilde x_i'\beta)^2 + (\check x_i'\beta)^2 \right) = o(1), \quad (ii) \ \frac{ \lambda_1^2}{\sum_{\ell=1}^r \lambda_\ell^2} = o(1),
			\end{align}
			and \th\ref{ass:reg} holds, then $\mathbb{V}[\hat \theta]^{-1/2} (\hat \theta - \theta ) \xrightarrow{d} \mathcal{N}(0,1)$.
		\end{theorem}
	
		\begin{proof}
			The proof involves two steps: First, we decompose $\hat \theta$ into a weighted sum of two terms of the type described in \th\ref{lem:clt}. Second, we use \th\ref{lem:clt} to show joint asymptotic normality of the two terms. The conclusion that $\hat \theta$ is asymptotically normal is immediate from there. 
			
			\noindent \textbf{\textit{Decomposition}}
		 		The difference between $\hat \theta$ and $\theta$ is
				\begin{align}
					\hat \theta - \theta &
					= \sum_{i=1}^n \left( 2\tilde x_i'\beta - \check x_i'\beta \right) \varepsilon_{i} + \sum_{i=1}^n \sum_{\ell \neq i} C_{i\ell} \varepsilon_{i} \varepsilon_\ell,
				\end{align}
				where these two terms are uncorrelated and have variances
				\begin{align}
					V_\mathcal{S} = \sum_{i=1}^n (2\tilde x_i' \beta - \check x_i'\beta)^2 \sigma_i^2 \quad \text{and} \quad V_\mathcal{U} = 2\sum_{i=1}^n \sum_{\ell \neq i} C_{i\ell}^2 \sigma_i^2 \sigma_\ell^2.
				\end{align}
				Thus we write $\mathbb{V}[\hat \theta]^{-1/2}(\hat \theta - \theta) = \omega_1 \mathcal{S}_n + \omega_2 \mathcal{U}_n$ where
				\begin{align}
					\mathcal{S}_n &= V_\mathcal{S}^{-1/2}\sum_{i=1}^n \left( 2\tilde x_i'\beta - \check x_i'\beta \right) \varepsilon_{i}, && \mathcal{U}_n = V_{\mathcal{U}}^{-1/2} \sum_{i=1}^n \sum_{\ell \neq i} C_{i\ell} \varepsilon_{i} \varepsilon_\ell, \\
								\omega_1 &= \sqrt{V_\mathcal{S}/\mathbb{V}[\hat \theta]}, && \omega_2 = \sqrt{V_\mathcal{U}/\mathbb{V}[\hat \theta]}.
				\end{align}

			\noindent \textbf{\textit{Asymptotic Normality}}
				We will argue along converging subsequences. Move to a subsequence where $\omega_1$ converges. If the limit is zero, then $\mathbb{V}[\hat \theta]^{-1/2}(\hat \theta - \theta) = \omega_2 \mathcal{U}_n + o_p(1)$ and so it follows from \th\ref{res2} below and \th\ref{thm3}(ii) that $\hat \theta$ is asymptotically normal. Thus we consider the case where the limit of $\omega_1$ is nonzero.
				
				In the notation of \th\ref{lem:clt} we have
				\begin{align}
				\dot w_i = \frac{\left( 2\tilde x_i'\beta - \check x_i'\beta \right)}{V_\mathcal{S}^{1/2}} \quad \text{and} \quad W_{i\ell} = \frac{C_{i\ell}}{V_\mathcal{U}^{1/2}}.
				\end{align}
		
				For \th\ref{lem:clt}(i) we have
				\begin{align}
				\max_i \dot w_i^2 \le 4\omega_1\inverse \max_i \frac{(\tilde x_i'\beta)^2 + (\check x_i'\beta)^2 }{\mathbb{V}[\hat \theta]} = o(1),
				\end{align}
				where the last equality follows from \th\ref{thm3}(i) and the nonzero limit of $\omega_1$. 
				
				For \th\ref{lem:clt}(ii) we show instead that $\text{trace}(W^4) = o(1)$. It can be shown that for all $n$, $\text{trace}(C^4) \le c_U \cdot \text{trace}(B^4) = c_U \cdot \text{trace}(\tilde A^4) \le c_U  \lambda_1^2 \cdot \text{trace}(\tilde A^2)$ and ${V}_\mathcal{U} \ge c_L \min_i \sigma_{i}^4 \cdot \text{trace}(\tilde A)$, where the finite and nonzero constants $c_U$ and $c_L$ do not depend on $n$ (but depend on $\min_i M_{ii}$ which is bounded away from zero). Thus, \th\ref{ass:reg} implies that
				\begin{align}
				\text{trace}(W^4) \le  \frac{c_U  \lambda_1^2 \cdot \text{trace}(\tilde A^2)}{(c_L \min_i \sigma_{i}^4 \cdot \text{trace}(\tilde A^2))^2} = O\left( \frac{\lambda_1^2 }{\text{trace}(\tilde A^2)} \right) = o(1)
				\end{align}
				where the last equality follows from \th\ref{thm3}(ii).
		\end{proof}
		
		\begin{theorem}\th\label{thm4S}
			If $\max_i \mathsf{w}_{iq}' \mathsf{w}_{iq} = o(1)$, $\mathbb{V}[\hat \theta_q]\inverse \max_i \left( (\tilde x_{iq}'\beta)^2 + (\check x_{iq}'\beta)^2 \right) = o(1)$, and \th\ref{ass:reg,ass:eig} holds, then 
					\begin{enumerate}
				\item $\mathbb{V}[(\hat{\mathsf{b}}_q',\hat \theta_q)']^{-1/2}
				\left(
				(\hat{\mathsf{b}}_q',\hat \theta_q)'
				- \E[(\hat{\mathsf{b}}_q',\hat \theta_q)']
				\right)
				\xrightarrow{d} \mathcal{N}\left(0, I_{q+1} \right)$
				\item  $\hat \theta = \sum_{\ell = 1}^{q} \lambda_\ell\left(\hat b_\ell^2-\mathbb{V}[\hat b_{\ell}]\right) + \hat \theta_q + o_p(\mathbb{V}[\hat \theta]^{1/2})$
			\end{enumerate}
			for
			\begin{align}
				\mathbb{V}[(\hat {\mathsf{b}}_q',\hat \theta_q)'] &= \sum_{i=1}^n \begin{bmatrix}
					\mathsf{w}_{iq}\mathsf{w}_{iq}' \sigma_i^2 & 2\mathsf{w}_{iq}\left(\sum_{\ell \neq i} C_{i\ell q} x_\ell'\beta \right) \sigma_i^2 \\
					2\mathsf{w}_{iq}'\left(\sum_{\ell \neq i} C_{i\ell q} x_\ell'\beta \right) \sigma_i^2  & 4 \left(\sum_{\ell \neq i} C_{i\ell q} x_\ell'\beta \right)^2 \sigma_i^2 + 2 \sum_{\ell \neq i} C_{i\ell q}^2 \sigma_i^2 \sigma_\ell^2
				\end{bmatrix},
			\end{align}
			$C_{i \ell q} = B_{i \ell q} - 2\inverse M_{i \ell} \left( M_{ii}\inverse B_{ii q}   +  M_{\ell \ell}\inverse B_{\ell \ell q}  \right)$, $B_{i \ell q} = x_i' S_{xx}^{-1/2} \tilde A_q S_{xx}^{-1/2}  x_\ell$, $\tilde A_q = \sum_{\ell=q+1}^{r} \lambda_\ell q_\ell q_\ell',$ $\tilde x_{iq} = \sum_{\ell =1}^n B_{i\ell q} x_\ell$, and $\check x_{iq} = \sum_{\ell=1}^n M_{i\ell} M_{\ell\ell} \inverse B_{\ell\ell q} x_\ell$.
		\end{theorem}

		\begin{proof}
	
			The proof involves two steps: First, we write $\hat \theta$ as the sum of (1a) a quadratic function applied to $\hat {\mathsf{b}}_q$, (1b) an approximation error which is of smaller order than $\mathbb{V}[\hat \theta]$, and (2) a weighted sum of two terms, $\mathcal{S}_{n}$ and $\mathcal{U}_n$, of the type described in \th\ref{lem:clt}. Second, we use \th\ref{lem:clt} to show that $(\hat {\mathsf{b}}_q',\mathcal{S}_{n},\mathcal{U}_n)' \in \R^{q+2}$ is jointly asymptotically normal. %Third, we use the continuous mapping theorem to establish convergence in distribution along convergent subsequences.
			
			\noindent \textbf{\textit{Decomposition and Approximation}}
				We have that
				\begin{align}
					\hat \theta &= \sum_{\ell =1}^{q} \lambda_\ell (\hat b_\ell^2 - \mathbb{V}[\hat b_\ell]) +\hat \theta_q  + o_p(\mathbb{V}[\hat \theta]^{1/2}) 
					\quad \text{for} \quad 
					\hat \theta_q = \sum_{i=1}^n \sum_{\ell \neq i} C_{i\ell q} y_i y_\ell
					\shortintertext{since}
					\hat \beta' A \hat \beta &= \sum_{\ell =1}^{q} \lambda_\ell \hat b_\ell^2 + \sum_{i=1}^n \sum_{\ell = 1}^n B_{i\ell q} y_i y_\ell  
					\shortintertext{and}
					\sum_{i = 1}^n B_{ii} \hat \sigma_i^2 &= \sum_{i=1}^n B_{ii\mathbf{1}}\sigma_i^2 + \sum_{i=1}^n B_{iiq}\hat \sigma_i^2  + \sum_{i=1}^n B_{ii,-q}( \hat \sigma_i^2 -  \sigma_i^2) \\
					&= \sum_{\ell =1}^{q} \lambda_\ell \mathbb{V}[\hat b_\ell] + \sum_{i=1}^n B_{iiq}\hat \sigma_i^2  + o_p(\mathbb{V}[\hat \theta]^{1/2})
				\end{align}
				where $B_{ii,-q} = B_{ii} - B_{ii q}$ and it follows from $\max_i \mathsf{w}_{iq}' \mathsf{w}_{iq} = o(1)$ and the calculations in the proof of \th\ref{thm2} that the mean zero random variable $\sum_{i=1}^n B_{ii,-q}( \hat \sigma_i^2 -  \sigma_i^2)$ is $o_p(\mathbb{V}[\hat \theta]^{1/2})$. 
				
				We will further center and rescale $\hat \theta_q$ by writing
				\begin{align}
					\mathbb{V}[\hat \theta_q]^{-1/2} \left(\hat \theta_q- \E[\hat \theta_q]\right) = \omega_1 \mathcal{S}_n + \omega_2 \mathcal{U}_n
				\end{align}
				where
				\begin{align}
				\mathcal{S}_{n} &= {V}_{\mathcal{S}}^{-1/2}\sum_{i=1}^n \left( 2\tilde x_{iq}'\beta - \check x_{iq}'\beta \right) \varepsilon_{i}, &&
				\mathcal{U}_{n} = {V}_{\mathcal{U}}^{-1/2} \sum_{i=1}^n \sum_{\ell \neq i} C_{i\ell q} \varepsilon_{i} \varepsilon_\ell, \\
				{V}_{\mathcal{S}} &=  \sum_{i=1}^n \left( 2\tilde x_{iq}'\beta - \check x_{iq}'\beta \right)^2 \sigma_{i}^2, 
				&& 
				{V}_{\mathcal{U}} = 2 \sum_{i=1}^n \sum_{\ell \neq i} C^2_{i\ell q} \sigma^2_{i} \sigma^2_\ell, \\
				\omega_1 &= \sqrt{{V}_{\mathcal{S}}/\mathbb{V}[\hat \theta_q]},
				&& 
				\omega_2 = \sqrt{{V}_{\mathcal{U}}/\mathbb{V}[\hat \theta_q]},
				\end{align}
				and $\mathcal{U}_{n}$ is uncorrelated with both $\mathcal{S}_n$ and $\hat {\mathsf{b}}_q$.

			\noindent \textbf{\textit{Asymptotic Normality}}
				As in the proof of \th\ref{thm3}, we will argue along converging subsequences and therefore move to a subsequence where $\omega_1$ converges. If the limit is zero, then the conclusion of the theorem follows from \th\ref{lem:clt} applied to $(\mathbb{V}[v'\hat {\mathsf{b}}_q]^{-1/2}(v'\hat {\mathsf{b}}_q-\E[v'\hat {\mathsf{b}}_q]),\mathcal{U}_n)'$ for $v \in \R^{q}$ with $v'v=1$. Thus we consider the case where the limit of $\omega_1$ is nonzero.
				
				Next we use \th\ref{lem:clt} to show that 
				\begin{align}
					\left(\frac{v'\hat {\mathsf{b}}_q-\E[v'\hat {\mathsf{b}}_q] +u\mathcal{S}_{n}}{\mathbb{V}[\hat {\mathsf{b}}_q +u\mathcal{S}_{n}]^{1/2}},\mathcal{U}_{n}\right)' \xrightarrow{d} \mathcal{N}(0,I_2)
				\end{align}
				for any non-random $(v',u)' \in \R^{q+1}$ with $v'v + u^2=1$. In the notation of \th\ref{lem:clt} we have
				\begin{align}
				\dot w_i = \frac{v' \mathsf{w}_{iq} + u {V}_{\mathcal{S}}^{-1/2}\left( 2\tilde x_{iq}'\beta - \check x_{iq}'\beta \right)} {\mathbb{V}[\hat {\mathsf{b}}_q +u\mathcal{S}_{n}]^{1/2}}  \quad \text{and} \quad W_{i\ell} = \frac{C_{i\ell q}}{{V}_\mathcal{U}^{1/2}}.
				\end{align}

				A simple calculation shows that $\mathbb{V}[v'\hat {\mathsf{b}}_q +u\mathcal{S}_{n}] \ge \min_i \sigma_i^2 \gg 0$, so $\max_i \dot w_i^2 = o(1)$ follows from \th\ref{thm4}(i), \th\ref{thm4}(ii), and $\omega_1$ being bounded away from zero.

				Similarly, we have as in the proof of \th\ref{thm3} that
				\begin{align}
					\text{trace}(C_q^4) &\le c \text{trace}(B_q^4) \le c \lambda_{q+1}^2 \sum_{\ell = q+1}^r \lambda_\ell^2
					\quad \text{and} \quad
					{V}_\mathcal{U}^{2} \ge \omega_2^{-4} \min_i \sigma_i^8 \text{trace}(\tilde A^2)^2
				\end{align}
				for $C_q = (C_{i\ell q})_{i,\ell}$ and $B_q = (B_{i\ell q})_{i,\ell}$, so \th\ref{ass:reg,ass:eig} yield $\text{trace}(W^4) = o(1)$.
		\end{proof}

	\subsubsection{Proof of a Central Limit Theorem}
	The proof of \th\ref{lem:clt} uses the notation and verifies the conditions of Lemmas A2.1 and A2.2 in \cite{soelvsten2017robust} referred to as SS2.1 and SS2.2, respectively. First, we show marginal convergence in distribution of $\mathcal{S}_n$ and $\mathcal{U}_n$. Then, we show joint convergence in distribution of $\mathcal{S}_n$ and $\mathcal{U}_n$. Let $V_n = (v_1,\dots,v_n)$ where $\{v_i\}_i$ are as in the setup of \th\ref{lem:clt}.
	
	Before starting we note that $\max_i \sigma^{-2}_i = O(1)$ and $2\sum_{i=1}^n \sum_{\ell \neq i} W_{i\ell}^2 \sigma_{i}^2 \sigma_{\ell}^2 =1$ implies that $\text{trace}(W^2) = \sum_{i=1}^n \sum_{\ell \neq i} W_{i\ell}^2 = O(1)$ and therefore that
	\begin{align}
		\lambda_{\max}(W^2) = o(1) \Leftrightarrow \text{trace}(W^4)=o(1).
	\end{align}
	
	\subsubsection*{Marginal Distributions}
	\begin{res}\th\label{res1}
		$\max_i \E[v_i^4] + \sigma_i^{-2}= O(1)$, $\sum_{i=1}^n \dot w_i^2 \sigma_i^2=1$, and \th\ref{lem:clt}(i) implies that $\mathcal{S}_n \xrightarrow{d} \mathcal{N}(0,1)$.
	\end{res}

	In the notation of SS2.1 we have,
	\begin{align}
		\Delta^0_i \mathcal{S}_n = \dot w_i v_i
		\quad \text{and} \quad E[T_n \mid V_n] = 1 + \tfrac{1}{2} \sum_{i=1}^n \dot w_i^2 (v_i^2 - \sigma_i^2),
	\end{align}
	and it follows from $\max_i \E[v_i^4] + \sigma_i^{-2}= O(1)$, $\sum_{i=1}^n \dot w_i^2 \sigma_i^2=1$, and \th\ref{lem:clt}(i) that
	\begin{align}
		E[T_n \mid V_n] \xrightarrow{\mathcal{L}^1} 1, \quad \sum_{i=1}^n \E[(\Delta_i^0 \mathcal{S}_n)^2] = 1, \qquad \sum_{i=1}^n \E[(\Delta^0_i \mathcal{S}_n)^4] \le \max_i \frac{\E[v_i^4]}{\sigma_i^{2}}  \dot w_i^2 = o(1),
	\end{align}
	so \th\ref{res1} follows from SS2.1.
	
	\begin{res}\th\label{res2}
		$\max_i \E[v_i^4] + \sigma_i^{-2}= O(1)$, $2\sum_{i=1}^n \sum_{\ell \neq i} W_{n,i\ell}^2 \sigma_{n,i}^2 \sigma_{n,\ell}^2 =1$, and \th\ref{lem:clt}(ii) implies that $\mathcal{U}_n \xrightarrow{d} \mathcal{N}(0,1)$.
	\end{res}

	In the notation of SS2.1 we have,
	\begin{align}
		\Delta^0_i \mathcal{U}_n = 2 v_i \sum_{\ell \neq i} W_{i\ell} v_\ell
		\quad \text{and} \quad E[T_n \mid V_n] = \sum_{i=1}^n \sum_{\ell \neq i} \sum_{k \neq i} (v_i + \sigma_i^2) W_{i\ell} W_{ik} v_\ell v_k,
	\end{align}
	and
	\begin{align}
		\sum_{i=1}^n \E[(\Delta_i^0 \mathcal{U}_n)^2] = 2, \qquad \sum_{i=1}^n \E[(\Delta_i^0 \mathcal{U}_n)^4] \le 2^5 \max_i \E[v_i^4]^2 \max_i \sigma_i^{-4} \max_i \sum_{\ell \neq i} W_{i\ell}^2,
	\end{align}
	where $\max_i \sum_{\ell \neq i} W_{i\ell}^2 \le \sqrt{\text{trace}(W^4)} = o(1)$. Now, split $E[T_n \mid V_n] - 1$ into three terms
	\begin{align}
		a_n &= \sum_{i=1}^n \sum_{\ell \neq i} \sigma_i^2 W_{i\ell}^2 (v_\ell + v_\ell^2 - \sigma_\ell^2) \\
		b_n &= 2 \sum_{i=1}^n \sum_{\ell \neq i} \sum_{k \neq i,\ell} \sigma_k^2 W_{\ell k} W_{ik} v_i v_\ell  + \sum_{i=1}^n \sum_{\ell \neq i} W_{i\ell}^2 v_i (v_\ell^2 - \sigma_\ell^2) \\
		c_n &= \sum_{i=1}^n \sum_{\ell \neq i} \sum_{k \neq i, \ell} W_{i\ell} W_{ik} (v_i^2 - \sigma_i^2) v_\ell v_k.
	\end{align}
	
	\subsubsection*{Interlude: Convergence in $\mathcal{L}^1$}
	$a_n, b_n,$ and $c_n$ are a linear sum, a quadratic sum, and a cubic sum. We will need to treat similar sums later, so we record some simple sufficient conditions for their convergence. For brevity, let $\sum_{i \neq \ell}^n = \sum_{i=1}^n \sum_{\ell \neq i},$ and $\sum_{i \neq \ell \neq k}^n = \sum_{i=1}^n \sum_{\ell \neq i} \sum_{k \neq i,\ell}$, etc. We use the notation $u_i =(v_{i1},v_{i2},v_{i3},v_{i4}) \in \R^4$ to denote independent random vectors in order that the result applies to combinations of $v_i$ and $v_i^2 - \sigma_i^2$ as in $a_n$, $b_n$, and $c_n$ above. For the inferential results we will also treat quartic sums, so we provide the sufficient conditions here.
	
	\begin{res}\th\label{res:L1}
		Let $S_{n1} = \sum_{i=1}^n \omega_i v_{i1}$, $S_{n2} = \sum_{i\neq\ell}^n \omega_{i\ell} v_{i1} v_{\ell 2}$, $S_{n3} = \sum_{i \neq \ell \neq k}^n \omega_{i\ell k} v_{i1} v_{\ell 2} v_{k 3}$, and
		$S_{n4} = \sum_{i \neq \ell \neq k \neq m}^n \omega_{i\ell k m} v_{i1} v_{\ell 2} v_{k3} v_{m4}$
		%, and $S_{n5} = \sum_{i \neq \ell \neq k \neq m \neq j}^n \omega_{i\ell k m j} v_i v_\ell v_k v_m v_j$
		 where the weights $\omega_i$, $\omega_{i \ell}$, $\omega_{i \ell k}$, and $\omega_{i \ell k m}$ are non-random.
		 Suppose that $\E[u_i]=0$, $\max_i \E[u_i'u_i] = O(1)$.
		\begin{enumerate}
			\item If $\sum_{i=1}^n \omega_i^2 = o(1)$, then $S_{n1} \xrightarrow{\mathcal{L}^1} 0$.
			
			\item If $\sum_{i\neq\ell}^n \omega_{i\ell}^2 = o(1)$, then $S_{n2} \xrightarrow{\mathcal{L}^1} 0$.
			
			\item If $\sum_{i \neq \ell \neq k}^n \omega_{i\ell k}^2 = o(1)$, then $S_{n3} \xrightarrow{\mathcal{L}^1} 0$.
			
			\item If $\sum_{i \neq \ell \neq k \neq m}^n \omega_{i\ell k m}^2 = o(1)$, then $S_{n4} \xrightarrow{\mathcal{L}^1} 0$.
%			
%			\item  If $\sum_{i \neq \ell \neq k \neq m \neq j}^n \omega_{i\ell k m j}^2 = o(1)$, then $S_{n5} \xrightarrow{\mathcal{L}^1} 0$.
		\end{enumerate}
	\end{res}

	Consider $S_{n3}$, the other results follows from the same line of reasoning. In the notation of SS2.2 we have,
	\begin{align}
		\Delta_i^0 S_{n3}= v_{i1} \sum_{\ell \neq i} \sum_{k \neq i,\ell} \omega_{i\ell k} v_{\ell 2} v_{k 3} + v_{i2} \sum_{\ell \neq i} \sum_{k \neq i,\ell} \omega_{\ell i k} v_{\ell 1} v_{k 3} + v_{i3} \sum_{\ell \neq i} \sum_{k \neq i,\ell} \omega_{\ell k i} v_{\ell 1} v_{k 2}.
	\end{align}
	Focusing on the first term we have,
	\begin{align}
		\sum_{i=1}^n \E\left[\left(v_{i1} \sum_{\ell \neq i} \sum_{k \neq i,\ell} \omega_{i\ell k} v_{\ell 2} v_{k 3}\right)^2\right] 
		&\le \max_{i} \E[u_i'u_i]^3 \sum_{i \neq \ell \neq k}^n \left(\omega_{i\ell k}^2+\omega_{i\ell k}\omega_{ik\ell}\right) \\
		&\le 2\max_{i} \E[u_i'u_i]^3 \sum_{i \neq \ell \neq k}^n \omega_{i\ell k}^2,
	\end{align} 
	so the results follows from SS2.2, $\sum_{i \neq \ell \neq k}^n \omega_{i\ell k}^2 = o(1)$, and the observation that the last bound also applies to the other two terms in $\Delta_i^0 S_{n3}$.
	
%	Consider $S_{n4}$. In the notation of SS2.2 we have,
%	\begin{align}
%	\Delta_i^0 S_{n4} &= v_{i1} \sum_{\ell \neq i} \sum_{k \neq i,\ell} \sum_{m \neq i,\ell,k} \omega_{i\ell k m} v_{\ell 2} v_{k 3} v_{m 4} + v_{i2} \sum_{\ell \neq i} \sum_{k \neq i,\ell} \sum_{m \neq i,\ell,k} \omega_{\ell i k m} v_{\ell 1} v_{k 3} v_{m4}\\
%	&+ v_{i3} \sum_{\ell \neq i} \sum_{k \neq i,\ell} \sum_{m \neq i,\ell,k} \omega_{\ell k i m} v_{\ell 1} v_{k 2} v_{m4} + v_{i4} \sum_{\ell \neq i} \sum_{k \neq i,\ell} \sum_{m \neq i,\ell,k} \omega_{\ell k m i} v_{\ell 1} v_{k 2} v_{m3}.
%	\end{align}
%	Focusing on the first term we have,
%	\begin{align}
%		&\sum_{i=1}^n \E\left[\left(v_{i1} \sum_{\ell \neq i} \sum_{k \neq i,\ell} \sum_{m \neq i,\ell,k} \omega_{i\ell k m} v_{\ell 2} v_{k 3} v_{m 4}\right)^2\right]\\
%		& \le \max_i \E[u_i'u_i]^4 \sum_{i \neq \ell \neq k \neq m}^n \omega_{i\ell km}(\omega_{i\ell km} + \omega_{i m \ell k} + \omega_{i km\ell} + \omega_{i \ell m k} + \omega_{i k \ell m} + \omega_{i m k \ell}) \\
%		& \le 6 \max_i \E[u_i'u_i]^4 \sum_{i \neq \ell \neq k \neq m}^n \omega_{i\ell km}^2.
%	\end{align}
	
	\subsubsection*{Marginal Distributions, Continued}
	
	To see how $a_n \xrightarrow{\mathcal{L}^1} 0$, $b_n \xrightarrow{\mathcal{L}^1} 0$ and $c_n \xrightarrow{\mathcal{L}^1} 0$ follows from \th\ref{res:L1}, let $\bar W_{i \ell} = \sum_{k = 1}^n W_{ik} W_{k\ell}$ and note that $\text{trace}(W^4) = \sum_{i=1}^n \sum_{\ell =1}^n \bar W_{i \ell}^2$. We have
	\begin{align}
		\sum_{i=1}^n \left(\sum_{\ell \neq i} \sigma_\ell^2 W_{i\ell}^2 \right)^2 
		\le \max_i \sigma_i^4 \sum_{i=1}^n \bar W_{i i}^2. \\
		\sum_{i=1}^n \sum_{\ell \neq i} \left(\sum_{k \neq i,\ell} \sigma_k^2 W_{\ell k} W_{ik} \right)^2 \le \max_i \sigma_i^4 \sum_{i=1}^n \sum_{\ell =1}^n \bar W_{i \ell}^2 \\
		\sum_{i=1}^n \sum_{\ell \neq i} W_{i\ell}^4 = O\left(\max_{i,\ell} W_{i\ell}^2\right) \\
		\sum_{i=1}^n \sum_{\ell \neq i} \sum_{k \neq i, \ell} W_{i\ell}^2 W_{ik}^2 
		= O\left( \max_i \sum_{\ell \neq i} W_{i\ell}^2\right),
	\end{align}
	all of which are $o(1)$ as $\text{trace}(W^4) = o(1)$.
	
	\subsubsection*{Joint Distribution}
	
	Let $(u_1,u_2)' \in R^2$ be given and non-random with $u_1^2 + u_2^2 = 1$. Define $\mathcal{W}_n = u_1 \mathcal{S}_n + u_2 \mathcal{U}_n$. \th\ref{lem:clt} follows if we show that $\mathcal{W}_n \xrightarrow{d} \mathcal{N}(0,1)$. In the notation of SS2.1 we have,
	\begin{align}
		\Delta_i^0 \mathcal{W}_n &= u_1 \dot w_i v_i + u_2 2 v_i \sum_{\ell \neq i} W_{i \ell} v_\ell
		\shortintertext{and}
		\E[T_n \mid V_n] &= u_1^2 \left(1 + \tfrac{1}{2} \sum_{i=1}^n \dot w_i^2 (v_i^2 - \sigma_i^2)\right) + u_2^2 \sum_{i=1}^n \sum_{\ell \neq i} \sum_{k \neq i} (v_i + \sigma_i^2) W_{i\ell} W_{ik} v_\ell v_k \\
		& + u_1 u_2 3 \sum_{i=1}^n \sum_{\ell \neq i} (v_i^2 + \sigma_i^2) \dot w_i W_{i \ell} v_j.
	\end{align}
	The proofs of \th\ref{res1} and \th\ref{res2} showed that
	\begin{align}
		\sum_{i=1}^n \E[(\Delta_i^0 \mathcal{W}_n)^2] = O(1), \ \sum_{i=1}^n \E[(\Delta_i^0 \mathcal{W}_n)^4] = o(1)
	\end{align} 
	and that the first two terms of $\E[T_n \mid V_n]$ converge to $u_1^2 + u_2^2 = 1$. Thus the lemma follows if we show that the ``conditional covariance''
	\begin{align}
		3 \sum_{i=1}^n \sum_{\ell \neq i} (v_i^2 + \sigma_i^2) \dot w_i W_{i \ell} v_j
	\end{align}
	converges to $0$ in $\mathcal{L}^1$. This conditional covariance involves a linear and a quadratic sum so
	\begin{align}
		\sum_{i=1}^n \left(\sum_{\ell \neq i} \sigma_\ell^2 w_\ell W_{i \ell}\right)^2 &\le \max_i \sigma_i^4 \max_\ell \lambda^2_\ell(W) \sum_{i=1}^n \dot w_i^2 = O(\max_\ell\lambda^2_\ell(W))\\
		\sum_{i=1}^n \sum_{\ell \neq i}  \dot w_i^2 W_{i \ell}^2 &\le \sum_{i=1}^n \sum_{\ell \neq i}  W_{i \ell}^2  \max_i \dot w_i^2 = O( \max_i \dot w_i^2 )
	\end{align}
	ends the proof.

	\subsection{Asymptotic Variance Estimation}\label{app:inf}
	
	This appendix provides restatements and proofs of \th\ref{lem:varCS,lem:4} which establish consistency of the proposes standard error estimators that rely on sample splitting. Furthermore, it gives adjustments to those standard errors that guarantee existence whenever two independent unbiased estimators of $x_i'\beta$ cannot be formed. However, these adjustments may provide a somewhat conservative assessment of the uncertainty in $\hat \theta$ as further investigated in the simulations of Section \ref{sec:MC}. 
	
	\begin{lemma}\th\label{lem:var2S}
		For $s =1,2$, suppose that $\widehat{x_i'\beta}_{-i,s} = \sum_{\ell \neq i}^n P_{i\ell,s} y_\ell$ satisfies $\sum_{\ell \neq i}^n P_{i\ell,s} x_\ell'\beta=x_i'\beta$, $P_{i\ell,1} P_{i\ell,2} =0$ for all $\ell$, and $\lambda_{\max}(P_sP_s') = O(1)$. 
		\begin{enumerate}
			\item If the conditions of \th\ref{thm3} hold and $\abs{\mathcal{B}} = O(1)$, then 
			$\frac{\hat \theta - \theta}{\hat {\mathbb{V}}[\hat \theta]^{1/2}} \xrightarrow{d} \mathcal{N}(0,1).$ 
			\item If the conditions of \th\ref{thm3} hold, then $\liminf_{n \rightarrow \infty}\Pr\left( \theta \in \left[ \hat \theta \pm z_{\alpha} \hat{\mathbb{V}}[\hat \theta]^{1/2}  \right] \right) \ge 1-\alpha.$
		\end{enumerate}
	\end{lemma}
		
	\begin{proof}
		The proof continues in two steps: First, we show that $\hat {\mathbb{V}}[\hat \theta]$ has a positive bias which is of smaller order than ${\mathbb{V}}[\hat \theta]$ when $\abs{\mathcal{B}} = O(1)$. Second, we show that $\hat {\mathbb{V}}[\hat \theta] - \E[\hat {\mathbb{V}}[\hat \theta]] = o_p({\mathbb{V}}[\hat \theta])$. When combined with \th\ref{thm3}, these conclusions imply the two claims of the lemma.
		
		\noindent \textbf{\textit{Bias of $\hat {\mathbb{V}}[\hat \theta]$}}
		For the first term in $\hat {\mathbb{V}}[\hat \theta]$, a simple calculation shows that
		\begin{align}
		\E\left[ 4\sum_{i=1}^n \left( \sum_{\ell \neq i} C_{i\ell} y_\ell  \right)^2 \tilde \sigma_i^2 \right] 
		&= 4\sum_{i=1}^n \left(\sum_{\ell \neq i} C_{i\ell} x_\ell'\beta \right)^2 \sigma_i^2 + 4\sum_{i=1}^n \sum_{\ell \neq i} C_{i\ell}^2 \sigma_i^2 \sigma_\ell^2 \\
		&+4\sum_{i=1}^n \sum_{\ell \neq i} \sum_{m =1 }^n C_{mi} C_{m\ell} (P_{mi,1} P_{m\ell,2} + P_{mi,2} P_{m\ell,1}) \sigma_i^2 \sigma_\ell^2 \\
		&= \mathbb{V}[\hat \theta] + 2\sum_{i=1}^n \sum_{\ell \neq i} \tilde C_{i\ell} \sigma_i^2 \sigma_\ell^2.
		\end{align}
		For the second term in $\hat {\mathbb{V}}[\hat \theta]$, we note that if $P_{ik,-\ell}P_{\ell k,-i}=0 \text{ for all } k$, then independence between error terms yield $\E [\widehat{\sigma_i^2 \sigma_\ell^2} ] = \E [ \hat \sigma_{i,-\ell}^2] \E[ \hat \sigma_{\ell,-i}^2 ] = \sigma_i^2 \sigma_\ell^2.$ Otherwise if $P_{i\ell,1} + P_{i\ell,2}=0$, then 
		\begin{align}
		\E\left[\widehat{\sigma_i^2 \sigma_\ell^2}\right]
		&= \E\left[ \left(\varepsilon_i - \sum_{j \neq i} P_{ij,1} \varepsilon_j\right)\left(\varepsilon_i - \sum_{k \neq i} P_{ik,2} \varepsilon_k\right) \left( x_\ell'\beta + \varepsilon_\ell\right) \left(\varepsilon_\ell - \sum_{m \neq \ell} P_{\ell m, -i} \varepsilon_m\right)  \right] \\
		&= \sigma_i^2 \sigma_\ell^2 + x_\ell'\beta \E\left[ \left(\varepsilon_i - \sum_{j \neq i} P_{ij,1} \varepsilon_j\right)\left(\varepsilon_i - \sum_{k \neq i} P_{ik,2} \varepsilon_k\right) \sum_{m \neq \ell} P_{\ell m, -i} \varepsilon_m \right]
		\end{align}
		where the second term is zero since $P_{\ell i, -i} =0$ and $P_{ij,1}P_{ij,2}=0$ for all $j$. The same argument applies with the roles of $i$ and $\ell$ reversed when $P_{\ell i,1} + P_{\ell i,2}=0$. 
		
		Finally, when $(i,\ell) \in \mathcal{B}$ we have
		\begin{align}
			\E\left[\widehat{\sigma_i^2 \sigma_\ell^2}\right] &= \left(\sigma_i^2\left( \sigma_\ell^2 + ((x_\ell-\bar x)'\beta)^2 \right) + O\left(\frac{1}{n}\right)\right) 1_{\{\tilde C_{i\ell}<0\}}
		\end{align}
		where the remainder is uniform in $(i,\ell)$ and stems from the use of $\bar y$ as an estimator of $\bar x'\beta$. Thus for sufficiently large $n$, $\E[\tilde C_{i\ell} \widehat{\sigma_i^2 \sigma_\ell^2}]$ is smaller than $\tilde C_{i\ell} {\sigma_i^2 \sigma_\ell^2}$ leading to a positive bias in $\hat {\mathbb{V}}[\hat \theta]$. This bias is
		\begin{align}
			\sum_{(i,\ell) \in \mathcal{B}} \tilde C_{i\ell} \sigma_i^2 \left( \sigma_\ell^2 1_{\{\tilde C_{i\ell}>0\}} + ((x_\ell-\bar x)'\beta)^2 1_{\{\tilde C_{i\ell}<0\}} \right) + O\left(\frac{1}{n} {\mathbb{V}}[\hat \theta]\right)
		\end{align}
		which is ignorable when  $\abs{\mathcal{B}}= O(1)$.
		
		\noindent \textbf{\textit{Variability of $\hat {\mathbb{V}}[\hat \theta]$}}
		Now, $\hat {\mathbb{V}}[\hat \theta] - \E[\hat {\mathbb{V}}[\hat \theta]]$ involves a number of terms all of which are linear, quadratic, cubic, or quartic sums. \th\ref{res:L1} provides sufficient conditions for their convergence in $\mathcal{L}^1$ and therefore in probability. We have already treated versions of linear, quadratic, and cubic terms carefully in the proof of \th\ref{lem:clt}. Thus, we report here the calculations for the quartic terms (details for the remaining terms can be provided upon request) as they also highlight the role of the high-level condition $\lambda_{\max}(P_sP_s') = O(1)$ for $s=1,2$.
		
		The quartic term in $4\sum_{i=1}^n \left( \sum_{\ell \neq i} C_{i\ell} y_\ell  \right)^2 \tilde \sigma_i^2$ is $\sum_{i \neq \ell \neq  m \neq k }^n\omega_{i\ell mk} \varepsilon_i \varepsilon_\ell \varepsilon_m \varepsilon_k$ where
		\begin{align}
			\omega_{i\ell mk} = \sum_{j=1}^n C_{ji} C_{j\ell} M_{jm,1} M_{jk,2}
			\quad \text{and} \quad
			M_{i\ell,s} = 
			\begin{cases}
			1, & \text{if } i=\ell, \\
			-P_{i\ell,s}, & \text{if } i \neq \ell.
			\end{cases}
		\end{align}
		Letting $\odot$ denote Hadamard (element-wise) product and $M_s = I_n - P_s$, we have
		\begin{align}
			\sum_{i \neq \ell \neq  m \neq k }^n\omega_{i\ell mk}^2 
			&\le \sum_{i , \ell ,  m , k }^n\omega_{i\ell mk}^2 = \sum_{j,j'} (C^2)_{jj'}^2 (M_1 M_1')_{jj'}(M_2 M_2')_{jj'} \\
			& = \text{trace}\left( (C^2 \odot C^2) (M_1M_1' \odot M_2M_2') \right) \\
			&\le \lambda_{\max}\left( M_1M_1' \odot M_2M_2' \right) \text{trace}\left( C^2 \odot C^2 \right) 
			=O\left( \text{trace}\left( C^4\right) \right) = o\left( \mathbb{V}[\hat \theta]^2  \right)
		\end{align}
		where $\lambda_{\max}\left( M_1M_1' \odot M_2M_2' \right) = O(1)$ follows from $\lambda_{\max}(P_sP_s') = O(1)$ and we established the last equality in the proof of \th\ref{thm3}. The quartic term involved in $2\sum_{i=1}^n \sum_{\ell \neq i} \tilde C_{i\ell} \widehat{\sigma_i^2 \sigma_\ell^2}$ has variability of the same order as $\sum_{i \neq \ell \neq  m \neq k }^n\omega_{i\ell mk} \varepsilon_i \varepsilon_\ell \varepsilon_m \varepsilon_k$ where
		\begin{align}
			\omega_{i\ell mk} = \tilde C_{i\ell} M_{im,1} M_{lk,1} + \sum_{j=1}^n \tilde C_{ij} M_{im,1} M_{jk,1} M_{j\ell,2}.
		\end{align}
		Letting $\tilde C = (\tilde C_{i\ell})_{i,\ell}$, we find that
		\begin{align}
			\sum_{i \neq \ell \neq  m \neq k }^n\omega_{i\ell mk}^2  
			&\le 2\sum_{i , \ell}^n \tilde C_{i\ell}^2 (M_1M_1')_{ii}  (M_2M_2')_{\ell \ell}  +2\sum_{j,j'} \sum_{i}^n \tilde C_{ij} \tilde C_{ij'} (M_1M_1')_{ii} (M_1M_1')_{jj'} (M_2M_2')_{jj'}  \\
			&=O\left( \sum_{i , \ell}^n \tilde C_{i\ell}^2 + \text{trace}\left( (\tilde C^2 \odot M_1M_1') (M_1M_1' \odot M_2M_2') \right)  \right) \\
			&=O\left( \text{trace}\left( \tilde C^2\right)  \right).
		\end{align}
		We have $\tilde C = C \odot C + 2(C \odot P_1)'(C \odot P_2) + 2(C \odot P_2)'(C \odot P_1)$, from which we obtain that
		\begin{align}
			\text{trace}( \tilde C^2) = O\left( \left(\max_{i,\ell} C_{i\ell}^2 + \lambda_{\max}(C^2)\right) \text{trace}(C^2) \right)  = o\left(\mathbb{V}[\hat \theta]^2 \right)
		\end{align}
		where we established the last equality in the proof of \th\ref{thm3}.
	\end{proof}

	Section \ref{sec:varq} proposed standard errors for the case of $q>0$, but left a few details to the appendix since the definitions were completely analogous to the previous lemma. Those definitions are $\tilde C_{i\ell q} = C_{i\ell q}^2 + 2\sum_{m =1 }^n C_{mi q} C_{m\ell q} (P_{mi,1} P_{m\ell,2} + P_{mi,2} P_{m\ell,1})$ where $C_{i\ell q}$ was introduced in the proof of \th\ref{thm4} and is of the form $C_{i \ell q} = B_{i \ell q} - 2\inverse M_{i \ell} \left( M_{ii}\inverse B_{ii q}   +  M_{\ell \ell}\inverse B_{\ell \ell q}  \right)$ for $B_{i \ell q} =B_{i\ell} - \sum_{s=1}^q \lambda_s w_{is} w_{\ell s}$.

	Furthermore, the proposed standard error estimator relies on
	\begin{align}
		\widetilde{\sigma_i^2 \sigma_\ell^2} &= 
		\begin{cases}
		\hat \sigma_{i,-\ell}^2 \cdot \hat \sigma_{\ell,-i}^2, & \text{if }  P_{ik,-\ell}P_{\ell k,-i}=0 \text{ for all } k, \\
		\tilde \sigma_{i}^2 \cdot \hat \sigma_{\ell,-i}^2,  & \text{else if } P_{i\ell,1} + P_{i\ell,2}=0, \\
		\hat \sigma_{i,-\ell}^2 \cdot	\tilde \sigma_{\ell}^2,  & \text{else if } P_{\ell i,1} + P_{\ell i,2}=0, \\
		\hat \sigma_{i,-\ell}^2 \cdot (y_\ell - \bar y)^2 \cdot 1_{\{\tilde C_{i\ell q} < 0\}}, & \text{otherwise.}
		\end{cases}
	\end{align}
	\begin{lemma}\th\label{lem:var3S}
		For $s =1,2$, suppose that $\widehat{x_i'\beta}_{-i,s}$ satisfies $\sum_{\ell \neq i}^n P_{i\ell,s} x_\ell'\beta=x_i'\beta$, $P_{i\ell,1} P_{i\ell,2} =0$ for all $\ell$, and $\lambda_{\max}(P_sP_s') = O(1)$ where $P_s = (P_{i\ell,s})_{i,\ell}$.
		\begin{enumerate}
			\item If the conditions of \th\ref{thm4} hold and $\abs{\mathcal{B}}= O(1)$, then $\varSigma_q\inverse
			\hat\varSigma_q \xrightarrow{p} I_{q+1}.$
			
			\item If the conditions of \th\ref{thm4} hold, then $\liminf_{n \rightarrow \infty} \Pr\left( \theta \in \hat C_{\alpha,q}^{\theta} \right) \ge 1-\alpha.$
		\end{enumerate}
		
	\end{lemma}

	The following provides a proof of the first claim of this lemma, while we postpone a proof of the second claim to the end of Appendix \ref{app:conf}.	
	
	\begin{proof}
		The statements $\mathbb{V}[\hat {\mathsf{b}}_q]\inverse	\hat{\mathbb{V}}[\hat {\mathsf{b}}_q]\xrightarrow{p} I_{q} $ and $\mathbb{V}[\hat \theta_q]\inverse\hat{\mathbb{V}}[\hat \theta_q] \xrightarrow{p} 1$ follow by applying the arguments in \th\ref{thm2S,lem:var2S}. Thus we focus on the remaining claim that
		\begin{align}
			\delta(v) := \frac{\hat{\mathcal{C}}[v'\hat {\mathsf{b}}_q,\hat \theta_q]-{\mathcal{C}}[v'\hat {\mathsf{b}}_q,\hat \theta_q]}{\mathbb{V}[v'\hat {\mathsf{b}}_q]^{1/2} \mathbb{V}[\hat \theta_q]^{1/2}} \xrightarrow{p} 0
			\quad \text{where} \quad
			\hat{\mathcal{C}}[v'\hat {\mathsf{b}}_q,\hat \theta_q] = 2 \sum_{i=1}^n v' \mathsf{w}_{iq}\left(\sum_{\ell \neq i} C_{i\ell q} y_\ell \right) \tilde \sigma_i^2
		\end{align} 
		for all non-random $v \in \R^{q}$ with $v'v=1$. 
		
		\noindent\textbf{\textit{Unbiasedness of $\hat{\mathcal{C}}[v'\hat {\mathsf{b}}_q,\hat \theta_q]$}} 
		Since $\tilde \sigma_i^2$ is unbiased for $\sigma_i^2$, it follows that
		\begin{align}
			\E\left[\hat{\mathcal{C}}[v'\hat {\mathsf{b}}_q,\hat \theta_q]\right] 
			&= 2 \sum_{i=1}^n v' \mathsf{w}_{iq}\left(\sum_{\ell \neq i} C_{i\ell q} x_\ell'\beta \right) \sigma_i^2 + 2 \sum_{i=1}^n v' \mathsf{w}_{iq}\left(\sum_{\ell \neq i} C_{i\ell q} \E[\varepsilon_\ell \tilde \sigma_i^2 ] \right) 
			= {\mathcal{C}}[v'\hat {\mathsf{b}}_q,\hat \theta_q]
		\end{align}
		as split sampling ensures that $\E[\varepsilon_\ell \tilde \sigma_i^2 ]$ for $\ell \neq i$.
		
		\noindent \textbf{\textit{Variability of $\hat{\mathcal{C}}[v'\hat {\mathsf{b}}_q,\hat \theta_q]$}} Now, $\hat{\mathcal{C}}[v'\hat {\mathsf{b}}_q,\hat \theta_q]-{\mathcal{C}}[v'\hat {\mathsf{b}}_q,\hat \theta_q]$ is composed of the following linear, quadratic, and quartic sums:
		{\small
		\begin{align}
			&\sum_{i=1}^n v'\mathsf{w}_{iq} \left[ \left(\varepsilon_i^2-\sigma_i^2\right)\sum_{\ell \neq i} C_{i\ell q} x_\ell'\beta + \sigma_i^2\sum_{\ell \neq i} C_{i\ell q} \varepsilon_\ell + \sum_{\ell \neq i} C_{i\ell q} \sigma_\ell^2 \sum_{k \neq \ell} \left(M_{i\ell,1} M_{ik,2} + M_{i\ell,2} M_{ik,1}\right) \varepsilon_k \right] \\
			&\sum_{i=1}^n v'\mathsf{w}_{iq} \Bigg[ \sum_{\ell \neq i} C_{i\ell q} x_\ell'\beta \sum_{m} \sum_{k \neq m} M_{im,1}   M_{ik,2} \varepsilon_m\varepsilon_k + \sum_{\ell \neq i} C_{i\ell q} \varepsilon_\ell \left(\varepsilon_i^2-\sigma_i^2\right) \\
			&\phantom{\sum_{i=1}^n v'\mathsf{w}_{iq} \Bigg[}+ \sum_{\ell \neq i} C_{i\ell q} \sum_{k \neq \ell} \left(M_{i\ell,1} M_{ik,2} + M_{i\ell,2} M_{ik,1}\right) \varepsilon_k\left(\varepsilon_\ell^2-\sigma_\ell^2\right) \Bigg] \\
			& \sum_{i=1}^n v'\mathsf{w}_{iq} \sum_{\ell \neq i} C_{i\ell q}  \sum_{m \neq \ell} \sum_{k \neq m,\ell} M_{im,1}   M_{ik,2} \varepsilon_\ell\varepsilon_m\varepsilon_k
		\end{align}
		}%
		These seven terms are $o_p(\mathbb{V}[v'\hat {\mathsf{b}}_q]^{1/2} \mathbb{V}[\hat \theta_q]^{1/2})$ by \th\ref{res:L1} as outlined in the following.
		{\small	
		\begin{align}
			&\sum_{i=1}^n (v'\mathsf{w}_{iq})^2 \left( \sum_{\ell \neq i} C_{i\ell q} x_\ell'\beta \right)^2 = O( \max_{i} \mathsf{w}_{iq}'\mathsf{w}_{iq} \mathbb{V}[\hat \theta_q] ) = o(\mathbb{V}[v'\hat {\mathsf{b}}_q] \mathbb{V}[\hat \theta_q] ) \\
			&\sum_{\ell=1}^n \left(\sum_{i=1}^n v'\mathsf{w}_{iq} C_{i\ell q}\right)^2 = O(\lambda_{\max}( C_q^2 ) \mathbb{V}[v'\hat {\mathsf{b}}_q] ) = O( \lambda_{q+1}^2 \mathbb{V}[v'\hat {\mathsf{b}}_q] ) = o(\mathbb{V}[v'\hat {\mathsf{b}}_q] \mathbb{V}[\hat \theta_q] ) \\
			&\sum_{k=1}^n \left(\sum_{i=1}^n v'\mathsf{w}_{iq} \sum_{\ell} C_{i\ell q} M_{i\ell,1} M_{ik,2}  \right)^2 = O( \max_{i} \mathsf{w}_{iq}'\mathsf{w}_{iq} \text{trace}(C_qM_1 \odot C_qM_1) ) = o(\mathbb{V}[v'\hat {\mathsf{b}}_q] \mathbb{V}[\hat \theta_q] ) \\
			&\sum_{m=1}^n \sum_{k=1}^n \left( \sum_{i=1}^n v'\mathsf{w}_{iq} \sum_{\ell \neq i} C_{i\ell q} x_\ell'\beta M_{im,1}   M_{ik,2} \right)^2 = O\left( \sum_{i=1}^n (v'\mathsf{w}_{iq})^2 \left( \sum_{\ell \neq i} C_{i\ell q} x_\ell'\beta \right)^2 \right) \\
			&\sum_{i=1}^n \sum_{\ell \neq i} C_{i\ell q} ^2(v'\mathsf{w}_{iq})^2 = O( \max_{i} \mathsf{w}_{iq}'\mathsf{w}_{iq} \mathbb{V}[\hat \theta_q] ) \\
			&\sum_{k=1}^n \sum_{\ell=1}^n \left( \sum_{i=1}^n v'\mathsf{w}_{iq} C_{i\ell q} M_{i\ell,1} M_{ik,2} \right)^2 = O\left( \mathbb{V}[v'\hat {\mathsf{b}}_q]\lambda_{max}( (C_q \odot M_1)(C_q \odot M_1)' ) \right) = o(\mathbb{V}[v'\hat {\mathsf{b}}_q] \mathbb{V}[\hat \theta_q] ) \\
			&\sum_{\ell=1}^n \sum_{m=1}^n \sum_{k=1}^n \left( \sum_{i=1}^n v'\mathsf{w}_{iq} C_{i\ell q}  M_{im,1}   M_{ik,2} \right)^2 = O\left( \mathbb{V}[v'\hat {\mathsf{b}}_q] \lambda_{max}( C_q^2) \right) \qedhere
		\end{align}
		}%
	\end{proof}
	
	\subsubsection{Conservative Variance Estimation}\label{sec:cons}
	
	The standard error estimators considered in the preceding two lemmas relied on existence of the independent and unbiased estimators $\widehat{x_i'\beta}_{-i,1}$ and $\widehat{x_i'\beta}_{-i,2}$. This part of the appendix creates an adjustment for observations where these estimators do not exist. The adjustment ensures that one can obtain valid inference as stated in the lemma at the end of the subsection.

	For observations where it is not possible to create $\widehat{x_i'\beta}_{-i,1}$ and $\widehat{x_i'\beta}_{-i,2}$, we construct $\widehat{x_i'\beta}_{-i,1}$ to satisfy the requirements in \th\ref{lem:4} and set $P_{i\ell,2}=0$ for all $\ell$ so that $\widehat{x_i'\beta}_{-i,2}=0$. Then we define $\mathcal{Q}_i = 1_{\{\max_{\ell} P_{i\ell,2}^2=0\}}$ as an indicator that $\widehat{x_i'\beta}_{-i,2}$ could not be constructed as an unbiased estimator.

	Based on this we let 
	\begin{align}
		\hat{\mathbb{V}}_2[\hat \theta] &= 4\sum_{i=1}^n \left( \sum_{\ell \neq i} C_{i\ell} y_\ell  \right)^2 \tilde \sigma_{i,2}^2 - 2 \sum_{i=1}^n \sum_{\ell \neq i} \tilde C_{i\ell} \widehat{\sigma_i^2 \sigma_\ell^2}_2
	\end{align}
	where $\tilde \sigma_{i,2}^2 =(1-\mathcal{Q}_i) \tilde \sigma_{i}^2 + \mathcal{Q}_i (y_i - \bar y)^2$ and 
	{\small
	\begin{align}
	\widehat{\sigma_i^2 \sigma_\ell^2}_2 &= 
	\begin{cases}
	\hat \sigma_{i,-\ell}^2 \cdot \hat \sigma_{\ell,-i}^2, & \text{if }  P_{ik,-\ell}P_{\ell k,-i}=0 \text{ for all } k \text{ and } \mathcal{Q}_{i\ell}=\mathcal{Q}_{\ell i}=0 \\
	\tilde \sigma_{i}^2 \cdot \hat \sigma_{\ell,-i}^2,  & \text{else if } P_{i\ell,1} + P_{i\ell,2}=0\text{ and } \mathcal{Q}_i =\mathcal{Q}_{\ell i} =0, \\
	\hat \sigma_{i,-\ell}^2 \cdot	\tilde \sigma_{\ell}^2,  & \text{else if } P_{\ell i,1} + P_{\ell i,2}=0 \text{ and } \mathcal{Q}_\ell =\mathcal{Q}_{i \ell} =0, \\
	\hat \sigma_{i,-\ell}^2 \cdot (y_\ell - \bar y)^2 \cdot 1_{\{\tilde C_{i\ell} < 0\}}, & \text{else if } \mathcal{Q}_{i\ell}=0, \\
	(y_i - \bar y)^2 \cdot \hat \sigma_{\ell,-i}^2 \cdot 1_{\{\tilde C_{i\ell} < 0\}}, & \text{else if } \mathcal{Q}_{\ell i}=0, \\
	(y_i - \bar y)^2 \cdot (y_\ell - \bar y)^2 \cdot 1_{\{\tilde C_{i\ell} < 0\}}, & \text{otherwise}
	\end{cases}
	\end{align}
	}%
	where we let $\mathcal{Q}_{i\ell} = 1_{\{P_{i\ell,1}\neq 0 \neq \mathcal{Q}_i \}}$. The defintion of $\hat{\mathbb{V}}_2[\hat \theta]$ is such that $\hat{\mathbb{V}}_2[\hat \theta]=\hat{\mathbb{V}}[\hat \theta]$ when two independent unbiased estimators of $x_i'\beta$ can be formed for all observations, i.e., when $\mathcal{Q}_i=0$ for all $i$.

	Similarly, we let
	\begin{align}
		\hat\varSigma_{q,2} &= \sum_{i=1}^n \begin{bmatrix}
		\mathsf{w}_{iq} \mathsf{w}_{iq}' \hat \sigma_{i,2}^2 & 2 \mathsf{w}_{iq}\left(\sum_{\ell \neq i} C_{i\ell q} y_\ell \right) \tilde \sigma_{i,2}^2 \\
		2 \mathsf{w}_{iq}'\left(\sum_{\ell \neq i} C_{i\ell q} y_\ell \right) \tilde \sigma_{i,2}^2  & 4 \left(\sum_{\ell \neq i} C_{i\ell q} y_\ell \right)^2 \tilde \sigma_i^2 - 2 \sum_{\ell \neq i} \tilde C_{i\ell q}^2 \widetilde{\sigma_i^2 \sigma_\ell^2}_2
		\end{bmatrix}
	\end{align}
	where $\hat \sigma_{i,2}^2 = (1-\mathcal{Q}_i) \hat \sigma_i^2 + \mathcal{Q}_i(y_i - \bar y)^2$ and $\widetilde{\sigma_i^2 \sigma_\ell^2}_2$ is defined as $\widehat{\sigma_i^2 \sigma_\ell^2}_2$ but using $\tilde C_{i\ell q}$ instead of $\tilde C_{i\ell}$. 
	
	The following lemma shows that these estimators of the asymptotic variance leads to valid inference when coupled with the confidence intervals proposed in Sections \ref{sec:dist} and \ref{sec:variance}.
	
	\begin{lemma}\th\label{lem:conservative}
		Suppose that $\sum_{\ell \neq i}^n P_{i\ell,1} x_\ell'\beta=x_i'\beta$, either $\sum_{\ell \neq i}^n P_{i\ell,2} x_\ell'\beta=x_i'\beta$ or $\max_{\ell} P_{i\ell,2}^2=0$, $P_{i\ell,1} P_{i\ell,2} =0$ for all $\ell$, and $\lambda_{\max}(P_sP_s') = O(1)$ where $P_s = (P_{i\ell,s})_{i,\ell}$. 
		\begin{enumerate}
			\item If the conditions of \th\ref{thm3} hold, then $\liminf_{n \rightarrow \infty}\Pr\left( \theta \in \left[ \hat \theta \pm z_{\alpha} \hat{\mathbb{V}}_2[\hat \theta]^{1/2}  \right] \right) \ge 1-\alpha$.
			
			\item If the conditions of \th\ref{thm4} hold, then $\liminf_{n \rightarrow \infty} \Pr\left( \theta \in C_{\alpha}^{\theta}(\hat \varSigma_{q,2}) \right) \ge 1-\alpha$.

		\end{enumerate}
		
	\end{lemma}

	The following provides a proof of the first claim of this lemma, while we postpone a proof of the second claim to the end of Appendix \ref{app:conf}.	

	\begin{proof}
		As in the proof of \th\ref{lem:varCS} it suffices to show that $\hat{\mathbb{V}}_2[\hat \theta]$ has a positive bias in large samples and that $\hat{\mathbb{V}}_2[\hat \theta] - \E[\hat{\mathbb{V}}_2[\hat \theta]]$ is $o_p({\mathbb{V}}[\hat \theta])$. The second claim involves no new arguments relative to the proof of \th\ref{lem:varCS} and is therefore omitted. Thus we briefly report the positive bias in $\hat{\mathbb{V}}_2[\hat \theta]$.
		
		We have that
		\begin{align}
			\E\left[\hat{\mathbb{V}}_2[\hat \theta]\right] &= {\mathbb{V}}[\hat \theta] + 4\sum_{i : \mathcal{Q}_i=1} \left( \sum_{\ell \neq i} C_{i\ell} x_\ell'\beta  \right)^2 ((x_i-\bar x)'\beta)^2 \\
			&+2 \sum_{(i,\ell) \in \mathcal{B}_1} \tilde C_{i\ell} \sigma_i^2 \left( \sigma_\ell^2 1_{\{\tilde C_{i\ell}>0\}} + ((x_\ell-\bar x)'\beta)^2 1_{\{\tilde C_{i\ell}<0\}} \right)  \\
			&+2 \sum_{(i,\ell) \in \mathcal{B}_2} \tilde C_{i\ell} \sigma_\ell^2 \left( \sigma_i^2 1_{\{\tilde C_{i\ell}>0\}} + ((x_i-\bar x)'\beta)^2 1_{\{\tilde C_{i\ell}<0\}} \right) \\
			&+2 \sum_{(i,\ell) \in \mathcal{B}_3} \tilde C_{i\ell} \left( \sigma_i^2\sigma_\ell^2 1_{\{\tilde C_{i\ell}>0\}} + \left( 2 \sigma_i^2 ((x_\ell-\bar x)'\beta)^2 + ((x_i-\bar x)'\beta (x_\ell-\bar x)'\beta)^2\right)1_{\{\tilde C_{i\ell}<0\}} \right) \\
			&+ O\left(\frac{1}{n} {\mathbb{V}}[\hat \theta]\right)
 		\end{align}
		where the remainder stems from estimation of $\bar y$ and $\mathcal{B}_1, \ \mathcal{B}_2, \ \mathcal{B}_3 $ refers to pairs of observations that fall in each of the three last cases in the definition of $\widehat{\sigma_i^2 \sigma_\ell^2}_2$.
	\end{proof}

	\subsection{Inference with Nuisance Parameters}\label{app:conf}

	This Appendix starts by defining curvature and accompanying critical value for a given curvature as introduced in Section \ref{sec:variance}. Then it derives the closed form representation of $C_{\alpha}^\theta(\tilde \varSigma_1)$ for any variance matrix $\tilde \varSigma_1 \in \R^{2\times 2}$ where for general $q$ we have
	\begin{align}
		C_{\alpha}^\theta(\tilde \varSigma_q) &= \left[ \min_{(\dot b_1,\dots,\dot b_q,\dot \theta_q)'\in \mathsf{E}_\alpha(\tilde \varSigma_q)} \sum_{\ell=1}^{q} \lambda_\ell \dot b_\ell^2 + \dot \theta_q , \max_{(\dot b_1,\dots,\dot b_q,\dot \theta_q)'\in \mathsf{E}_\alpha(\tilde \varSigma_q)} \sum_{\ell=1}^{q} \lambda_\ell \dot b_\ell^2 + \dot \theta_q \right]
		\shortintertext{and}
		\mathsf{E}_\alpha(\tilde \varSigma_q) &= \left\{ (\mathsf{b}_q',\theta_q)' \in \R^{q+1} : \begin{pmatrix}
		\hat{\mathsf{b}}_q - \mathsf{b}_q \\ \hat \theta_q - \theta_q
		\end{pmatrix}'\tilde \varSigma_q\inverse \begin{pmatrix}
		\hat{\mathsf{b}}_q - \mathsf{b}_q \\ \hat \theta_q - \theta_q
		\end{pmatrix} \le z_{\alpha,{\kappa}(\tilde \varSigma_q)}^2\right\}.
	\end{align}
	Finally, it proofs validity of $\hat C_{\alpha,q}^\theta = C_{\alpha}^\theta(\hat \varSigma_q)$ and $C_{\alpha}^\theta(\hat \varSigma_{q,2})$ for any fixed $q$. As for $\hat \varSigma_q$ and $\hat \varSigma_{q,2}$, we partition $\tilde \varSigma_q$ into $\tilde \varSigma_q = \begin{bmatrix}
	\tilde{\mathbb{V}}[\hat {\mathsf{b}}_q] & \tilde{\mathbb{C}}[\hat {\mathsf{b}}_q,\hat \theta_{q}]'\\
	\tilde{\mathbb{C}}[\hat {\mathsf{b}}_q,\hat \theta_{q}] & \tilde{\mathbb{V}}[\hat \theta_{q}]
	\end{bmatrix}$ with $\tilde{\mathbb{V}}[\hat \theta_{q}] \in \R$. In Section \ref{sec:variance}, $\hat C_{\alpha,q}^\theta = C_{\alpha}^\theta(\hat \varSigma_q)$, $\hat{\mathsf{E}}_{\alpha,q} = \mathsf{E}_\alpha(\hat \varSigma_q)$, and $\hat \kappa_q = {\kappa}(\hat \varSigma_q)$.

	\subsubsection{Preliminaries}

	\noindent \textbf{\textit{Critical value function}}
	For a given curvature $\kappa >0$ and confidence level $1-\alpha$, the critical value function $z_{\alpha,\kappa}$ is the $(1-\alpha)$'th quantile of
	\begin{align}\label{eq:rho}
		\rho\left(\chi_{q},\chi_{1},\kappa\right) = \sqrt{ \chi^2_{q} + \left(\chi_{1} + \frac{1}{\kappa}\right)^2  } - \frac{1}{\kappa}
	\end{align}
	where $\chi^2_{q}$ and $\chi^2_{1}$ are independently distributed variates from the $\chi$-squared distribution with $q$ and $1$ degrees of freedom, respectively. $\rho\left(\chi_{q},\chi_{1},\kappa\right)$ is the Euclidean distance from $(\chi_{q},\chi_{1})$ to the circle with center $(0,-\frac{1}{\kappa})$ and radius $\frac{1}{\kappa}$. The critical value function at $\kappa=0$ is the limit of $z_{\alpha,\kappa}$ as $\kappa \downarrow 0$, which is the $(1-\alpha)$'th quantile of a central $\chi^2_{1}$ random variable. See \cite{andrews2016geometric} for additional details.

	\noindent \textbf{\textit{Curvature}}
	The confidence interval $C_{\alpha}^\theta(\tilde \varSigma_q)$ inverts hypotheses of the type $H_0 : \theta = c$ versus $H_1 : \theta \neq c$ based on the value of the test statistic
	\begin{align}
		\min_{\mathsf{b}_q,\theta_q : g(\mathsf{b}_q,\theta_q,c)=0} \begin{pmatrix}
		\hat {\mathsf{b}}_q - \mathsf{b}_q \\ \hat \theta_q - \theta_q
		\end{pmatrix}'\tilde \varSigma_q \inverse \begin{pmatrix}
		\hat {\mathsf{b}}_q - \mathsf{b}_q \\ \hat \theta_q - \theta_q
		\end{pmatrix} 
	\end{align} 
	where $g(\mathsf{b}_q,\theta_q,c) = \sum_{\ell=1}^{q} \lambda_\ell \dot b_{\ell}^2 + \theta_q - c$ and $\mathsf{b}_q = (\dot {b}_{1},\dots,\dot {b}_{q})'$. This testing problem depends on the manifold $S = \{ x = \tilde \varSigma_q^{-1/2}(\mathsf{b}_q,\theta_q)' : g(\mathsf{b}_q,\theta_q,c) = 0  \}$ for which we need an upper bound on the maximal curvature. We derive this upper bound using the parameterization $\mathbf{x}({\dot y}) = \tilde \varSigma_q^{-1/2}(\dot y_1,\dots,\dot y_{q},c - \sum_{\ell=1}^{q} \lambda_\ell \dot y_\ell^2)'$ which maps from $\R^{q}$ to $S$, is a homeomorphism, and has a Jacobian of full rank:	
	\begin{align}
		d\mathbf{x}({\dot y}) = \tilde \varSigma_q^{-1/2} 
		\begin{bmatrix} 		\diag(1,\dots,1) \\		-2\lambda_1 \dot y_1, \dots, -2\lambda_{q} \dot y_{q}		\end{bmatrix}
	\end{align}
	The maximal curvature of $S$, $\kappa(\tilde \varSigma_q)$, is then given as $\kappa(\tilde \varSigma_q) = \max_{\dot y \in \R^{q}} \kappa_{\dot y}$ where
	\begin{align}
		\kappa_{\dot y} &= \sup_{u \in \R^{q}} \frac{\norm{(I-P_{\dot y})V (u \odot u) }}{\norm{d\mathbf{x}({\dot y}) u}^2},
		\qquad V =  \tilde \varSigma_q^{-1/2} 
		\begin{bmatrix} 		0 \\		-2\lambda_1, \dots, -2\lambda_{q}	\end{bmatrix},
	\end{align}
	and $P_{\dot y} = d\mathbf{x}({\dot y})(d\mathbf{x}({\dot y})'d\mathbf{x}({\dot y}))\inverse d\mathbf{x}({\dot y})'$. See \cite{andrews2016geometric} for additional details. %Here, $V$ is the matrix of second derivatives of $\mathbf{x}({\dot y})$.

	\noindent \textbf{\textit{Curvature when $q=1$}}	
	In this case the maximization over $u$ drops out and we have 
	\begin{align}
		\kappa(\tilde \varSigma_1) = \max_{\dot y \in \R} \frac{\sqrt{V'V - \frac{(v'V)^2}{v'v}}}{v'v} \qquad \text{where } v = \tilde \varSigma_1^{-1/2}(1,-2\lambda_1 {\dot y})'
	\end{align}
	and $V = \tilde \varSigma_1^{-1/2}(0,-2\lambda_1)$. The value ${\dot y}^* = - \frac{\tilde \rho \tilde{\mathbb{V}}[\hat \theta_{q}]}{2 \lambda_1 \tilde{\mathbb{V}}[\hat b_1]}$ for $\tilde \rho = \frac{\tilde{\mathbb{C}}[\hat b_1,\hat \theta_{q}]}{\tilde{\mathbb{V}}[\hat b_1]^{1/2}\tilde{\mathbb{V}}[\hat \theta_{q}]^{1/2}}$ is both a minimizer of $v'v$ and $(v'V)^2$, so we obtain that $\kappa(\tilde \varSigma_1) = \frac{2\abs{\lambda_1} \tilde{\mathbb{V}}[\hat b_1] }{\tilde{\mathbb{V}}[\hat \theta_{q}]^{1/2}(1-\tilde \rho^2)^{1/2}}$.

	\noindent \textbf{\textit{Curvature when $q>1$}}	
	In this case we first maximize over $\dot y$ and then over $u$. For a fixed $u$ we want to find
	\begin{align}
		\max_{{\dot y} \in \R^{q}} \frac{\sqrt{V_u'V_u - V_u'P_{\dot y} V_u}}{v_{u,{\dot y}}'v_{u,{\dot y}}},
		\qquad \text{where } V_u = \tilde \varSigma_q^{-1/2} (0,-2\sum_{\ell=1}^{q}\lambda_\ell u_\ell^2), 
		\quad v_{u,{\dot y}} = \tilde \varSigma_q^{-1/2} (u', -2u'D_q{\dot y})',
 	\end{align}
	and $D_q = \text{diag}(\lambda_1,\dots,\lambda_q)$. The value for ${\dot y}$ that solves $-2D_q{\dot y} = \tilde{\mathbb{V}}[\hat {\mathsf{b}}_q]\inverse \tilde{\mathbb{C}}[\hat {\mathsf{b}}_q,\hat \theta_q]$ sets $P_{\dot y} V_u =0$ and minimizes $v_{u,{\dot y}}'v_{u,{\dot y}}$. Thus we obtain
	{\small
	\begin{align}
		\kappa(\tilde \varSigma_q) &= \frac{2 \max_{u \in \R^{q}} \frac{\abs{u'D_qu}}{u'\tilde{\mathbb{V}}[\hat {\mathsf{b}}_q]\inverse u}}{\left(\tilde{\mathbb{V}}[\hat \theta_{q}] - \tilde{\mathbb{C}}[\hat {\mathsf{b}}_q,\hat \theta_q]'\tilde{\mathbb{V}}[\hat {\mathsf{b}}_q]\inverse \tilde{\mathbb{C}}[\hat {\mathsf{b}}_q,\hat \theta_q] \right)^{1/2}}   
		=\frac{2\abs{\dot{ \dot \lambda}_1(\tilde{\mathbb{V}}[\hat {\mathsf{b}}_q]^{1/2} D_q \tilde{\mathbb{V}}[\hat {\mathsf{b}}_q]^{1/2})}}{\left(\tilde{\mathbb{V}}[\hat \theta_{q}] - \tilde{\mathbb{C}}[\hat {\mathsf{b}}_q,\hat \theta_q]'\tilde{\mathbb{V}}[\hat {\mathsf{b}}_q]\inverse \tilde{\mathbb{C}}[\hat {\mathsf{b}}_q,\hat \theta_q] \right)^{1/2}}
	\end{align}
	}%
	where $\dot{ \dot \lambda}_1(\cdot)$ is the eigenvalue of largest magnitude. This formula simplifies to the one derived above when $q=1$.

	\subsubsection{Closed Form Representation of $C_{\alpha}^\theta(\tilde \varSigma_1)$}
	An implicit representation of $C_{\alpha}^\theta(\tilde \varSigma_1)$ is
	\begin{align}
		C_{\alpha}^\theta(\tilde \varSigma_1) &= \left[ \lambda_1 b_{1,-}^2 + \theta_{1,-}  , \lambda_1 b_{1,+}^2 + \theta_{1,+}\right]
	\end{align}
	where $b_{1,\pm}$ and $\theta_{1,\pm}$ are solutions to
	{\small
	\begin{align}
	b_{1,\pm} &= \hat b_{1} \pm z_{\alpha,\kappa(\tilde \varSigma_1)} \left( \tilde{\mathbb{V}}[\hat b_{1}](1-\tilde a(b_{1,\pm}))\right)^{1/2}, \label{eq:pm1}
	\\
	\theta_{1,\pm} &=\hat \theta_1 - \tilde \rho \frac{\tilde{\mathbb{V}}[\hat \theta_1]^{1/2}}{\tilde{\mathbb{V}}[\hat b_{1}]^{1/2}}(\hat b_{1} - b_{1,\pm}) \pm z_{\alpha,\kappa(\tilde \varSigma_1)} \left(\tilde{\mathbb{V}}[\hat \theta_1] (1-\tilde \rho^2) \tilde a(b_{1,\pm}) \right)^{1/2} \label{eq:pm2}
	\end{align}
	}%
	for $\tilde a(\dot b_1)= \left( 1 + \left(  \frac{\text{sgn}(\lambda_1)\kappa(\tilde \varSigma_1) \dot b_{1}}{\tilde{\mathbb{V}}[\hat b_{1}]^{1/2}} + \frac{\tilde \rho}{\sqrt{1-\tilde \rho^2}} \right)^2 \right)\inverse$.
	
	This construction is fairly intuitive. When $\hat \rho =0$, the interval has endpoints that combine 
	\begin{align}
	\lambda_1\left(\hat b_{1} \pm z_{\alpha,\kappa(\tilde \varSigma_1)} \left( \tilde {\mathbb{V}}[\hat b_{1}](1-\tilde a(b_{1,\pm}))\right)^{1/2} \right)^2
	\quad \text{and} \quad
	\hat \theta_q \pm z_{\alpha,\kappa(\tilde \varSigma_1)}\left( \tilde {\mathbb{V}}[\hat \theta_q]  a(b_{1,\pm}) \right)^{1/2}
	\end{align}
	where $ a(\dot b_1)$ estimates the fraction of $\mathbb{V}[\hat \theta]$ that stems from $\hat \theta_1$ when $\E[\hat b_1] = \dot b_1$. When $\hat \rho$ is non-zero, $C_{\alpha}^\theta(\tilde \varSigma_1)$ involves an additional rotation of $(\hat b_{1},\hat \theta_1)'$. This representation of $C_{\alpha}^\theta(\tilde \varSigma_1)$ is however not unique as \eqref{eq:pm1},\eqref{eq:pm2} can have multiple solutions. Thus we derive the representation above together with an additional side condition that ensures uniqueness and represents $b_{1,\pm} $ and $\theta_{1,\pm}$ as solutions to a fourth order polynomial.

	\noindent \textbf{\textit{Derivation}}	
	The upper end of $C_{\alpha}^\theta(\tilde \varSigma_1)$ is found by noting that maximization over a linear function in $\theta_1$ implies that the constraint must bind at the maximum. Thus we can reformulate the bivariate problem as a univariate problem
	{\small
	\begin{align}
		\max_{(\dot b_1,\dot \theta_1)\in \mathsf{E}_\alpha(\tilde \varSigma_1)} \lambda_1  \dot b_1^2 +  \dot \theta_1 
		=\max_{ \dot b_1} \lambda_1 \dot b_1^2 + \hat \theta_1 - \tilde \rho \tfrac{\tilde{\mathbb{V}}[\hat \theta_1]^{1/2}}{\tilde{\mathbb{V}}[\hat b_1]^{1/2}}(\hat b_1 - \dot b_{1}) + \sqrt{\tilde{\mathbb{V}}[\hat \theta_1] (1-\tilde \rho^2) \left(z_{\alpha,\kappa(\tilde \varSigma_1)}^2 - \tfrac{(\hat b_1 - \dot b_1)^2}{\tilde{\mathbb{V}}[\hat b_1]} \right)}
	\end{align}
	}%
	where we are implicitly enforcing the constraint on $\dot b_1$ that the term under the square-root is non-negative. Thus we will find a global maximum in $\dot b_1$ and note that it satisfies this constraint. The first order condition for a maximum is
	\begin{align}
		2 \lambda_1 \dot b_{1} + \tilde \rho \tfrac{\tilde{\mathbb{V}}[\hat \theta_1]^{1/2}}{\tilde{\mathbb{V}}[\hat b_1]^{1/2}} + \tfrac{\hat b_1 - \dot b_{1}}{\tilde{\mathbb{V}}[\hat b_1]} \sqrt{\tfrac{\tilde{\mathbb{V}}[\hat \theta_1] (1-\hat \rho^2) }{z_{\alpha,\kappa(\tilde \varSigma_1)}^2 - \frac{(\hat b_1 - \dot b_{1})^2}{\tilde{\mathbb{V}}[\hat b_1]} }} = 0
	\end{align}
	which after a rearrangement and squaring of both sides yields $\frac{(\hat b_1 - \dot b_{1})^2}{\tilde{\mathbb{V}}[\hat b_1]} = (1- a(\dot b)) z_{\alpha,\kappa(\tilde \varSigma_1)}^2$. This in turn leads to the representation of $b_{1,\pm}$ given in \eqref{eq:pm1}. All solutions to this equation satisfies the implicit non-negativity constraint since any solution $\dot b$ satisfies
	\begin{align}
		z_{\alpha,\kappa(\tilde \varSigma_1)}^2 - \frac{(\hat b_1 - \dot b_1)^2}{\tilde{\mathbb{V}}[\hat b_1]} =  a(\dot b_1) z_{\alpha,\kappa(\tilde \varSigma_1)}^2 > 0.
	\end{align}
	A slightly different arrangement of the first order condition reveals the equivalent quartic condition
	\begin{align}\label{eq:zeroes}
		\tfrac{(\hat b_1 - \dot b_1)^2}{\tilde{\mathbb{V}}[\hat b_1]}\left( 1 + \left( \tfrac{\text{sgn}(\lambda_1) \kappa(\tilde \varSigma_1) \dot b_{1}}{\tilde{\mathbb{V}}[\hat b_1]^{1/2}} + \tfrac{\tilde \rho}{\sqrt{1-\tilde \rho^2}}\right)^2 \right) = \left(\frac{\text{sgn}(\lambda_1)\kappa(\tilde \varSigma_1) \dot b_{1}}{\tilde{\mathbb{V}}[\hat b_1]^{1/2}} + \tfrac{\tilde \rho}{\sqrt{1-\tilde \rho^2}}\right)^2 z_{\alpha,\kappa(\tilde \varSigma_1)}^2
	\end{align}
	which has at most four solutions that are given on closed form. Thus the solution $b_{1,+}$ can be found as the maximizer of
	\begin{align}
		\lambda_1 \dot b_1^2 + \hat \theta_1 - \tilde \rho \tfrac{\tilde{\mathbb{V}}[\hat \theta_1]^{1/2}}{\tilde{\mathbb{V}}[\hat b_1]^{1/2}}(\hat b_1 - \dot b_{1}) + 
		z_{\alpha,\kappa(\tilde \varSigma_1)}\left( \tilde {\mathbb{V}}[\hat \theta_q]  a(\dot b_{1}) \right)^{1/2}
	\end{align}
	among the at most four solutions to \eqref{eq:zeroes}. More importantly, the maximum is the upper end of $C_{\alpha}^\theta(\tilde \varSigma_1)$. Now, for the minimization problem we instead have
	{\small
		\begin{align}
		\min_{(\dot b_1,\dot \theta_1)\in \mathsf{E}_\alpha(\tilde \varSigma_1)} \lambda_1  \dot b_1^2 +  \dot \theta_1 
		=\min_{ \dot b_1} \lambda_1 \dot b_1^2 + \hat \theta_1 - \tilde \rho \tfrac{\tilde{\mathbb{V}}[\hat \theta_1]^{1/2}}{\tilde{\mathbb{V}}[\hat b_1]^{1/2}}(\hat b_1 - \dot b_{1}) - \sqrt{\tilde{\mathbb{V}}[\hat \theta_1] (1-\tilde \rho^2) \left(z_{\alpha,\kappa(\tilde \varSigma_1)}^2 - \tfrac{(\hat b_1 - \dot b_1)^2}{\tilde{\mathbb{V}}[\hat b_1]} \right)}
		\end{align}
	}%
	which when rearranging and squaring the first order condition again leads to \eqref{eq:zeroes} as a necessary condition for a minimum. Thus $b_{1,-}$ and the lower end of $C_{\alpha}^\theta(\tilde \varSigma_1)$ can be found by minimizing
	\begin{align}
		\lambda_1 \dot b_1^2 + \hat \theta_1 - \tilde \rho \tfrac{\tilde{\mathbb{V}}[\hat \theta_1]^{1/2}}{\tilde{\mathbb{V}}[\hat b_1]^{1/2}}(\hat b_1 - \dot b_{1}) - 
		z_{\alpha,\kappa(\tilde \varSigma_1)}\left( \tilde {\mathbb{V}}[\hat \theta_q]  a(\dot b_{1}) \right)^{1/2}
	\end{align}
	over the at most four solutions to \eqref{eq:zeroes}.

%	The same set of calculations yield that the lower end of the confidence set can be found by minimizing
%	\begin{align}
%	f_-(b_{1,-}) = \lambda_1 b_{1,-}^2 + \hat \theta_1 - \hat \rho \frac{\hat{\mathbb{V}}[\hat \theta_1]^{1/2}}{\hat{\mathbb{V}}[\hat b_1]^{1/2}}(\hat b_1 - b_{1,-}) - z_{\hat{\kappa}} \left(\hat{\mathbb{V}}[\hat \theta_1] (1-\hat \rho^2) \hat a_- \right)^{1/2}
%	\end{align}
%	over the real solutions to \eqref{eq:zeroes} that are also solutions to the first and second order conditions
%	\begin{align}
%		b_{1,-} &= \hat b_1 - z_{\hat \kappa} \left( \hat{\mathbb{V}}[\hat b_1](1-\hat a_-)\right)^{1/2}\\ 
%	\text{sgn}(\lambda_1)\hat \kappa z_{\hat \kappa} &\ge - \hat a_-^{-3/2} \left(1-\hat \rho^2\right).
%	\end{align}
	
	\subsubsection{Asymptotic Validity}
	
	\begin{lemma}\th\label{lem:Sinf}
		If $\varSigma_q \inverse \hat \varSigma_q \xrightarrow{p} I_{q+1}$ and the conditions of \th\ref{thm4} hold, then 
		\begin{align}
		\liminf_{n \rightarrow \infty} \Pr\left( \theta \in \hat C_{\alpha,q}^{\theta} \right) \ge 1-\alpha.
		\end{align}
	\end{lemma}

	\begin{proof}
		The following two conditions are the inputs to the proof of Theorem 2 in \cite{andrews2016geometric}, from which it follows that
		\begin{align}
		\liminf_{n \rightarrow \infty} \Pr\left( \theta \in \hat C_{\alpha,q}^{\theta} \right) 
		= \liminf_{n \rightarrow \infty} \Pr\left( \min_{({\mathsf{b}}_q',\theta_q)': g(\mathsf{b}_q,\theta_q,\theta)=0} \begin{pmatrix}
		\hat {\mathsf{b}}_q - {\mathsf{b}}_q \\ \hat \theta_q - \theta_q
		\end{pmatrix}'\hat \varSigma_q \inverse \begin{pmatrix}
		\hat {\mathsf{b}}_q - {\mathsf{b}}_q \\ \hat \theta_q - \theta_q
		\end{pmatrix} \le z_{\alpha,\hat{\kappa}_q}^2 \right) 
		\ge 1-\alpha
		\end{align}
		where $g(\mathsf{b}_q,\theta_q,\theta) = \sum_{\ell=1}^{q} \lambda_\ell \dot b_{\ell}^2 + \theta_q - \theta$ and $\mathsf{b}_q = (\dot {b}_{1},\dots,\dot {b}_{q})'$.
		
		Condition (i) requires that $\hat \varSigma_q^{-1/2}
		\left(
		(\hat {\mathsf{b}}_q',\hat \theta_q)'
		- \E[(\hat {\mathsf{b}}_q',\hat \theta_q)']
		\right)
		\xrightarrow{d} \mathcal{N}\left(0, I_{q+1} \right),$ which follows from \th\ref{thm4} and $\varSigma_q \inverse \hat \varSigma_q \xrightarrow{p} I_{q+1}$. 
		
		Condition (ii) is satisfied if the conditions of Lemma 1 in \cite{andrews2016geometric} are satisfied. To verify this, take the manifold
		\begin{align}
			\tilde S &= \left\{ \dot x \in \R^{q+1} : \tilde g(\dot x) = 0 \right\}
		\shortintertext{for}
			\tilde g(\dot x) &= \dot x' \hat \varSigma_q^{1/2} \begin{bmatrix}			D_q & 0 \\ 0 & 0			\end{bmatrix}\hat \varSigma_q^{1/2} \dot x + 
			(2 \E[\hat{\mathsf{b}}_q]',1)\begin{bmatrix}			D_q & 0 \\ 0 & 1			\end{bmatrix}\hat \varSigma_q^{1/2} \dot x.
		\end{align}
		The curvature of $\tilde S$ is $\hat \kappa$, $\tilde g(0)=0$, and $\tilde g$ is continuously differentiable with a Jacobian of rank 1. These are the conditions of Lemma 1 in \cite{andrews2016geometric}.
	\end{proof}

	\begin{proof}[Proof of the second claims in \th\ref{lem:var3S,lem:conservative}]
		The proof contains two main parts. One part is to establish that the biases of $\hat \varSigma_{q}$ and $\hat \varSigma_{q,2}$ are positive semidefinite in large samples, and that $\E[\hat \varSigma_{q}]\inverse \hat \varSigma_{q} - I_{q+1}$ and $\E[\hat \varSigma_{q,2}]\inverse \hat \varSigma_{q,2} - I_{q+1}$ are $o_p(1)$. These arguments are analogues to those presented in the proofs of \th\ref{lem:var3S,lem:conservative} and are therefore only sketched. The other part is to show that this positive semidefinite asymptotic bias in the variance estimator does not alter the validity of the confidence interval based on it. We only cover $\hat \varSigma_{q,2}$ as that estimator simplifies to $\hat \varSigma_{q}$ when the design is sufficiently well-behaved.
		
		\noindent \textbf{\textit{Validity}}
		First, we let $\mathsf{Q} \mathsf{D} \mathsf{Q}'$ be the spectral decomposition of $\E[\hat \varSigma_{q,2}]^{-1/2} \varSigma_{q} \E[\hat \varSigma_{q,2}]^{-1/2}$. Here, $\mathsf{Q} \mathsf{Q}' = \mathsf{Q}' \mathsf{Q} = I_{q+1}$ and all diagonal entries in the diagonal matrix $ \mathsf{D}$ belongs to $(0,1]$ in large samples. Now,
		{\small
		\begin{align}
			\Pr\left( \theta \in C_{\alpha}^{\theta}(\hat \varSigma_{q,2}) \right) 
			&= \Pr\left( \min_{({\mathsf{b}}_q',\theta_q)':  g(\mathsf{b}_q,\theta_q,\theta)=0} \begin{pmatrix}
			\hat {\mathsf{b}}_q - {\mathsf{b}}_q \\ \hat \theta_q - \theta_q
			\end{pmatrix}'\E[\hat \varSigma_{q,2}] \inverse \begin{pmatrix}
			\hat {\mathsf{b}}_q - {\mathsf{b}}_q \\ \hat \theta_q - \theta_q
			\end{pmatrix} \le z_{\alpha,\kappa(\E[\hat \varSigma_{q,2}])}^2 \right) + o(1)
		\end{align}
		}%
		where the minimum distance statistic above satisfies
		\begin{align}
			\min_{({\mathsf{b}}_q',\theta_q)':  g(\mathsf{b}_q,\theta_q,\theta)=0} \begin{pmatrix}
			\hat {\mathsf{b}}_q - {\mathsf{b}}_q \\ \hat \theta_q - \theta_q
			\end{pmatrix}'\E[\hat \varSigma_{q,2}] \inverse \begin{pmatrix}
			\hat {\mathsf{b}}_q - {\mathsf{b}}_q \\ \hat \theta_q - \theta_q
			\end{pmatrix} 
			= \min_{x \in S_2} (\xi - x)'(\xi - x)
		\end{align}
		where $S_2 = \{x : x = \mathsf{Q}'\E[\hat \varSigma_{q,2}]^{-1/2}\left(
		({\mathsf{b}}_q', \theta_q)'
		- \E[(\hat{\mathsf{b}}_q',\hat \theta_q)']
		\right), g(\mathsf{b}_q,\theta_q,\theta)=0  \}$ and the random vector $\xi = \mathsf{Q}'\E[\hat \varSigma_{q,2}]^{-1/2}\left(
		(\hat {\mathsf{b}}_q', \hat \theta_q)'
		- \E[(\hat{\mathsf{b}}_q',\hat \theta_q)']
		\right)$ has the property that $\mathsf{D}^{-1/2 }\xi \xrightarrow{d} \mathcal{N}(0,I_{q+1})$. From the geometric consideration in \cite{andrews2016geometric} it follows that $S_2$ has curvature of $\kappa(\E[\hat \varSigma_{q,2}])$ since curvature is invariant to rotations. Furthermore,
		\begin{align}
			\min_{x \in S_2} (\xi - x)'(\xi - x) &\le \rho^2\left(\norm{\xi_{-1}},\abs{\xi_1},\kappa(\E[\hat \varSigma_{q,2}])\right) \\
			&\le  \rho^2\left(\norm{(\mathsf{D}^{-1/2}\xi)_{-1}},\abs{(\mathsf{D}^{-1/2}\xi)_1},\kappa(\E[\hat \varSigma_{q,2}])\right) \\
		\end{align}
		where $\xi = (\xi_1,\xi_{-1}')'$ and $\mathsf{D}^{-1/2}\xi = ((\mathsf{D}^{-1/2}\xi)_1,(\mathsf{D}^{-1/2}\xi)_{-1}')$ and the first inequality follows from the proof of Theorem 1 in \cite{andrews2016geometric}. Thus
		\begin{align}
			\liminf_{n \rightarrow \infty} \Pr\left( \theta \in C_{\alpha}^{\theta}(\hat \varSigma_{q,2}) \right) 
			&= \liminf_{n \rightarrow \infty} \Pr\left(\min_{x \in S_2} (\xi - x)'(\xi - x) \le z_{\alpha,\kappa(\E[\hat \varSigma_{q,2}])}^2 \right)  \\
			&\ge \liminf_{n \rightarrow \infty} \Pr\left( \rho^2\left({\chi_{q}},{\chi_{1}},\kappa(\E[\hat \varSigma_{q,2}])\right) \le z_{\alpha,\kappa(\E[\hat \varSigma_{q,2}])}^2 \right) = 1-\alpha
		\end{align}
		since $(\norm{\xi_{-1}},\abs{\xi_1}) \xrightarrow{d} (\chi_{q},\chi_{1})$.
		
		\noindent \textbf{\textit{Bias and variability in $\hat \varSigma_{q,2}$}} We finish by reporting the positive semidefinite bias in $\hat \varSigma_{q,2}$. We have that
		\begin{align}
		\E\left[\hat \varSigma_{q,2}\right] &= \varSigma_{q} + \sum_{i : \mathcal{Q}_i=1} \sigma_i^2 \begin{pmatrix} \mathsf{w}_{iq} \\ 2\sum_{\ell \neq i} C_{i\ell} x_\ell'\beta	\end{pmatrix}\begin{pmatrix} \mathsf{w}_{iq} \\ 2\sum_{\ell \neq i} C_{i\ell} x_\ell'\beta	\end{pmatrix}' + \begin{bmatrix}
		0 & 0 \\ 0 & \mathcal{B}
		\end{bmatrix}
		+O\left(\frac{1}{n} {\mathbb{V}}[\hat \theta]\right)
		\end{align}
		where
		\begin{align}
			\mathcal{B}
			&=2 \sum_{(i,\ell) \in \mathcal{B}_1} \tilde C_{i\ell q} \sigma_i^2 \left( \sigma_\ell^2 1_{\{\tilde C_{i\ell q}>0\}} + ((x_\ell-\bar x)'\beta)^2 1_{\{\tilde C_{i\ell q}<0\}} \right)  \\
			&+2 \sum_{(i,\ell) \in \mathcal{B}_2} \tilde C_{i\ell q} \sigma_\ell^2 \left( \sigma_i^2 1_{\{\tilde C_{i\ell q}>0\}} + ((x_i-\bar x)'\beta)^2 1_{\{\tilde C_{i\ell q}<0\}} \right) \\
			&+2 \sum_{(i,\ell) \in \mathcal{B}_3} \tilde C_{i\ell q} \left( \sigma_i^2\sigma_\ell^2 1_{\{\tilde C_{i\ell q}>0\}} + \left( 2 \sigma_i^2 ((x_\ell-\bar x)'\beta)^2 + ((x_i-\bar x)'\beta (x_\ell-\bar x)'\beta)^2\right)1_{\{\tilde C_{i\ell q}<0\}} \right) 
		\end{align}
		for  $\mathcal{B}_1, \ \mathcal{B}_2, \ \mathcal{B}_3 $ referring to pairs of observations that fall in each of the three last cases in the definition of $\widetilde{\sigma_i^2 \sigma_\ell^2}_2$.
	\end{proof}
	
		\subsection{Verifying Conditions}\label{app:sufficiency}
	
	\begin{customthm}{1}
		The only non-immediate conclusions are that:
		\begin{align}
		\mathbb{V}[\hat \theta]\inverse \max_{i} (\tilde x_i'\beta)^2 &= O\left(\frac{\max_i (x_i'\beta)^2/n^2}{\min_i \sigma_i^2 \text{trace}(\tilde A^2)}\right) = O\left(\frac{\max_i (x_i'\beta)^2}{r}\right) \\
		\mathbb{V}[\hat \theta]\inverse \max_{i}(\check x_i'\beta)^2 &= O\left( \frac{\max_{i,j} M_{jj}^{-2} \left(P_{jj} - \frac{1}{n}\right)^2 (x_j'\beta)^2 \left( \sum_{\ell=1}^n \abs{M_{i\ell}} \right)^2/n^2}{\min_i \sigma_i^2 \text{trace}(\tilde A^2)} \right) \\
		& = O\left( \frac{\max_{i,j} (x_j'\beta)^2 \left( \sum_{\ell=1}^n \abs{M_{i\ell}} \right)^2}{r} \right).
		\end{align}
	\end{customthm}
	
	\begin{customthm}{2}
		We first derive the representations of $\hat \sigma_\alpha^2$ given in section \ref{sec:examples}.
		When there are no common regressors, the representation in \eqref{eq:AP} follows from $B_{ii} = \frac{1}{n T_{g(i)}}\left(1 - \frac{T_{g(i)}}{n}\right)$ and
		\begin{align}
		\hat \sigma_g^2 = \frac{1}{T_g} \sum_{t=1}^{T_g} y_{gt} \left(y_{gt} - \frac{1}{T_g-1}\sum_{s \neq t} y_{gs} \right) = \frac{1}{T_g} \sum_{i: g(i)=g} \hat \sigma_{i}^2
		\end{align}
		which yields that
		\begin{align}
		\sum_{i=1}^n B_{ii} \hat \sigma_i^2 = \frac{1}{n} \sum_{g=1}^N \left(1 - \frac{T_{g}}{n}\right) \hat \sigma_g^2.
		\end{align}
		With common regressors, it follows from the formula for block inversion of matrices that
		\begin{align}
		\tilde X' &= A S_{xx}\inverse \begin{bmatrix} D' \\ X' \end{bmatrix} = \frac{1}{n} \begin{bmatrix} \left(D' - \bar d \mathbf{1}_n'\right)\left(I - X \left(X'(I-P_D)X'\right)\inverse X'(I-P_D) \right) \\ 0\end{bmatrix} \\
		& = \frac{1}{n} \begin{bmatrix} D' - \bar d \mathbf{1}_n' - \hat \varGamma' X'(I-P_D) \\ 0\end{bmatrix} 
		\end{align}
		where $D = (d_1,\dots,d_n)'$, $X = (x_{g(1)t(1)},\dots,x_{g(n)t(n)})'$, $P_D = DS_{dd}\inverse D'$, $\mathbf{1}_n = (1,\dots,1)'$, and $S_{dd} = D'D$. Thus it follows that
		\begin{align}
		\tilde x_i = \frac{1}{n}\begin{pmatrix}
		d_i - \bar d - \hat \varGamma'(x_{g(i)t(i)} - \bar x_{g(i)}) \\ 0
		\end{pmatrix}.
		\end{align}
		
		The no common regressors claims are immediate. With common regressors we have
		\begin{align}
		P_{i\ell} = T_{g(i)}\inverse \mathbf{1}_{\{g(i)=g(\ell)\}} + n\inverse (x_{g(i)t(i)} - \bar x_{g(i)})'W\inverse (x_{g(\ell)t(\ell)} - \bar x_{g(\ell)}) = T_{g(i)}\inverse \mathbf{1}_{\{i=\ell\}} + O(n\inverse)
		\end{align}
		where $W = \frac{1}{n}\sum_{g=1}^N \sum_{t=1}^T (x_{gt} - \bar x_g)(x_{gt} - \bar x_g)'$ so $P_{ii} \le C <1$ in large samples. The eigenvalues of $\tilde A$ are equal to the eigenvalues of
		\begin{align}
		\frac{1}{n}\left(I_N - n S_{dd}^{-1/2} \bar d \bar d' S_{dd}^{-1/2}  \right)\left( I_N + \frac{1}{n} S_{dd}^{1/2} D'X W\inverse X'D S_{dd}^{-1/2} \right)
		\end{align}
		which in turn satisfies that $\frac{c_1}{n} \le \lambda_\ell \le \frac{c_2}{n}$ for $\ell = 1,\dots,N-1$ and $c_2 \ge c_1 > 0$ not depending on $n$.
		$w_i'w_i = O(P_{ii})$ so \th\ref{thm2} applies when $N$ is fixed and $\min_g T_g \rightarrow \infty$. Finally,
		\begin{align}
		\max_i \mathbb{V}[\hat \theta]\inverse(\tilde x_i'\beta )^2 &= O \left( \frac{\max_{g,t} \alpha_g^2 + \norm{x_{gt}}^2 \frac{1}{n}\sum_{i=1}^n \norm{x_{g(i)t(i)}}^2 \sigma^2_\alpha}{N}\right) \\
		\max_i \mathbb{V}[\hat \theta]\inverse(\check x_i'\beta )^2 &= O \left( \frac{\max_{i,j} (x_j'\beta)^2 \left( \sum_{\ell=1}^n \abs{M_{i\ell}} \right)^2}{N}\right)
		\end{align}
		and $\sum_{\ell =1}^n \abs{M_{i\ell}} = O(1)$ so  \th\ref{thm3} applies when $N \rightarrow \infty$.
		
		We finish this example with a setup where an unbalanced panel leads to a bias and inconsistency in $\hat \theta_{\text{HO}}$. Consider 
		\begin{align}
		y_{{g} t} =\alpha_{{g}} + \varepsilon_{{g} t} && ({g}=1,\dots,N, \ t = 1,\dots,T_{g})
		\end{align}
		where $N$ is even, $(T_g=2,\E[\varepsilon_{g t}^2]=2\sigma^2)$ for $g \le N/2$ and $(T_g=3,\E[\varepsilon_{g t}^2]=\sigma^2)$ for $g > N/2$, and the estimand is,
		\begin{align}
		\theta = \frac{1}{n} \sum_{g=1}^N T_g \alpha_g^2 \qquad \text{where } n=\sum_{g=1}^N T_g = \frac{5N}{2}.
		\end{align}
		Here we have that $\tilde A = I_N/n$ and $\text{trace}(\tilde A^2) = N/n^2 = o(1)$ as $n \rightarrow \infty$ so the leave-out estimator is consistent. Furthermore,
		\begin{align}
		nB_{ii} = P_{ii} =\begin{cases}
		\frac{1}{2}, &\text{if } i \le N, \\  \frac{1}{3}, & \text{otherwise},
		\end{cases}
		&&
		\sigma_{i}^2 = \begin{cases}
		2\sigma^2, &\text{if } i \le N, \\  \sigma^2,  & \text{otherwise},
		\end{cases}
		\end{align}
		so
		\begin{align}
		\E[\tilde \theta] -\theta &= \sum_{i=1}^n B_{ii} \sigma_i^2 = \frac{\sigma^2}{n} \left( N + \frac{N}{2}\right) = \frac{3\sigma^2}{5}, \\
		\E[\hat \theta_\text{HO}] - \theta &= \sigma_{nB_{ii},\sigma_i^2} +S_B \frac{n}{n-N} \sigma_{P_{ii},\sigma_i^2} 
		=\frac{2\sigma^2}{50} + \frac{2}{3} \times \frac{2\sigma^2}{50} = \frac{\sigma^2}{15}.
		\end{align}
	\end{customthm}

	\begin{customthm}{3}
		$\tilde A$ is diagonal with $N$ diagonal entries of $ \frac{1}{n}\frac{T_g}{S_{zz,g}}$, so $\lambda_g = \frac{1}{n}\frac{T_g}{S_{zz,g}}$ for $g=1,\dots,N$. $\text{trace}(\tilde A^2) \le \frac{\lambda_1}{\min_g S_{zz,g}} \frac{1}{n} \sum_{g=1}^N T_g = O(\lambda_1)$. $\max_i w_i'w_i = \max_{g,t} \frac{(z_{gt} - \bar z_g)^2}{S_{zz,g}} = o(1)$ when $\min_g S_{zz,g} \rightarrow \infty$. Furthermore, $\mathbb{V}[\hat \theta]\inverse  = O(\frac{n^2}{N})$, so
		\begin{align}
		\mathbb{V}[\hat \theta]\inverse \max_{i} (\tilde x_i'\beta)^2 &= O \left( \max_{g,t} \frac{z_{gt}^2 \delta_g^2}{N S_{zz,g}} \right) = o(1),
		\end{align}
		and $M_{i\ell}=0$ if $g(i) \neq g(\ell)$ so
		\begin{align}
		\mathbb{V}[\hat \theta]\inverse \max_i (\check x_i'\beta)^2 &= O \left( \max_g \left(\frac{n\sum_{i : g(i)=g } B_{ii}}{\sqrt{N}}\right)^2 \right) = O \left( \max_g \left(\frac{T_{g}}{\sqrt{N} S_{xx,g}}\right)^2 \right) = o(1)
		\end{align}
		both under the condition that $N \rightarrow \infty$ and $\frac{\sqrt{N} S_{xx,1}}{T_{1}}\rightarrow \infty$. Used above:
		\begin{align}
		P_{i\ell} &= T_{g(i)}\inverse \mathbf{1}_{\{g(i) = g(\ell)\}} + \frac{(z_{g(i)t(i)} - \bar z_{g(i)}) (z_{g(i)t(\ell)} - \bar z_{g(i)})}{S_{zz,g(i)}}  \mathbf{1}_{\{g(i) = g(\ell)\}} \\
		B_{ii} &= \frac{1}{n} \frac{z_{g(i)t(i)} - \bar z_{g(i)}}{S_{zz,g(i)}} \frac{T_{g(i)}}{S_{zz,g(i)}}.
		\end{align}
		Finally, 
		\begin{align}
		\max_i \mathsf{w}_{iq}'\mathsf{w}_{iq} &= \max_{t} \frac{(z_{1t} - \bar z_1)^2}{S_{zz,1}} = o(1) \\
		\mathbb{V}[\hat \theta_q]\inverse \max_{i} (\tilde x_{iq}'\beta)^2 &= O \left( \max_{g \ge 2,t} \frac{z_{gt}^2 \delta_g^2}{N S_{zz,g}} \right) = o(1), \\
		\mathbb{V}[\hat \theta_q]\inverse \max_i (\check x_{iq}'\beta)^2 &= O \left( \max_{g \ge 2} \left(\frac{T_{g}}{\sqrt{N} S_{xx,g}}\right)^2 \right) = o(1)
		\end{align}
		under the conditions that $\frac{\sqrt{N}}{T_2} S_{zz,2} \rightarrow \infty$ and $S_{zz,1} \rightarrow \infty$. Thus, \th\ref{thm4} applies when $ \frac{\sqrt{N}}{T_{1}} S_{zz,{1}} = O(1)$.
	\end{customthm}

	\begin{customthm}{4}
			Let $\dot f_{i} = (\mathbf{1}_{\{j({g},t)=0\}},f_i')' =(\mathbf{1}_{\{j({g},t)=0\}},\mathbf{1}_{\{j({g},t)=1\}},\dots,\mathbf{1}_{\{j({g},t)=J\}})'$ and define the following partial design matrices with and without dropping $\psi_0$ from the model:
		\begin{align}
		S_{ff} &= \sum_{i=1}^n f_i f_i', &  S_{\dot f \dot f} &= \sum_{i=1}^n \dot f_i \dot f_i', &
		S_{\Delta f \Delta f} &= \sum_{g=1}^N \Delta f_g \Delta f_g', &  S_{\Delta \dot f \Delta \dot f} &= \sum_{g=1}^N \Delta \dot f_g \Delta \dot f_g',
		\end{align}
		where $\Delta \dot f_g = \dot f_{i(g,2)} -  \dot f_{i(g,1)}$. Letting $\dot D$ be a diagonal matrix that holds the diagonal of $S_{\Delta \dot f \Delta \dot f}$ we have that
		\begin{align}
		E = \dot D S_{\dot f \dot f}\inverse
		\quad \text{and} \quad 
		\mathcal{L} = \dot D^{-1/2} S_{\Delta \dot f \Delta \dot f} \dot D^{-1/2}.
		\end{align}
		$S_{\Delta \dot f \Delta \dot f}$ is rank deficient with $S_{\Delta \dot f \Delta \dot f} \mathbf{1}_{J+1} = 0$ from which it follows that the non-zero eigenvalues of $E^{1/2} \mathcal{L} E^{1/2}$ (which are the non-zero eigenvalues of $S_{\dot f \dot f}\inverse S_{\Delta \dot f \Delta \dot f}$) are also the eigenvalues of $S_{\Delta f \Delta f} (S_{ff}\inverse + \frac{\mathbf{1}_{J} \mathbf{1}_{J}'}{S_{\dot f \dot f,11}})$. Finally, from the Woodbury formula we have that $A_{ff}$ is invertible with
		\begin{align}
		A_{ff} \inverse = n(S_{ff} - n \bar f \bar f')\inverse = n\left( S_{ff}\inverse  + n\frac{S_{ff}\inverse \bar f \bar f' S_{ff}\inverse}{1 - n\bar f' S_{ff}\inverse \bar f}  \right) = n\left( S_{ff}\inverse  + \frac{\mathbf{1}_{J} \mathbf{1}_{J}'}{S_{\dot f \dot f,11}} \right),
		\end{align}
		so
		\begin{align}
		\lambda_\ell = \lambda_\ell(A_{ff} S_{\Delta f \Delta f}\inverse) = \frac{1}{\lambda_{J+1-\ell}(S_{\Delta f \Delta f} A_{ff}\inverse)} = \frac{1}{n \lambda_{J+1-\ell}(E^{1/2} \mathcal{L} E^{1/2})}.
		\end{align}
		With $E_{jj}=1$ for all $j$, we have that
		\begin{align}
		\frac{\lambda_1^2}{\sum_{\ell=1}^J \lambda_\ell^2} = \frac{\dot \lambda_J^{-2}}{\sum_{\ell=1}^J \dot \lambda_\ell^{-2}} \le \frac{4}{(\sqrt{J} \dot \lambda_J)^2}
		\end{align}
		since $\dot \lambda_\ell \le 2$ \citep[][Lemma 1.7]{chung1997spectral}.	An algebraic definition of Cheeger's constant $\mathcal{C}$ is
		\begin{align}
		\mathcal{C} = \min_{X \subseteq \{0,\dots,J\} : \sum_{j \in X} \dot D_{jj}\le \frac{1}{2}\sum_{j=0}^J \dot D_{jj}} \frac{- \sum_{j \in X} \sum_{k \notin X } S_{\Delta \dot f \Delta \dot f,jk}}{ \sum_{j \in X} \dot D_{jj} }
		\end{align}
		and it follows from the Cheeger inequality $\dot \lambda_J \ge 1 - \sqrt{1-\mathcal{C}^2}$ \cite[][Theorem 2.3]{chung1997spectral} that $\sqrt{J} \dot \lambda_J \rightarrow \infty$ if $\sqrt{J} \mathcal{C} \rightarrow \infty$.
		
		For the stochastic block model we consider $J$ odd and order the firms so that the first $(J+1)/2$ firms belongs to the first block, and the remaining firms belong to the second block. We assume that $\Delta \dot f_g$ is generated \emph{i.i.d.} across $g$ according to 
		\begin{align}
		\Delta \dot f = \mathrm{W} (1-\mathrm{D}) + \mathrm{B} \mathrm{D}
		\end{align}
		where $(\mathrm{W},\mathrm{B},\mathrm{D})$ are mutually independent, $P(\mathrm{D}=1) = 1 - P(\mathrm{D}=0) = p_b \le \frac{1}{2}$, $\mathrm{W}$ is uniformly distributed on $\{ v \in \R^{J+1} : v'\mathbf{1}_{J+1}=0, v'v=2, \max_j v_j =1, v'c=0 \}$, and $\mathrm{B}$ is uniformly distributed on $\{ v \in \R^{J+1} : v'\mathbf{1}_{J+1}=0, v'v=2, \max_j v_j =1, (v'c)^2=4 \}$ for $c =  (\mathbf{1}_{(J+1)/2}',-\mathbf{1}_{(J+1)/2}')'$. In this model $E_{jj}=1$ for all $j$. The following lemma characterizes the large sample behavior of $S_{\Delta \dot f \Delta \dot f}$ and $\mathcal{L}$. Based on this lemma it is relatively straightforward (but tedious) to verify the high-level conditions imposed in the paper.
		
		\begin{lemma}
			Suppose that $\frac{\log(J)}{n p_b} + \frac{J \log(J)}{n} \rightarrow 0$ as $n \rightarrow \infty$ and $J \rightarrow \infty$. Then
			\begin{align}
			\norm*{ \underline{\mathcal{L}}^\dagger \tfrac{J+1}{n} S_{\Delta \dot f \Delta \dot f} - I_{J+1} + \tfrac{\mathbf{1}_{J+1} \mathbf{1}_{J+1}'}{J+1} } =o_p\left( 1 \right) 
			\quad \text{and} \quad
			\norm*{\underline{\mathcal{L}}^\dagger \mathcal{L} - I_{J+1} + \tfrac{\mathbf{1}_{J+1} \mathbf{1}_{J+1}'}{J+1} } = o_p\left( 1 \right) 
			\end{align}
			where $\underline{\mathcal{L}} = I_{J+1} - \frac{\mathbf{1}_{J+1} \mathbf{1}_{J+1}'}{J+1}  - (1-2p_b)\frac{cc'}{J+1}$ and $\norm{\cdot}$ returns the largest singular value of its argument. Additionally, $\max_{\ell} \dot{\underline\lambda}_\ell\inverse \abs*{{\dot{\lambda}_\ell - \dot{\underline\lambda}_\ell}} = o_p(1)$
			where $\dot{\underline\lambda}_1 \ge \dots \ge \dot{\underline\lambda}_J$ are the non-zero eigenvalues of $\underline{\mathcal{L}}^\dagger$.
		\end{lemma}
		
		\begin{proof}
			First note that 
			\begin{align}
			\tfrac{J+1}{n} \E[S_{\Delta \dot f \Delta \dot f}] - \underline{\mathcal{L}} = \tfrac{2+2p_b}{J-1}  \left(I_{J+1} - \tfrac{\mathbf{1}_{J+1} \mathbf{1}_{J+1}'}{J+1} - \tfrac{cc'}{J+1} \right) + \tfrac{4p_b}{J-1} \tfrac{cc'}{J+1},
			\end{align}
			and $\underline{\mathcal{L}}^\dagger = I_{J+1} - \frac{\mathbf{1}_{J+1} \mathbf{1}_{J+1}'}{J+1}  - \left(1-\tfrac{1}{2p_b}\right)\frac{cc'}{J+1}$, so
			\begin{align}
			\norm*{ \underline{\mathcal{L}}^\dagger \tfrac{J+1}{n} \E[S_{\Delta \dot f \Delta \dot f}] - I_{J+1} + \tfrac{\mathbf{1}_{J+1} \mathbf{1}_{J+1}'}{J+1} } &= \norm*{ \tfrac{2+2p_b}{J-1}  \left(I_{J+1} - \tfrac{\mathbf{1}_{J+1} \mathbf{1}_{J+1}'}{J+1} - \tfrac{cc'}{J+1} \right) + \tfrac{2}{J-1} \tfrac{cc'}{J+1} } \\&= \tfrac{2+2p_b}{J-1}
			\end{align}
			Therefore, we can instead show that $\norm{S} = o_p(1)$ for the zero mean random matrix
			\begin{align}
			S = (\underline{\mathcal{L}}^\dagger)^{1/2} \tfrac{J+1}{n} \left( S_{\Delta \dot f \Delta \dot f} - \E[S_{\Delta \dot f \Delta \dot f}] \right) (\underline{\mathcal{L}}^\dagger)^{1/2} = \sum_{g=1}^N s_g s_g' - \E[s_g s_g']
			\end{align}
			where $s_g = \sqrt{\frac{J+1}{n}} \Delta \dot f_g - \frac{\sqrt{2p_b} - 1}{\sqrt{2p_b n}} \Delta \dot f_g'c \frac{c}{\sqrt{J+1}}$. Now since
			\begin{align}
			s_g's_g = O\left( \frac{J}{n} + \frac{1}{np_b} \right)
			\quad \text{and} \quad
			\norm*{\sum_{g=1}^N \E[s_g s_g' s_gs_g']} = O\left( \frac{J}{n} + \frac{1}{ n p_b} \right)
			\end{align}
			it follows from \cite[][Corollary 7.1]{oliveira2009concentration} that $\Pr(\norm{S}\ge t) \le 2(J+1) e^{-\frac{t^2 (\frac{J}{n} + \frac{1}{np_b})}{c(8+4t)}}$ for some constant $c$ not depending on $n$. Letting $t \propto \sqrt{\frac{\log(J/\delta_n)}{n p_b} + \frac{J \log(J/\delta_n)}{n}}$ for $\delta_n$ that approaches zero slowly enough that $\frac{\log(J/\delta_n)}{n p_b} + \frac{J \log(J/\delta_n)}{n} \rightarrow 0$ yields the conclusion that $\norm{S}= o_p(1)$.
			
			Since $\mathcal{L} = \dot D^{-1/2} S_{\Delta \dot f \Delta \dot f} \dot D^{-1/2}$ the second conclusion follows from the first if $\norm{\tfrac{J+1}{n} \dot D - I_{J+1}} = o_p(1)$. We have $\tfrac{J+1}{n} \E[\dot D] = I_{J+1}$ and $\tfrac{J+1}{n}\dot D_{jj} = \tfrac{J+1}{n} \sum_{g=1}^N (\Delta \dot f_g'e_j)^2$ where $e_j$ is the $j$-th basis vector in $\R^{J+1}$ and $\Pr((\Delta \dot f_g'e_j)^2 =1) = 1 - \Pr((\Delta \dot f_g'e_j)^2 =0) = \frac{2}{J+1}$. Thus it follows from $\mathbb{V}(\tfrac{J+1}{n} \dot D_{jj}) \le 2 \frac{J+1}{n}$ and standard exponential inequalities that $\norm{\tfrac{J+1}{n} \dot D - I_{J+1}} = \max_j \abs{\tfrac{J+1}{n} \dot D_{jj}-1} = o_p(1)$ since $\frac{J \log(J)}{n} \rightarrow 0$.
			
			Finally, we note that $\norm*{\underline{\mathcal{L}}^\dagger \mathcal{L} - I_{J+1} + \tfrac{\mathbf{1}_{J+1} \mathbf{1}_{J+1}'}{J+1} } \le \epsilon$ implies $$v'\underline{\mathcal{L}}v(1-\epsilon) \le v'{\mathcal{L}}v \le v'\underline{\mathcal{L}}v(1+\epsilon) $$ which together with the Courant-Fischer min-max principle yields $(1-\epsilon) \le \frac{\dot{\lambda}_j}{\dot{\underline\lambda}_j}  \le (1+\epsilon)$.
		\end{proof}
		
		Next, we will verify the high-level conditions of the paper in a model that uses $\frac{n}{J+1}\underline{\mathcal{L}}$ in place of $S_{\Delta \dot f \Delta \dot f}$ and $\frac{1}{n}\underline{\mathcal{L}}^\dagger$ in place of $\tilde A$ and $\frac{n}{J+1} I_{J+1}$ in place of $\dot D$. Using an underscore to denote objects from this model we have
		\begin{align}
		\max_g \underline P_{gg} &= \max_g \tfrac{J+1}{n}\Delta \dot f_g' \underline{\mathcal{L}}^\dagger \Delta \dot f_g = 2\tfrac{J+1}{n} + 2\tfrac{(1-2p_b)}{np_b} = o(1), \\
		\text{trace}(\underline{\tilde A}^2) &= \frac{\text{trace}((\underline{\mathcal{L}}^{\dagger})^2)}{n^2} = \frac{J-1}{n^2} + \frac{1}{4(np_b)^2} = o(1), \\
		\frac{\underline \lambda_1^2}{\sum_{\ell=1}^J \underline \lambda_\ell^2} & = \frac{1}{\underline{\dot \lambda}_J^{2}\text{trace}((\underline{\mathcal{L}}^{\dagger})^2)} = \frac{1}{(J-1)4p_b^2 + 1}
		\end{align}
		which is $o(1)$ if and only if $\sqrt{J} p_b \rightarrow \infty$, and $\frac{\underline \lambda_2^2}{\sum_{\ell=1}^J \underline \lambda_\ell^2} \le \frac{1}{J}$. Furthermore,
		\begin{align}
		\max_g \underline{\mathsf{w}}_{g1}^2 & = \max_g \left(\frac{c'(\underline{\mathcal{L}}^\dagger)^{1/2} \Delta \dot f_g}{\sqrt{n}}\right)^2 
		=  \left(\frac{2}{\sqrt{2p_bn}}\right)^2 = \frac{2}{np_n} = o(1), \\
		\max_g (\underline{\tilde x}_g'\beta)^2 &= \max_g  \left(\frac{1}{n} \psi' \underline{\mathcal{L}}^\dagger \Delta \dot f_g \right)^2 \le \frac{2}{n^2} \left[ \max_g (\Delta \dot f_g'\psi)^2 + \left(1 - \frac{1}{2p_b}\right)^2(\bar \psi_{cl,1} - \bar \psi_{cl,2})^2 \right] \\
		&= O\left( \frac{1}{n^2} + \frac{1}{(np_b)^2} \right)
		\end{align}
		which is $o\left( \mathbb{V}[\hat \theta] \right)$ if $\sqrt{J} p_b \rightarrow \infty$ as $\text{trace}(\underline{\tilde A}^2) = O(\mathbb{V}[\hat \theta])$ and
		\begin{align}
		\max_g (\underline{\tilde x}_{g1}'\beta)^2  = \max_g \left(\frac{1}{n} \psi' \Delta \dot f_g \right)^2 = O\left(\frac{1}{n^2}\right) = o\left( \mathbb{V}[\hat \theta] \right).
		\end{align}
		Finally,
		\begin{align}
		\max_g (\check x_{g}'\beta)^2 &= O\left( \sum_{g=1}^N B_{gg}^2 \right) = O\left( \max_g B_{gg} \text{trace}({\tilde A}) \right)
		\shortintertext{where}
		\max_g \underline B_{gg} &= \max_g \Delta \dot f_g' \frac{J+1}{n^2} (\underline{\mathcal{L}}^\dagger)^2 \Delta \dot f_g = 2\frac{J+1}{n^2} + \frac{1- 4p_b^2}{(np_b)^2} = O\left(\text{trace}(\underline {\tilde A^2})\right)\\
		\text{trace}(\underline{\tilde A}) &= \frac{J-1}{n} + \frac{1}{2p_b n} = o(1)
		\end{align}
		so $\max_g \underline B_{gg} \text{trace}(\underline {\tilde A}) = O(\text{trace}(\underline {\tilde A^2}))o(1).$
		
		Finally, we use the previous lemma to transfer the above results to their relevant sample analogues.
		{\footnotesize
			\begin{align}
			\max_g \abs{P_{gg} - \underline P_{gg}} 
			&= \max_g \abs{\Delta \dot f_g' (S_{\Delta \dot f \Delta \dot f}^\dagger - \tfrac{J+1}{n}\underline {\mathcal{L}}^\dagger)\Delta \dot f_g} \\
			&= \tfrac{J+1}{n} \max_g \abs*{\Delta \dot f_g' (\underline {\mathcal{L}}^\dagger)^{1/2} \left(\underline {\mathcal{L}}^{1/2} \tfrac{n}{J+1}S_{\Delta \dot f \Delta \dot f}^\dagger \underline {\mathcal{L}}^{1/2} - I_{J+1} + \tfrac{\mathbf{1}_{J+1} \mathbf{1}_{J+1}'}{J+1}\right)(\underline {\mathcal{L}}^\dagger)^{1/2}\Delta \dot f_g} \\
			&=O\left(\norm*{ \underline{\mathcal{L}}^\dagger \tfrac{J+1}{n} S_{\Delta \dot f \Delta \dot f} - I_{J+1} + \tfrac{\mathbf{1}_{J+1} \mathbf{1}_{J+1}'}{J+1} }\right) \max_{g} \underline P_{gg} = o\left(\max_{g} \underline P_{gg}\right)\\
			\abs*{\text{trace}({\tilde A}^2-\underline{\tilde A}^2)} &= \abs*{\sum_{\ell=1}^J \frac{1}{n^2 \dot \lambda_\ell^2} - \frac{1}{n^2\underline {\dot\lambda}_\ell^2}} 
			= \text{trace}(\underline{\tilde A}^2) O\left( \max_\ell \abs*{\frac{\dot \lambda_\ell - \underline {\dot\lambda}_\ell}{\underline {\dot\lambda}_\ell}} \right) = o_p\left(\text{trace}(\underline{\tilde A}^2) \right) \\
			\abs*{\frac{\lambda_1^2}{\sum_{\ell=1}^J \lambda_\ell^2} - \frac{\underline \lambda_1^2}{\sum_{\ell=1}^J \underline {\lambda}_\ell^2}} &= \frac{\underline \lambda_1^2}{\sum_{\ell=1}^J \underline \lambda_\ell^2} O\left( \frac{\abs{\dot \lambda_J - \underline{\dot \lambda}_J}}{\underline{\dot\lambda}_J} + \frac{\abs*{\text{trace}(\underline{\tilde A}^2-{\tilde A}^2)}}{\text{trace}(\underline{\tilde A}^2)}  \right) = o_p(1)
			\end{align}
		}
		with a similar argument applying to $\frac{\lambda_2^2}{\sum_{\ell=1}^J \lambda_\ell^2} - \frac{\underline \lambda_2^2}{\sum_{\ell=1}^J \underline {\lambda}_\ell^2}$. Furthermore,
		{\footnotesize
			\begin{align}
			\max_g {\mathsf{w}}_{g1}^2 &=  \max_g \left( \Delta \dot f_g (\tfrac{J+1}{n} \underline{\mathcal{L}}^\dagger)^{1/2}(\underline{\mathcal{L}} \tfrac{n}{J+1}S_{\Delta \dot f \Delta \dot f}^\dagger )^{1/2} q_1 \right)^2 \le \norm{(\underline{\mathcal{L}} \tfrac{n}{J+1}S_{\Delta \dot f \Delta \dot f}^\dagger )^{1/2} } \max_{g} \underline P_{gg} = o_p(1)  \\
			\shortintertext{and $\max_g \abs{({\tilde x}_g'\beta)^2 - (\underline{\tilde x}_g'\beta)^2} =o_p(\text{trace}(\underline {\tilde A^2}))$ since}
			\max_g ({\tilde x}_g'\beta-  \underline{\tilde x}_g'\beta)^2 &= \tfrac{J+1}{n^2} \max_g \left( \Delta \dot f_g' \underline{\mathcal{L}}^\dagger \left( \underline{\mathcal{L}}S_{\Delta \dot f \Delta \dot f} \dot D - I_{J+1} + \tfrac{\mathbf{1}_{J+1} \mathbf{1}_{J+1}'}{J+1} \right)\tfrac{\psi}{\sqrt{J+1}} \right)^2 \\
			&\le \norm*{\underline{\mathcal{L}}S_{\Delta \dot f \Delta \dot f} \dot D - I_{J+1} + \tfrac{\mathbf{1}_{J+1} \mathbf{1}_{J+1}'}{J+1} } \max_g \underline B_{gg} \frac{\norm{\psi}^2}{J+1} \\
			&= o_p(\text{trace}(\underline {\tilde A^2}))
			\end{align}
		}%
		and this also handles $\max_i \abs{({\tilde x}_{g1}'\beta)^2 - (\underline{\tilde x}_{g1}'\beta)^2} =o_p(1)$ as the previous result does not depend on the behavior of $\sqrt{J} p_b$.
		Finally,
		{\footnotesize
			\begin{align}
			\max_g \abs{ B_{gg} - \underline B_{gg}} &= \frac{J+1}{n^2} \max_g  \abs*{ \Delta \dot f_g' \underline{\mathcal{L}}^\dagger \left( \tfrac{n}{J+1} \underline{\mathcal{L}} S_{\Delta \dot f \Delta \dot f}^\dagger  \dot D S_{\Delta \dot f \Delta \dot f}^\dagger \underline{\mathcal{L}}   - I_{J+1} + \tfrac{\mathbf{1}_{J+1} \mathbf{1}_{J+1}'}{J+1} \right)  \underline{\mathcal{L}}^\dagger\Delta \dot f_g  }\\
			&\le \norm*{ \tfrac{n}{J+1} \underline{\mathcal{L}} S_{\Delta \dot f \Delta \dot f}^\dagger \tfrac{J+1}{n} \dot D \tfrac{n}{J+1}S_{\Delta \dot f \Delta \dot f}^\dagger \underline{\mathcal{L}}   - I_{J+1} + \tfrac{\mathbf{1}_{J+1} \mathbf{1}_{J+1}'}{J+1} } \max_g \underline B_{gg} \\
			&= o_p(\max_g \underline B_{gg}) \\
			\abs*{\text{trace}(\underline{\tilde A} - \tilde A)} &= \abs*{\sum_{\ell=1}^J \frac{1}{n \dot \lambda_\ell} - \frac{1}{n\underline {\dot\lambda}_\ell}}  = \text{trace}(\underline{\tilde A}) O\left( \max_\ell \abs*{\frac{\dot \lambda_\ell - \underline {\dot\lambda}_\ell}{\underline {\dot\lambda}_\ell}} \right) = o_p\left(\text{trace}(\underline{\tilde A}) \right) \\
			\end{align}	
		}
		
	\end{customthm}

\end{appendices}

\end{document}